\documentclass[a4paper]{article}
\usepackage{fullpage}
\usepackage{parskip}
\usepackage{amssymb,amsmath}
\usepackage{amsthm}
\usepackage[T1]{fontenc}
\usepackage{libertinus}
\usepackage{microtype}
\usepackage{slantsc} %
\usepackage{graphicx}
\usepackage[utf8]{inputenc}
\usepackage{subcaption}
\usepackage{booktabs}
\usepackage{tabularx}
\usepackage{algorithm,algorithmicx}
\usepackage{algpseudocode, algxpar}
\usepackage{url}
\usepackage[hypertexnames=false]{hyperref}       
\usepackage{cleveref}    
\usepackage{circledtext}
\usepackage{newunicodechar}
\usepackage{float}
\usepackage{longtable}  
\usepackage{enumitem}
\usepackage{caption}

\urldef{\mailsb}\path|changryeol@yonsei.ac.kr|
\urldef{\mailsa}\path|changryeol@kaist.ac.kr|

\title{Implementation of Polynomial NP-Complete Algorithms\\
Based on  the NP Verifier Simulation Framework
}

\author{\small
  Changryeol Lee\thanks{Independent Researcher \mailsa}
  \thanks{The primary research for this work was conducted while the author was affiliated with the Department of Software, Yonsei University.}
}
\date{\small \today}

\newtheorem{theorem}{Theorem}
\newtheorem{lemma}{Lemma}
\newtheorem{sublemma}{Sublemma}

\newtheorem{corollary}{Corollary}
\newtheorem{definition}{Definition}
\newtheorem{remark}{Remark}

\newcommand{\In}{\textbf{In}}
\newcommand{\Out}{\textbf{Out}}
\newcommand{\InOut}{\textbf{In/Out}}
\newcommand{\bigO}{\mathcal{O}}
\algdef{SE}[CLASS]{Class}{EndClass}[1]{$\mathbf{class}\ \text{#1}$}

\algdef{SE}[PYDEF]{PyDef}{EndPyDef}[2]{$\mathbf{def}\ \text{#1(#2)}$:}

\begin{document}

\maketitle
\begin{abstract}
While prior work established a verifier-based polynomial-time framework for NP, explicit deterministic machines for concrete NP-complete problems have remained elusive.
 In this paper, we construct fully specified deterministic certificate-oblivious verifier Turing machines for \textsc{SAT} and \textsc{Subset-Sum} within an improved NP verifier simulation framework.
 A key contribution of this work is the development of a functional implementation that bridges the gap between theoretical proofs and executable software.
 Our improved feasible-graph construction yields a theoretical reduction in the asymptotic polynomial degree, while enhanced edge extension mechanisms significantly improve practical execution speed.
 We show that these machines generate valid witnesses, extending the framework to deterministic \textsc{FNP} computation without increasing complexity. 
 The complete Python implementation behaves in accordance with the predicted polynomial-time bounds, and the source code, along with sample instances, is available in a public online repository.
\end{abstract}

\tableofcontents
\newpage
\section{Introduction}

The P versus NP problem asks whether every decision problem whose solutions can be
verified in polynomial time can also be solved deterministically in polynomial time.
In recent work by Lee~\cite{lee2025PNP}, a graph-theoretic deterministic computation
model based on the notion of a \emph{feasible graph} was introduced, enabling a
polynomial time simulation of nondeterministic verification on deterministic
machines.

Throughout this paper, we adopt this construction, referred to as the
\emph{Polynomial-Time NP Verifier Simulation Framework via Feasible Graph}, or
simply the \emph{NP verifier simulation framework}.
Within this framework, the feasible graph provides the central structural
tool for deterministic verification for partial and extendable verifier computations.

The previous work establishes the existence of such a framework at a formal level.
It specifies the relevant computational objects—such as computation graphs, walks,
and feasibility conditions—and presents algorithms designed to support the proof of
correctness.
However, certain components required to obtain an executable procedure for concrete
NP decision problems are intentionally left abstract.
In particular, no explicit finite state Turing machines are provided for NP
verifiers, and several operational details necessary for direct translation into
executable code—such as concrete data-structure layouts, indexing schemes, and update
procedures—are omitted.
As a consequence, the framework does not, by itself, yield a closed algorithm that
can be directly executed on specific NP problem instances.

The present paper addresses this gap by providing an implementation-oriented
realization of the NP verifier simulation framework.
Our objective is not to alter or strengthen the core proof of the original paper, but
to make its construction explicit, executable, and empirically verifiable on finite
inputs.
To this end, we supplement the original framework with concrete Turing machine
descriptions, explicit data structures, and algorithmic realizations that adhere
closely to the original definitions.

On the theoretical side, we construct explicit, finite state, finite-symbol Turing
machines for representative NP-complete problems, including SAT and
Subset-Sum.
These constructions instantiate the NP verifiers required by the framework and allow
the NP verifier simulation framework to be applied to concrete decision problems.
The constructions preserve the polynomial bounds on structural parameters such as
width and height and enable a refined, implementation-level complexity analysis.

With respect to feasible graph construction, the original paper employs algorithms
chosen primarily to facilitate the proof of feasible-walk preservation.
In the present work, we instead adopt algorithms that follow the formal definition of
feasibility more directly.
This change does not modify the notion of feasibility or affect correctness; rather,
it replaces a proof-oriented construction with a definition-faithful algorithmic
realization that is more intuitive and admits a lower effective complexity.

In addition, we restrict the set of computation edges that are subject to verification
and extension.
Whereas the original construction verifies a conservative superset of admissible
edges, the present work identifies, through auxiliary structural arguments, a smaller
subset of edges whose verification suffices to preserve acceptance behavior.
These arguments are independent of the core P=NP proof and serve solely to control
combinatorial growth in the implementation.
The relationship between the two constructions is monotone: every edge examined in
the present implementation is also examined in the original framework.
Consequently, correctness of the restricted implementation implies correctness of the
original, more conservative construction.

As a result of these choices, the effective polynomial degree of the simulation is
substantially reduced, making execution on small inputs feasible.
For example, under the restricted candidate-set formulation, each computation edge
requires verification against nearly $\bigO(wh)$ relevant boundary conditions, 
representing only a small fraction of the $\bigO(wh^2)$ bound arising from the fully conservative treatment.
Similarly, feasible graph recomputation costs decrease as the graph is pruned, without
affecting correctness.

Beyond decision capability, the simulator realizes a functional computation in the
sense of FNP.
Whenever an input instance is accepted, the algorithm reconstructs an explicit NP
certificate by tracing a verified computation walk via \emph{floor edges}.
This extension requires only minor structural modifications and does not increase the
asymptotic time complexity of the simulation.

The resulting implementation prioritizes faithfulness to the original definitions and
proof structure over performance engineering.
No aggressive heuristic tuning is applied, and even problem instances that are
trivially solvable by conventional algorithms may require substantial runtime within
this framework.
Nevertheless, the implementation suffices to validate the construction on small
inputs, typically of length 50--200 (and up to 940 in limited cases), and to confirm
that observed behavior matches the theoretical expectations.

In summary, this paper demonstrates that the NP verifier simulation framework admits a fully
explicit realization at the level of Turing machines, data structures, and algorithms.
The implementation closes the gap between theoretical constructiveness and
executability, while preserving the logical foundation and conservative guarantees of
the original proof.

Our main contributions are summarized as follows:
\begin{enumerate}
    \item We provide explicit, finite state, finite-symbol Turing machines for
          representative NP-complete problems, enabling the  NP verifier simulation framework
          to be applied to concrete decision instances.
    \item We present definition-faithful algorithms for feasible graph construction
          and edge verification that preserve correctness while admitting lower
          effective complexity.
    \item We introduce a restricted edge-verification scheme, justified by auxiliary
          structural arguments, that reduces combinatorial growth without altering
          acceptance behavior.
    \item We specify concrete data structures and implementation details corresponding
          directly to the formal objects of the framework.
    \item We extend the simulator to output FNP-style certificates via floor-edge
          reconstruction, without increasing asymptotic complexity.
\end{enumerate}

The remainder of the paper is organized as follows.
\Cref{sec:related_works} reviews related research.
\Cref{sec:background} summarizes the computation model and NP verifier simulation framework.
\Cref{sec:TMs} presents the problem-specific Turing machines.
\Cref{sec:improvements_on_feasiblegraph} discusses feasible graph construction.
\Cref{sec:improvements_on_edge_extension} presents the restricted edge-extension
mechanism, including FNP-style extensions.
\Cref{sec:implementation_detail} details the implementation and data structures.
\Cref{sec:experimental_evaluation} reports experimental validation.
\Cref{sec:discussions} discusses limitations and future directions.
\Cref{sec:conclusion} concludes the paper.

To support reproducibility, the Python source code and sample input instances are made available at \url{https://github.com/changryeol-hub/poly-np-sim}.

\section{Related Work}\label{sec:related_works}
The notion of a Turing Machine was introduced by Turing in his seminal 1936
paper on computable numbers, where he formalized the concept of mechanical
computation and demonstrated that elementary arithmetic operations such as
addition and subtraction can be realized by a finite set of symbolic
manipulation rules~\cite{turing1936computable}.  This work established the
foundational model for algorithmic computation and provided the basis for
later formalizations of effective procedures and complexity-theoretic classes.
Importantly, Turing’s construction was not merely existential: it explicitly
described how concrete computations are carried out step by step on a
machine, a feature that continues to influence implementation-oriented
theoretical work.

Building on this computational model, complexity theory later introduced
resource-bounded classes defined in terms of Turing Machines.  Cook
formalized the class NP via polynomial time verifiability and proved that
Boolean satisfiability is NP-complete, thereby establishing SAT as a canonical
representative of NP~\cite{cook1972time}.  Karp’s subsequent reductions showed
that a wide range of combinatorial problems can be reduced to SAT in
polynomial time, reinforcing the central role of NP-completeness in
theoretical computer science~\cite{karp1972reducibility}.

Algorithmic research on SAT has traditionally focused on high-level search and resolution techniques,
 rooted in the Davis--Logemann--Loveland (DLL/DPLL) procedure~\cite{davis1962machine} and its modern clause-learning extensions.
  While these methods are complete and highly effective in practice, they remain exponential-time algorithms in
the worst case and do not provide polynomial time decision procedures for SAT
or general NP-complete problems.  In particular, they operate at a level of
abstraction that is largely decoupled from the underlying machine model used
to define NP itself.

In contrast, Lee~\cite{lee2025PNP} introduced the feasible-graph model and
proved, at an abstract level, that deterministic polynomial time simulation
of NP verifiers can be achieved by verifying valid computation transitions
within a graph-based computation model.  Central to this framework is the
notion of a feasible graph, which preserves valid computation walks
corresponding to legal Turing Machine transitions.  The NP verifier simulation
framework adopts a strictly machine level viewpoint: its input is
an explicit NP verifier Turing Machine, and all subsequent constructions and
correctness guarantees are defined relative to that machine.  Consequently,
instantiating the framework for a concrete NP problem requires an explicit
construction of the corresponding verifier Turing Machine.

Accordingly, the present work returns to the machine level perspective
inherent in Turing’s original formulation.  Rather than designing new
resolution-based search heuristics, we focus on explicitly constructing a
verifier Turing Machine for SAT and analyzing its deterministic simulation
through structured computation graphs.  The goal is not to compete with
state-of-the-art SAT solvers, but to make the verification process itself
concrete, executable, and inspectable at the level of machine transitions.

Furthermore, while the original framework relied on an exhaustive edge extension strategy—incurring significant computational overhead to ensure no potential edges were overlooked—we have substantially improved efficiency by restricting the verification targets through more intuitive direct extensions and structural constraints on candidate edges. In addition, we reduce the dominant polynomial degree of the framework by employing a more definition-faithful feasible graph construction, providing an enhanced polynomial-time bound for the verification process. 

\section{Background and Preliminaries} \label{sec:background}

This section summarizes the computation model and terminology introduced in
Lee~\cite{lee2025PNP}, which we adopt throughout this paper with minimal
modification.
The structure and wording of the definitions closely follow those in
\cite{lee2025PNP} and are reproduced here to fix notation and to make the
present paper self-contained.
Except where explicitly stated, no conceptual changes to the underlying model
are introduced.

\newcommand{\blank}{\epsilon}
\newcommand{\spacedelim}{\_} 
\newcommand{\IPrec}{\mathrm{IPrec}}
\newcommand{\ISucc}{\mathrm{ISucc}}
\newcommand{\Prev}{\mathrm{Prev}}
\newcommand{\Next}{\mathrm{Next}}
\newcommand{\ipred}{\mathrm{ipred}}
\newcommand{\isucc}{\mathrm{isucc}}
\newcommand{\state}{\mathrm{state}}
\newcommand{\tier}{\mathrm{tier}}
\newcommand{\nextIndex}{\mathrm{next\_index}}
\newcommand{\nextState}{\mathrm{next\_state}}
\newcommand{\lastState}{\mathrm{last\_state}}
\newcommand{\lastSymbol}{\mathrm{last\_symbol}}
\newcommand{\NIL}{\texttt{NIL}}
\newcommand{\indexOf}{\mathrm{index}}
\renewcommand{\output}{\mathrm{output}}
\renewcommand{\symbol}{\mathrm{symbol}}
\newcommand{\dir}{\operatorname{dir}}
\newcommand{\qinit}{q_{\mathrm{init}}}
\newcommand{\qacc}{q_{\mathrm{acc}}}
\newcommand{\qrej}{q_{\mathrm{rej}}}
\newcommand{\vdown}[1]{{#1}_{\downarrow}}
\newcommand{\vup}[1]{{#1}_{\uparrow}}
\newcommand{\vbot}[1]{{#1}_{\bot}} 
\newcommand{\Lfixed}{L_{\mathrm{fixed}}}
\newcommand{\init}{\mathrm{init}}
\newcommand{\term}{\mathrm{term}}
\newcommand{\Incoming}{\mathrm{In}}
\newcommand{\Outgoing}{\mathrm{Out}}
\newcommand{\TCase}{\mathrm{TCase}}
\newcommand{\timeOf}{\mathrm{time}}
\newcommand{\nextOf}{\mathrm{next}}
\newcommand{\prevOf}{\mathrm{prev}}

\subsection{Deterministic Turing Machines and Computation Graph} 

\begin{definition}[Deterministic Acceptor Turing Machine]
A deterministic single-tape acceptor Turing machine (TM) is a tuple
\[
M = (Q, \Sigma, \Gamma, \delta, q_{\mathsf{init}}, q_{\mathsf{acc}}, q_{\mathsf{rej}}),
\]
where:
\begin{itemize}
  \item $Q$ is a finite set of states.
  \item $\Sigma$ is the input alphabet, not containing the blank symbol $\blank$.
  \item $\Gamma$ is the tape alphabet, satisfying $\Sigma \subseteq \Gamma$ and $\blank \in \Gamma$.
  \item $q_{\mathsf{init}} \in Q$ is the initial state.
  \item $q_{\mathsf{acc}},\, q_{\mathsf{rej}} \in Q$ are the designated accepting and rejecting halting states, with $q_{\mathsf{acc}} \neq q_{\mathsf{rej}}$.
  \item $\delta : (Q \setminus \{q_{\mathsf{acc}}, q_{\mathsf{rej}}\}) \times \Gamma 
    \to Q \times \Gamma \times \{-1,+1\}$  
    is the transition function.  On reading a symbol $a\in\Gamma$ in state $q$, the machine writes $b\in\Gamma$, moves its head left or right, and enters the next state $q'$.
\end{itemize}
A configuration of $M$ is a triple $(q, L, i_h)$, where $q \in Q$ is the current state, $i_h \in \mathbb{Z}$ is the head position, and $L$ represents the tape symbols for each cell index $i \in \mathbb{Z}$.
\end{definition}
An \textbf{acceptor Turing machine} is a Turing machine that always halts in either an accepting state or a rejecting state. 

\begin{definition}[NP via Acceptor Verification]
A language $\mathcal{L} \subseteq \Sigma^*$ belongs to \textsf{NP} if there exists a deterministic acceptor Turing machine $M$ (the verifier) and polynomials $p(\cdot)$ and $q(\cdot)$ such that for every $X \in \Sigma^*$ and  $Y \in \Sigma^*$ with $|Y| \le q(|X|)$:
\begin{itemize}
    \item $M$ halts on the input string $X\#Y$ within $p(|X| + |Y| + 1)$ steps, where the $+1$ accounts for the delimiter in the concatenated input.
    \item $X \in \mathcal{L} \iff \exists Y$ such that $M$ reaches an \texttt{accept} state on input $X\#Y$.
\end{itemize}
\end{definition}
It is important to emphasize that not every string $Y \in \Sigma^*$ within the length bound $q(|X|)$ is necessarily a valid or well-formatted certificate for $X$. These candidates may consist of arbitrary symbol sequences that do not constitute a correct proof; the verifier $M$ \textbf{is designed to} reject such invalid candidates by the definition. The essential property of \textsf{NP} is the existence of \textit{at least one} string $Y$ (among the exponential space of all possible candidates) that causes $M$ to reach an \texttt{accept} state when $X \in \mathcal{L}$. Throughout this paper, we use the term \textbf{problem instance} to refer to the fixed sequence $X\#$, and we use \textbf{verifier} to refer to the deterministic acceptor $M$ to maintain consistency with the standard computational model.

\begin{definition}[Computation Node]\label{def:computation_node}
Given a deterministic Turing machine $M = (Q, \Sigma, \Gamma, \delta, \qinit, \qacc, \qrej)$, a \textbf{computation node} is defined as a 6-tuple $v = (i, t, q, \sigma, \vdown{q}, \vdown{\sigma})$, representing a local configuration of a tape cell, where:
\begin{itemize}
    \item $i \in \mathbb{Z}$ representing the tape cell index, denoted by $\indexOf(v)$.
    \item $t \in \mathbb{N}_0$ representing the \textbf{tier}, the number of transitions that have occurred at cell $i$ prior to this configuration, denoted by $\tier(v)$.
    \item $q \in Q$ representing the state after the $t$-th transition at cell index $i$, denoted by $\state(v)$.
    \item $\sigma \in \Gamma$ representing the tape symbol after the $t$-th transition at cell index $i$, denoted by $\symbol(v)$.
    \item $\vdown{q} \in Q \cup \{\bot\}$ representing the former state right before the $t$-th transition at cell $i$, denoted by $\lastState(v)$. If $t=0$, $\vdown{q} = \bot$.
    \item $\vdown{\sigma} \in \Gamma \cup \{\bot\}$ representing the former tape symbol right before the $t$-th transition at cell $i$, denoted by $\lastSymbol(v)$. If $t=0$, $\vdown{\sigma} = \bot$.
\end{itemize} 
A node $v$ is called an \textbf{initial node} if $i=0, t=0, q=\qinit, \vdown{q} = \bot$, and $\vdown{\sigma} = \bot$. The set of initial nodes $V_0$ denotes the starting configuration of $M$.
The transition function $\delta$ is uniquely defined for the pair $(\state(v), \symbol(v))$. For each node $v$, we denote the resulting state, written symbol, and head direction as $\nextState(v)$, $\output(v)$, and $\dir(v)$, respectively, satisfying:
\[ \delta(\state(v), \symbol(v)) = (\nextState(v), \output(v), \dir(v)) \]
\end{definition}

\begin{definition}[Transition Case] \label{def:transition_case}
A \textbf{transition case} $T$ (specifically $V_{i,t}^{q,\sigma}$) is defined as the set of all computation nodes $v \in V$ such that $\indexOf(v)=i$, $\tier(v)=t$, $\state(v)=q$, and $\symbol(v)=\sigma$.

\textbf{Deterministic Projections for Cases:} 
Since all nodes $v \in T$ share the same state $q$ and symbol $\sigma$, the deterministic projections are uniquely determined by the case $T$. We denote them as $\nextState(T)$, $\output(T)$, and $\dir(T)$, satisfying:
\[ (\nextState(T), \output(T), \dir(T)) = \delta(q, \sigma) \]
\end{definition}

Throughout this paper, we let $\TCase(v)$ represent the transition case where the computation node $v$ belongs to  for convenience.

\begin{definition}[Computation Graph]
Given a deterministic Turing machine $M$, the \textbf{computation graph} $G = (V, E)$ is a directed graph where $V$ is the set of all possible computation nodes as defined in Definition \ref{def:computation_node}. A directed edge exists from node $u$ to node $v$, denoted by $(u, v) \in E$, only if:
\[ |\,\indexOf(v) - \indexOf(u)\,| = 1 \]
When the context is clear, we refer to such a graph simply as a \emph{computation graph} $G$. When specified, we refer to $G$ as a \textbf{computation graph with a set of initial nodes $V_0$} to denote a computation graph whose initial nodes are only those in $V_0$.
For a computation graph $G = (V, E)$, an $e$-augmented graph for an edge $e = (u, v)$ is defined as $G \cup \{e\} = (V \cup \{u, v\}, E \cup \{e\})$. More generally, for a set of edges $E'$, the $E'$-augmented graph is denoted as $G \cup E' = (V \cup V', E \cup E')$, where $V'$ is the set of all vertices incident to any edge in $E'$. Similarly, we denote by $G - e$ (or $G - E'$) the graph $(V, E \setminus \{e\})$ (or $(V, E \setminus E')$), representing the removal of edges while preserving the underlying vertex set.
\end{definition}

A \emph{computation graph} $G = (V,E)$ represents TM configurations as nodes, where edges correspond to single-step transitions. Key concepts—including \emph{index-precedent}, \emph{index-succedent}, \emph{computation walk}, \emph{footmarks}, and \emph{folding nodes}—are defined within the computation graph, following the conventions established in the original work \cite{lee2025PNP}

\begin{definition}\label{def:width-height-definition}
Given a computation graph $G$, the \textbf{width} of $G$ (denoted by $w$) is defined as the difference between the maximum and minimum indices of the vertices in $G$:
\[ w = \max_{v \in V(G)} \indexOf(v) - \min_{v \in V(G)} \indexOf(v). \]
The \textbf{height} of $G$ (denoted by $h$) is defined as the maximum tier among all vertices in $G$:
\[ h = \max_{v \in V(G)} \tier(v). \]
\end{definition}

\begin{remark}\label{prop:graph_size_bound}
Since edges exist only between nodes $u, v$ such that $|\indexOf(u) - \indexOf(v)| = 1$, the total number of edges $|E(G)|$ is bounded by:
\[ |E(G)| \le w \cdot h^2 \cdot |Q|^2 \cdot |\Gamma|^2, \]
where $Q$ is the set of states and $\Gamma$ is the tape alphabet of the underlying Turing machine.
\end{remark}

\begin{definition}
Given computation graph $G$ and its subgraph $H$, a \textbf{boundary edge} of $H$ in $G$ is defined as any edge $(u, v)$ in $G$ such that $u \in V(H)$ and $v \notin V(H)$.
\end{definition}

\begin{definition}[Edge Index and Direction]\label{def:edge_index_dir} 
For an edge $e = (u, v)$ in a computation graph $G$, the \textbf{index of an edge} is defined as the minimum index of its endpoints: $\indexOf(e) = \min(\indexOf(u), \indexOf(v))$. The \textbf{direction of an edge} represents the head's displacement: $\dir(e) = \indexOf(v) - \indexOf(u)$.
\end{definition}

\begin{definition}[Edge Slice]\label{def:edge_slice}
For each integer $i \in \mathbb{Z}$, the \textbf{edge slice} $E_i$ is the set of all edges in $G$ with index $i$:
\[ E_i = \{ e \in E(G) \mid \indexOf(e) = i \} \]
An edge slice is treated as a formally indexed set $(i, E_i)$. In particular, $E_i$ remains a distinguished, index-specific slice even if $E_i = \emptyset$. 
\end{definition}

Geometrically, an edge slice $E_i$ represents the set of all possible transitions that cross the boundary between tape cell $i$ and cell $i+1$. Specifically, an edge $e=(u,v)$ belongs to $E_i$ if it corresponds to a move from $i$ to $i+1$ (where $\dir(e)=1$) or from $i+1$ to $i$ (where $\dir(e)=-1$). This construction allows us to analyze the computation as a flow of information across discrete spatial boundaries.

\begin{definition}[Index-Predecessor] \label{def:index-predecessor}
Let $W$ be a walk in $G$. The \textbf{index-predecessor} of a node $v$ on $W$, denoted by $\ipred_W(v)$,  is the last node $p$ appearing before $v$ on $W$ such that $\indexOf(p) = \indexOf(v)$. If $v$ is the first node on $W$ with that index, we define $\ipred_W(v) = \bot$.

Similarly, the \textbf{index-predecessor} of an edge $e$ on $W$, denoted by $\ipred_W(e)$, is the last edge $e_p$ appearing before $e$ on $W$ such that $\indexOf(e_p) = \indexOf(e)$. If no such edge exists, $\ipred_W(e) = \bot$.
\end{definition}

\begin{definition}[Index-Successor] \label{def:index-successor}
Let $W$ be a walk in $G$. The \textbf{index-successor} of a node $v$ on $W$, denoted by $\isucc_W(v)$, is the first node $s$ appearing after $v$ on $W$ such that $\indexOf(s) = \indexOf(v)$. If $v$ is the last node on $W$ with that index, we define $\isucc_W(v) = \bot$.

Similarly, the \textbf{index-successor} of an edge $e$ on $W$, denoted by $\isucc_W(e)$, is the first edge $e_s$ appearing after $e$ on $W$ such that $\indexOf(e_s) = \indexOf(e)$. If no such edge exists, $\isucc_W(e) = \bot$.
\end{definition}

\begin{definition}[Computation Walk] \label{def:computation_walk}
A \textbf{computation walk} $W = (v_0, v_1, \dots, v_n)$ is a sequence of computation nodes representing the step-by-step execution of a Turing machine $M$. Formally, for each step $k$, the node $v_k$ represents the 6-tuple configuration of the tape cell pointed to by the head after $k$ transitions. Equivalently, a walk $\mathcal{W}$ is a computation walk if and only if for every node $v_k$, the following consistency conditions are satisfied:

\begin{enumerate}
    \item \textbf{Initial Vertex Condition:} $v_0$ is an \textbf{initial node} (i.e., $\indexOf(v_0) = 0$, $\tier(v_0) = 0$, and $\state(v_0) = q_0$, where $q_0$ is the initial state of the corresponding Turing machine $M$).
    \item \textbf{Tier Consistency:} The tier of $v_k$ reflects its sequential visit count to cell $\indexOf(v_k)$:
    \[ \tier(v_k) = \begin{cases} \tier(\ipred_{W}(v_k)) + 1 & \text{if } \ipred_{W}(v_k) \neq \bot \\ 0 & \text{if } \ipred_{W}(v_k) = \bot \end{cases} \]

    \item \textbf{State/Symbol History (Index-Predecessor Flow):} The node $v_k$ inherits the state and symbol recorded during the cell's previous visit:
	    \[ 
	    \begin{aligned}
	        &(\lastState(v_k), \lastSymbol(v_k)) \\
	      =& 
	      \begin{cases} 
	          (\state(\ipred_{W}(v_k)), \symbol(\ipred_{W}(v_k))) & \text{if } \ipred_{W}(v_k) \neq \bot \\ 
	          (\bot, \bot) & \text{if } \ipred_{W}(v_k) = \bot 
	      \end{cases} 
	    \end{aligned}
	    \]
    \item \textbf{Transition Consistency:} The transition $\delta(\state(v_k), \symbol(v_k)) = (q', \sigma', d)$ at node $v_k$ uniquely determines the configuration of subsequent nodes in the walk:
    \begin{itemize}
        \item \textbf{Head Flow:} The resulting state $q' = \nextState(v_k)$ must match the state of the next node; that is $\nextState(v_k)= \state(v_{k+1})$ for $k < n$.
        \item \textbf{Cell Flow:} The written symbol $\sigma' = \output(v_k)$ must match the symbol read by the cell's next visitor, $\output(v_k) = \symbol(\isucc_{W}(v_k))$ if $\isucc_{W}(v_k) \neq \bot$.
        \item \textbf{Displacement:} The direction $d = \dir(v_k)$ must satisfy the spatial constraint $\dir(v_k) = \indexOf(v_{k+1}) - \indexOf(v_k)$.
    \end{itemize}
\end{enumerate}
A computation walk $W$ is maximal if no $W+e$ computation walk exists for any $e \in E(G) \setminus E(W)$ in a computation grah $G$.
\end{definition}

From the definition, given computation walk $W$, we can see that $\tier(\isucc_W(v))=\tier(v)+1$, and $\tier(v)=\tier(\ipred_W(v))+1$ for a vertex $v$ on a walk $W$,
 which means no vertex can be repeated in walk $W$.
Hence, a computation walk can be regarded as a \textbf{computation path}, and the two notions are considered equivalent.
Furthermore, since no vertex or edge is repeated, a computation walk may also be treated as a graph—specifically, an induced subgraph of the full computation graph—consisting of the same vertex and edge set as the walk itself.
This allows us to apply graph-theoretic notions directly to walks when appropriate.

\begin{definition}[Surface of a Computation Walk]\label{def:surface}
For a computation walk $W$, the \textbf{surface} $\mathcal{S}(W)$ is defined as a sequence of transition cases indexed by cell positions $i \in \mathbb{Z}$:
\[ \mathcal{S}(W) = \{ T_i \mid i \in \mathbb{Z} \} \]
where each $T_i$ represents the status of cell $i$ after the execution of $W$, satisfying:
\begin{itemize}
    \item $T_i = T$ if $T$ is the transition case of the \textbf{last node} $v \in W$ such that $\indexOf(v) = i$.
    \item $T_i = \bot$ if no node $v \in W$ satisfies $\indexOf(v) = i$ (i.e., cell $i$ has not been visited).
\end{itemize}

\end{definition}

\begin{definition}\label{def:footmarks_graph}
Let $\mathcal{W}$ be a set of computation walks. The \textbf{footmarks graph} $F(\mathcal{W})$ is the graph defined by the union of all vertices and edges appearing in the walks of $\mathcal{W}$. Formally:
\[ 
V(F(\mathcal{W})) = \bigcup_{W \in \mathcal{W}} V(W), \quad E(F(\mathcal{W})) = \bigcup_{W \in \mathcal{W}} E(W). 
\]
When no ambiguity arises, we may \textit{simply} refer to $F(\mathcal{W})$ as the \emph{footmarks}. For an edge $e$ (or a set of edges $E$), we refer to $F(\mathcal{W}) + e$ (or $F(\mathcal{W}) + E$) as the \textbf{$e$-augmented} (or \textbf{$E$-augmented}) \textbf{footmarks}.
\end{definition}

\begin{definition}
A computation node $v$ is called a \textbf{folding node} (or \textbf{folding vertex}) if there exist two edges $e$ and $f$ incident to $v$, where $e$ is an incoming edge, $f$ is an outgoing edge, and $\indexOf(e) = \indexOf(f)$.
In this case, both $e$ and $f$ are called \textbf{folding edges} of $v$. \textit{If needed, we refer to $e$ as the incoming folding edge and $f$ as the outgoing folding edge of $v$.}
\end{definition}

\begin{definition}[Index-Precedent] \label{def:index-precedent_nodes}
The \textbf{index-precedent} of a computation node $v \in V(G)$ with $\tier(v) > 0$ is the unique transition case $P$ satisfying:
\[ (\indexOf(P), \tier(P), \state(P), \symbol(P)) = (\indexOf(v), \tier(v) - 1, \lastState(v), \lastSymbol(v)). \]
This unique case is formally denoted by $\IPrec_G(v)$, or simply $\IPrec(v)$ when the graph context is clear. Since $P$ is a transition case of a deterministic Turing machine, all nodes $u \in P$ share a unique transition rule, and consequently, all outgoing edges from $P$ share the same direction $\dir(P)$ determined by $\delta(\state(P), \symbol(P))$.
\end{definition}

\begin{definition}[Previous and Next Edges]
Let $G = (V, E)$ be a directed computation graph, and let $W = (e_1, e_2, \dots, e_k)$ be a computation walk in $G$. For any edge $e \in E$, we define the following:
\begin{itemize}
    \item \textbf{Walk-based previous and next edge:} \\
    If $e = e_i$ for some $1 < i < k$, then we define the previous and next edge within the walk $W$ as:
    \[
    \mathrm{prev}_W(e_i) = e_{i-1}, \quad \mathrm{next}_W(e_i) = e_{i+1}.
    \]
    For boundary cases, $\mathrm{prev}_W(e_1)$ and $\mathrm{next}_W(e_k)$ are undefined. \\
    These are individual edges, not sets, and are written in lowercase.

    \item \textbf{Graph-based previous and next edges:} \\
    Regardless of any walk, we define the set of graph-adjacent edges:
    \[
    \Prev_G(e) = \{ e' \in E \mid \term(e') = \init(e) \}, \quad
    \Next_G(e) = \{ e' \in E \mid \init(e') = \term(e) \}.
    \]
    These denote the sets of incoming and outgoing edges adjacent to $e$ in the graph $G$, and are written in capitalized form to reflect their set-valued nature.
If the context is clear, $_G$ can be omitted.
\end{itemize}
\end{definition}

\begin{definition}\label{def:index-precedent_index-succedent_edges}
Let $(u,v)$ be an edge in a computation graph. Then:

\begin{enumerate}
	\item \textbf{Index-precedent edges:}  
	The \emph{index-precedent} of $e=(u,v)$, denoted by $\vdown{e}$, is the set of edges $\vdown{e} = (v', u')$ such that:
	\begin{itemize}
	    \item $v' \in \IPrec(v)$, and
	    \item there exists an \emph{IPrec-folding node chain} $(u_0, \dots, u_m)$ ($m \ge 0$) such that $u_0 = u$, $u_m = u'$, $u_{i+1} \in \IPrec(u_i)$ for all $0 \le i < m$, and each $u_i$ is a folding node for $0 < i < m$.
	\end{itemize}
	An edge $\vdown{e}$ is referred to as a \emph{direct index-precedent} if $m \le 1$ (equivalently, $u = u'$ or $u' \in \IPrec(u)$); otherwise, $\vdown{e}$ is referred to as an \emph{indirect index-precedent}.

    \item \textbf{Index-succedent edges:}  
	The \emph{index-succedent} of $e=(u,v)$, denoted by $\vup{e}$, is the set of edges $\vup{e} = (v', u')$ such that: 
	\begin{itemize}
		\item $u' \in \ISucc(u)$, and 
		\item there exists an \emph{ISucc-folding node chain} $(v_0, v_1, \dots, v_n)$ ($n \ge 0$) such that $v_0 = v, v_n = v'$, $v_{i+1} \in \ISucc(v_i)$ for all $0 \le i < n$, and each $v_i$ is a folding node for $0 < i < n$.
    \end{itemize}
	An edge $\vup{e}$ is referred to as an \emph{direct index-succedent} if $n \le 1$ (equivalently, $v = v'$ or $v' \in \ISucc(v)$); otherwise $\vup{e}$ is referred to as an \emph{indirect index-succedent}
\end{enumerate}

 For an edge $e$, we write $\IPrec_G(e)$ and $\ISucc_G(e)$ to denote
    its index-precedent and index-succedent sets.
    When the underlying graph $G$ is clear from context, the subscript $G$
    may be omitted.
\end{definition}

The operators $\IPrec(\cdot)$ and $\ISucc(\cdot)$ are used for both
vertices and edges as in the original paper, since the intended meaning is determined by the argument type, and no ambiguity
arises in context.

\begin{definition}
A \textbf{merging edge} is an incoming edge to a computation node $v$ whose in-degree is greater than $1$ and out-degree is not $0$ in a computation graph.
A \textbf{splitting edge} is an outgoing edge from a computation node $v$ whose out-degree is greater than $1$ and in-degree is not $0$ in a computation graph.
\end{definition}

\subsection{Feasible Graph and Computation Walk Verification}
In this subsection, we review the notion of a grid-aligned footmark graph and a feasible graph.

The certificate-oblivious verifier Turing machine is a verifier whose head movements are independent of the certificates. 
A grid-aligned footmark graph is a footmark graph such that all computation walks are coordinate-coincident, as established in the original paper.
\begin{definition}[Certificate-Oblivious Verifier TM]
A verifier Turing machine $M$ is \textbf{certificate-oblivious} with respect to a problem instance $L_{fixed}$ if the tape head trajectory $H_M(L_{fixed}, Y, t)$ is uniquely determined by the certificate length $m = |Y|$ and the time step $t$, 
independent of the content of $Y$, for all $t$ up to the halting time $T(L_{fixed}, Y)$.
Formally, for any two certificates $Y, Y'$ with $|Y| = |Y'|$, the following holds:
$$H_M(L_{fixed}, Y, t) = H_M(L_{fixed}, Y', t) \quad \text{for all } t \le \min(T(L_{fixed}, Y), T(L_{fixed}, Y'))$$
\end{definition}

\begin{definition}[Grid-Aligned Footmarks]
A footmark graph $G = (V, E)$ is said to be \textbf{Grid-Aligned} if for any two computation walks $W_1, W_2 \in \mathcal{W}$ and any step $k$ such that $|W_1| \ge k$ and $|W_2| \ge k$, 
the vertices $v_k \in W_1$ and $u_k \in W_2$ satisfy:
$$(\tier(v_k), \indexOf(v_k)) = (\tier(u_k), \indexOf(u_k))$$
Specifically, for any vertex $v$ reached at the $k$-th step of a computation walk, we define $\timeOf_G(v) = k$, which is well-defined due to the grid-aligned property. 
For any subgraph $G' \subseteq G$, the coordinates and time step of any vertex $u \in G'$ are defined relative to the original context of $G$.
To simplify notation, we omit the subscript $G$ when the context is clear; thus, $\timeOf(u)$ refers to $\timeOf_G(u)$ for any $u \in G' \subseteq G$.
\end{definition}

\begin{remark}[Certificate-Oblivious to Grid-Aligned Correspondence]\label{rem:cotm_to_grid_aligned}
It was shown in \cite{lee2025PNP} that a Certificate-Oblivious Verifier TM $M_{\mathrm{COTM}}$ generates a footmark graph $G$ of $\mathcal{W}$ that satisfies the Grid-Aligned property where $\mathcal{W}$ is a set of all computation walks corresponding to all certificates for a particular problem instance.
\end{remark}

For a computation walk in a computation graph, its \emph{floor edges} and \emph{ceiling edges} characterize the earliest and latest segments of the walk, respectively, as defined in the original paper.
\begin{definition}
Let $W$ be a computation walk in a computation graph.

An edge $e \in W$ is called a \textbf{ceiling edge} if it has no index-successor edge in the walk $W$.

An edge $e \in W$ is called a \textbf{floor edge} if it has no index-predecessor edge in the walk $W$.
\end{definition}

\begin{remark}[Floor Edge Condition]\label{rem:floor_edge_condition}
It was shown in \cite{lee2025PNP} that an edge $e=(u,v)$ is a floor edge if and only if $\tier(v)=0$.
We will use this fact throughout the paper without further proof.
\end{remark}

\begin{definition}[Ceiling-Adjacent Edge]\label{def:ceiling_adjacent}
Let $e = (v, w)$ be an edge in a computation graph $G$, and let $E_f$ be a set of designated final edges. 
An edge $f = (u, v')$ is said to be \emph{weakly ceiling-adjacent} to $e$ toward $E_f$ if there exists a vertex sequence $(v_0, v_1, \dots, v_n)$ such that:
\begin{itemize} 
    \item $v_0 = v'$ and $v_0$ is a non-folding node (the terminal node of $f$);
    \item for all $0 \le i < n$, $v_i \in \IPrec_G(v_{i+1})$;
    \item every vertex $v_i$ for $0 < i < n$ is a folding node;
    \item either ($v_n = v$ and $v$ is a folding node) or ($v_n = w$ and $e = (v, w) \in E_f$).
\end{itemize}
If, in addition to being weakly ceiling-adjacent, there exists a path from $f$ to $e$ in $G$ such that no edge in the path—other than $f$ itself—shares the same index as $f$, then $f$ is said to be \emph{ceiling-adjacent} to $e$ toward $E_f$.
 This includes cases where $e$ and $f$ are directly adjacent.
When the context is clear, we simply refer to $f$ as being ceiling-adjacent to $e$.
\end{definition} 
If $e \in E_f$, there may exist two distinct ceiling-adjacent edges $f$ satisfying
$\indexOf(f) = \indexOf(e) \pm 1$.
For non-final edges (i.e., $e \notin E_f$), any ceiling-adjacent edge $f$ satisfies
$\indexOf(f) = \indexOf(e) - \dir(e)$.

While floor edge is globally well-defined within the graph structure, independent of any local walk or surface assignment, ceiling edge is solely dependent on the computation walk it belongs to, since the same edge can be ceiling edge or not depends on the walk.
Thus, we require a graph-theoretic notion of ceiling edge that holds globally as well, rather than being dependent on a particular walk.
To formulate this, we introduce a family of edges that play a symmetric role to floor edges in a computation graph, not in a walk. 

\begin{definition}[Cover and Ex-Cover Edges] \label{def:cover_edges}
Let $G$ be a computation graph and $E_f$ be the set of designated final edges of the computation walks. 

Let $C = (c_0, c_1, \dots, c_k)$ be a finite sequence of edges such that $c_0 = e$ and $c_k \in E_f$. Depending on the connectivity and adjacency properties of this chain, we categorize $e$ as follows:

\begin{enumerate}
    \item $e$ is a \textbf{cover edge} if for every $j = 0, \dots, k-1$, the edge $c_j$ is \emph{ceiling-adjacent} to $c_{j+1}$. The set of all cover edges is denoted by $\widehat{C}$.
    
    \item $e$ is an \textbf{ex-cover edge} if the following broader conditions are satisfied:
    \begin{itemize}
        \item for every $j = 0, \dots, k-1$, $c_j$ is \textbf{weakly ceiling-adjacent} to $c_{j+1}$;
        \item for every $j = 0, \dots, k-1$, there exists a path in $G$ from $c_j$ to $c_k$.
    \end{itemize}
    The set of all such edges is denoted by $\widehat{C}_{ex}$.
\end{enumerate}
We refer to the sequence $C$ as a \emph{cover edge chain} or an \emph{ex-cover edge chain}, respectively.
\end{definition}

\begin{remark}[Ceiling Edges as Cover Edges]
Following the results in \cite{lee2025PNP}, every ceiling edge is categorized as a cover edge. We adopt this fact throughout our analysis. Note that since the set of cover edges is a subset of ex-cover edges by definition, every ceiling edge is also an ex-cover edge.
\end{remark}

With the notions of floor and cover edges in place, we now define the following concept of a step-pendant edge, which includes vertical pendant behavior within the graph.

\begin{definition}[Ex-Pendant Edge]\label{def:ex_pendant_edge}
Let $e \in E(G)$ be an edge with index $i$ in the computation graph $G$.  
We say that $e$ is \emph{left-pendant} if there exists no edge in $E$ adjacent to $e$ with index $i - 1$, and the node of $e$ with index $i$ is not a folding node.  
We say that $e$ is \emph{right-pendant} if there exists no edge adjacent to $e$ with index $i + 1$, and the node of $e$ with index $i+1$ is not a folding node.  
If $e$ is either left-pendant or right-pendant, we say that $e$ is \emph{ex-pendant}, and if $e$ is both left-pendant and right-pendant, we say that $e$ is \emph{both-pendant}.
\end{definition}

\begin{definition}[Step-Pendant Edge] \label{def:step_pendant_edge}
Let $G$ be a computation graph with initial nodes $V_0$, and $E_f \subseteq E(G)$ be a set of designated final edges. Let $E_{\mathrm{init}}$ denote the set of all outgoing edges from $V_0$, and $\widehat{C}_{ex}$ be the set of ex-cover edges toward $E_f$ as defined in \cref{def:cover_edges}.

An edge $e = (u,v) \in E(G)$ is said to be \textbf{step-pendant with respect to $E_f$} if it satisfies any of the following conditions:
\begin{enumerate}
    \item \textbf{Horizontally Step-Pendant:} $e$ is an ex-pendant edge such that ($e \notin E_{\mathrm{init}} \cup E_f$)  or ($e$ is both-pendant and $e \in (E_{\mathrm{init}} \cup E_f) \setminus (E_{\mathrm{init}} \cap E_f)$);
    \item \textbf{Vertically Step-Pendant:} 
    \begin{itemize}
        \item $\IPrec_G(e) = \emptyset$ and $\mathrm{tier}(v) > 0$;
        \item $\ISucc_G(e) = \emptyset$ and $e \notin \widehat{C}_{ex}$, implying that $e$ is not a cover edge toward $E_f$.
    \end{itemize}
\end{enumerate}
Henceforth, we denote an edge satisfying these conditions as an \textbf{$E_f$-step-pendant edge}.
\end{definition}

The following definition incorporates vertical adjacency relationships to ensure consistency with index-based transitions:
\begin{definition}[Step-Adjacency] \label{def:step_adjacency}
An edge $e \in E(G)$ is said to be \textbf{step-adjacent} to an edge $f$ if and only if $e$ is adjacent to $f$ in the standard graph-theoretic sense, or $e \in \ISucc_G(f)$ or $e \in \IPrec_G(f)$.
\end{definition}

To define the key notion of a feasible graph as a residual component obtained from another, we first introduce the concept of a step-extended component based on the above notion of step-pendant edges.

\begin{definition}[Step-Extended Component]
\label{def:step_extended_component}
Given a computation graph $G$ and a set of \textbf{base edges} $E_R \subseteq E(G)$, the \textbf{step-extended component} of $E_R$ is the subgraph $C_E \subseteq G$ defined recursively as follows:

Let $C_E^{(0)}$ be the subgraph of $G$ consisting of the edges in $E_R$ and their incident vertices. For each $i \ge 0$, define $C_E^{(i+1)}$ by adding to $C_E^{(i)}$ all edges $e \in E(G) \setminus E(C_E^{(i)})$ (and their incident vertices) such that:
\begin{enumerate}
    \item[(i)] $e$ is step-adjacent to some edge in $E(C_E^{(i)})$, and
    \item[(ii)] $e$ is a step-pendant edge in the subgraph $G - E(C_E^{(i)})$.
\end{enumerate}

The process terminates at step $k$ when no such edges can be added. The resulting subgraph $C_E^{(k)}$ is called the \textbf{maximal step-extended component} of $E_R$ in $G$, denoted by $\mathsf{MSEC}_G(E_R)$.
\end{definition}

Now, we define the feasible graph.

\begin{definition}[Feasible Graph] \label{def:feasible_graph}
Given a computation graph $G$ with initial nodes $V_0$ and a set of designated final edges $E_f \subseteq E(G)$, let $E_{\mathrm{init}}$ denote the set of all outgoing edges from $V_0$. The set of \textbf{base edges} $E_R$ is defined as:
\[
E_R := \{ e \in E(G) \mid e \text{ is $E_f$-step-pendant}\}.
\]

The \textbf{feasible graph} $G_f$, denoted by $\mathsf{Feasible}(G)$, is obtained from $G$ by removing the edges and internal vertices of the maximal step-extended component $\mathsf{MSEC}_G(E_R)$, and any remaining isolated vertices $I(G)$:
\[
G_f := (G - E(\mathsf{MSEC}_G(E_R))) - I(G).
\]
\end{definition}

\begin{definition}[Feasible Walk and Feasible Edge] \label{def:feasible_walk_edge}]
Let $G$ be a feasible graph with respect to a set of designated final edges $E_f$. We define a feasible walk and a feasible edge in $G$ as follows:

\begin{itemize}
    \item A \textbf{feasible walk} is a computation walk in $G$ that contains at least one edge from $E_f$.
    \item A \textbf{feasible edge} is any edge that belongs to at least one \textit{feasible walk}.
\end{itemize}

For precision, we use the phrase \textbf{with respect to $E_f$} to denote this classification. When $E_f = \{e_f\}$ is a singleton, we refer to these as walks \textbf{to $e_f$}. For consistency, the term \textbf{feasible walk} and \textbf{feasible edge} is applied not only to the computed feasible graph but also to the original computation graph from which the feasible graph is derived.
\end{definition}

Besides the feasible graph related terms, we define the types of walks and edges for the verfication purpose.

\begin{definition}[Computing-Redundant and Computing-Futile Edges]\label{def:computing-redundant-futile-edge}
Let $G$ be a subgraph of the $e_t$-augmented footmarks $G_U$ for a verification target edge $e_t$.
\begin{itemize}
    \item \textbf{Computing-targeted walk:} A valid computation walk in $G$ that reaches a \textit{verification target edge} in $e_t$.
    \item \textbf{Computing-futile walk:} A maximal computation walk in $G$ that does not reach any edge in $e_t$.
    \item \textbf{Nested computing-futile walk:} A walk $W$ in $G$ is a \textit{nested computing-futile walk} if it is a subwalk of another computing-futile or computing-targeted walk, and there exists some edge $f \in E(G_U) \setminus E(G)$ such that the $f$-augmented walk $W + f$ forms a valid computation walk in $G_U$. 
    \item \textbf{Computing-effective edge:} An edge $e \in E(G)$ that belongs to some computing-targeted walk in $G$ to $e_t$.
    \item \textbf{Computing-essential edge:} An edge $e \in E(G)$ that is strictly contained within every valid computing-targeted walk in $G$ to $e_t$.
    \item \textbf{Computing-redundant edge:} A computing-effective edge $e \in E(G)$ for some computing-targeted walk in $G$ to $e_t$, such that there exists another computing-targeted walk to $e_t$ in the feasible graph of $G - e$.
    \item \textbf{Computing-futile edge:} An edge $e \in E(G)$ that does not belong to any computing-targeted walk in $G$ to $e_t$. Equivalently, removing $e$ from $G$ does not destroy the existence of any computing-targeted walk to $e_t$.
\end{itemize}
\end{definition}

\begin{definition}[Extendable and Extended Computing-Futile Edge]\label{def:extended_futile_edge}
Let $\mathcal{W}$ be a set of computation walks, and let $G_U$ denote $e_t$-augmented footmarks graph for $\mathcal{W}$ with target verification edge $e_t$.  
Let $G \subseteq G_U$ be the current feasible graph, and let $W \subset W' \in \mathcal{W}$ be a computing-futile walk with respect to $e_t$.  

\begin{itemize}
    \item An edge $e$ of $E(W') \setminus E(G)$ is called an \emph{extendable computing-futile edge} if $W$ is not a nested computing-futile walk, and the sequence $W+e$ (the concatenation of $W$ and $e$) is the valid computation walk.  
   
    \item An \emph{extended computing-futile edge} is an extendable futile edge (from $E_o$) that is included in the augmented graph $G + E_o$ where $E_o$ is the set of all extendable edges in $G$.
\end{itemize}
\end{definition}

\subsection{Overview of Polynomial-Time NP Verifier Simulation Framework}\label{sec:np-framework-summary}
This subsection provides a high-level, implementation-oriented summary of the
Polynomial-Time NP Verifier Simulation Framework introduced in \cite{lee2025PNP}.
The presentation is intentionally schematic: several proof-oriented details
and auxiliary constructions from the original work are omitted or simplified,
and the description here should not be read as a formal restatement of the
original algorithms.
Its sole purpose is to clarify the control flow and algorithmic roles of the
main components used in the implementation developed in this paper.

\begin{enumerate}
    \item \textbf{Computation Walks and Footmarks.}
    Given a certificate-oblivious verifier Turing Machine $M$ and an input instance,
    valid executions of $M$ are represented as \emph{computation walks},
    each corresponding to a concrete sequence of TM transitions.
    The union of all such walks forms the grid-aligned \emph{footmarks graph},
    which represents all consistent computation fragments
    discovered so far.

    \item \textbf{Candidate Edge Extension through Verification.}
    Starting from the initial footmarks, the algorithm attempts to extend computation walks by introducing new candidate edges that represent additional transition fragments. 
    These extensions are not immediately accepted; instead, they are treated as candidates subject to rigorous verification. 
    A walk verification procedure is applied to each candidate edge, leveraging the feasible graph to ensure global consistency. 
    If there exists at least one computation walk consistent with both the current footmarks and the candidate edge, the edge is validated and incorporated into the footmarks graph. 
    Otherwise, edges that cannot belong to any such valid walk are identified as futile and discarded.

    \item \textbf{Walk Verification and Pruning.}
    The verification mechanism is an iterative pruning process designed to eliminate computing-futile or computing-redundant edges while preserving valid computation walks directed toward the verification target. The identification of such edges is itself an iterative procedure:
    \begin{itemize}
    \item First, a computation walk is selected within the current graph. If this walk contains the target edge, the verification is complete.
    \item Otherwise, the algorithm simulates the removal of the first splitting edge in the walk to observe whether the graph undergoes a \textbf{total collapse}. A collapse indicates that the removed edge was an essential edge required for all valid computation walks.
    \item Following a collapse, the algorithm identifies a specific computing-futile or computing-redundant edge: the first edge in another computation walk that does not belong to the walk containing the identified essential edge, as observed in the graph state immediately preceding the collapse.
    \end{itemize}

    \item \textbf{Feasible Graph as a Verification Tool.}
    To determine whether a candidate edge can belong to at least one valid computation walk directed toward the verification target, 
    the framework constructs and maintains a \textbf{feasible graph}. This graph is not intended to represent the computation itself; 
    rather, it serves as a rigorous trimming and filtering structure. 
	Its primary function is to exclude edges that cannot participate in any locally consistent computation walk while strictly preserving all edges belonging to valid computation walks that reach the designated final edge. 
	Notably, the graph collapses to an empty set if the essential edge preceding the second-splitting edge is missing.

    \item \textbf{Iteration and Termination.}
    This process is repeated until either an accepting computation walk
    is identified or no further candidate extensions remain.
    In the former case the input is accepted; in the latter case it is rejected.
\end{enumerate}

This algorithmic skeleton clarifies the separation between
the semantic object of computation (the footmarks graph)
and the verification mechanism using the feasible graph,
and serves as the basis for the refinements and implementations
presented in this paper.

\begin{algorithm}[!ht]
\caption{Verification-Based Edge Extension (schematic)}
\label{alg:edge-extension}
\begin{algorithmic}[1]
\Require
Verifier TM $M$, tape string $L$ for NP problem, and certification length $m$.
\Ensure
\textbf{Yes} if an accepting computation exists; otherwise \textbf{No}

\State $G \gets$ the compution graph with Initial edges $E_0$.
\State $E_{\mathrm{cand}} \gets$ Initial boundary edge \Comment{Outgoing edges from the nodes of $E_0$}

\While{\True}
    \State $E_{\mathrm{valid}} \gets \emptyset$

    \ForAll{$e = (u,v) \in E_{\mathrm{cand}}$}
        \If{\Call{VerifyEdgeExtension}{$G, e$} = \True}
            \If{$v$ is an accepting computation node}
                \State \Return \textbf{Yes}
            \EndIf
            \State $E_{\mathrm{valid}} \gets E_{\mathrm{valid}} \cup \{e\}$
        \EndIf
    \EndFor

    \If{$E_{\mathrm{valid}} = \emptyset$}
        \State \Return \textbf{No} \Comment{No further extension possible}
    \EndIf

    \State $G \gets G + E_{\mathrm{valid}}$
    \State Update boundary edges $E_{\mathrm{cand}}$ of $G$ with $M$, $L$, $m$.
\EndWhile
\end{algorithmic}
\end{algorithm}

\begin{algorithm}[!ht]
\caption{VerifyEdgeExtension$(G, e_t)$}
\label{alg:verify-edge}
\begin{algorithmic}[1]
\Require Footmarks graph $G$, candidate edge $e_t$
\Ensure \True~ if $e_t$ is extendable; otherwise \False
\Function{VerifyEdgeExtension}{$G, e_t$}
\State $G' \gets$ Compute Feasible Graph of $G+e_t$ with the final edge $e_t$

\If{$G' = \emptyset$}
    \State \Return \False
\EndIf

\While{$G' \neq \emptyset$}
    \State $H \gets$ the copy of $G'$; $H' \gets \emptyset$ \Comment{$H'$: For storing the state right before collapse}
    \State $W \gets$ Empty List 
    \While{$H \neq \emptyset$} \Comment{Escape when total collapse occurs} 
        \State Construct a maximal computation walk $W$ in $H$  
	    \If{$e_t \in W$} [$W$ is a computing-targeted walk]
	        \State \Return \True
	    \EndIf 
	    \State Let $e$ the first splitting edge of $W$ if exists otherwise the final edge of it.
            \State $H' \gets$ the copy of $H$
            \State $H \gets$ Compute Feasible Graph of $G'-e$ with the final edge $e_t$
    \EndWhile
    \State $H \gets$  Compute Feasible Graph of $H'-e$ with preserving all the computing-futile walks.
    \State Identify the diverged edge $e'$  from $W$  in $H$
    \If{no such edge exists}
        \State \Return \False
    \EndIf

    \State $G' \gets$ Compute Feasible Graph of $G'-e'$ with the final edge $e_t$
\EndWhile
\State \Return \False
\EndFunction
\end{algorithmic}
\end{algorithm} 

We emphasize that the role of walk verification is solely to determine
whether an edge can participate in at least one valid computation walk.
Acceptance is not a property of verification, but is determined directly
by the target computation node of an edge, since each computation node
encodes the current state of the verifier Turing Machine.
The procedure \textproc{ComputeFeasibleGraph()} does not restrict the graph
to edges belonging exclusively to feasible walks. Instead, it removes
edges that cannot participate in any feasible computation walk, while
preserving all feasible walks that exist in the input graph.
At each iteration, the algorithm attempts to follow an arbitrary maximal
computation walk in the current graph. Feasibility is determined only
afterward through verification; no feasibility assumption is made during
walk construction.

Although feasible walks do not explicitly appear as outputs at this level,
they constitute the semantic units preserved by the feasible graph.

\paragraph{Summary of Time Complexity Bounds (Original Framework).}
We summarize the asymptotic time complexity bounds of the original framework~\cite{lee2025PNP} using consistent notation. 
The overall time complexity of the framework is bounded by $\bigO(p(n)^{20})$ for a certificate-oblivious verifier Turing machine $M$ with input size $n$ where $M$ halts at most $p(n)$ steps.
A more precise analysis can be conducted by considering the structural parameters of the computation graph. 
Let $G$ be the NP computation graph of $M$ with width $w$ and height $h$, and let $T_f$ denote the running time of \textproc{ComputeFeasibleGraph()}.

\begin{itemize}
    \item \textbf{Feasible graph construction.}
    \[
        T_f = \bigO\bigl(w^{2} h^{4} (h\log h + \log w)\bigr),
    \]
    which removes most of edges that cannot belong to any feasible
    computation walk while preserving all feasible walks.

    \item \textbf{Verification per candidate edge.}
    Determining whether a single edge participates in some feasible walk
    (via \textproc{VerifyExistenceOfWalk()}) requires
    \[
        \bigO\bigl(w^{2} h^{4} T_f\bigr).
    \]

    \item \textbf{Extension of verified candidate edges.}
    In \textproc{ExtendByVerifiableEdges()}, at most $\mathcal{O}(wh^2)$ candidate edges are examined for verification. 
    In each iteration, at least one edge is extended, and the process terminates within at most $\mathcal{O}(wh^2)$ iterations.
    Since each edge is extended at most once, and each edge is examined at most once per extension step, the total cost of this phase is
    \[
        \bigO(w^2h^4 \cdot w^2 h^4 T_f) = \bigO(w^4 h^8 T_f),
    \]
    
    \item \textbf{Total simulation cost.}
    The algorithm
    \textproc{SimulateVerifierForAllCertificates()}
    performs at most $\bigO(wh^2)$ edge extensions, yielding
    \[
        \bigO\bigl(w^4 h^8 T_f\bigr) = \bigO\bigl(w^6 h^{12} (h \log h + \log w)\bigr).
    \]
\end{itemize}

Since $M$ halts within $p(n)$ steps, then $w = \bigO(p(n))$ and $h = \bigO(p(n))$. 
Consequently, the total running time is bounded by$$ \bigO\bigl(p(n)^{19} \log p(n)\bigr), $$
confirming that the framework operates in polynomial time relative to the verifier's input size. 
For an NP problem instance of size $n'$, the complexity is $\bigO(p'(n') \log p'(n'))$, where $p'$ is a polynomial function of the verifier's complexity with respect to the problem size.

In this subsection, we have summarized the Polynomial-Time NP Verifier Simulation Framework via Feasible Graphs as introduced in the original paper, 
providing a sufficient foundation for the implementation and refinements presented in the subsequent chapters.

\newcommand{\circled}[1]{\text{\textcircled{#1}}}
\section{Explicit Turing Machines for NP-Complete Problems} \label{sec:TMs}

This section presents three explicit deterministic certificate-oblivious Turing machines constructed within the NP verifier simulation framework. 
Theoretically, as introduced in the original framework, a certificate-oblivious Turing machine can be constructed mechanically from an ordinary Turing machine, 
as detailed in Appendix \cref{subsec:sat_tm_input_dependent}. 
However, this approach incurs quadratic overhead; thus, we construct certificate-oblivious machines directly, 
achieving the same time complexity as their ordinary counterparts.

The machines are intentionally ordered to emphasize conceptual clarity and
generality, rather than increasing formal restriction.

We begin with an input-dependent construction whose state set is allowed to
grow with the input size.
Although this machine is not a fixed-state Turing Machine in the classical
sense, it provides the most transparent illustration of how verifier logic is
realized within the NP verifier simulation framework.

We then present a fixed-state construction for SAT with Boolean certificates.
This machine adheres strictly to the classical definition of a Turing Machine,
with constant alphabet and state set, and operates within an
$\bigO(n^{2})$ time bound with $w = h = \bigO(n)$.
While more intricate, this construction demonstrates that the framework admits
fully uniform and polynomially bounded realizations.

Finally, we present a Subset-Sum machine whose certificate alphabet is not
restricted to $\{T,F\}$.
This construction illustrates that the framework applies uniformly to NP
verifiers with arbitrary finite certificate alphabets, underscoring the
generality of the approach.

In the first three subsections, we introduce certificate-oblivious vertifier Turing machines designed to handle \textbf{well-formed inputs} only. 
In these constructions, the certificate string is restricted to a specific subset of the input alphabet, reflecting practical verification scenarios. 
In the subsequent subsections, we extend these models to handle \textbf{all input symbols} of the Turing machines through an input sanitization process,
 demonstrating the robustness and formal correctness of the framework even under unconstrained input conditions.

Before describing the concrete transition systems of the machines,
we fix the tape layout and symbol encoding conventions that are shared
by all constructions in this section.
These conventions are implementation-level choices and are not part
of the abstract framework defined in the preliminaries.

Throughout this paper, the input is placed starting at tape cell $0$ and occupies cells from $0$ through $n-1$.  
The head initially points to cell $0$, all other cells contain $\blank$ except the input cells,  
and the tape is assumed to extend infinitely in both directions.  

We adopt a digit-based tape representation for numerical fields using the tape input alphabet
\[
\Sigma \subseteq \{0,1,\dots,9, \, @,\, -, \, T,\, F,\, \&,\,
  \#,\, \spacedelim, \, ; \},
\]
where each symbol has a fixed operational interpretation.
In particular, `\_' denotes a space delimiter used to distinguish it from $\blank$. 
In addition, the machine uses auxiliary working symbols
\texttt{?}, \texttt{!}, $\sim$, \texttt{\$}, and the circled digits
\circled{0} \dots \circled{9}.
The circled digits act as reversible digit markers and are always
restored to their original digit values,
whereas the remaining auxiliary symbols are non-reversible working
markers used for intermediate overwriting during computation.

All transformations performed by the transition table preserve this encoding.

\begin{remark}
The encoding conventions fixed above are used uniformly across all
machines constructed in this section.
They are chosen to minimize transition complexity and to allow direct
implementation of verifier logic without auxiliary reductions.
We also treat the transition function as partial: any undefined transition
(or transition to a non-instantiated parameterized state) is interpreted
as an immediate rejecting transition.
\end{remark}
\subsection{SAT Turing Machine with Input-Dependent States} \label{subsec:sat_tm_input_dependent}

In this subsection, we present a deterministic certificate-oblivious verifer Turing machine 
for SAT whose control states are constructed, prior to execution, in a manner
that depends on the input, in particular on the number of variables.

Unlike the fixed-state constructions to be introduced in the following subsection,
 this machine does not require explicit arithmetic operations on the tape—such as subtraction, borrow propagation,
 or digit-wise decrement—as it encodes variable indices directly into the control states,
 thereby facilitating more intuitive comprehension.

The purpose of this construction is to demonstrate that the NP verifier
simulation framework does not inherently rely on a fixed finite set of control
states.
While the original proof assumes a Turing Machine with a fixed finite state
space, the simulation framework itself remains valid as long as the total number
of states is polynomially bounded in the size of the input.
In particular, the framework accommodates machines whose control states are
generated from the input, provided that this generation is effective and
polynomially bounded.

\paragraph{Input Format}
The input to $M_{\mathrm{SAT}}$ consists of a CNF formula followed
by an explicit Boolean assignment counter.
An example input has the form
\[
1\_2\_3\&4\_5\_6\&7\_8\_9\&-9\_10\_1\&-2\_6\_1\&3\_5\_1\&-4\_2\_10\#TFTFFFTTTT 
\]
Here each clause is encoded as three integer literals separated by a single
space symbol `\_'.
The symbol `$\&$' acts as a clause delimiter, and the symbol `\#' separates
the clause region from the assignment region.
Negative literals are indicated by a preceding minus sign `\texttt{-}'.
We emphasize that the space symbol `\texttt{\_}' is an explicit delimiter
between literals and variable indices, and is \emph{not} the blank symbol of
the Turing machine tape.

The suffix following `\#' encodes a Boolean assignment as a sequence of
truth values over the alphabet $\{\texttt{T},\texttt{F}\}$.
The assignment region, which serves as a certificate region, 
is treated as a finite counter that the machine systematically decrements to enumerate all possible truth assignments.

\paragraph{High-Level Structure}

The machine operates on a single tape containing a CNF formula and a candidate assignment. Rather than separate assignment and evaluation phases, the machine performs real-time evaluation while scanning literals, proceeding as follows:
\begin{enumerate}
\item \textbf{Index Encoding and Value Retrieval.}
Upon encountering a literal, the machine captures its variable index into the control states.
 It then moves forward past the delimiter (\texttt{\#}) to the assignment region, retrieves the corresponding truth value, and returns to the literal's position—marked with a temporary marker—to evaluate the result.
\item \textbf{On-the-fly Clause Evaluation.}
The machine evaluates each clause dynamically based on the retrieved values:
\begin{itemize}
\item \textbf{Check Mode:} If the current clause has not yet been satisfied, the machine continues to fetch and evaluate the next literal within the same clause. If no more literals exist in the clause, it rejects.
\item \textbf{Skip Mode:} As soon as a literal evaluates to \textit{True}, the clause is deemed satisfied. The machine then process remaining literals in that clause returning to the \textsc{Skip} state.
\end{itemize}
\item \textbf{Termination and Result.}
If all clauses are successfully satisfied (i.e., every clause reaches \textit{Skip} mode), the machine halts in \texttt{Accept} state. If any clause is fully evaluated to \textit{False} without a single \textit{True} literal, the machine transitions to a rejection state.
\end{enumerate}

\paragraph{Transition-Level Structure} 
\begin{itemize}
\item \textbf{Input-Dependent State Encoding of Variable Indices} \
The TM constructs variable indices dynamically within its control states. Upon scanning a digit $D$ 
 of a literal, the machine transitions from state \textsc{Inc}.$j$ to \textsc{Inc.}($10j+D$), effectively performing the calculation  through state changes.
  This process continues until a non-digit symbol (e.g., \texttt{\_}) is encountered, at which point the complete index 
 is stored in the state (\texttt{Forward}.$i$). This mechanism allows the machine to identify variables without using a separate counter on the tape.
\item \textbf{Direct Value Fetching via Index Decoupling} \
Once the index $i$ is encoded, the TM moves forward to the assignment region. Upon crossing the delimiter (\texttt{\#}), it enters a \textsc{Fetch}.$i$ state.
  For each cell it traverses in the assignment region, it decrements the index stored in the state suffix.
  When the index reaches zero, the TM retrieves the truth value (\texttt{T} or \texttt{F}) at that position.
   This value, along with the literal's polarity (positive or negative), is stored in a return state (\textsc{Backward.T}/\textsc{Backward.F}), which carries the information back to the original literal position.
\item \textbf{On-the-fly Evaluation: Check and Skip Modes} \
After returning to the literal and marking it, the TM enters one of two primary control modes:
\begin{enumerate}
\item \textbf{Check Mode:} If the evaluated literal is \textit{False}, the machine remains in \texttt{Check} mode and advances to the next literal in the current clause to find a satisfying value.
 If the machine encounters the clause delimiter \texttt{\&} without having found a true literal, it rejects, as a single unsatisfied clause renders the entire CNF formula false.
\item \textbf{Skip Mode:} If the literal is \textit{True}, the clause is satisfied. The machine enters \texttt{Skip} mode, the processing of remaining literals are redundant and only for certificate-oblivious property until the next clause delimiter (\texttt{\&}).
\end{enumerate}
The process resets to \texttt{Check} mode at the start of each new clause and repeats until the delimiter \texttt{\#} is reached again, signaling the completion of the formula evaluation.
\end{itemize} 

\begin{definition}[Deterministic SAT Turing Machine with Input-Dependent States]
\label{def:sat_tm_input_dependent_states}
Let $\varphi$ be a Boolean formula in CNF form, encoded on a single tape as described in earlier paragraph.
We define a deterministic single-tape Turing machine
\[
M_{\mathrm{SAT}}^{\mathrm{ID}}
= (Q_\varphi, \Sigma, \Gamma, \delta_\varphi, \qinit, \qacc, \qrej)
\]
with the following properties.
As required by the definition of a Turing machine, the state set of
$M_{\mathrm{SAT}}^{\mathrm{ID}}$ is finite, as the machine is defined for a fixed number of variables $m$, where $m$ is treated as a constant.
\begin{itemize}
    \item $Q$ is a finite set of states consisting of the control states
    listed in \cref{tab:sat_variable_transitions}, including the distinguished
    halting states $\qacc$ and $\qrej$.
    Parameterized state names such as $\textproc{Forward.N}$ or
    $\textproc{Dec.(N-1)}$ denote finite families of states instantiated variable up to index $m$; 
    therefore, they contribute only a constant factor to $|Q|$.

    \item $\Sigma$ is the input alphabet encoding a CNF formula together with
    an explicit Boolean assignment, using symbols
    $\{0,\ldots,9,-,\&, \_, \#, T, F\}$.

    \item $\Gamma$ is the tape alphabet satisfying
    $\Sigma \subseteq \Gamma$; it additionally includes the blank symbol $\blank$ and auxiliary symbols \texttt{`?'}, \texttt{`!'} and \texttt{`S'}.

    \item $\qinit = \textsc{Check.Forwarded}$ is the initial state, in which
    the machine begins scanning the input from left to right.

    \item $\delta : Q\setminus \{ \qacc, \qrej \} \times \Gamma \to Q \times \Gamma \times \{-1,+1\}$
    is the transition function defined explicitly by the rules given in
    \cref{tab:sat_variable_transitions}.

    \item $\qacc = \textsc{Accept}$ is the accepting state,
    entered if and only if all clauses are verified to be \textit{True} under
    the given assignment.

    \item $\qrej = \textsc{Reject}$ is the rejecting state,
    entered upon detection of an invalid encoding or a falsified clause.
\end{itemize}
\end{definition}

\begin{table}[!ht]
\caption{Transition function of the certificate-oblivious inputdepent-state Turing machine}
\label{tab:sat_variable_transitions}
\centering
\renewcommand{\arraystretch}{1.15}
\begin{tabular}{llllll}
\toprule
Current State & Read Symbol & Next State & Write & Move & Remark \\
\midrule
\multicolumn{6}{l}{\emph{Assign.Check phase}} \\
\midrule
Check        & \texttt{\spacedelim} & Check        & \texttt{\spacedelim} & R & Skip blanks \\
Check        & \texttt{-} & Not         & \texttt{\_} & R & Negation detected \\
Check        & \texttt{D} & Inc.D       & \texttt{?} & R & Begin number parsing \\
Not         & \texttt{D} & Inc.D       & \texttt{!} & R & Mark negated literal \\
Skip			& \texttt{D} & Inc.D  	  & \texttt{S} & R & Maker for satisfied \\
Skip        & \texttt{\&} & Check        & \texttt{\spacedelim} & R & Move to next clause \\
Skip        & \texttt{\#} & Accept      & \texttt{\spacedelim} & R & All clauses processed \\
Skip        & \texttt{*} & Skip        & \texttt{\spacedelim} & R & Bypass the literal \\
Check        & \texttt{\&} & Reject      & \texttt{\spacedelim} & R & False clause \\
Check        & \texttt{\#} & Reject      & \texttt{\spacedelim} & R & False clause termination \\

\midrule
\multicolumn{6}{l}{\emph{Assign.Inc phase}} \\
\midrule
Inc.N       & \texttt{\spacedelim} & Forward.N  & \texttt{\spacedelim} & R & End of literal \\
Inc.N       & \texttt{\&} & Forward.N  & \texttt{\&} & R & Clause boundary \\
Inc.N       & \texttt{\#} & Dec.(N-1)  & \texttt{\#} & R & Final literal \\
Inc.N       & \texttt{D} & Inc.(10N+D) & \texttt{\spacedelim} & R & Decimal expansion \\

\midrule
\multicolumn{6}{l}{\emph{Assign.Forward phase}} \\
\midrule
Forward.N   & \texttt{*} & Forward.N  & \texttt{*} & R & Scan to assignment area \\
Forward.N   & \texttt{\#} & Dec.(N-1)  & \texttt{\#} & R & Reached truth value\\

\midrule
\multicolumn{6}{l}{\emph{Assign.Dec phase}} \\
\midrule
Dec.N       & \texttt{T} & Dec.(N-1)  & \texttt{T} & R & Foward/Decrement in certificate area \\
Dec.N       & \texttt{F} & Dec.(N-1)  & \texttt{F} & R & Foward/Decrement in certificate area \\
Dec.0       & \texttt{T} & Backward.T & \texttt{T} & L & Selected literal true \\
Dec.0       & \texttt{F} & Backward.F & \texttt{F} & L & Selected literal false \\

\midrule
\multicolumn{6}{l}{\emph{Assign.Backward phase}} \\
\midrule
Backward.T  & \texttt{*} & Backward.T & \texttt{*} & L & Return to clause \\
Backward.F  & \texttt{*} & Backward.F & \texttt{*} & L & Return to clause \\
Backward.T  & \texttt{?} & Skip       & \texttt{\spacedelim} & R & True clause  \\
Backward.F  & \texttt{?} & Check       & \texttt{\spacedelim} & R & False literal \\
Backward.T  & \texttt{!} & Check       & \texttt{\spacedelim} & R & False from negated literal \\
Backward.F  & \texttt{!} & Skip       & \texttt{\spacedelim} & R & True from negated literal \\
Backward.T  & \texttt{S} & Skip       & \texttt{\spacedelim} & R & Restore True clause  \\
Backward.F  & \texttt{S} & Skip       & \texttt{\spacedelim} & R & Restore True clause \\
\bottomrule
\end{tabular}
\\
The states ($\cdot N$) of the TM are constructed during preprocessing
based on the input variable count.
\end{table}
\paragraph{Transition Symbols.}
We summarize the tape symbols and meta-symbols used in the transition
specification of the variable-state Turing machine.

\begin{itemize}
  \item \textbf{digit}: A concrete numeric symbol
        $\{0,1,\dots,9\}$ appearing on the tape, representing part of a
        decimal-encoded variable index.

  \item \textbf{$D$}: A meta-symbol denoting any digit for a digit-encoded variable index.
        States of the form \texttt{Inc.$D$} represent families of states parameterized by a concrete integer value
        and are instantiated during preprocessing for each digit.

  \item \textbf{$N$}: A meta-variable representing the integer value currently
        encoded in the state name.
        Transitions such as \texttt{Inc.(10$N$+D)} and \texttt{Dec.(N-1)}
        describe arithmetic updates on this encoded value and are realized
        concretely by instantiating the corresponding successor states.

  \item \textbf{? / ! / S}: \texttt{?} and \texttt{!} literal markers written during the \texttt{Check} phase,
        indicating whether the current literal is positive (\texttt{?})
        or negated (\texttt{!}). \texttt{S} literal marker is the clause is already satisfied.

  \item \textbf{T / F}: The truth values read from the assignment region of the tape,
        corresponding to the valuation of a variable under the current
        certificate.

  \item \textbf{*}: A wildcard symbol that is preserved during traversal and
        ignored by most transition rules.

  \item \textbf{\#}: A delimiter separating the formula region from the
        assignment region and marking the end of clause processing.

  \item \textbf{\&}: A clause separator symbol delimiting individual clauses
        within the formula region.
\end{itemize}

\paragraph{Bounded Parameterized States}
The parameterized state families $\mathrm{Inc}.N$, $\mathrm{Forward}.N$, and $\mathrm{Dec}.N$
are instantiated only for indices
\[
0 \le N \le k,
\]
where $k$ is the number of variables determined during preprocessing.
Hence, for every fixed input instance, the machine is a finite state Turing machine.

Any transition that would produce a state outside this range
(e.g., $\mathrm{Inc}.(10N+D)$ with value $>k$,
or $\mathrm{Dec}.(-1)$)
is treated as undefined and immediately causes the machine to enter the
\textsc{Reject} state.
Consequently, malformed literals or references to non-existent variables are
rejected during the parsing phase.

\begin{lemma}[Soundness of Certificate-Oblivious Input-Dependent Verifer]
\label{lem:soundness_sat_id}
If $M_{\mathrm{SAT}}^{\mathrm{ID}}$ halts in the \texttt{Accept} state on
input $\Phi \# \mathcal{V}$, then the valuation $\mathcal{V}$ is a satisfying assignment of the CNF
formula $\Phi$.
\end{lemma}

\begin{proof}
We prove soundness by contradiction, based on the operational semantics induced
by the transition function of $M_{\mathrm{SAT}}^{\mathrm{ID}}$.

\begin{itemize}
    \item \textbf{Contradictory assumption.}  
    Assume that $M_{\mathrm{SAT}}^{\mathrm{ID}}$ halts in the
    \texttt{Accept} state, but the valuation $\mathcal{V}$ does \emph{not}
    satisfy $\Phi$.
    Then there exists a clause $\mathcal{C}$ in $\Phi$ such that every literal 
    $\ell$ in $\mathcal{C}$ evaluates to \texttt{False} under $\mathcal{V}$.

    \item \textbf{Acceptance implies successful processing of all clauses.}  
    By inspection of the transition table, the \texttt{Accept} state is reachable
    only after the machine has scanned the entire formula region, successfully
    processing each clause delimited by \texttt{\&}, and has reached the terminal
    symbol \texttt{\#}.
    In particular, the machine advances past a clause only by entering the
    \texttt{Skip} phase for that clause.

    \item \textbf{Entering \texttt{Skip} requires a satisfied literal.}  
    For any fixed clause $\mathcal{C}$, the machine enters \texttt{Skip} only if
    at least one literal in $\mathcal{C}$ is verified as satisfied.
    For each literal, the machine deterministically fetches the literal as follows:
        \begin{itemize}
        \item marks the position with a literal marker and parses the variable index using the digit-reading transitions and the state family \texttt{Inc.(10$N$+$D$)};
        \item moves to the valuation region through the \texttt{Fetch.$N$} state and enters the corresponding \texttt{Dec.$N$} state, where $N$ represents the exact offset to the valuation region;
        \item reaches the valuation cell $\mathcal{V}(i)$ by decrementing $N$ to zero;
        \item returns to the literal marker via the appropriate \texttt{Backward.T} or \texttt{Backward.F} state and \textbf{checks} consistency with the literal polarity(\texttt{?}, \texttt{!}). 
        			If the literal is \texttt{S}, then the  verification yields \texttt{True} regardless of the suffix of the state.
        \end{itemize}
    Only if this verification yields \texttt{True} does the machine enter \texttt{Skip}, and if the machine reaches the clause end \texttt{\&} with \texttt{Skip} state by using the \texttt{S} marker for each following fetch.

    \item \textbf{Deriving the contradiction.}  
    By the assumption, the clause $\mathcal{C}$ is unsatisfied by $\mathcal{V}$,
    so no literal in $\mathcal{C}$ can cause the machine to enter
    \texttt{Skip}.
    Hence, the machine cannot advance past $\mathcal{C}$ and therefore cannot
    reach the \texttt{Accept} state, contradicting the initial assumption.

    \item \textbf{Conclusion.}  
    No such unsatisfied clause exists.
    Therefore, every clause in $\Phi$ contains at least one literal that
    evaluates to \texttt{True} under $\mathcal{V}$, and $\mathcal{V}$ satisfies
    $\Phi$.
\end{itemize}
\end{proof}

\begin{lemma}[Completeness of Certificate-Oblivious Input-Dependent Verifer]
\label{lem:completeness_sat_id}
If a valuation $\mathcal{V}$ is a satisfying assignment of the CNF formula $\Phi$, then
$M_{\mathrm{SAT}}^{\mathrm{ID}}$ enters the \texttt{Accept} state on
input $\Phi \# \mathcal{V}$.
\end{lemma}

\begin{proof}
We prove completeness by following the deterministic execution of
$M_{\mathrm{SAT}}^{\mathrm{ID}}$ under a satisfying valuation
$\mathcal{V}$.

\begin{itemize}
    \item \textbf{Existence of a satisfying literal in each clause.}  
    By assumption, the valuation $\mathcal{V}$ satisfies $\Phi$.
    Hence, for every clause $\mathcal{C}$ in $\Phi$, there exists at least one literal $\ell \in \mathcal{C}$ 
    such that the assignment of the corresponding value to $\ell$ evaluates to \texttt{True}.

    \item \textbf{Correct evaluation of literals.}  
    During the \texttt{Check} phase for a fixed clause $\mathcal{C}$, the machine
    sequentially scans all literals in $\mathcal{C}$.
    For each literal, it:
    \begin{itemize}
        \item parses the variable index using the digit-reading transitions and
              the \texttt{Inc.(10$N$+D)} state family;
        \item enters the corresponding \texttt{Dec.$i$} state, where $i$ encodes
              the exact offset to the valuation region;
        \item reaches and reads the valuation cell $\mathcal{V}(i)$ associated
              with the literal;
        \item returns to the clause region via the appropriate
              \texttt{Backward.T} or \texttt{Backward.F} transition, and checks the truth value.
    \end{itemize}
    This process correctly computes the truth value of each literal under
    $\mathcal{V}$.

    \item \textbf{Deterministic advancement past satisfied clauses.}  
    When a literal $\ell$ satisfying $\mathcal{C}$ is encountered, the machine
    deterministically transitions into the \texttt{Skip} state.
    Once \texttt{Skip} state achieved in a clause, the \texttt{Skip} state is restored after fetching each remaining symbol in the clause via  \texttt{S} marker.

    \item \textbf{Progress through the entire formula.}  
    Since every clause contains at least one satisfying literal, the machine
    successfully enters \texttt{Skip} for each clause delimiter \texttt{\&}.
    Consequently, it advances monotonically through the entire formula region.

    \item \textbf{Reaching acceptance.}  
    After processing the final clause, the machine reaches the terminal symbol
    \texttt{\#}, upon which it enters the \texttt{Accept} state by definition of
    the transition function.
\end{itemize}

\noindent
Therefore, whenever $\mathcal{V}$ satisfies $\Phi$,
$M_{\mathrm{SAT}}^{\mathrm{ID}}$ deterministically halts in the
\texttt{Accept} state.
\end{proof}

\begin{lemma}
The Turing machine $M_{\mathrm{SAT}}^{\mathrm{ID}}$ is a certificate-oblivious verifier for SAT if only `T', `F' are allowed as certificate symbols (via preprocessing).

\begin{proof}
Let $L_{\text{fixed}}$ be a fixed problem instance of SAT. Let $H_i(Y)$ denote the head position of $M_{\mathrm{SAT}}$ after the $i$-th transition, and $L_{i,j}(Y)$ denote the symbol at the $j$-th tape cell after the $i$-th transition, given a certificate $Y \in \mathcal{C}$ of length $m = |Y|$.

We categorize the state space $Q$ into three computational phases: the \textbf{Clause Evaluation phase} ($Q_{\text{eval}}$: states with prefixes $\{\texttt{Check}, \texttt{Skip}, \texttt{Not}\}$), the \textbf{Fetch phase} ($Q_{\text{fetch}}$: state with the prefixes $\{\texttt{Forward}, \texttt{Inc}, \texttt{Dec}\}$), and the \textbf{Backward/Assignment phase} ($Q_{\text{sub}}$: states with prefix $\{\texttt{Backward}$.

We assert the following properties for all $i \ge 0$ and any $Y_1, Y_2 \in \mathcal{C}$ with $|Y_1| = |Y_2| = m$:
\begin{enumerate}
    \item \textbf{Head Invariance:} $H_i(Y_1) = H_i(Y_2)$.
    \item \textbf{Data Invariance:} $L_{i,j}(Y_1) = L_{i,j}(Y_2)$ whenever $L_{i,j}(Y_1), L_{i,j}(Y_2) \notin \{\text{`T', `F'}\}$ and $L_{i,j}(Y_1), L_{i,j}(Y_2) \notin \{\text{`!', `?', `S'}\}$.
    \item \textbf{Phase Invariance:} The machine is in the same phase for both $Y_1$ and $Y_2$.
\end{enumerate}

At the start of the computation, the head position $H_0(Y) = 0$ is uniform for all $Y$. The tape content $L_{0,j}$ is identical for all $j < |L_{\text{fixed}}|$. For the certificate region $|L_{\text{fixed}}| \le j < |L_{\text{fixed}}| + m$, $L_{0,j} \in \{T, F\}$ holds by the assumption that the certificate alphabet is restricted to $\{T, F\}$. For the empty tape area ($j < 0$ or $j \ge |L_{\text{fixed}}| + m$), $L_{0,j} = \epsilon$. Thus, the assertions hold trivially for $i=0$.

Assume, for the sake of contradiction, that there exists a minimum index $i>0$ such that the obliviousness property is violated for two certificates $Y_1, Y_2$ with $|Y_1| = |Y_2| = m$. That is, either $H_i(Y_1) \neq H_i(Y_2)$, or there exists a position $j$ such that $L_{i,j}(Y_1) \neq L_{i,j}(Y_2)$ where $L_{i,j}(Y_1), L_{i,j}(Y_2) \notin \{\text{`T', `F'}\}$. We examine the transition at step $i$ for each phase:

\begin{itemize}
    \item \textbf{Fetch:} Except for \texttt{Dec.0} state, every transition is identical, since every symbol processed in this phase does not contain symbols in $\{\text{`T', `F'}\}$ or $\{\text{`!', `?', `S'}\}$; thus $H_i(Y_1) = H_i(Y_2)$ and $L_{i,j}(Y_1) = L_{i,j}(Y_2)$, contradiction. For state \texttt{Dec.0}, the machine process symbols only in $\{\text{`T', `F'}\}$ and transit to Backward/Assignment phase with the \texttt{Backward.B} state to the left direction where suffix \texttt{`B'} is corresponding boolean symbol while preserving the tape symbol, which means $H_i(Y_1) = H_{i-1}(Y_1)$ and $L_{i,j}(Y_1), L_{i,j}(Y_2) \in \{\text{`T', `F'}\}$, contradicting the assumption.

	\item \textbf{Backward/Assignment:} Except for the symbols in $\{\text{`!', `?', `S'}\}$, every transition is identical independent of the state suffix category $.B$ (where $B \in \{\text{`T', `F'}\}$), which means $H_i(Y_1) = H_i(Y_2)$ and $L_{i,j}(Y_1) = L_{i,j}(Y_2)$, contradicting the assumption. For the symbols $\{\text{`!', `?', `S'}\}$, the direction of transition is identical to right ($H_i(Y_1) = H_{i-1}(Y_1) + 1 = H_{i-1}(Y_2) + 1 = H_i(Y_2)$), the output of the transition is identical to \spacedelim, and transit to clause evaluation phase (leading to the contradiction..
    
    \item \textbf{Clause Evaluation:} Transitions move right ($+1$). Thus, $H_i(Y_1) = H_{i-1}(Y_1) + 1 = H_{i-1}(Y_2) + 1 = H_i(Y_2)$. Regarding the tape content, every symbol $0$ is replaced by $\spacedelim$ independent of the internal state (e.g., \texttt{Check}, \texttt{Skip}, \texttt{Not}, thus $L_{i,j}(Y_1) = L_{i,j}(Y_2) = \text{`\_'}$. Finally, upon encountering the digit symbol `0'-`9', the machine deterministically transitions to the Fetch phase (\texttt{Inc.D}), unless the machine rejects. This consistency contradicts the minimality of $i$.
\end{itemize}

Since the trajectory $H_i(Y)$ is invariant and independent of the certificate content for a particular problem instance, no such minimal index $i$ can exist. Therefore, the execution graph satisfies the Grid-aligned Transition property.
\end{proof}
\end{lemma}

\begin{lemma}[Time and Space Complexity of Certificate-Oblivious Input-Dependent Verifer]\label{lem:time_space_sat_id}
Let $\Phi$ be a CNF formula of encoding length $n$.
For a fixed certificate (truth assignment), the input-dependent SAT Turing
Machine $M_{\mathrm{SAT}}^{\mathrm{ID}}$ runs in time $\bigO(n^{2})$ and
uses $\bigO(n)$ space.
\end{lemma}

\begin{proof}
The complexity bounds follow from the structural properties of the machine
execution for a fixed certificate.

\begin{itemize}
  \item \textbf{Clause-region traversal.}
  The machine performs a left-to-right traversal of the clause region of the
  tape.
  The total length of the clause region is $\bigO(n)$, and each literal symbol is
  encountered at most once during this traversal.

  \item \textbf{Literal evaluation cost.}
  For each literal $\ell = x_i$ encountered in the clause region, the machine
  evaluates its truth value by fetching it from the certificate region of the tape.
  This involves a traversal whose length is proportional to the distance 
  between the clause region and the $i$-th certificate cell, which is $\bigO(n)$ in 
  the worst case. After reading the certificate cell, the machine returns to 
  the clause region to resume evaluation. Hence, each literal evaluation 
  incurs an $\bigO(n)$ time cost.

  \item \textbf{Total number of literal evaluations.}
  The total number of literals appearing in the formula is $\bigO(n)$.
  Since the certificate values are not erased after evaluation, and the same 
  variable may appear in multiple clauses, value fetch operations may be 
  repeated for the same variable across different clauses.

  \item \textbf{Total running time.}
  By combining the $\bigO(n)$ time required for each literal evaluation with 
  the $\bigO(n)$ total number of literals, the total running time for a fixed 
  certificate is $\bigO(n^2)$.

  \item \textbf{Tape usage.}
  The machine operates solely on the input tape containing the formula and
  the fixed certificate.
  All auxiliary information is encoded using a constant number of marker
  symbols that overwrite existing tape cells.
  No additional work tape or unbounded storage is used, and the number of
  non-blank tape cells remains $\bigO(n)$ throughout the computation.

  \item \textbf{Control-state dependence.}
  Although the number of control states depends on the number of variables and
  is polynomially bounded, this dependence affects only the finite control and
  does not contribute to tape space usage.
\end{itemize}

Therefore, for a fixed certificate, the total running time of
$M_{\mathrm{SAT}}^{\mathrm{ID}}$ is $\bigO(n^{2})$, and the space usage
is $\bigO(n)$.
\end{proof}

\begin{remark}
Although the control-state space of
$M_{\mathrm{SAT}}^{\mathrm{ID}}$ depends on the input,
its size is polynomially bounded.
Consequently, this construction remains compatible with the NP verifier
simulation framework, which requires polynomial bounds on both time and state
space but does not require a fixed finite set of states.
\end{remark}

\begin{theorem}[Input-Dependent Polynomial-Time Verifier for \textsc{SAT}]
\label{thm:sat-id-verifier}
The input-dependent Turing machine
$M_{\mathrm{SAT}}^{\mathrm{ID}}$ decides whether a certificate is accepted by the verifier for the
\textsc{SAT} problem.

More precisely, for every input string $W = \Phi \# \mathcal{V}$ ,
where $\Phi$ is a CNF formula of length $n$ and $\mathcal{V}$ is a truth assignment encoded as
a certificate,
\begin{itemize}
    \item $M_{\mathrm{SAT}}^{\mathrm{ID}}$ accepts $W$ if and only if
    the valuation $\mathcal{V}$ satisfies the formula $\Phi$, and
    \item $M_{\mathrm{SAT}}^{\mathrm{ID}}$ halts in $\bigO(n^2)$ time and
    uses $\bigO(n)$ tape space.
\end{itemize}
Consequently, \textsc{SAT} admits a deterministic polynomial time verifier using
linear tape space.
\end{theorem}

\begin{proof}
The claim follows by combining the correctness and resource bounds established previously.

\begin{itemize}
\item \textbf{Soundness.}
By \cref{lem:soundness_sat_id}, if $M_{\mathrm{SAT}}^{\mathrm{ID}}$ accepts an input $\Phi \# \mathcal{V}$, 
then for every clause of $\Phi$, the machine has identified a literal that evaluates to \texttt{True} under $\mathcal{V}$. 
Hence, the valuation $\mathcal{V}$ satisfies $\Phi$, and no invalid certificate is accepted.

\item \textbf{Completeness.}
By \cref{lem:completeness_sat_id}, if a valuation $\mathcal{V}$ satisfies the CNF formula $\Phi$, 
then for every clause, the machine encounters a satisfying literal during the \texttt{Check} phase 
and deterministically transitions through the \texttt{Skip} mechanism. 
Thus, $M_{\mathrm{SAT}}^{\mathrm{ID}}$ reaches the \texttt{Accept} state on input $\Phi \# \mathcal{V}$.

\item \textbf{Time and space complexity.}
By \cref{lem:time_space_sat_id}, for an input where $\Phi$ has length $n$, 
the machine $M_{\mathrm{SAT}}^{\mathrm{ID}}$ halts in $\bigO(n^2)$ time and uses $\bigO(n)$ tape space. 
In particular, the verifier runs in deterministic polynomial time.
\end{itemize}

Combining the above arguments, $M_{\mathrm{SAT}}^{\mathrm{ID}}$ is a correct deterministic polynomial-time verifier for the \textsc{SAT} problem.
\end{proof}

\paragraph{Discussion}
This input-dependent Turing machine construction departs from the classical model in that the control states are not finite independently of the input.
 The purpose of this design is not to optimize the traditional Turing machine,
  but to demonstrate that a fixed control-state set is required only relative to a specific certificate and need not be uniform across all possible inputs.
Presenting this construction first emphasizes that the NP verifier simulation framework does not fundamentally rely on a fixed finite state set.
 Instead, input-dependent state machines should be viewed as an extension that illustrates the framework's adaptability.
 The Turing machine introduced in the following subsection strictly conforms to the classical Turing machine model, although its transition logic is more complex to define.

\subsection{SAT Turing Machine With Fixed States} \label{subsec:sat_tm_fixed}

In this subsection, we describe a deterministic certificate-oblivious Turing machine
$M_{\mathrm{SAT}}$ that verifies satisfiability of Boolean formulas
in CNF form under an explicit assignment encoding.
As required by the formal definition of a Turing machine, $M_{\mathrm{SAT}}$
has a finite set of states and a finite tape alphabet.
What is relevant here is that the number of states is a fixed constant,
independent of the input size and the number of variables.
All numerical generality is handled solely through tape symbols
and structured scan patterns.

\paragraph{Tape Structure.}
The input format, tape alphabet, and overall tape structure of this
fixed SAT Turing Machine are identical to those described in the
previous subsection.
Conceptually, the tape is divided into the following regions:
\begin{enumerate}
    \item a \emph{clause region}, containing the CNF formula encoded as
    integer literals with explicit space delimiters and clause separators
    `\texttt{\&}'.
    This region also serves as a working area, in which temporary markers and
    modified symbols are written during decrement operations, borrow
    propagation, and clause evaluation;

    \item a \emph{certificate region}, located to the right of the symbol
    `\texttt{\#}', encoding a Boolean vector over the alphabet
    $\{\texttt{T},\texttt{F}\}$.
    This vector represents the current candidate truth assignment, with entries
    ordered according to the variable indices.
    During execution, whenever a variable is assigned, 
    the corresponding certificate symbol is erased and replaced with `\_'. 
\end{enumerate}

\paragraph{Symbol Conventions.}
The symbol `\texttt{D}' denotes an arbitrary decimal digit in $\{0,\dots,9\}$.
The symbols `\texttt{T}' and `\texttt{F}' denote Boolean truth values
\textsc{True} and \textsc{False}, respectively.

\paragraph{High Level Semantics.}
The machine operates by repeating the following phases:

    \begin{enumerate}
    \item \textbf{Right-to-Left Decrement Phase (Fecth and Assignment):} 
    \begin{itemize}
        \item The head fetches the first truth value from the certificate region and replaces it with \spacedelim.
        \item The head moves leftward over the clause region, \textbf{decrementing} each variable's index. 
        \item If an index reaches zero, the corresponding variable is considered \textbf{``assigned''} for evaluation. Otherwise, the value \textbf{is decremented} without placement.
    \end{itemize}

    \item \textbf{Left-to-Right Scan Phase (Clause Evaluation):} 
    \begin{itemize}

	\item The head scans from left to right across the entire clause region.
	\item Leading zeros are erased by replacement with `\_'.
	\item Unassigned literals within clauses that are not yet completely assigned are recorded in the control state using the state schema \texttt{Unknown.S} and \texttt{UnknownTerm.S}.
	Here \texttt{UnknownTerm} denotes a literal currently being processed,
	while \texttt{Unknown} denotes a variable whose value is not yet determined.
	(The parameter \texttt{S} specifies the scan mode and will be explained below.)
        \item Clauses are evaluated using only assigned variables:
        \begin{itemize}
            \item If any assigned variable satisfies a clause, the machine
            transitions to \texttt{Skip} to bypass further evaluation.
            \item If all variables in a clause are assigned but the clause evaluates
            to \texttt{False}, the machine enters \texttt{Reject}.
        \end{itemize}
        \item After the evaluation, the head enters to the certificate region to fetch the next variable.
    \end{itemize}

    \item \textbf{Iteration and Halting:} 
	\begin{itemize}
	    \item The machine \textbf{repeats} the fetch/assignment and scan phases for all variables and clauses.
	    \item If every clause \textbf{evaluates to} \texttt{True} (i.e., all scans end in \texttt{Skip.Forwarded}), the machine halts in \texttt{Accept}. 
	    The state \texttt{Skip.Forwarded} indicates that all clauses \textbf{already evaluate to} \texttt{True} under the current assignment prefix; 
	    thus, no further verification is required once it reaches the end of the clause region `\#'.
	    \item Otherwise, it \textbf{fetches} the next assignment from the certificate region and \textbf{repeats} the process.
    \end{itemize}
\end{enumerate}
The decrement operation is the mechanism for preventing multiple fetching for a single variable.
It renormalizes variable indices so that the next single certificate symbol
corresponds to the variable whose index becomes zero.
If an index reaches $0$, the corresponding variable is considered ``assigned''.
(This occurs precisely for the variable whose index was $1$ before the decrement, i.e., the variable aligned with the next certificate symbol.)
After assignment, the certificate is shortened and the remaining indices
are shifted, allowing the process to repeat until all variables are assigned.

\paragraph{Key Features.}
\begin{itemize}
    \item \textbf{Single Fetch and Multiple Assignment:} 
    The assignment region is treated as an area holding truth values for the variables not yet processed.
    Variables are processed in a fixed order across all clauses, rather than clause by clause.

    \item \textbf{State-Based Unknowns:} Unassigned variables are tracked via the control states (\textsc{Unknown.S}, \textsc{UnknownTerm.S}),
     where $S \in \{\textsc{Forwarded}, \textsc{Free}\}$. This allows the TM to perform partial evaluations without the need to mark intermediate results on the tape.

    \item \textbf{Forwarded Satisfaction Propagation:}
    During clause scanning, the machine maintains two scan modes,
    \textsc{Free} and \textsc{Forwarded}. 
    the \textsc{Forwarded} scan mode indicates that all clauses to the left of the current
    position have already evaluated to \texttt{True} under the current assignment
    prefix, and this invariant is propagated across subsequent clauses.
    Acceptance occurs only if the clause region end marker is reached while in
    \textsc{Forwarded} mode.

    \item \textbf{Single Tape and Fixed States:} 
    Numerical operations, assignment management, and clause evaluation are
    encoded entirely through tape symbols and state transitions; the number
    of states is independent of input size or number of variables.
\end{itemize}

We define a deterministic single-tape Turing machine
$M_{\mathrm{SAT}}$
that verifies satisfiability of Boolean formulas in CNF form under an
explicit assignment encoding.

All numerical processing, including variable indices and the fetch position, 
is handled exclusively through tape symbols and structured scan patterns, rather than through state expansion.

\begin{definition}[Deterministic Certificate-Oblivious SAT Verifier Turing Machine]\label{def:tm_sat_fixed}
The deterministic single-tape Turing machine
$M_{\mathrm{SAT}}$ is defined as a tuple
\[
M_{\mathrm{SAT}}
= (Q, \Sigma, \Gamma, \delta, \qinit, \qacc, \qrej),
\]
where the components are specified as follows.

\begin{itemize}
    \item $Q$ is a finite set of states consisting of the control states
    listed in \cref{tab:sat_fixed_transitions_forward,tab:sat_fixed_transitions_backward},
     including the distinguished halting states $\qacc$ and $\qrej$.

    \item $\Sigma$ is the input alphabet encoding a CNF formula together with
    an explicit Boolean assignment, using symbols $\{0,\ldots,9,-,\&, \_, \#, T, F \}$.

    \item $\Gamma$ is the tape alphabet satisfying
    $\Sigma \subseteq \Gamma$, and additionally includes empty tape symbol $\blank$.

    \item $\qinit = \textsc{Check.Forwarded}$ is the initial state, in which
    the machine begins scanning the input from left to right.

    \item $\delta : Q \times \Gamma \to Q \times \Gamma \times \{-1,+1\}$
    is the transition function defined explicitly by the rules given in
    \cref{tab:sat_fixed_transitions_forward,tab:sat_fixed_transitions_backward}.

    \item $\qacc = \textsc{Accept}$ is the accepting state,
    entered if and only if all clauses are verified to be satisfied under
    the given assignment.

    \item $\qrej = \textsc{Reject}$ is the rejecting state,
    entered upon detection of an invalid encoding or a falsified clause.
\end{itemize}

\end{definition}
The corresponding transition table fully determines the machine's behavior, including
digit-wise subtraction, borrow propagation, literal evaluation, and clause
verification.
All numerical generality, including variable indices and the fetch positions is handled exclusively through tape symbols and structured scan
modes, rather than through state expansion.

\begin{table}[!ht] \label{fixed_clause_evaluation}
\centering
\caption{Certificate-Oblivious Transition Function Part 1} \label{tab:sat_fixed_transitions_forward}
\begin{tabular}{llllll}
\textbf{State} & \textbf{Read} & \textbf{Next} & \textbf{Write} & \textbf{Move} & \textbf{Comment} \\
\hline
Check.S & $\_$ & Check.S & $\_$ & R & bypass delimiters \\
Check.S & $-$ & CheckNot.S & $-$ & R & negated state \\
Check.S & $0$ & Unknown.S & $\_$ & R & erase leading zero \\
Check.S & D & UnknownTerm.S & D & R & enter variable term \\
Check.S & T & Skip.S & T & R & clause satisfied \\
Check.S & F & Check.S & F & R & continue scan \\
Check.S & \& & Reject & $\_$ & R & premature rejection \\
Check.S & \# & Reject & $\_$ & R & reject in the final clause \\
\hline
CheckNot.S & $\_$ & CheckNot.S & $\_$ & R & bypass delimiters \\
CheckNot.S & T & Check.S & T & R & false from negation\\
CheckNot.S & F & Skip.S & F & R & satisfied from negation\\
CheckNot.S & D & UnknownTerm.S & D & R & negation to unassigned term\\
CheckNot.S & 0 & Unknown.S & $\_$ & R & erasing leading zero \\
\hline
Unknown.S & $\_$ & Unknown.S & $\_$ & R & bypass \\
Unknown.S & 0 & Unknown.S & $\_$ & R & erasing leading zero \\
Unknown.S & D & UnknownTerm.S & D & R & enter variable term \\
Unknown.S & T & Skip.S & T & R & satisfied \\
Unknown.S & F & Unknown.S & F & R & undecided \\
Unknown.S & $-$ & UnknownNot.S & $-$ & R & negated unknown \\
Unknown.S & \& & Check.Free & \& & R & next clause \\
Unknown.S & \# & Fetch & \# & R & enter fetch phase \\
\hline
UnknownNot.S & $\_$ & UnknownNot.S & $\_$ & R & bypass \\
UnknownNot.S & T & Unknown.S & T & R & undecided from negation \\
UnknownNot.S & F & Skip.S & F & R & satisfied from negation\\
UnknownNot.S & D & UnknownTerm.S & D & R & enter variable term \\
UnknownNot.S & 0 & Unknown.S & $\_$ & R & erase leading zero \\
\hline
UnknownTerm.S & D & UnknownTerm.S & D & R & scan digits \\
UnknownTerm.S & $\_$ & Unknown.S & $\_$ & R & exit variable term \\
UnknownTerm.S & \& & Check.Free & \& & R & next clause \\
UnknownTerm.S & \# & Fetch & \# & R & enter fetch phase \\
\hline
Skip.S & * & Skip.S & $ * $ & R & skip satisfied clause \\
Skip.S & $0$ & Skip.S & $\_$ & R & erase leading zero \\
Skip.S & D & SkipTerm.S & D & R & enter variable term \\
SkipTerm.S & D & SkipTerm.S & D & R & enter variable term \\
SkipTerm.S & \_ & Skip.S & \_ & R & exit variable term \\
Skip.Free & \& & Check.Free & \& & R & next clause \\
Skip.Free & \# & Fetch & \# & R & enter fetch phase \\
Skip.Forwarded & \& & Check.Forwarded & \& & R & forwarded next clause \\
Skip.Forwarded & \# & Accept & \# & R & all clauses satisfied \\
SkipTerm.Free & \& & Check.Free & \& & R & next clause \\
SkipTerm.Free & \# & Fetch & \# & R & enter fetch phase \\
SkipTerm.Forwarded & \& & Check.Forwarded & \& & R & forwarded next clause \\
SkipTerm.Forwarded & \# & Accept & \# & R & all clauses satisfied \\
\hline
Fetch & $\_$ & Fetch & $\_$ & R & bypass space delims \\
Fetch & T & Backward.T & $\_$ & L & fetch True value\\
Fetch & F & Backward.F & $\_$ & L & fetch False value\\
Fetch & * & Reject & $\_$ & L & reject \\

\hline
\end{tabular}
\caption*{Clause evaluation and fetch transitions.
Here $\_$ denotes a space delimiter and $D$ denotes a decimal digit.}
\end{table}

\begin{table}[!ht] \label{tab:fixed_transition_subtraction}
\centering \caption{Certificate-Oblivious Transition Function Part 2} \label{tab:sat_fixed_transitions_backward}
\begin{tabular}{llllll}
\textbf{State} & \textbf{Read} & \textbf{Next} & \textbf{Write} & \textbf{Move} & \textbf{Comment} \\
\hline
Backward.B & * & Backward.B & * & L & scan left \\
Backward.B & 1 & BackwardFrom1.B & 0 & L & subtract 1 \\
Backward.B & 0 & Borrow.B & 9 & L & borrow \\
Backward.B & D & BackwardInTerm.B & D$-1$ & L & decrement digit \\
Backward.B & $\epsilon$ & Check.Forwarded & $\epsilon$ & R & evaluation pahse \\
\hline
Borrow.B & 0 & Borrow.B & 9 & L & propagate borrow \\
Borrow.B & D & BackwardInTerm.B & D$-1$ & L & resolve borrow \\
\hline
BackwardInTerm.B & D & BackwardInTerm.B & D & L & scan left in variable \\
BackwardInTerm.B & $\_$ & Backward.B & $\_$ & L & exit literal(left) \\
BackwardInTerm.B & \& & Backward.B & \& & L & exit clause(left) \\
BackwardInTerm.B & $-$ & Backward.B & $-$ & L & exit lieteral(left) \\
BackwardInTerm.B & $\epsilon$ & Check.Forwarded & $\epsilon$ & R & evaluation phase \\
\hline
BackwardFrom1.B & D & BackwardInTerm.B & D & L & no borrow needed \\
BackwardFrom1.B & $\_$ & Assign.B & $\_$ & R & move for assignment \\
BackwardFrom1.B & $-$ & Assign.B & $-$ & R & move for assignment \\
BackwardFrom1.B & \& & Assign.B & \& & R &  move for assignment \\
BackwardFrom1.B & $\epsilon$ & Assign.B & $\epsilon$ & R & assignment site \\
\hline
Assign.B & 0 & Backward.B & B & L & finalize assignment \\
\hline
\end{tabular}
\caption*{Backward arithmetic and assignment transitions.
State suffix `.B` is instantiated to either `.T` or `.F` by symbolic lifting generally, and either `.1' or `.0' for \textsc{Borrow.B}}
\end{table}

\begin{remark}[Symbolic States and Symbols]
States suffixed by \texttt{.S}, \texttt{.B}, and \texttt{.D} denote scan modes, 
boolean variables, and parameterized digits, respectively. 
A transition labeled by $D$ applies uniformly to all decimal digits $0$--$9$, 
and a state suffixed by \texttt{.B} is instantiated to either \texttt{.T} 
or \texttt{.F}. This notation is purely representational; the induced 
Turing machine has a fixed, finite transition function.
\end{remark}
\begin{lemma}[Soundness of Acceptance]
\label{lem:tm_sat_fixed_soundness}
If the Turing machine $M_{\mathrm{SAT}}$ halts in the
\textsc{Accept} state on input $\Phi \#\mathcal{V}$, then the Boolean
assignment encoded by the certificate string $\mathcal{V}$ satisfies the CNF
formula $\varphi$.
\begin{proof}
The machine $M_{\mathrm{SAT}}$ enters the \textsc{Accept} state
if and only if it reaches the end-of-input marker $\#$ while remaining
in the \textsc{Skip.Forwarded} state.

By the transition rules, the state \textsc{Skip.Forwarded} can encounter 
`\#' only after a complete evaluation of the clause region, during which 
every clause has been verified to be \texttt{True} under the variable 
substitutions performed so far. Thus, acceptance implies that during 
the final forward scan, every clause has been fully determined, 
and no clause evaluates to \texttt{False}.

We examine how a clause is verified as \texttt{True}. When all former 
clauses have evaluated to \texttt{True}, a clause is marked satisfied 
(by entering \textsc{Skip.Forwarded}) if and only if at least one of its 
literals evaluates to \texttt{True}. A positive literal evaluates to 
\texttt{True} only when its variable has been replaced by \texttt{T}, 
and a negative literal evaluates to \texttt{True} only when its variable 
has been replaced by \texttt{F}. No clause can be marked satisfied without 
such a substitution.

Variable substitutions are produced exclusively during the \textsc{Fetch} 
phase. Each \textsc{Fetch} operation selects the next certificate symbol 
(\texttt{T} or \texttt{F}) and initiates a backward scan that decrements 
all remaining variable indices by one. When an index reaches zero, 
the corresponding variable is identified and replaced consistently 
throughout the clause region by the fetched Boolean value. 
If an index does not reach zero, no substitution occurs for that 
variable during the current iteration.

Because all variable counters are decremented by one in each
\textsc{Fetch} phase, the $i$-th variable is substituted precisely when
the $i$-th certificate symbol is fetched.
Thus, every substitution performed by the machine corresponds exactly
to the Boolean assignment encoded by the certificate string $\mathcal{V}$.

Since acceptance requires that every clause is satisfied during the
final forward scan, it follows that under the assignment specified by
$\mathcal{V}$, each clause of $\varphi$ evaluates to \texttt{True}.
Therefore, the certificate $\mathcal{V}$ satisfies the CNF formula $\varphi$.
\end{proof}
\end{lemma}

\begin{lemma}[Completeness of Certificate-Oblivious SAT Verifier]\label{lem:tm_sat_fixed_completeness}
Let $\Phi \# \mathcal{V}$ be a well-formed input such that the Boolean assignment
encoded by the certificate string $\mathcal{V}$ satisfies the CNF formula $\Phi$.
Then the Turing machine $M_{\mathrm{SAT}}$ does not enter the
\textsc{Reject} state and eventually halts in the \textsc{Accept} state.
\end{lemma}

\begin{proof}
Assume that the certificate $\mathcal{V}$ satisfies the formula $\varphi$.
We show that $M_{\mathrm{SAT}}$ cannot enter the
\textsc{Reject} state during its execution.

By inspection of the transition rules, the machine enters
\textsc{Reject} only if it encounters a clause whose literals have all
been evaluated and found to be false.
Unassigned variables, represented by the states
\textsc{Unknown.S} and \textsc{UnknownTerm.S} where
$\textsc{S} \in \{\textsc{Forwarded}, \textsc{Free}\}$, never trigger a rejection.

During the iteration of the fetch/assignment and evaluation phases, when a corresponding variable is identified and substituted, 
the satisfying literal causes the machine to transition into the \textsc{Skip.Free} or \textsc{Skip.Forwarded} state for that clause. 
Once a clause enters a \textsc{Skip.S} state, it keep \textsc{Skip} prefix for the remainder of the clause and cannot lead to a rejection.

Thus, the iteration of fetch/assignment and evaluation continues until every clause evaluates to \texttt{True}, 
as there is no clause that evaluates to \texttt{False} under the satisfying assignment. 
Furthermore, since $\mathcal{V}$ satisfies $\Phi$, every clause in $\Phi$ contains at least one literal that evaluates to \texttt{True} under the assignment encoded by $\mathcal{V}$.
Therefore, no clause can cause the machine to enter the \textsc{Reject} state.

After all variables are processed, the machine completes a forward scan
in the \textsc{Skip.Forwarded} state and reaches the end marker $\#$, at which point it enters \textsc{Accept}.

Hence, $M_{\mathrm{SAT}}$ halts in \textsc{Accept} whenever the certificate $\mathcal{V}$ satisfies the formula $\varphi$.
\end{proof}

A crucial feature of $M_{\mathrm{SAT}}$ is its symmetric transition structure regarding the certificate symbols $T$ and $F$. For any state $q$ and index $i$ in the \textsc{Fetch} or \textsc{Assignment} phases, the transition function $\delta(q, T)$ and $\delta(q, F)$ are defined to result in identical head movements and state transitions, differing only in the symbol written to the tape. This structural symmetry ensures that the computation trajectory—defined by the sequence of coordinate pairs $(\text{index}, \text{tier})$—remains invariant regardless of the specific Boolean values encoded in the certificate.

\begin{lemma}
The Turing machine $M_{\mathrm{SAT}}$ is a certificate-oblivious verifier for SAT if only `T', `F' are allowed as certificate symbols (via preprocessing).

\begin{proof}
Let $L_{\text{fixed}}$ be a fixed problem instance of SAT. Let $H_i(Y)$ denote the head position of $M_{\mathrm{SAT}}$ after the $i$-th transition, and $L_{i,j}(Y)$ denote the symbol at the $j$-th tape cell after the $i$-th transition, given a certificate $Y \in \mathcal{C}$ of length $m = |Y|$.

We categorize the state space $Q$ into three computational phases: the \textbf{Clause Evaluation phase} ($Q_{\text{eval}}$: states with prefixes $\{\texttt{Check}, \texttt{Skip}, \texttt{Unknown}\}$), the \textbf{Fetch phase} ($Q_{\text{fetch}}$: state $\{\texttt{Fetch}\}$), and the \textbf{Subtraction/Assignment phase} ($Q_{\text{sub}}$: states with prefixes $\{\texttt{Backward}, \texttt{Borrow}, \texttt{Assign}\}$).

We assert the following properties for all $i \ge 0$ and any $Y_1, Y_2 \in \mathcal{C}$ with $|Y_1| = |Y_2| = m$:
\begin{enumerate}
    \item \textbf{Head Invariance:} $H_i(Y_1) = H_i(Y_2)$.
    \item \textbf{Data Invariance:} $L_{i,j}(Y_1) = L_{i,j}(Y_2)$ whenever $L_{i,j}(Y_1), L_{i,j}(Y_2) \notin \{\text{`T', `F'}\}$.
    \item \textbf{Phase Invariance:} The machine is in the same phase for both $Y_1$ and $Y_2$.
\end{enumerate}

At the start of the computation, the head position $H_0(Y) = 0$ is uniform for all $Y$. The tape content $L_{0,j}$ is identical for all $j < |L_{\text{fixed}}|$. For the certificate region $|L_{\text{fixed}}| \le j < |L_{\text{fixed}}| + m$, $L_{0,j} \in \{\text{`T', `F'}\}$ holds by the assumption that the certificate alphabet is restricted to $\{T, F\}$. For the empty tape area ($j < 0$ or $j \ge |L_{\text{fixed}}| + m$), $L_{0,j} = \epsilon$. Thus, the assertions hold trivially for $i=0$.

Assume, for the sake of contradiction, that there exists a minimum index $i>0$ such that the obliviousness property is violated for two certificates $Y_1, Y_2$ with $|Y_1| = |Y_2| = m$. That is, either $H_i(Y_1) \neq H_i(Y_2)$, or there exists a position $j$ such that $L_{i,j}(Y_1) \neq L_{i,j}(Y_2)$ where $L_{i,j}(Y_1), L_{i,j}(Y_2) \notin \{\text{`T', `F'}\}$. We examine the transition at step $i$ for each phase:

\begin{itemize}
    \item \textbf{Fetch:} During this phase, the machine processes tape symbols. If $L_{i-1,j}(Y_1) = L_{i-1,j}(Y_2) = \text{`\_'}$, the transition moves the head as $H_i(Y_1) = H_{i-1}(Y_1) + 1 = H_{i-1}(Y_2) + 1 = H_i(Y_2)$. Conversely, if $L_{i-1,j}(Y_1) = L_{i-1,j}(Y_2) \in \{\text{`T', `F'}\}$, the transition updates the head position as $H_i(Y_1) = H_{i-1}(Y_1) - 1 = H_{i-1}(Y_2) - 1 = H_i(Y_2)$, and the machine transitions to the Subtraction/Assignment phase (\texttt{Backward.B} for $B \in \{\text{`T', `F'}\}$). Any other symbols, specifically `$\epsilon$', result in a rejection while moving the head in the left direction. Since the branch decision depends solely on the delimiter structure and not the specific truth values, this consistency contradicts the assumption of divergence at $i$.
	
	\item \textbf{Subtraction/Assignment:} Every transition is identical except for \texttt{Assign.B} state, independent of the state suffix category $.B$ (where $B \in \{\text{`T', `F'}\}$), as detailed in \cref{tab:fixed_transition_subtraction}. Since the transition rules do not depend on the state category $B$ except for the retention of the suffix, the head trajectory and tape updates remain identical for any state other than \texttt{Assign.B} state. In the \texttt{Assign.B} state, the tape symbol is updated to the corresponding value $B \in \{T, F\}$; however, this update consistently satisfies the data invariance property, as it is a deterministic result of the assignment phase. Furthermore, upon completion at step $i$, the machine deterministically transitions to the Clause Evaluation phase (\texttt{Check.Forwarded}) for any certificate. This consistency contradicts the minimality of $i$.
    
    \item \textbf{Clause Evaluation:} Transitions move right ($+1$). Thus, $H_i(Y_1) = H_{i-1}(Y_1) + 1 = H_{i-1}(Y_2) + 1 = H_i(Y_2)$. Regarding the tape content, every leading $0$ is replaced by $\_$ independent of the internal state (e.g., \texttt{Check}, \texttt{Skip}, \texttt{Unknown}, \texttt{CheckNot}, \texttt{UnknownNot}), thus $L_{i,j}(Y_1) = L_{i,j}(Y_2) = \text{`\_'}$. Since all other tape symbols are preserved, $L_{i,j}(Y_1) = L_{i,j}(Y_2)$ holds. Finally, upon encountering the delimiter '\#', the machine deterministically transitions to the Fetch phase (\texttt{Fetch}), unless the machine rejects. This consistency contradicts the minimality of $i$.
\end{itemize}

Since the trajectory $H_i(Y)$ is invariant and independent of the certificate content for a particular problem instance, no such minimal index $i$ can exist. Therefore, the execution graph satisfies the Grid-aligned Transition property.
\end{proof}
\end{lemma}

\begin{lemma}[Time and Space Complexity of Certificate-Oblivious SAT Verifier]
\label{lem:tm_sat_fixed_time}
Let $\Phi$ be a CNF formula of encoding length $n$. 
For a fixed certificate (truth assignment) of length $m < n$, the fixed-state SAT verifier Turing machine $M_{\mathrm{SAT}}$ runs in $\bigO(nm)$ time and uses $\bigO(n)$ space.

\begin{proof}
We analyze the complexity of the verifier by tracking the access frequency of each variable and literal occurrence.

For the time complexity, observe that the verifier $M_{\mathrm{SAT}}$ processes variables sequentially. Each variable is fetched exactly once and assigned to its occurrences via a backward scan.
This process costs $\bigO(n)$ time per variable, as the head moves forward over each tape cell only once for assignment before returning to the backward scan, totaling at most $3n$ transitions per variable iteration.
After a variable is fetched, $M_{\mathrm{SAT}}$ enters an evaluation phase 
during which it scans the clause region to update the status of each 
clause (satisfaction, falsification, or uncertainty) under the current 
partial assignment. During this phase, the head moves exclusively in the 
forward direction. Since the total encoding length is $\bigO(n)$, this evaluation phase costs $\bigO(n)$ time per variable. 
Crucially, once a variable is fetched, it is permanently replaced by the symbol `\spacedelim', ensuring that each variable is processed exactly once. 
Therefore, for $m$ variables, the total running time is bounded by $\bigO(n) \cdot m = \bigO(nm)$. 
Finally, if all variables are processed without satisfying the formula, the verifier halts and rejects, guaranteeing termination.

For the space complexity, note that the verifier never introduces new tape symbols beyond a constant-size alphabet. 
Both the backward fetch and forward evaluation phases operate within the initial tape boundaries of size $n + m + 1 = \bigO(n)$. 
The head position remains within a constant offset from the active regions of the input. 
Consequently, the total number of tape cells accessed is linear in the original input size, yielding a space complexity of $\bigO(n)$.
\end{proof}
\end{lemma}

\begin{theorem}[Polynomial-Time Verifier for SAT]\label{thm:3sat-verifier}
The Turing machine $M_{\mathrm{SAT}}$ defined in
\cref{def:tm_sat_fixed} is a correct NP verifier for the
\textsc{SAT} problem.

More precisely, for every input string $W$  encoding a CNF formula $\Phi$ of length $n$
 together with a Boolean certificate $\mathcal{V}$ of length $m<n$,
\begin{itemize}
    \item $M_{\mathrm{SAT}}$ accepts $W$ if and only if the certificate
    $\mathcal{V}$ evaluates the formula $\Phi$ to \textsc{True}, and
    \item $M_{\mathrm{SAT}}$ halts in $\bigO(n^2)$ time and uses $\bigO(n)$ tape
    space.
\end{itemize}
Hence, \textsc{SAT} admits a deterministic quadratic-time and linear-space
verifier.
\begin{proof}
The claim follows by combining soundness, completeness, termination, and the
resource bounds established previously.

\begin{itemize}
\item \textbf{Soundness.} By \cref{lem:tm_sat_fixed_soundness}, 
if $M_{\mathrm{SAT}}$ accepts an input $W$, then the certificate $\mathcal{V}$ 
evaluates the formula $\Phi$ to \textsc{True}. 
Thus, no invalid certificate is ever accepted.

\item \textbf{Completeness.} 
By \cref{lem:tm_sat_fixed_completeness}, 
for every certificate $\mathcal{V}$ that evaluates $\Phi$ to \textsc{True}, 
$M_{\mathrm{SAT}}$ deterministically reaches the accepting 
configuration. 
Thus, every valid certificate admits an accepting computation.

\item \textbf{Termination and resource bounds.} 
By \cref{lem:tm_sat_fixed_time}, the machine halts after finitely many steps and runs in $\bigO(nm)$ time using 
$\bigO(n)$ tape space. Since $m < n$, the time complexity is $\bigO(n^2)$. 
In particular, the verifier operates in deterministic polynomial time.
\end{itemize}

Combining the above arguments, $M_{\mathrm{SAT}}$ is a correct
polynomial-time verifier for the \textsc{SAT} problem.
\end{proof}
\end{theorem}

\paragraph{Conclusion.}
The total running time of $M_{\mathrm{SAT}}$ is $\bigO(n^2)$.
The machine uses only the original input tape and a constant number of auxiliary
symbols, and therefore requires $\bigO(n)$ space.

\subsection{Subset-Sum Turing Machine} \label{subsec:subset-sum-tm}
In this subsection, we describe the certificate-oblivious verifier Turing machine $M_{\mathrm{SS}}$ for the Subset-Sum problem.  
The machine $M_{\mathrm{SS}}$ employs a finite set of states and a finite tape alphabet, including standard digit symbols $\{0, 1, \dots, 9\}$, 
circled digits $\{\text{\circled{0}, \circled{1}, \dots, \circled{9}}\}$, and structural markers $\{\texttt{@}, \texttt{\#}, \texttt{|}, 
\texttt{\_}, \texttt{x}$. For the formal definition of transition rules, we use the symbols $D$ and $M$ as parameters representing arbitrary digits 
to describe the matching and subtraction phases concisely.
 The machine receives a target integer $T$ and a list of integers $a_1,a_2,\dots,a_n$ as an input, 
all encoded in decimal.  It then determines whether a subset within the certificate area, excluding those masked by 'x', sums to $T$.
The semantics are specified at the level of state transitions and head movements, 
together with the meaning of each state in the computational workflow.

The input format is
\[
1\_@1\_3\_5\_7\_10\_20\#x\_3\_5\_7\_xx\_xx
\]
where the left-hand side specifies the target sum,
the segment between ``@'' and ``\#'' encodes the set of input elements,
and the segment to the right of ``\#'' represents a proposed certificate.

As previously discussed, we assume the input is well-formed for the 
purposes of this subsection. Specifically, the target sum is 
delimited by an underscore (\texttt{\_}). The order of certificate element is
the same as that of input set element except that unused number is represented by 
repetition of  'x' with the same length. Thus, the 
length of the certificate is exactly same as the encoded set, any other length of certificate would be rejected.

The redundant x's are inserted as dummy symbol to preserve certificate-oblivious property .

The Turing machine employs auxiliary tape symbols, such as 
$\sim$, $|$, $\$$, and $\text{\circled{0}, \dots, \circled{9}}$, to support 
digit matching, erasing, and subtraction, acrting as positional or erasure markers. Its operation interleaves 
digit matching with subtraction: starting from the most significant 
digit of the certificate, the machine matches the digit or $x$ at the exactly 
same location in the input set region,  the digit matched $x$ replaced by $0$, and the delimiter also matched to the delimiter in the same way.
Once a complete whole certificate is matched, the head round trip between input set region and the sum area to perform 
digit-wise subtraction corresponding to the matched element.

This design separately performs matching and subtraction.

\paragraph{Tape Structure.}
The tape consists of the following regions:
\begin{enumerate}
    \item a \emph{target sum region}, initially holding $T$ and used as a subtractive workspace
        in which selected elements $a_i$ are successively subtracted according to the given certificate after the entire matching process.
  
	\item the \emph{input set region}, containing the decimal encodings of
	$a_1,\dots,a_n$ separated by delimiters and  bounded by the marker `@' on the left and `\#' on the right.

    \item a \emph{certificate region}, in which the machine records the
        certificate symbols encoding a candidate solution, specifying
        which elements $a_i$ are selected in the subset; the unselected elements are represented by the same length of x's.

    \item a right-unbounded blank region; the tape is also left-unbounded, consisting entirely of blank symbols.
\end{enumerate}

The certificate region is intentionally designed to store a general
symbols over a finite alphabet, rather than a Boolean vector.
While certificates for SAT naturally consist of truth values,
the Subset-Sum construction demonstrates that the  NP verifier simulation framework
and the associated Turing machines operate uniformly over arbitrary
symbolic certificates, such as sequences of integers or digit encodings.

\newcommand{\q}[1]{\text{`$#1$'}}

We use $D$ to denote a generic digit from $\{0, \dots, 9\}$, and we use the circled notation \circled{D} to denote its corresponding symbol in the machine’s internal alphabet, used as a positional marker. 
Thus, \circled{D} represents one of \circled{0}, \circled{1}, \dots, \circled{9}.

\paragraph{Symbol conventions}
The machine relies on several delimiter and auxiliary symbols whose semantics are summarized as follows:
\begin{itemize}
	\item The delimiter \texttt{\_} serves as an element-level delimiter. In the certificate region, it marks the end of a certificate element and triggers number-level match verification. 
		In the target-sum region, a trailing \texttt{\_} guarantees a well-defined least significant digit for subtraction.

	\item During processing, an underscore \texttt{\_} may be temporarily replaced by \texttt{|} as a positional boundary marker, indicating that all digits and delimiters up to this marker have been matched or processed.

    \item The symbol \texttt{x} denotes the placeholder for unused set element for the certificate oblivious property.

    \item The symbol \texttt{\$} is used as a subtraction positional
        marker, tracking the digit location currently involved in
        subtraction.

    \item The symbol \texttt{$\sim$} represents a deleted character in the
        certificate region, indicating that the corresponding digit has
        already been consumed and should be ignored in subsequent scans.

    \item Marked digits \circled{D} denote digits that have been
        matched or processed.  In the matching phase, they record
        positional correspondence as temporary processed boundary marker; in the subtraction phase, they
        identify the most recently processed digit .  All circled digits are restored to their original
        values once the corresponding operation is completed.
\end{itemize}

\paragraph{High-Level Semantics.}
The machine operates as a certificate-driven verifier for the Subset-Sum problem, 
utilizing deterministic, fixed-length certificates for each problem instance. 
Digit matching and arithmetic operations are coordinated through explicit verification stages, ensuring that subtraction is performed only after the entire matching phase has been successfully completed.

Conceptually, the execution proceeds through the following stages, each implemented as a structured block of states within the transition table:

\begin{enumerate}
    \item \textbf{Forward scan.}
    The machine scans the tape from left to right until it reaches 
    the delimiter \texttt{\#}, which separates the problem instance from the
    certificate region.

    \item \textbf{Selection of a certificate digit.}
	In the state \textsc{FindDigitToMatch}, the machine scans the certificate region to locate the next unmatched digit, proceeding from left to right. Once such a digit is found, it is erased by replacing it with the deletion marker $\sim$, 
	and the machine initiates a backward digit-matching procedure upon encountering the blank symbol $\epsilon$

	\item \textbf{Digit/Delim matching.}
	The states \textsc{BackwardToMatch.M} and \textsc{MatchPosition.M} perform digit-wise matching within the input set region to the right of \texttt{@}. 
	Specifically, \textsc{BackwardToMatch.M} searches for a previously matched digit and restores it to its original value. 
	If the corresponding digit is $M$, the state \textsc{MatchPosition.M} temporarily marks the newly matched digit (e.g., by replacing it with \text{\circled{D}}) to record the positional correspondence, indicating that the prefix has been successfully matched. 
	Delimiters are matched in a similar manner using the states \textsc{BackwardToMatch.Delim}, \textsc{MatchPosition.Delim}, and the temporary marker $\vert{}$.
	To satisfy the certificate-oblivious property, unused elements in the certificate are represented by 'x'. The states \textsc{BackwardToMatch.X} and \textsc{MatchPosition.X} perform the corresponding positional matching by replacing the digit with $0$.

	\item \textbf{Subtraction after entier matching.}
	Only after all elements have been successfully matched does the machine enter the subtraction phase. 
	The states \textsc{SumArea.M}, \textsc{Subtract.M}, and \textsc{Borrow.B} perform digit-wise subtraction of the matched element from the target-sum region. 
	This process proceeds from the least significant digit (LSD), utilizing the positional marker \texttt{\$} and the circled digit for alignment. 
	Crucially, the trailing delimiter `\_' in the target-sum region immediately preceding \texttt{@} serves as the LSD marker for the subtraction. 
	The subtraction propagates across all digits until $\epsilon$ is encountered, a requirement of the certificate-oblivious property. 
	Upon completion, the matched element is erased, and all temporary markers are restored to their original symbols by the state \textsc{BackwardToRestore}.

	\item \textbf{Final sum verification.}
	After all elements have been matched and the corresponding subtractions have been performed, the machine scans the target-sum region to check whether the remaining digits are exactly zeros. 
	If only zeros (or their marked equivalents) remain, the machine enters the accepting state; otherwise, it rejects.

\end{enumerate}

\paragraph{Invariant}
After the completion of the \(j\)-th subtraction phase,
the target-sum region encodes
\[
    S - \sum_{i=1}^{j} a_{\pi(i)},
\]
where each \(a_{\pi(i)}\) is the unique input-set element whose digit sequence has been fully matched and subsequently subtracted
during the \(i\)-th subtraction phase.

In particular, the value stored in the target-sum region changes only
upon the completion of a subtraction phase and remains unchanged throughout the preceding matching phases.

\begin{definition}[Subset-Sum Verifier Turing Machine] \label{def:sum_of_subset_verifier_tm}

The deterministic single-tape Turing machine
\[
M_{\mathrm{SS}}
=
(Q,\Sigma,\Gamma,\delta,\qinit,\qacc,\qrej)
\]
is defined as follows.

\begin{itemize}
\item
$Q$ is a finite set of states consisting of all states explicitly appearing
in the transition table in \cref{tab:tm_subsetsum_transitions}, including
parameterized state families such as $\textproc{Backward.M}$,
$\textproc{MatchPosition.M}$, $\textproc{Subtract.M}$, and $\textproc{Borrow.B}$.
Each parameterized family is instantiated over a finite domain
($M,D \in \{0,\dots,9\}$, and therefore contributes only a constant number of
states.

\item
$\Sigma$ is the input alphabet, consisting of decimal digits
$\{0,\dots,9\}$ and delimiter symbols
$\texttt{\_}$, \texttt{@}, \texttt{\#}, and \texttt{x}.

\item
$\Gamma$ is the tape alphabet, extending $\Sigma$ with auxiliary marker
symbols used during verification, including
$\sim$, $|$, $\$$, \circled{0}$,\dots,$\circled{9},  and blank symbol $\epsilon$.

\item
$\delta$ is the deterministic transition function defined explicitly by the
rules in \cref{tab:tm_subsetsum_transitions}.
It implements digit matching, marker-based bookkeeping, and digit-wise
subtraction with borrow propagation.

\item
$\qinit = \textproc{Forward}$ is the initial state, in which the
machine begins by scanning the input from left to right.

\item
$\qacc = \textproc{Accept}$ is the accepting halting state.

\item
$\qrej = \textproc{Reject}$ is the rejecting halting state.
\end{itemize}
\end{definition} 
No sanitization is performed on the problem instance area of the tape
(the segment encoding the target sum and the element list).
Consequently, repeated elements are processed independently,
and the verifier therefore accepts instances in which the element list
represents a multiset rather than a set.

\begin{table}[H]
\centering
\small
\begin{tabular}{llllll}
\hline
\textbf{State} & \textbf{Read} & \textbf{Write} & \textbf{Move} & \textbf{Next State} & \textbf{Comment} \\
\hline
Forward & \# & \# & +1 & FindDigitToMatch & Enter certificate area \\
Forward & * & * & +1 & Forward & Scan input \\
FindDigitToMatch & $\sim$ & $\sim$ & +1 & FindDigitToMatch & Skip deleted symbols \\
FindDigitToMatch & M & $\sim$ & -1 & BackwardToMatch.M & Select M digit to match \\
FindDigitToMatch & \_ & $\sim$ & -1 & BackwardToMatch.Delim & Select Delim to match \\
FindDigitToMatch & x & $\sim$ & -1 & BackwardToMatch.X & Select X digit to match \\
FindDigitToMatch & * & $\sim$ & -1 & Reject & Invalid symbol \\
FindDigitToMatch & $\epsilon$ & $\epsilon$ & -1 & BackwardAfterMatching & End of certificate \\
BackwardToMatch.M & * & * & -1 & BackwardToMatch.M & Scan left \\
BackwardToMatch.M & | & \_ & +1 & MatchPosition.M & Found candidate \\
BackwardToMatch.M & \circled{D} & D & +1 & MatchPosition.M & Found matched \\
BackwardToMatch.M & x & \_ & +1 & Reject & Conflict \\
BackwardToMatch.M & @ & @ & +1 & MatchPosition.M & End of set \\
BackwardToMatch.Delim & * & * & -1 & BackwardToMatch.Delim & Scan left \\
BackwardToMatch.Delim & | & \_ & +1 & MatchPosition.Delim & Found candidate \\
BackwardToMatch.Delim & \circled{D} & D & +1 & MatchPosition.Delim & Found matched \\
BackwardToMatch.Delim & x & 0 & +1 & MatchPosition.Delim & Mark 0 \\
BackwardToMatch.Delim & @ & @ & +1 & MatchPosition.Delim & End of set \\
BackwardToMatch.X & * & * & -1 & BackwardToMatch.X & Scan left \\
BackwardToMatch.X & | & \_ & +1 & MatchPosition.X & Found candidate \\
BackwardToMatch.X & \circled{D} & \_ & +1 & Reject & Conflict \\
BackwardToMatch.X & x & 0 & +1 & MatchPosition.X & Match X \\
BackwardToMatch.X & @ & @ & +1 & MatchPosition.X & End of set \\
MatchPosition.M & M & \circled{M} & +1 & Forward & Mark digit \\
MatchPosition.Delim & \_ & | & +1 & Forward & Mark delimiter \\
MatchPosition.X & D & x & +1 & Forward & Mark X \\
BackwardAfterMatching & * & * & -1 & BackwardAfterMatching & Skip \\
BackwardAfterMatching & \circled{M} & \$ & -1 & BackwardToSubtract.M & Prepare subtract \\
BackwardAfterMatching & x & \$ & -1 & BackwardToSubtract.0 & Prepare subtract \\
SbtractionDigit & M & \$ & -1 & BackwardToSubtract.M & Prepare subtract \\
SbtractionDigit & \_ & \$ & -1 & BackwardToRestore & Prepare restore \\
SbtractionDigit & @ & @ & -1 & CheckSum & Check result \\
BackwardToSubtract.M & * & * & -1 & BackwardToSubtract.M & Scan left \\
BackwardToSubtract.M & @ & @ & -1 & SumArea.M & Go to sum area \\
BackwardToRestore & \circled{D} & D & -1 & BackwardToRestore & Restore digit \\
BackwardToRestore & | & \_ & -1 & BackwardToRestore & Restore separator \\
BackwardToRestore & * & * & -1 & BackwardToRestore & Skip \\
BackwardToRestore & $\epsilon$ & $\epsilon$ & +1 & ForwardToSubtract & Return to forward \\
SumArea.M & D & D & -1 & SumArea.M & Traverse sum \\
SumArea.M & | & | & -1 & SumArea.M & Skip \\
SumArea.M & \_ & | & -1 & Subtract.M & LSD subtraction \\
SumArea.M & \circled{D} & D & -1 & Subtract.M & LSD subtraction \\
Subtract.M & D & \circled{D}-\circled{M} & -1 & Borrow.B & Subtract digit \\
Borrow.0 & D & D & -1 & Borrow.0 & No borrow \\
Borrow.0 & $\epsilon$ & $\epsilon$ & +1 & ForwardToSubtract & Return \\
Borrow.1 & 0 & 9 & -1 & Borrow.1 & Propagate borrow \\
Borrow.1 & D & D-1 & -1 & Borrow.0 & Resolve borrow \\
Borrow.1 & $\epsilon$ & \_ & +1 & Reject & Negative result \\
CheckSum & | & \_ & -1 & CheckSum & Skip \\
CheckSum & 0 & \_ & -1 & CheckSum & Skip \\
CheckSum & \circled{0} & \_ & -1 & CheckSum & Skip \\
CheckSum & $\epsilon$ & \_ & +1 & Accept & Success \\
CheckSum & * & \_ & -1 & Reject & Nonzero remain \\
\hline
\end{tabular}
\caption{Transition table for the Certificate-Oblivious Subset-Sum TM $M_{\mathrm{SS}}$.} 
\label{tab:tm_subsetsum_transitions} 
\end{table}
\paragraph{Remark on Symbol Conventions and Matching Semantics.}
The transition table employs schematic symbols to compactly represent
families of transitions. Their meanings and application order are defined
as follows.

\begin{enumerate}
    \item \textbf{Digit symbols.}
    The symbol $D$ denotes an arbitrary decimal digit in $\{0,\dots,9\}$.
    The symbol $\circled{D}$ denotes the corresponding circled digit
    $\{\circled{0},\dots,\circled{9}\}$, used for temporary bookkeeping
    during matching and subtraction.
    Within a single transition rule, repeated occurrences of $D$
    (or $\circled{D}$) always refer to the \emph{same concrete digit}
    read from the tape.

    \item \textbf{Matched digit symbol $M$.}
    The symbol $M$ denotes the specific digit currently stored in the
    parameterized state (i.e., the digit selected for matching from the
    certificate).
    Its marked counterpart $\circled{M}$ denotes the marked version
    of that same digit.
    Whenever both $M$ and $D$ are applicable to the current tape symbol,
    \emph{$M$ takes precedence over $D$}, enforcing exact digit matching
    before generic digit handling.

    \item \textbf{Subtraction notation $\circled{D}-\circled{M}$}
    The schematic write symbol $\circled{D}-\circled{M}$ denotes
    digit-wise subtraction modulo $10$.
    If the subtraction requires borrowing, the machine transitions into
    a borrow state $B\in\{0,1\}$, where $B=1$ indicates an active borrow to be
    propagated to the next more significant digit.
    The resulting digit is written in marked form, and the borrow value is
    encoded in the subsequent control state.

    \item \textbf{Borrow parameter $B$.}
    States of the form $\textproc{Borrow}.B$ encode whether a borrow is
    pending ($B=1$) or resolved ($B=0$).
    Borrow propagation proceeds leftward until the blank symbol is encountered, 
    when underflow is the machine rejects.

    \item \textbf{Wildcard symbol $*$.}
    The symbol $*$ denotes a fallback transition applicable to any tape
    symbol not matched by a more specific rule.
    Wildcard transitions are considered \emph{only after} attempting
    matches for concrete symbols, $M$, $\circled{M}$, $D$, or
    $\circled{D}$.
    Writing $*$ leaves the tape symbol unchanged.

    \item \textbf{Matching priority.}
    For a given state and tape symbol, transitions are applied in the
    following order of specificity:
    \[
        \text{concrete symbol}
        \;>\;
        M
        \;>\;
        D
        \;>\;
        *
    \]
    This priority order guarantees that exact digit matches are handled
    before generic digit transitions, and that wildcard rules serve strictly
    as a fallback mechanism.
\end{enumerate}

\begin{lemma}[Token Identification and Single-Use Correspondence]
\label{lem:token-identification}
During any accepting computation of $M_{\mathrm{SS}}$,
each confirmed certificate element corresponds to exactly one uniquely
determined delimiter-bounded input-set element.

Let $c_j$ be the $j$-th delimiter-bounded integer encoded in the
certificate region that reaches the confirmed matching phase.
Then it always matches input-set element $a_j$ such that:

\begin{enumerate}
    \item Starting from the most significant digit, every digit of $c_j$
    is matched sequentially with the digits of $a_j$.
    \item The match terminates exactly at the delimiter of both tokens,
    establishing equality of the encoded integers $c_j=a_j$ rather than
    prefix agreement.
    \item The subtraction phase is executed over the same digit interval,
    from least significant digit to most significant digit.
    \item After completion, all digits of $a_j$ are erased and can never
    participate in future matches.
\end{enumerate}

Consequently, each certificate element $c_j$ is used exactly once and
corresponds to exactly one input-set element $a_j$, and no input element
can be reused in later phases.
\end{lemma}

\begin{proof}
Digits of a certificate element $b_j$ are processed in the state
\textsc{FindDigitToMatch} from the most significant to the least significant digit.
Backward matching tentatively aligns each digit with a candidate input token,
and the machine enters the next element matching only after all digits and delimiter token have been processed.

The confirmation scan verifies that all matched digits belong to a single
delimiter-bounded input-set element and that both token boundaries are reached.
Therefore equality $b_j=a_j$ is established before subtraction begins.

The subtraction routine operates on the same digit interval in reverse
direction to perform borrow propagation.
After subtraction, the digits are erased, preventing reuse.

Thus each subtraction corresponds to exactly one complete input-set element
and to exactly one certificate element.
\end{proof}

\begin{lemma}[Soundness of Certificate-Oblivious Subset-Sum Verifier]\label{lem:soundness-subsetsum}
Let $W$ be an input string encoded in the Subset-Sum format
\[
W = S\_@a_1\_a_2\_\cdots\_a_k\_\# c_1\_c_2\_\cdots\_c_k,
\]
where $S, a_i$ are nonnegative integers represented in decimal and $c_i$ is a nonnegative integer or a repetition of x's.
Let $b_j$ be th j-th integer in sequence $(c_1, \cdots, c_k)$. 
If $M_{\mathrm{SS}}$ enters the accepting configuration on input $W$,
then
\[
\{b_1,\dots,b_m\} \subseteq \{a_1,\cdots,a_k\}
\quad\text{and}\quad
\sum_{j=1}^m b_j = S.
\]
\end{lemma}

\begin{proof}
Assume that $M_{\mathrm{SS}}$ enters the accepting configuration
on input $W$.
We show that this implies
\[
\{b_1,\dots,b_m\} \subseteq \{a_1,\cdots,a_k\}
\quad\text{and}\quad
\sum_{j=1}^m b_j = S.
\]

\begin{itemize}
\item[(1)] \textbf{Acceptance implies zero target sum.}
The machine enters the accepting state only from the \textsc{CheckSum} state.
By construction of the transition rules in \textsc{CheckSum},
acceptance is possible only if the target-sum region $T$ wholely contains  zero(`0') digits.
Hence, at acceptance, the tape encodes
\[
T = 0.
\]

\item[(2)] \textbf{Zero target sum implies complete and correct subtraction.}
The target-sum region is initialized to $S$. The subtraction phase is triggered by the \textsc{BackwardToSubtract} state following \textsc{BackwardAfterMatching}. 
Upon entering this state at the LSD of a matched element, each digit is subtracted iteratively from the LSD to the MSD. 
This process is orchestrated by the states \textsc{Subtract.M} and \textsc{Borrow.B}, which return control to \textsc{BackwardToSubtract} after restoring the symbols once the operation is complete. 
Any unresolved borrow, digit underflow, or malformed subtraction during this phase forces an immediate transition to the rejecting state
Let $s_j$ denote the numerical value subtracted during the $j$-th subtraction phase. 
Reaching the \textsc{CheckSum} state with the target-sum region being exactly zero implies that every subtraction phase was executed successfully without rejection. 
Therefore, the final state of the tape satisfies:
\[
T = S - \sum_{j=1}^m s_j = 0, \quad \text{which implies} \quad S = \sum_{j=1}^m s_j.
\]

\smallskip
\item[(3)] \textbf{Successful subtraction implies matching of an input-set element.}
A subtraction phase is initiated only after the machine confirms that a 
digit-matching process terminated. Specifically, subtraction is triggered only after the matching  
phase (via \textsc{BackwardAfterMatching} indicating that every element has been matched. 

\item[(4)] \textbf{Element matching requires complete digit-wise correspondence.}
For each certificate element, the machine processes its digits from the most significant digit (MSD) to the least significant digit (LSD) in the state \textsc{FindDigitToMatch}. 
Each selected digit is tentatively marked within the corresponding input-set element, while the previously matched digit is restored to its original value by the backward matching procedure. 
Only after all digits of the delimiter-bounded token have been processed does the machine transition to the states \textsc{FindDigitToMatch}. 
A partially matched element leads the machine to reject. 
This ensures that the $j$-th successfully matched element $b_j$ to $i$-th input elements corresponds exactly to the $i$-th certificate element $c_i$; otherwise, the machine rejects. 
Furthermore, this mechanism ensures that $a_i$ is replaced by $0$'s when $c_i$ is a repetition of x's so that such $a_i$ does not contribute to sum 
implying that each $b_j$ is associated with a unique input-set element $a_{i_j}$ and ensuring $\{b_1, \dots, b_m\} \subseteq \{a_1, \dots, a_k\}$. 

After this phase, the subtraction phase is triggered only if every digit of the certificate element has been successfully matched, leaving the LSD of the last element properly marked. 

Since the value subtracted from the target-sum region is $s_j = a_{i_j}$, and the matching process guarantees $a_{i_j} = b_j$, it follows that $s_j = b_j$. 
Thus, each subtraction step corresponds exactly to one delimiter-bounded input-set element as per \cref{lem:token-identification}. 
Consequently, we conclude that $S = \sum_{j=1}^m b_j$ and $\{b_1, \dots, b_m\} \subseteq \{a_1, \dots, a_k\}$.

\end{itemize}
By combining stages (1) through (4), we establish that acceptance implies the existence of a subset of the input—excluding the masked elements—that sums exactly to $S$, thus completing the proof.
\end{proof}

\begin{lemma}[Completeness of Certificate-Oblivious Subset-Sum Verifier]\label{lem:completeness-subsetsum}
Let $W$ be an input string encoded in the Subset-Sum format:
\[
W = S\_@a_1\_a_2\_\cdots\_a_k \# c_1\_c_2\_\cdots\_c_k,
\]
where $S, a_i$ are nonnegative integers represented in decimal and $c_i$ is a nonnegative integer or a repetition of x's.
Let $b_j$ be th j-th integer in sequence $(c_1, \cdots, c_k)$. 
If $\{b_1, \dots, b_m\} \subseteq \{a_1, \dots, a_k\}$ and 
\[
\sum_{j=1}^m b_j = S,
\]
then $M_{\mathrm{SS}}$ enters the accepting configuration on input $W$.
\end{lemma}

\begin{proof}

The machine processes the certificate element sequentially, and within each
element matches digits in the order determined by the tape layout, starting
from the most significant digit.

\begin{itemize}
\item[(1)] \textbf{Successful element matching for each certificate element.}
In state \textsc{Forward}, the machine reaches the beginning of the next unprocessed certificate element and transitions to \textsc{FindDigitToMatch}. 
For each digit of the current certificate element, a corresponding digit exists in the input-set region at the same relative position, guaranteed by the validity of the certificate. 
State \textsc{BackwardToMatch.M} deterministically locates the matched digit and circles it, while simultaneously restoring the previously matched digit.
The boundary of each element is confirmed by $\textsc{BackwardToMatch.Delim}$ and $\textsc{MatchPosition.Delim}$.
Thus, $i$-th certificate is positioned at the exactly $i$-th input element and $j$-th matched element is identical to some element of $\{a_1, \dots, a_k\}$.

\item[(2)] \textbf{Replacing unused element with $0$ value}
When \textsc{FindDigitToMatch} encounters symbol $x$, the machine matches any digit to this symbol using temporary marker $x$ using $\textsc{BackwardToMatch.X}$ and $\textsc{MatchPosition.X}$
and the following matching process replacing $x$ with $0$.
This phase confirms that any unused element is replaced by $0$ value so that it cannot contribute to total sum.
 
\item[(3)] \textbf{Element-wise subtraction corresponding to the certificate.}
For each successfully confirmed element, the machine enters the subtraction states \textsc{BackwardToSubtract}.
Using the states \textsc{SumArea.M}, \textsc{Subtract.M}, and \textsc{Borrow.$B$},
the subtraction is performed from LSD to MSD, and the subtracted digit is erased. 
The digits marked during the matching phase encode exactly the value of the
current matched element $b_j$, and the subtraction performed corresponds
precisely to subtracting $b_j$ from the target-sum region.
Because
\[
\sum_{j=1}^m b_j = S,
\]
All borrow propagations terminate successfully and no underflow occurs; otherwise, the input is rejected. After processing all elements, the target-sum region contains only zeros.

\item[(4)] \textbf{Final zero check and acceptance.}
Upon reaching the end of the subtraction, the machine enters
\textsc{CheckSum}.
Since the target-sum region encodes only zeros, the machine reaches the accepting
state.
\end{itemize}

Hence, for every valid certificate, $M_{\mathrm{SS}}$ never enters a
rejecting state and eventually accepts, establishing completeness.
\end{proof}

We analyze the time and space complexity of $M_{\mathrm{SS}}$ as a function of the Subset-Sum problem instance length, 
rather than the total input length of the verifier TM. Let $n$ denote the length of the initial tape contents—including the target sum, 
the input-set encoding, and delimiters—excluding the certificate. We then assume that the certificate length is at most $n$

\begin{lemma}[Time and Space Complexity of Certificate-Oblivious Subset-Sum Verifier] \label{lem:time-space-subsetsum}
Let $n$ denote the length of the initial tape contents (including the target sum, the input set, and delimiters). 
Assuming the certificate length is at most $n$, the verifier machine $M_{\mathrm{SS}}$ defined in \cref{def:sum_of_subset_verifier_tm} halts in $ \bigO(n^2)$ time and uses $ \bigO(n)$ tape space.
\end{lemma}

\begin{proof}
We bound the time and space complexity separately.

\smallskip
\noindent\textbf{Time complexity.} 
The machine operates through a sequence of deterministic matching and subtraction phases, each implemented using linear tape scans, marker-based position recovery, and irreversible digit erasure. 
While individual scans are linear ($\bigO(n)$), certain phases are repeated proportional to the input length, resulting in quadratic total time.

\begin{enumerate}
\item \textbf{Forward scan and digit-to-match search.}
Before locating a matching certificate symbol, the machine performs a left-to-right scan until it reaches the `\#' delimiter. Each scan requires $\bigO(n)$ steps.
Since this is performed for each of the at most $n$ digits in the certificate, the cumulative time is $\bigO(n^2)$.
In the state \textsc{FindDigitToMatch}, the machine scans the certificate region to locate the most-significant unmatched digit. Since at most $n$ digits are processed, the total time for this phase is $\bigO(n^2)$.

\item \textbf{Backward matching and element verification.}
For each selected certificate symbol (digit, delimiter, or x), the machine scans leftward to locate a corresponding element in the input-set region. 
Each scan takes $\bigO(n)$, and since at most $n$ such operations are performed, the total time is $\bigO(n^2)$.

\item \textbf{Subtraction phase.}
Digit-wise subtraction with borrow propagation involves repeated scans between the matched element and the target-sum region. 
Each subtraction step requires $\bigO(n)$ time. With at most $n$ digits to be subtracted, the total time spent in this phase is $\bigO(n^2)$.

\item \textbf{Final sum check.}
A final linear scan of the target-sum region verifies that all digits are zero, requiring a single $\bigO(n)$ pass.
\end{enumerate}

Summing these contributions, the total running time is $\bigO(n^2) + \bigO(n^2) + \bigO(n^2) + \bigO(n) = \bigO(n^2)$.

\smallskip
\noindent\textbf{Space complexity.}
The machine never extends the non-blank area of the tape beyond the cells initially occupied by the input and the provided certificate. 
All auxiliary information is managed by overwriting existing symbols, and the head reverses its direction whenever it meets a blank symbol as a boundary, preventing the head from moving beyond those boundary blank cells. 
Thus, the total number of tape cells used, including the two boundary cells, is $\bigO(n)$.

\smallskip
Therefore, $M_{\mathrm{SS}}$ operates in quadratic time and linear space relative to the input length.
\end{proof}

The Subset-Sum verifier Turing machine $M_{\mathrm{SS}}$ therefore satisfies both soundness and completeness 
and runs in polynomial time with respect to the input size of corresponding Subset-Sum problem instance, making it suitable 
for integration into the NP verifier simulation framework described in the next 
section.

\begin{lemma}[Certificate Obliviousness of Subset-Sum Verifier]
\label{lem:oblivious-computational-phases}
Let $M_{\mathrm{SS}}$ be the Subset-Sum verifier Turing machine. For any fixed problem instance $L_{\text{fixed}}$ and certificate length $m = |Y|$, the head trajectory $H_i(Y)$ is determined solely by the iteration count $k$ and the problem instance $L_{\text{fixed}}$, and is independent of the certificate content $Y$.
\end{lemma}

\begin{proof}
The state space $Q$ is partitioned into two disjoint phases, $P_M$ (Matching) and $P_S$ (Subtraction), each further divided into forward ($S_{*, fwd}$) and backward ($S_{*, bwd}$) subphases:

\begin{itemize}
    \item \textbf{Matching Phase ($P_M = S_{M, fwd} \cup S_{M, bwd}$):}
    \begin{itemize}
        \item $S_{M, fwd}$ : the state \texttt{Forward} and the states with prefix in $\{\texttt{FindDigitToMatch, MatchPosition}\}$
        \item $S_{M, bwd}$ : the state in $\{\texttt{BackwardToMatch}\}$
    \end{itemize}
    \item \textbf{Intermediate Phase ($P_I$)}: the state \texttt{BackwardAferMatching}
    \item \textbf{Subtraction Phase ($P_S = S_{S, fwd} \cup S_{S, bwd}$):}
    \begin{itemize}
        \item $S_{S, fwd}$ : the states with the prefix in $\{\texttt{ForwardToSubtract, SumArea}\}$
        \item $S_{S, bwd}$ : the states with the prefix  in $\{\texttt{BackwardToSub}$, \texttt{BackwardToRes}, $\texttt{Subtract}, \texttt{Borrow}$, $\texttt{CheckSum}\}$
    \end{itemize}
\end{itemize}

For any certificate $Y$ of length $m$, the machine performs exactly $m$ iterations in each phase.

\begin{itemize}
    \item \textbf{Matching Phase ($P_M$):} At each iteration $k \in \{1, \dots, m\}$, the machine performs a \textit{Forward Subphase} and a \textit{Backward Subphase}. During the Forward subphase, the head moves from the fixed left marker $p_l^0$ (@) to the current fetch position $p_r^0 + k$. Since $k$ is determined strictly by the iteration count, the head trajectory is invariant for all $Y$. In the Backward subphase, the machine scans to the marker positions (circled digits, $|, x$) and shifts the marker index by $+1$. Because these marker shifts depend only on the delimiter structure—which is fixed for a given certificate length $m$—the turnaround positions remain independent of the specific digit values stored in the certificate.

	\item \textbf{Intermediate Phase ($P_{int} = \{ \texttt{BackwardAfterMatching} \}$):}
    This phase acts as a bridge, triggered upon encountering the right delimiter `$\epsilon$' in the Matching Phase. The head scans leftward to the designated marker, which is then replaced by the symbol `\$' to initiate the Subtraction Phase.

    \item \textbf{Subtraction Phase ($P_S$):} At each iteration $k \in \{1, \dots, m\}$, the machine employs the symbol `\$' as a dynamic right marker. In the Forward subphase, the head scans rightward up to the current position of `\$' (at $p_r^0 - k$), and in the Backward subphase, it scans leftward until the boundary `$\epsilon$'. Both the forward turnaround point (`\$') and the backward turnaround point ($\epsilon$) are updated deterministically as a function of the iteration $k$. Consequently, the head trajectory is invariant regardless of the digit values being processed.

    \item \textbf{Global Invariance:} Since the transition between subphases is triggered by reaching fixed spatial coordinates (markers $p_l, p_r, \epsilon, \$ $) rather than by the content of the certificate, the head position at any time $t$ is defined by a function $H_t = f(L_{\text{fixed}}, m, t)$. Thus, $H_t(Y_1) = H_t(Y_2)$ for all $Y_1, Y_2$ such that $|Y_1| = |Y_2| = m$.
\end{itemize}

Since the head trajectory $H_i(Y)$ is invariant for all certificates of the same length, the execution graph is certificate-oblivious, satisfying the Grid-aligned Transition property.
\end{proof}

\begin{theorem}[Polynomial-Time Verifier for Subset-Sum] \label{thm:subset-sum-verifier}
The Turing machine $M_{\mathrm{SS}}$ defined in
\cref{def:sum_of_subset_verifier_tm} is a correct NP verifier for the
\textsc{Subset-Sum} problem.

More precisely, for every input string $W$  encoding an instance
of \textsc{Subset-Sum} of length $n$ together with a certificate of length at most $n$,
\begin{itemize}
    \item $M_{\mathrm{SS}}$ accepts $W$ if and only if the certificate
    encodes a valid subset whose sum equals the target value, and
    \item $M_{\mathrm{SS}}$ halts in $\bigO(n^2)$ time and uses $\bigO(n)$ tape
    space.
\end{itemize}
Hence, \textsc{Subset-Sum} admits a deterministic quadratic-time and linear-space
verifier.
\begin{proof}
The claim follows by combining soundness, completeness, and the resource bounds
established previously.

\begin{itemize}
\item \textbf{Soundness.}
By \cref{lem:soundness-subsetsum},
if $M_{\mathrm{SS}}$ accepts an input $W$, then the certificate encoded
in $W$ corresponds to a subset of the input multiset whose elements sum exactly
to the target value.
Thus, no invalid certificate is ever accepted.

\item \textbf{Completeness.}
By \cref{lem:completeness-subsetsum},
for every valid certificate encoding a subset whose sum equals the target value,
$M_{\mathrm{SS}}$ deterministically reaches the accepting
configuration.
Thus, every valid instance admits an accepting computation.

\item \textbf{Time and space complexity.}
By \cref{lem:time-space-subsetsum},
on inputs of length $n$ the machine $M_{\mathrm{SS}}$ halts in $\bigO(n^2)$
time and uses $\bigO(n)$ tape space.
In particular, the verifier runs in deterministic polynomial time.
\end{itemize}

Combining the above arguments, $M_{\mathrm{SS}}$ is a correct
polynomial-time verifier for the \textsc{Subset-Sum} problem.
\end{proof}

\end{theorem}

The Subset-Sum machine therefore constitutes a finite deterministic verifier for 
a canonical NP-complete problem within our NP verifier simulation framework

\subsection{Restriction of Certificate Symbols for Input Sanitization}

In the original verifier-simulation framework, certificates are allowed to be arbitrary strings over the entire finite input alphabet of the Turing machine.
 Consequently, for full generality, each verifier Turing Machine must be equipped with a certificate-validation phase 
 that checks whether the given certificate conforms to the syntactic constraints required by the subsequent computation.

\paragraph{Sanitization Turing Machine.}

To preserve full compatibility with the original framework, certificate validation can be incorporated directly into the verifier by adding a finite preprocessing phase; 
in other words, the input to the Turing machine can be sanitized. In this phase, the Turing machine scans the certificate region 
to verify that all symbols belong to the allowed alphabet and to ensure that essential structural requirements are met—such as excluding non-certificate symbols or rejecting invalid ones. 
If an invalid symbol or a critical formatting error is detected, the machine halts in the \textsc{Reject} state; otherwise, it proceeds to the corresponding problem-specific verification phase.

For the two SAT TMs presented in \cref{subsec:sat_tm_input_dependent,subsec:sat_tm_fixed}, 
an explicit input-sanitization phase is required for complete conformance with the original verifier-simulation framework. 
In this case, the preprocessing phase scans the certificate region to ensure it consists solely of symbols from the alphabet $\{\texttt{T}, 
\texttt{F}\}$ and is terminated by a $\blank$ symbol before entering the main SAT verification routine.
Since the preprocessing phase only checks for allowed certificate symbols, the head movement remains independent of the certificate contents, thus preserving the certificate-oblivious property.

\begin{table}[H]
\centering
\footnotesize
\begin{tabular}{llllll}
\textbf{State} & \textbf{Read} & \textbf{Next} & \textbf{Write} & \textbf{Move} & \textbf{Comment} \\
\hline
InputCheck & \# & CertificateCheck & \# & R & enter certificate region \\
InputCheck & *  & InputCheck & * & R & skip formula symbols \\

CertificateCheck & T & CertificateCheck & T & R & valid certificate symbol \\
CertificateCheck & F & CertificateCheck & F & R & valid certificate symbol \\
CertificateCheck & $\epsilon$ & BackToBeginning & $\epsilon$ & L & end of certificate \\
CertificateCheck & * & Reject & \  & R & invalid certificate symbol \\

BackToBeginning & * & BackToBeginning & * & L & rewind to tape start \\
BackToBeginning & $\epsilon$ & Check & $\epsilon$ & R & start main verification \\
\hline
\end{tabular}
\caption{Transition table for input sanitization and certificate validation.
This preprocessing phase verifies that the certificate region contains only
symbols from the allowed alphabet $\{T,F\}$ before entering the main SAT
verification routine.} \label{tab:sanitization_sat}
\end{table}

\Cref{tab:sanitization_sat} presents the transition table for the input-sanitization phase of the input-dependent SAT Turing machine. 
For the fixed-state Turing machine construction primarily used in our implementation, the role of the \textsc{Check} state is replaced by the \textsc{Check.Forwarded} state. 
The preprocessing logic remains identical; only the entry state of the main verification routine is renamed to align with the fixed-state transition layout. 
For both TM variants, the \textsc{InputCheck} state naturally serves as the initial state $q_{\text{init}}$.

Similarly, for the Subset-Sum Turing machines presented above, an explicit input-sanitization phase is required in the original formulation. 
Since certificates are arbitrary strings, the machine must first verify that the certificate region contains only valid numeric symbols and permitted separators, 
and that no unused markers appear beyond the designated boundary. 
\Cref{tab:sanitization_sum_of_subset} presents a transition fragment implementing this validation, where the TM's initial state is \textsc{InputCheck}.

\begin{table}[H]\label{tab:sanitization_sum_of_subset}
\centering
\footnotesize
\begin{tabular}{llllll}
\textbf{State} & \textbf{Read} & \textbf{Next} & \textbf{Write} & \textbf{Move} & \textbf{Comment} \\
\hline
InputCheck & \# & CertificateCheck & \# & R & enter certificate region \\
InputCheck & *  & InputCheck & * & R & skip problem instance \\

CertificateCheck & D & CertificateCheck & D & R & valid digit in certificate \\
CertificateCheck & x & CertificateCheck & x & R & valid digit in certificate \\
CertificateCheck & \_ & CertificateCheck & \_ & R & valid separator / padding \\
CertificateCheck & $\epsilon$  & BackToBeginning & $\epsilon$ & L & end of certificate \\
CertificateCheck & * & Reject & \_ & R & invalid certificate symbol \\

BackToBeginning & * & BackToBeginning & * & L & rewind to tape start \\
BackToBeginning & $\epsilon$ & Forward & $\epsilon$ & R & begin subset-sum verification \\
\hline
\end{tabular}
\caption{Transition table for input sanitization and certificate validation in
the Sum-of-Subset Turing Machine.
This preprocessing phase ensures that the certificate region contains only
numeric symbols and permitted separators before entering the main computation.}
\end{table}

These input sanitization transitions of the machines runs in time linear in the length of the input tape and
uses only a constant number of additional states.
Therefore, composing such a sanitization phase with any of the verifier machines
defined above preserves polynomial time bounds and does not affect the
feasible-graph simulation results.

This construction is a concrete realization of the \emph{built-in input
sanitization} mechanism assumed in the original framework, where the verifier
is required to behave correctly on \emph{all} certificate strings.

\paragraph{Practical enumeration strategy.}
The original verifier model assumes certificates ranging over the full alphabet $\Sigma^*$ and therefore requires an explicit sanitization phase to reject malformed inputs. 
In the actual implementation of the NP verifier simulation framework, these sanitization phases are not executed explicitly. 
Instead, the enumeration procedure \textproc{GetNextFloorEdge()} is designed to generate certificates only over the allowed alphabet. 
As a result, all certificates passed to the verifier are composed of only valid certificate symbols by construction.
This design choice does not weaken the framework. All certificates that could possibly be accepted by the verifier under the original model are still enumerated, 
while certificates that would be rejected deterministically during sanitization are pruned \emph{a priori}. 
From a theoretical perspective, this is equivalent to composing the verifier with a sanitization Turing machine; 
from an implementation perspective, it is strictly more efficient, as it avoids redundant computation on trivially rejecting inputs.

In the remainder of this paper, inputs are therefore assumed to be composed of only the expected certificate symbols by construction. 
This assumption is fully consistent with the original NP verifier simulation framework and preserves all soundness, completeness, and complexity guarantees established above.

\subsection{Refined Bounds of the Footmarks Graph} \label{subsec:refined_footmarks_bounds}

In the analysis of the original framework, the width $w$ and height $h$ of the computation graph were bounded uniformly 
by a single polynomial $p(n)$ representing the running time of the verifier Turing machine, independent of the machine's internal structure. 
While this approach is sufficient for establishing polynomial-time simulability, it tends to overestimate the actual graph size within the NP verifier simulation framework. 
This is because the width and height can be significantly lower than the total time bound; for instance, the space bound, which determines the graph width, is typically much smaller than the time bound. 
In this subsection, we provide a refined analysis of the footmarks graph size for the maximal computation walks of the three TMs, mainly bounding their width and height.

For clarity in this subsection, we denote $n$ as the input size of the NP problem itself. This is justified by the fact that the tape input size $n_v$ of the corresponding verifier is at most linear in the problem size $n$ (specifically, $n_v \leq 2n$ for typical encodings), ensuring that this substitution does not affect the asymptotic time complexity calculations.

\paragraph{Input-dependent SAT verifier.}
For the input-dependent SAT verifier Turing Machine, the tape length is $\bigO(n)$ and each
individual computation walk has height
\[
h = \bigO(n).
\]

This bound holds because the machine performs a single left-to-right traversal followed by a returning backward traversal for each literal in the clause region, 
ensuring that the maximum tier index remains $\bigO(n)$, specifically $2n$. Consequently, the width $w$ is also $\bigO(n)$, specifically $n+1$, as it is constrained by the tape space bound. 
Notably, as discussed in the preceding section, this specific structure allows for a more refined analysis of the total vertex size and edge size.

Unlike the fixed-state construction, however, the control-state set $Q$ depends
on the input size.
In particular, variable-index--dependent states yield
\[
|Q| = \bigO(n).
\]

Consequently, for each tape index and tier, there may exist up to $\bigO(|Q|^{2})$ distinct computation nodes within a transition case.

Importantly, this increase in the state space does not affect the maximum tier of any individual computation walk, nor does it affect the height of the footmarks graph. 
Instead, the effect is purely an inflation of the node count within each transition case, which in turn amplifies the node and edge complexity of the computation graph.

In particular, since there are up to $h+1$ transition cases for each tape cell index, and each tape cell is associated with $\bigO(|Q|^2 h)$ vertices, 
the footmarks graph construction—where the index difference between incident nodes of an edge is at most 1—yields a total number of edges bounded by
\[
O\!\left(
  w \cdot (|Q|^{2} h)^{2}
\right)
=
O\!\left(
  n \cdot (n^{2} \cdot n)^{2}
\right)
=
O(n^{7}).
\]

\begin{remark}
The primary design objective of the input-dependent SAT construction is structural simplicity. 
This is achieved through a minimal tape alphabet—consisting only of clause symbols and truth values $\{T, F\}$—and a straightforward transition logic that reduces the overhead of rule matching. 
However, this simplicity comes at the cost of a state set that grows linearly with the input size. 
Consequently, as $n$ increases, the graph size of the input-dependent model significantly exceeds that of the fixed-state verifier, 
since the latter maintains a constant number of states, as detailed in the following subsection.
\end{remark}

\paragraph{SAT verifier with fixed states.}
For the SAT verifier Turing machine with fixed control states, each variable index digit is decremented during each fetch step, 
triggering clause-wise scans. As a result, the worst-case running time is quadratic in the input size, $\bigO(n^2)$, due to the repeated scans over the digit positions.
Nevertheless, when representing the computation with certificate length $m$ as a footmarks graph, the following bounds apply:
\begin{itemize}
\item The control-state set $Q$ is constant ($|Q| = \bigO(1)$).
\item The width $w$ corresponds to the tape space bound and grows linearly with $n$, such that $w = \bigO(n)$, specifically $n+m+1$.
\item The height $h$ corresponds to the maximum number of transitions per cell, which involves a constant number of backward fetching/assigning and forward evaluation passes, yielding $h = \bigO(n)$, specifially $2m+2$.
\end{itemize}
Thus, while the total configuration space is $\bigO(n^2)$, the structural complexity of the footmarks graph satisfies:
\[
|E(G)| = \bigO(w \cdot h^2) = \bigO(n \cdot n^2) = \bigO(n^3),
\]
characterized by linear width and linear height, leading to an overall cubic edge complexity in terms of the input size $n$.

\paragraph{Subset-Sum verifier}
For the \textsc{Subset-Sum} verifier Turing machine, the tape length is $\bigO(n)$, the certificate length $m$
and the computation consists of repeated linear scans combined with digit-wise matching and subtraction. 
Although individual tape cells may be revisited up to $\bigO(n)$ times during borrow propagation and repeated matching passes, 
these revisits occur sequentially for each certificate digit or delimiter.
As a result, the maximum tier of any node in a computation walk is $\bigO(n)$, specifically $3m$,
implying that the height of the footmarks graph is $h = \bigO(n)$. 
Furthermore, the width $w$—representing the number of cells reachable by the Turing machine head—is constrained by the space bound $\bigO(n)$, specifically $n+m+1$. 
Thus, the effective bounds satisfy:
\[w = \bigO(n), \qquad h = \bigO(n), \qquad |E(G)| = \bigO(w \cdot h^2) = \bigO(n \cdot n^2) = \bigO(n^3).
\]

\paragraph{Summary (Graph Parameters and Simulation Complexity)}

For all concrete (fixed-state) verifiers considered in this paper, including the fixed-state SAT verifier and the \textsc{Subset-Sum} verifier, 
the corresponding footmarks graph satisfies the following asymptotic bounds:
\[
    w = \bigO(n), \qquad h = \bigO(n), \qquad |E(G)| = \bigO(w h^{2}) = \bigO(n^{3}).
\]
Substituting these parameters into the footmarks graph construction, we obtain
\[
T_f = \bigO\bigl(w^{2} h^{4} (h\log h + \log w)\bigr) = \bigO(n^{7} \log n).
\]
Taking into account the cost of verifying and extending all candidate edges in the footmarks graph, 
the total running time of \textproc{SimulateVerifierForAllCertificates()} is bounded by
\[
    \bigO\bigl(w^4 h^8 T_f\bigr) = \bigO(n^{19} \log n),
\]
which confirms that the NP simulation runs in polynomial time relative to the input size $n$ for both problems. 
This observation does not apply to the input-dependent SAT construction, where the state-indexed expansion dominates the graph complexity.

It is important to note that the parameters $w$ and $h$ cannot be replaced by a generic running-time bound $p(n)$. 
If one naively sets $w, h \le p(n) = \bigO(n^2)$, the resulting bound would become $\bigO(n^{38} \log n)$, 
which severely overestimates the actual complexity. 
In contrast, when the parameters derived above are applied to the original NP verifier simulation framework of~\cite{lee2025PNP}, 
the same asymptotic bound $\bigO(n^{19} \log n)$ is obtained. 
This shows that the reduced degree arises from structural properties of the constructed computation graph dependent on the verifier TM rather than from any algorithmic modification.

Consequently, beyond constant-factor savings, further reductions in simulation cost can hardly be obtained by refining Turing machine transition rules alone, 
since the linear bound is an inherent bound for both width and height. 
Hence, any asymptotic improvement must instead arise from structural refinements of the computation graph itself, motivating a closer examination of the feasible graph construction.

\section{Efficient Construction of Feasible Graphs} \label{sec:improvements_on_feasiblegraph}

While the time complexity analysis establishes a uniform worst-case upper bound of
$\bigO(n^{38} \log n)$ for \textproc{SimulateVerifierForAllCertificates()}, this bound is not representative for all concrete verifier constructions.

In particular, for both the \textsc{Subset-Sum} Turing Machine and \textsc{SAT},
a refined analysis of the computation graph shows that the width and height
admit substantially smaller polynomial bounds, yielding an effective theoretical
runtime of $\bigO(n^{19}\log n)$ for these constructions, as discussed in the previous section.

Despite this reduction, the resulting bound remains too large for practical execution. This motivates a refined feasible-graph construction designed for efficiency. 
In the original framework, the feasible graph was refined primarily through iterative additive sweep operations over index-adjacent edges. 
While these operations are polynomially bounded, they incur significant overhead due to repeated global sweep passes and delayed convergence.

In this work, we reorganize and improve the feasible-graph construction by explicitly separating and accelerating two critical components: 
\emph{cover-edge computation} and \emph{feasible-graph trimming}. 
By computing cover edges before the main feasible-graph refinement, we reduce the initial computational cost, which in turn significantly improves overall performance.

\subsection{Efficient Cover Edge Computation}
\label{subsec:cover-edge}

Cover edges represent the minimal structural upper marginal connections required to preserve potential accepting computation walks. 
In the original approach, ceiling-adjacent relationships between edges were determined by repeatedly examining paths between candidate edge pairs, 
which led to a substantial number of redundant reachability checks.

We improve this process by deferring connectivity verification. Instead of explicitly checking all candidate paths during the initial cover-edge construction, 
we first generate a superset of potential cover edges using local adjacency (specifically, \emph{weak-ceiling adjacency}). 
Global connectivity is then verified in a single batch process at the conclusion of the construction.

Although this approach may initially include more redundant edges (referred to as \emph{ex-cover edges}), 
this modification reduces the complexity of path-existence checks. Specifically, 
it transitions from a polynomial dependent on the number of candidate edge pairs to a much simpler inclusion check within connected components. 
This yields a significant asymptotic improvement in the efficiency of cover-edge computation.

\begin{algorithm}
\caption{Collect Cover Edges Connected by Paths}
\label{alg:collect_edges_with_path}
\begin{algorithmic}[1]
\Input{$G$: computation graph, $C_0$: candidate cover edge set, $E_f$: set of final edges}
\Output{Filtered ex-cover edge set $C$}
\Function{CollectEdgesWithPath}{$G, C_0, E_f$}
    \State Initialize queue $Q \gets E_f$
    \State Initialize visited edge set $E_v \gets \emptyset$
    \State Initialize $C$ as an empty index-partitioned edge set
    \While{$Q$ is not empty}
        \State Dequeue edge $e=(u,v)$ from $Q$
        \If{$e \in E_v$} \Continue \EndIf
        \State Let $i \gets \indexOf(e)$
        \If{$e \in C_0[i]$}
            \State Add $e$ to $C_i$
        \EndIf
        \State Add $e$ to $E_v$
        \State Enqueue all incoming edges of $u$ into $Q$
    \EndWhile
    \State \Return $C$
\EndFunction
\end{algorithmic}
\end{algorithm}

\begin{algorithm}
\caption{Improved Computation of Cover Edges}
\label{alg:improved_compute_cover_edges}
\begin{algorithmic}[1]
\Input{$G$: subgraph of the footmarks graph, $E_f$: set of final edges}
\Output{Set $C$ of cover edges of $G$ with respect to $E_f$}
\Function{ComputeCoverEdges}{$G, E_f$}
    \State Initialize $C$ as an empty index-partitioned edge set
    \ForAll{$e=(u,v) \in E_f$}
        \State Add $e$ to $C_i$ where $i=\indexOf(e)$
    \EndFor
    \State Let $V_v \gets$ an empty set 
    \State Initialize queue $Q \gets E_f$
    \While{$Q$ is not empty}
        \State Dequeue an edge $f$ from $Q$
        \State Let $E_c \gets$ \Call{GetWeakCeilingAdjacentEdges}{$G, f, E_f, V_v$} \Comment{No revisited vertices (\cref{alg:weakly_ceiling_adjacent_edges})}
        \ForAll{$e=(u,v) \in E_c$}
            \State Let $i \gets \indexOf(e)$
            \If{$e \notin C_i$}
                \State Add $e$ to $C_i$
                \State Enqueue $e$ into $Q$
            \EndIf
        \EndFor
    \EndWhile
    \State Let $C \gets$ \Call{CollectEdgesWithPath}{$G, C, E_f$} \Comment{Phase 2: Path-based filtering}
    \State \Return $C$
\EndFunction
\end{algorithmic}
\end{algorithm}

\paragraph{Algorithmic Rationale}
The original algorithm for computing cover edges explicitly verified the existence of a path between each ceiling edge and its ceiling-adjacent candidate at the moment of discovery. 
While convenient for correctness proofs, this formulation incurs substantial computational overhead.

The improved algorithm separates these concerns into two distinct phases. 
In the first phase, it computes the transitive closure of ceiling adjacency starting from $E_f$, 
collecting all potential cover-edge candidates—specifically, all edges within the weakly ceiling-adjacent edge chain. 
In the second phase, path existence is verified through a single backward traversal from $E_f$, thereby filtering out edges that lack the necessary connectivity.

This reorganization preserves all valid cover edges: an edge is included if it is weakly ceiling-adjacent to another ceiling edge and connected to it by a path in the graph. 
The primary difference lies in the inclusion of certain redundant \emph{ex-cover edges}, 
which are weakly ceiling-adjacent to a ceiling edge and connected to the designated final edges $E_f$, but may not meet the stricter criteria of the original definition.

\begin{sublemma}[Completeness and Ex-Cover Property of Improved \textproc{ComputeCoverEdges()}]
\label{sublem:completeness_excover_edges_improved}

Let $G$ be a computation graph, and let $E_f$ be a designated final set.
Then, the set $C$ returned by the improved algorithm \textproc{ComputeCoverEdges()} 
(\cref{alg:improved_compute_cover_edges}) satisfies:

\begin{enumerate}
    \item \emph{Completeness:} $C$ contains all cover edges toward $E_f$.
    \item \emph{Ex-cover property:} Every edge in $C$ is an ex-cover edge toward $E_f$.
\end{enumerate}

Consequently, for any computation walk ending with some edge of $E_f$ in $G$, all ceiling edges along the walk are included in $C$.

\begin{proof}
We establish the two properties separately
\begin{itemize}
\item \textbf{Completeness:}
Suppose, for the sake of contradiction, that there exists a cover edge toward $E_f$ not contained in $C$. 
By the definition of a cover edge, there must exist a cover-edge chain $(c_0, c_1, \dots, c_k)$ such that $c_k \in E_f$. 
Let $c_i$ be the edge with the largest index in this chain such that $c_i \notin C$. If $i = k$, then $c_k \in E_f$. 
Since the algorithm initializes the candidate set with $E_f$, $c_k$ must be in $C$, which is a contradiction. 
Thus, $i < k$. By our choice of $i$, the successor $c_{i+1}$ is already included in $C$. 
Since $c_i$ is ceiling-adjacent to $c_{i+1}$, it is, by definition, also weakly ceiling-adjacent. 
Therefore, $c_i$ is enqueued and added to the candidate set during the transitive closure phase of \textproc{ComputeCoverEdges()}. 
Furthermore, since $c_i$ is a cover edge, there exists a path from $c_i$ to $c_{i+1}$, and by induction, a path from $c_i$ to some $c_k \in E_f$. 
Consequently, the reachability filter in \textproc{CollectEdgesWithPath()} will retain $c_i$ in the final set $C$. 
This contradicts the assumption $c_i \notin C$.
\item \textbf{Ex-cover property:}
Let $e \in C$ be any edge returned by the algorithm. 
We must show that $e$ is an ex-cover edge toward $E_f$.
By the construction in \textproc{ComputeCoverEdges()}, any edge added to the candidate set must belong to a weakly ceiling-adjacent chain $(c_0, c_1, \dots, c_k)$ where $c_0 = e$ and $c_k \in E_f$. 
This satisfies the first structural requirement of an ex-cover edge.
Subsequently, the second phase, \textproc{CollectEdgesWithPath()}, performs a backward traversal from $E_f$ to verify reachability. 
An edge $e$ is retained in $C$ if and only if there exists a path in $G$ from $e$ to some edge in $E_f$. 
Since $e$ is both part of a weakly ceiling-adjacent chain leading to $E_f$ and possesses a valid path to $E_f$, it satisfies the definition of an ex-cover edge.
\end{itemize}
Therefore, $C$ preserves all cover edges while ensuring that every included edge qualifies as an ex-cover edge.
\end{proof}
\end{sublemma}

\begin{sublemma}[Time Complexity of Improved \textproc{ComputeCoverEdges}]
\label{lem:time_improved_computing_cover_edges}

Let $G$ be a computation graph of width $w$ and height $h$. 
Then the improved algorithm \textproc{ComputeCoverEdges()} runs in time 
$\bigO(w h^{3} \log(wh))$.

\begin{proof}
Note that each edge slice layer contains $\bigO(h^{2})$ edges, and the total number of edges is $|E(G)| = \bigO(w h^{2})$. The algorithm operates in two phases.

\paragraph{Phase 1: Ceiling-adjacency expansion.}
The queue $Q$ initially contains all final edges $E_f$. Each edge is enqueued at most once, as it is added to the cover-edge set $C$ immediately upon discovery and is never re-enqueued.

For a given edge $e$, the set of ceiling-adjacent edges is confined to neighboring edge slices. 
First, since each edge slice contains $\bigO(h^{2})$ edges and checking inclusion in $E_f$ costs $\bigO(\log(wh))$, finding ceiling-adjacent edges requires $\bigO(h^2 \log(wh))$ time per edge, as described in the \cref{alg:weakly_ceiling_adjacent_edges}. 
Second, the number of ceiling-adjacent edges processed per edge is $\bigO(h^{2})$, and checking for existence in $C_i$ followed by insertion into $C_i$ costs at most $\bigO(\log h)$, as $C_i$ contains at most $h^2$ edges. 
Furthermore, inserting these ceiling-adjacent edges into $Q$ costs $\bigO(\log(wh))$ since $Q$ can contain at most $|E| = \bigO(wh^2)$ edges. 
Consequently, the work per edge in the while-loop is bounded by $\bigO(h^2 \log(wh) )$.

A key optimization in this framework, as detailed in Appendix \cref{alg:weakly_ceiling_adjacent_edges}, is the use of a visited node set $V_v$. 
Although there are $\bigO(wh^2)$ edges in total, if an incident node has already been visited, the algorithm returns immediately after a membership check in $V_v$, which costs only $\bigO(\log(wh))$.
Since the full expansion logic is executed at most once per node (totaling $\bigO(wh)$ nodes), the cumulative complexity for these exhaustive checks is $\bigO(wh \cdot (h^2 \log(wh))) = \bigO(wh^3 \log(wh))$.
Including the linear scan of edges that hit the visited set, the total work for Phase 1 is bounded by:
\[
\bigO(wh^2 \cdot \log(wh)) + \bigO(wh^3 \log (wh)) = \bigO(wh^3 \log(wh)).
\]

\paragraph{Phase 2: Path-based filtering.}
In the second phase, \textproc{CollectEdgesWithPath()} performs a backward traversal starting from the final edges $E_f$. 
While the visited set $E_v$ ensures each edge's predecessors are explored only once, an edge can be enqueued multiple times up to its out-degree (backward in-degree), which is $\bigO(h)$. 
Thus, the total number of enqueued elements is $\bigO(|E(G)| \cdot h) = \bigO(wh^3)$.

For each popped edge, the membership test in $E_v$ costs $\bigO(\log(wh))$. 
If the edge is not in $E_v$, the cost for checking and inserting into $C_i$ is $\bigO(\log h)$, and exploring incoming edges takes $\bigO(h)$. 
The cumulative cost is dominated by the membership tests for all enqueued elements:
\[
\bigO(|E(G)| \cdot h \log(wh)) = \bigO(wh^3 \log(wh)).
\]

\paragraph{Total complexity.}
Combining both phases, where Phase 1 is $\bigO(wh^3 \log h)$, the overall running time is:
\[
\bigO(wh^3 \log h) + \bigO(wh^3 \log(wh)) = \bigO(wh^3 \log(wh)).
\]

Thus, the improved algorithm \textproc{ComputeCoverEdges()} executes in $\bigO(w h^{3} \log(wh))$ time, reflecting the cost of ceiling-adjacency expansion and subsequent path-based filtering.
\end{proof}
\end{sublemma}

\begin{remark}[Ordered Set Assumption]
For the purpose of the worst-case time complexity analysis of \textproc{ComputeCoverEdges()}, 
we assume that all sets (such as the visited edge set and cover edge set) are implemented as ordered sets, 
regardless of specific language-internal implementations (e.g., Python's hash-based sets). 
This ensures that membership checks, insertions, and deletions are performed in logarithmic time per operation, thereby validating the stated $\bigO(w h^{3} \log(wh))$ bound.
\end{remark}

\begin{remark}[Practical Cover-Edge Density and Performance]
The worst-case time complexity of \textproc{ComputeCoverEdges()} is derived under the assumption that a large majority of total edges could potentially be identified as cover edges. 
In practice, however, the cover-edge set typically constitutes only a small fraction of the total edge set. 
Consequently, the empirical runtime is often significantly lower than the formal $\bigO(w h^3 \log(wh))$ upper bound, reflecting the inherent sparsity of cover-edge structures in typical computation graphs.
\end{remark}

\paragraph{Relation to the Original Cover-Edge Set.}
The current construction computes \emph{ex-cover edges} rather than strictly limiting the output to the original cover edges. 
By definition, every cover edge is an ex-cover edge; thus, the resulting set serves as a structural superset of the original cover-edge set. 
This enlargement is a strategic choice: while our primary interest lies in ceiling edges and their associated structures, 
adopting the broader ex-cover edge framework facilitates a more efficient computational procedure.

This modification preserves the correctness of the algorithm. Since any feasible computation walk consists of ceiling edges and the edges residing below them, 
all such valid paths remain intact within the enlarged graph. Although additional edges are introduced, they merely represent non-feasible transitions. 
Unless a verification is confirmed prematurely, these edges are to be eliminated during the subsequent walk verification stages. 

Consequently, this construction maintains completeness by providing a safe over-approximation of the intermediate graph, 
allowing for a significantly more efficient construction without compromising the ultimate preservation of feasible walks.

\subsection{Complexity Improvement in Feasible Graph Construction} \label{subsec:rejected-walk-pruning}
In this subsection, we present a refined algorithm for computing the feasible graph. 
This algorithm preserves the same semantic definition as established in the original work while significantly improving both practical efficiency and asymptotic time complexity. 
Importantly, the fundamental notion of a \emph{feasible graph} remains unchanged: 
it is the subgraph devoid of step-pendant edges that preserves all feasible walks --- namely, computation walks terminating at the designated final edge set $E_f$.

The performance gain stems not from altering the underlying definition, but from adopting a more direct elimination strategy. 
This strategy removes infeasible edges based on localized structural conditions, rather than relying on repeated global sweeps over index layers.

Once the cover edges are computed, the feasible graph is initialized as the reachable component from the initial nodes and subsequently refined. Instead of expensive sweep operations, 
we first identify \emph{step-pendant edges} by detecting: (i) ex-pendant edges that are neither initial nor designated-final, (ii) non-floor edges lacking index-precedent edges, and (iii) non-cover edges lacking index-succedent edges. 
These edges cannot participate in any feasible walk and are thus eliminated from the feasible graph. 

After that, we include the removed edges as the initial step-extended component. Then, we iteratively remove step-pendant edges that are step-adjacent to the removed step-extended component until no further removal is possible. 
This approach terminates within a bounded number of iterations, avoiding repeated global scans and achieving a genuine asymptotic improvement in the construction of the feasible graph.

\paragraph{Original Sweep-Based Construction}

We briefly recall the sweep-based feasible-graph construction from~\cite{lee2025PNP}. 
The algorithm computes the feasible graph through repeated bidirectional sweeps over the edge slices. 
In each iteration, index-succedent edge chains are expanded upward from the floor edges using index-adjacency conditions and are subsequently filtered downward from the cover edges according to index-precedent chains.

While this method faithfully preserves all feasible computation walks by explicitly reconstructing potential components in both directions, its complexity is hindered by several inefficiency factors:
\begin{itemize}
    \item \textbf{Repeated global sweeps:} The algorithm requires multiple passes over $w$ horizontal edge slices.
    \item \textbf{High per-slice cost:} Step-adjacency conditions must be re-evaluated for edge slices of height $h$, involving $\bigO(h^2)$ edges per slice.
    \item \textbf{Slow convergence:} The refinement process may converge slowly if only a small number of edges are eliminated per global sweep.
\end{itemize}

Consequently, the sweep-based approach may require up to $\bigO(m)$ global iterations in the worst case, where $m = |E(G)| = \bigO(w h^2)$ denotes the total number of edges. 
Since each iteration involves up to $\bigO(w h^{2} (h\log h + \log w))$ edge interactions, the overall worst-case upper bound is:
\[
\bigO(m \cdot w h^{2} (h\log h + \log w)) = \bigO(w^2 h^4 (h\log h + \log w)).
\]
This slow convergence, which dependent on the number of edges $m$ (or width $w$), highlights the necessity of the more direct, local elimination strategy proposed in this work.

\paragraph{Direct Elimination via Step-Pendant Edges}
The sweep-based construction repeatedly rebuilds feasibility from scratch, even though infeasibility typically originates from a small set of locally removed edges. 
In particular, once an edge loses its step-adjacent edges and becomes step-pendant, its removal is inevitable; further global sweeps merely rediscover this structural fact.

Instead of reconstructing the feasible graph through alternating upward and downward propagation, we reverse the viewpoint: 
we identify edges whose infeasibility is already determined and eliminate them directly. 
This approach is more faithful to the structural definition of the feasible graph, although the result may contain unnecessary isolated vertices. 
Since the feasible graph serves primarily as a computational substrate for subsequent verification stages, these isolated vertices do not affect the correctness and thus require no explicit handling.

The key observation is that such edges are characterized precisely by their \emph{step-pendant structure}. A step-pendant edge cannot participate in any feasible walk toward the designated final edge set, 
and this property is invariant unless a feasible walk is fundamentally altered by the removal of its constituent edges. Therefore, the feasible-graph construction can be reformulated as a monotone elimination process:

\begin{center}
    \textbf{feasible graph computation = iterative removal of step-pendant edges}
\end{center}

Since all edges lacking a step-adjacency chain from the initial nodes are eventually removed, and the step-pendant property can only propagate through step-adjacent edges upon their removal, it is sufficient to construct the feasible graph within the following \emph{step-reachable subgraph}.

\begin{definition}[Step-Reachable Subgraph]
\label{def:step_reachable_subgraph}
Given a computation graph $G = (V, E)$ and a set of initial edges $E_{\text{init}}$, an edge $e \in E$ is said to be \textbf{reachable via step-adjacency} (or simply \textbf{step-reachable}) from $E_{\text{init}}$ if there exists a sequence of edges $e_0, e_1, \dots, e_k$ such that $e_0 \in E_{\text{init}}$, $e_k = e$, and for each $0 \le j < k$, $e_{j+1}$ is \textbf{step-adjacent} to $e_j$. 

The \textbf{step-reachable subgraph} $H \subseteq G$ is the subgraph induced by the set of all step-reachable edges and their incident vertices. 
\end{definition}

Note that this definition implicitly accounts for both:
\begin{enumerate}
    \item[(i)] \textbf{Horizontal step-adjacency}: $e_{j+1} \in \Next_G(e_j) \cup \Prev_G(e_j)$, and
    \item[(ii)] \textbf{Vertical step-adjacency}: $e_{j+1} \in \IPrec_G(e_j) \cup \ISucc_G(e_j)$.
\end{enumerate}

By enumerating step-pendant edges within the step-reachable subgraph and removing them in a canonical order, the algorithm avoids global sweeps and performs only locally necessary updates. This converts the construction from repeated global propagation to a bounded number of local eliminations, yielding a strictly lower worst-case time complexity.

\begin{algorithm}[!ht]
\caption{Enumeration of Step-Pendant Edges with Step-Reachable subgraph}
\label{alg:get_step_pendant_edges_with_reachable_graph}
\begin{algorithmic}[1]

\Input
{$G$: computation graph,
 $V_0$: initial vertices,
 $E_f$: set of final edges,
 $C$: cover-edge index map}

\Output
{Set $E_r$ of step-pendant edges, Reachable graph $H$ from $V_0$}

\Function{GetStepPendentEdgesWithReachableGraph}{$G, C, V_0, E_f$}
    \State Let $H$ be an empty dynamic computation graph
    \State Let $E_r \gets \emptyset$ \Comment{Step-pendant edges}
    \State Let $E_0 \gets$ the set of outgoing edges of all $v_0 \in V_0$
    \State Let $Q \gets$ A deque with all the edges of $E_0$

    \While{$Q$ not empty}
        \State Let $e=(u,v) \gets$ Pop from $Q$
        \If{$e \in E(H)$} 
        	\State{\Continue}
	\EndIf
        \State add $e$ to $H$

        \If{$e \notin C[\indexOf(e)]$ \textbf{and} $|\ISucc_G(e)| = 0$}[No index-succedent and not a cover edge]
            \State add $e$ to $E_r$
        \Else
            \State add $\ISucc_G(e)$ to $Q$
        \EndIf

        \If{$\tier(v) > 0$ \textbf{and} $|\IPrec_G(e)| = 0$}[No index-precedent above the floor]
            \State add $e$ to $E_r$
        \Else
            \State add $\IPrec_G(e)$ to $Q$
        \EndIf

        \If{$e \notin E_f$}[No next edges; Forward exploration]
            \If{$|\Next_G(e)| = 0$}
                \State add $e$ to $E_r$
            \Else
            	\State enqueue $\Next_G(e)$ into $Q$
            \EndIf
        \EndIf

        \If{$u \notin V_0$}[No previous edges; Backward exploration]
            \If{$|\Prev_G(e)| = 0$}
                \State add $e$ to $E_r$
            \Else
            	\State enqueue $\Prev_G(e)$ into $Q$
            \EndIf
        \EndIf
    \EndWhile

    \State \Return $E_r, H$
\EndFunction

\end{algorithmic}
\end{algorithm}

\begin{sublemma}[Correctness of Step-Pendant Edge Enumeration on the Step-Reachable Subgraph]
\label{lem:step_pendant_edge_enumeration}
Let $G$ be a computation graph, $V_0$ the set of initial vertices,
$E_f$ the set of final edges, and $\widehat{C}_{ex}$ a given set of ex-cover
edges toward $E_f$.
Let $C$ denote the induced  indexed map of cover edges.

Let $(H, E_r)$ be the output of \cref{alg:get_step_pendant_edges_with_reachable_graph}.

Then $H$ is exactly the subgraph of $G$ step-reachable from $V_0$, and $E_r$ is exactly the set of step-pendant edges of $H$
(as defined in \cref{def:step_pendant_edge} with respect to $H$).
\end{sublemma}
\begin{proof}
We prove three statements: the reachability correctness of $H$, the soundness of $E_r$, and the completeness of $E_r$.

\begin{itemize}
\item \textbf{Step-Reachable Subgraph}:
The algorithm begins with the outgoing edges of $V_0$ and iteratively enqueues $\ISucc_G(e)$, $\IPrec_G(e)$, $\Next_G(e)$, and $\Prev_G(e)$ whenever they exist. 
By induction on the length of the step-adjacency sequence, every edge inserted into $H$ is reachable via step-adjacency from $V_0$. 
Conversely, let $e$ be any edge reachable from $V_0$ via step-adjacency. There exists a sequence $(e_0, e_1, \dots, e_k=e)$ starting from $V_0$ where each consecutive pair is step-adjacent. 
Since the algorithm performs a systematic exploration (BFS/DFS) over all such adjacencies, every $e_i$ in the sequence—and ultimately $e$—is processed and inserted into $H$.

\item \textbf{Soundness of Step-Pendency}:
An edge $e$ is added to $E_r$ only if it satisfies at least one of the conditions specified in the algorithm (e.g., lack of index-succedents while not being a cover edge, or lack of index-precedents at a non-zero tier). 
These conditions are identical to the structural criteria provided in \cref{def:step_pendant_edge}. 
Since the algorithm evaluates these adjacencies, any $e \in E_r$ is, by definition, a step-pendant edge of $H$.

\item \textbf{Completeness of Step-Pendency}:
Let $e$ be any step-pendant edge in $H$. Since $e \in H$, it is guaranteed to be popped from the queue $Q$ and processed. 
During processing, the algorithm evaluates all defining conditions for step-pendancy. 
Because $e$ is step-pendant, at least one of these conditional tests will succeed, resulting in $e$ being added to $E_r$. 
Thus, no step-pendant edge in $H$ is omitted.
\end{itemize}

By combining these parts, we conclude that $H$ is exactly the step-reachable subgraph from $V_0$ and $E_r$ is the exhaustive set of step-pendant edges within $H$.
\end{proof}

\begin{sublemma}[Time Complexity of Step-Pendant Edge Computation]
\label{lem:step_pendant_time_complexity}
Let $G$ be a computation graph of width $w$ and height $h$, and let $(H, E_r)$ be the output of \cref{alg:get_step_pendant_edges_with_reachable_graph}. 
Then the total running time of the procedure is
\[
    \bigO\bigl(|E(H)|(h \log(wh))\bigr) \subseteq \bigO\bigl(wh^3 \log(wh)\bigr).
\]
\end{sublemma}

\begin{proof}
The algorithm performs a graph exploration over the step-reachable subgraph $H$. Since visited edges are recorded in $E(H)$, each edge $e \in E(H)$ is enqueued and processed exactly once. 

For each processed edge $e$, the computational cost is determined by the following operations:
\begin{itemize}
    \item \textbf{Edge Enumeration}: The algorithm enumerates $\ISucc_G(e)$, $\IPrec_G(e)$, $\Prev_G(e)$, and $\Next_G(e)$. Each such set has a size of at most $\bigO(h)$. Following the indexing structure defined in the original framework, these enumerations and their associated index lookups take $\bigO(h \log h)$ time.
    \item \textbf{Membership and Queue Operations}: Checking or updating membership in $E(H)$ costs $\bigO(h)$ using the adjacency list structure. Membership tests for the cover map $C[\cdot]$ take $\bigO(\log h)$ via ordered sets of size $\bigO(h^2)$ within each edge slice. Operations on the dequeue $Q$ take $\bigO(1)$ time using a doubly linked list (deque).
    \item \textbf{Final Edge Verification}: Checking if $e \in E_f$ takes $\bigO(\log(wh))$ time. Given that $|E_f| \leq |E(G)| = \bigO(wh^2)$, the membership test using an ordered set costs $\bigO(\log(wh^2)) = \bigO(\log(wh))$.
\end{itemize}

Since $h \log(wh)$ asymptotically dominates both $\bigO(h\log h)$ and $\bigO(\log(wh))$, the total complexity per processed edge is $\bigO(h \log(wh))$. Summing over all edges in the step-reachable subgraph $H$, the total running time is:
\[
    O\bigl(|E(H)|(h \log(wh))\bigr).
\]
In the worst case, where $H = G$ and $|E(G)| = O(wh^2)$, the complexity is:
\[
    O\bigl(wh^2(h \log(wh))\bigr) = O(wh^3 \log(wh)).
\]
This completes the proof.
\end{proof}

\begin{algorithm}[!ht]
\caption{Improved Feasible Graph Construction via Step-Pendant Removal Extension}
\label{alg:improved_feasible_graph}
\begin{algorithmic}[1]
\Input
    {$G$: computation graph,
     $V_0$: initial vertices,
     $E_f$: designated final edges}
\Output
    {Feasible graph $H \subseteq G$}

\Function{ComputeFeasibleGraph}{$G, V_0, E_f$}
    \State $C \gets$ \Call{ComputeCoverEdges}{$H, E_f$} \Comment{Compute cover edges once}

    \State $(E_r, H) \gets$ \Call{GetStepPendentEdgesWithReachableGraph}{$G, C, V_0, E_f$}
 
    \State $Q$ be a queue with all the edges of $E_r$ \Comment{Initialize removal queue}

    \While{$Q$ not empty}
        \State $e \gets$ Dequeue frome Q

        \If{$e \notin E(H)$} 
        	\State \Continue
        \EndIf

        \If{$e$ is not a merging edge}[Forward propagation for non-merging edges]
            \State Enqueue all $\Next_H(e)$ into $Q$
        \EndIf

        \ForAll{$f \in \ISucc_H(e)$}[Propagate through index-succedents]
            \If{$|\IPrec_H(f)| = 1$}
                \State Enqueue $f$ into $Q$
            \EndIf
        \EndFor

        \ForAll{$f \in \IPrec_H(e)$}[Propagate through index-precedents that are not cover edges]
            \If{$f \notin C[\indexOf(f)]$ \textbf{and} $|\ISucc_H(f)| = 1$}
                \State Enqueue $f$ into $Q$
            \EndIf
        \EndFor

        \If{$e$ is not a splitting edge}[Backward propagation for non-splitting edges]
            \ForAll{$f \in \Prev_H(e) \setminus E_f$}
                \State Enqueue $f$ into $Q$
            \EndFor
        \EndIf

        \State remove $e$ from $H$
        \If{$e \in E_f$}  \label{algline:improved_feasible_graph:early_termination}
            \State remove $e$ from $E_f$
            \If{$|E_f|=0$} 
                \State{\Return Empty Graph}
            \EndIf
        \EndIf
    \EndWhile

    \State \Return $H$
\EndFunction
\end{algorithmic}
\end{algorithm}

\begin{lemma}[Edges Removed From the Maximal Step-Extended Component]
\label{lem:msec_removal_correctness}
Let $G$ be a computation graph, $V_0$ the set of initial vertices, and $E_R$ the initial set of step-pendant edges within the step-reachable subgraph $H \subseteq G$. The set of edges removed by \cref{alg:improved_feasible_graph} is exactly the maximal step-extended component $\mathsf{MSEC}_G(E_R)$. Consequently, the resulting graph $H$ (excluding isolated vertices) is is a feasible graph with respect to $E_f$.
\end{lemma}

\begin{proof}
The proof establishes that the iterative removal process in \cref{alg:improved_feasible_graph} precisely executes the recursive definition of $\mathsf{MSEC}_G(E_R)$.

\textbf{1. Initialization and Reachability:}
By \cref{lem:step_pendant_edge_enumeration}, the initial subgraph $H$ is the step-reachable subgraph from $V_0$, and $E_R$ contains all edges that are step-pendant in $H$. Any edge $e \in G \setminus H$ is not step-reachable from $V_0$ and thus cannot participate in any feasible walk. These edges are implicitly excluded from the feasible graph, consistent with the construction of the $\mathsf{MSEC}$.

\textbf{2. Soundness of Propagation (Local to Global):}
The algorithm uses a queue $Q$ to propagate the step-pendant property. We show that every edge $f$ added to $Q$ during the iteration is step-pendant in the current subgraph $H -e$.
\begin{itemize}
    \item \textbf{Horizontal Propagation:} 
    If $e$ is not a merging edge, its removal leaves its next edge $f \in \Next_H(e)$ without a previous edge, leading $f$ to be ex-pendant. 
    Conversely, if $e$ is not a splitting edge, its previous edges $f \in \Prev_H(e) \setminus E_f$ lose their next edge, leading $f$ to be ex-pendant.     
    \item \textbf{Vertical Propagation:} 
    \begin{itemize}
        \item If $f \in \ISucc_H(e)$ and $e$ was its only index-precedent ($|\IPrec_H(f)|=1$), $f$ becomes step-pendant upon removing edge $e$.
        \item If $f \in \IPrec_H(e)$, where $f$ is not a ex-cover edge, and $e$ was its only index-succedent ($|\ISucc_H(f)|=1$), $f$ becomes step-pendant upon removing edge $e$.
    \end{itemize}
\end{itemize}
Each case corresponds to the defining conditions of a step-pendant edge. Thus, $f$ is indeed a step-pendant edge in $H - e$. 
Now, let $e'$ be the first such edge in $\mathsf{MSEC}$ that triggers some edge $f$ in $H-e'$ to become step-pendant.
By the definition of maximal step-extended component, $f$ also should be removed, otherwise the component is not maximal which leads to a contradiction.
Thus, every removed edge is a member of $\mathsf{MSEC}_G(E_R)$.

\textbf{3. Completeness of Removal:}
Assume there exists an edge $e_s \in E(\mathsf{MSEC}_G(E_R))\cap E(H)$ that remains in the graph after the algorithm terminates. 
By the definition of the maximal step-extended component, there must be a step-adjacency chain from some initial edge in $E_R$ to $e_s$ that propagates the step-pendant property. 

If $e_s$ were an initially step-pendant edge, it would have been included in the initial queue $Q$ and subsequently removed. Thus, $e_s$ must have become step-pendant through the removal of its neighbors. 
Without loss of generality, let $e_s$ be the first edge in the chain that was not removed. 
Since the algorithm explores all step-adjacent neighbors of every removed edge and enqueues any step-adjacent neighbor that satisfies the step-pendant conditions, 
$e_s$ would necessarily have been enqueued when its preceding step-adjacent neighbor in the chain was processed, ensuring its subsequent removal. This creates a contradiction.
 Therefore, the termination of the loop with an empty $Q$ ensures that no step-pendant edges belonging to the maximal step-extended component remain in $H$.

\textbf{4. Early Termination:}
If all final edges $E_f$ are removed (\cref{algline:improved_feasible_graph:early_termination}), the algorithm returns an empty graph. 
This is sound because the absence of $E_f$ implies that no feasible walk exists to any edge of $E_f$, rendering all remaining edges step-pendant.

Therefore, the total set of removed edges is exactly $\mathsf{MSEC}_G(E_R)$, and the final graph $H$ without isolated vertices is a feasible graph of $G$.
\end{proof}

The improved algorithm replaces global sweeping with a queue-driven elimination
process centered on \emph{step-pendant edges}. 

Each edge is examined and removed at most once, and removal propagates only
through step-pendant adjacent edges (including index-precedents or index-succedents).  
This eliminates the need for repeated global sweeps and reconstruction of index layers.

Note that the only structural difference between the feasible graph in the original framework and the one constructed above is the use of ex-cover edges instead of standard cover edges. 
Consequently, the preservation property for feasible walks remains valid, as established in the following lemma.

\begin{lemma}[Preservation of Feasible Walks]
\label{lem:feasible_walk_preserved}
Let $H$ be the feasible graph constructed by \textproc{ComputeFeasibleGraph()} from a computation graph $G$ with initial vertices $V_0$ and designated final edges $E_f$, as described in \cref{alg:improved_feasible_graph}. Then every valid computation walk starting from $V_0$ and terminating at an edge in $E_f$ is preserved in $H$; i.e., all edges belonging to such a feasible walk remain in $H$ after the execution of the algorithm.

\begin{proof}
Suppose, for the sake of contradiction, that there exists a feasible walk $W = (v_0, e_0, v_1, e_1, \dots, v_k, e_k)$ where $e_k \in E_f$, such that $W$ is not entirely contained in $H$. Let $e \in W$ be the \emph{first removed edge} during the execution of \textproc{ComputeFeasibleGraph()}. By the algorithm's logic, $e$ must have been identified as a step-pendant edge at the time of its removal.

\paragraph{Case 1: $e$ is identified as step-pendant due to a lack of index-precedent edges.}
By the definition of step-pendancy, a floor edge cannot be step-pendant due to a lack of index-precedents. Thus, $e$ must be a non-floor edge. However, every non-floor edge in a feasible walk $W$ necessarily has an index-precedent edge within the same walk structure. Since $e$ is the \emph{first} removed edge of $W$, its index-precedent must still exist in $H$ at the time of $e$'s removal. Hence, $e$ cannot be step-pendant by this condition, leading to a contradiction.

\paragraph{Case 2: $e$ is identified as step-pendant due to a lack of index-succedent edges.}
If $e$ were a ceiling edge, it would belong to the ex-cover edge set $C$ toward $E_f$. By definition, ex-cover edges cannot be step-pendant due to a lack of index-succedents. Thus, $e$ must be a non-ceiling edge. However, every non-ceiling edge in a feasible walk $W$ has at least one index-succedent edge. Since $e$ is the first removed edge, its index-succedent in $W$ remains in $H$. This contradicts the assumption that $e$ is step-pendant by this condition.

\paragraph{Case 3: $e$ is identified as step-pendant due to a lack of previous edges.}
If $e$ is an initial edge (where $v_0 \in V_0$), it cannot be step-pendant due to a lack of previous edges by definition. If $e$ is a non-initial edge, it must have a previous edge in the walk $W$. Since $e$ is the first removed edge, this previous edge must still be present in $H$. Therefore, $e$ cannot be the first removed edge under this condition, a contradiction.

\paragraph{Case 4: $e$ is identified as step-pendant due to a lack of next edges.}
If $e \in E_f$, it cannot be step-pendant due to a lack of next edges by definition. If $e \notin E_f$, it must have a next edge in the walk $W$ toward $E_f$. Since $e$ is the first removed edge, its next edge in $W$ is still present in $H$. Thus, $e$ cannot be step-pendant under this condition, a contradiction.

\paragraph{Conclusion:}
In all cases, the assumption that a first removed edge exists within a feasible walk $W$ leads to a contradiction. Therefore, every feasible computation walk is entirely preserved in $H$.
\end{proof}
\end{lemma}

The feasible graph is intended solely to preserve all feasible computation walks with respect to $E_f$ while removing all the step-pendant edges. 
The difference between the original and current constructions lies in ex-cover edges $\widehat{C}_E$, which are not unique. In the original work, $\widehat{C}_E = \emptyset$ was implicitly assumed, whereas here it may be non-empty and is explicitly incorporated into the step-pendant edge enumeration.  
This generalization does not affect walk preservation: every feasible computation walk remains intact. The current construction thus aligns formally with the ex-cover definition, providing flexibility while maintaining correctness, making both the original and improved algorithms valid.

Since the feasible graph trims all the step-pendant edges,and the footmark graph is grid-aligned, the total collapse property in the original framework also preserved in the improved one, for the clarity we restate the lemma.
\begin{lemma}[Total Collapse of the Feasible Graph]\label{lem:total_collapse}
Let $G$ be a subgraph of grid-aligned footmark graph with a unique initial node hosting at least one valid feasible walk with respect to the designated final edge $e_f$. 
Let $G_f = \mathsf{Feasible}(G -e)$ denote the residual feasible graph obtained after removing a single edge $e \in E(G)$. 
If $e$ is an essential edge, which is contained in every valid feasible walk and located prior to the second splitting edge of any feasible walk, then the resulting feasible graph has no edges causing a total collapse.
\end{lemma}

\begin{proof}
We prove this structural collapse by establishing a stronger, generalized inductive hypothesis: \emph{For any subgraph $G$ of a footmark graph, if a set of removed feasible edges $E_r$ contains an essential edge $e$ before the second splitting edge, then the resulting feasible graph of $G - E_r$ is empty
where all target nodes $v$ of edges in $E_f$ share an identical time value $\timeOf(v)$.}
Note that an identical time value $\timeOf(v)$ implies an identical spatial coordinate $(index(v), tier(v))$ due to grid-aligned property.

We proceed by induction on the total number of edges $n = |E(G-E_r)|$.

To begin with, if  $|E(G-E_r)|=0$, the feasible graph is trivially empty.

For the induction hypothesis, assume that for any context graph $G-E_r$ containing exactly $i$ edges for such essential edge $e \in E_r$, the removal of such an essential edge guarantees total collapse, yielding $G_f = \emptyset$.

Now, we take an inductive step for $|E(G-E_r)|=i+1$.
The proof proceeds in two stages: first, we establish the existence of at least one step-pendant edge $e'$ within $G' \setminus E_r$ via contradiction.
 Second, we apply mathematical induction by pruning $e'$, showing that the remaining graph collapses while preserving the consistency of the feasible graph structure with respect to $E_f$.
Consider a graph $G' - E_r$ containing $i+1$ edges with an essential edge $e \in E_r$. 
Suppose, for contradiction,  that $G' - E_r$ contains no step-pendant edges. Let $F$ be the set of all edges whose target nodes possess the maximum time value in the grid-aligned graph. 
By the maximality of time, any edge $f \in F$ has neither next edges nor index-succedent edges. 
Consequently, $f$ must belong to $E_f$; otherwise, $f$ would be ex-pendant edge which is a step-pendant edge, triggering an immediate contradiction. 
Furthermore, if $e$ were a non-splitting edge, its immediate predecessors $\Prev(e)$ would become step-pendant upon $e$'s erasure. Thus, $e$ must be a splitting edge.

We categorize the removal of an edge $e$ into three exhaustive cases based on its essentiality in $G = G' - F$ and its relation to the set $F$:
(1) $e$ is essential, (2) $e \in F$ (non-essential), and (3) $e \notin F$ and is non-essential. We demonstrate that each case leads to a structural contradiction regarding the existence of step-pendant edges.

\begin{itemize}[leftmargin=*]
	\item \textbf{Case 1: $e \notin F$, and $e$ is essential in $G$.} \\
	Let $f$ be an edge in $F$. Note that throughout the following sub-cases, any shared incident node or step-adjacent structure between $e$ and $f$ forces identical $(index, tier)$ coordinates due to the grid-alignment property, which we use to derive structural contradictions.
	
	Suppose $f$ shares step-adjacent edges with any edge in $E_r$; then, by grid-alignment, these edges must share identical $(index, tier)$ coordinates with an incident node of $f$. This is impossible, as all edges with the maximum time value are constrained to belong to $F$.
		
	Since sharing any step-adjacent edge with $e$ leads to a structural contradiction, $f$ cannot share any such edges with $e$. Let $E_f^*$ denote the set of all previous edges of $F$. 
	Now, we examine whether a step-pendant edge $e'$ of $G$ remains step-pendant in $G' = G + F$ by investigating the step-adjacent edges of $f \in F$:
	
	\begin{itemize}[leftmargin=*]
	    \item \textbf{Case where $e' \in \ISucc(f)$ or $e' \in \Next(f)$:} This case is impossible, as $f$ is chosen to be an ex-pendant edge with no index-succedent edges.
	    
	    \item \textbf{Case where $e' \in \Prev(f)$:} Since $e' \in E_f^*$, the edge $e'$ cannot be step-pendant due to the presence of next-edges. Consequently, the addition of $f$ does not alter the step-pendancy status of $e'$, because the role of $f$ as a next-edge for $e'$ is already accounted for in the feasibility constraints.
	    
	    \item \textbf{Case where $e' \in \IPrec(f)$:} If $f$ is an index-succedent edge of $e'$, then $e'$ is a cover edge in $G$ toward $E_f^*$. This implies that $e'$ is not a step-pendant edge in $G$ due to the lack of index-succedent edges; thus, the addition of $f$ does not alter this status.
	\end{itemize}
	
	Consequently, the step-pendancy status of the step-adjacent edges of $e'$ in $G' - E_r$ remains identical to their configuration in $G - E_r$. Since the feasible graph of $G - E_r$ is empty by the IH, $G' - E_r$ must contain a step-pendant edge, which contradicts the initial assumption that $G' - E_r$ is step-pendant free.

	\item \textbf{Case 2: $e \in F$, implying $e$ is not essential in $G$.} \\
	By definition, any $e \in F$ is a designated final edge in $E_f$ since $e$ is the essential edge. As established, $e$ must be a splitting edge; otherwise, its removal would trigger an immediate step-pendancy in its previous edges, contradicting our assumption. 
	
	Since the first essential splitting edge in our grid-aligned model must be a floor edge that terminates in $E_f$ (by DTM and no step-pendency), any essential splitting edge $e$ must coincide with an edge in $E_f$. 
	
	Furthermore, if $e$ and $f$ (where $f \in F$) share a previous edge as a splitting junction, we examine the following structural sub-cases:
	\begin{itemize}[leftmargin=*]
	    \item \textbf{Both are Floor Edges:} The existence of $f$ provides an alternative feasible walk, which directly contradicts the essentiality of $e$.
	    \item \textbf{Either is a non-floor edge:} This configuration violates grid-aligned property, which explicitly excludes the existence of infeasible edges at the first splitting junction, rendering this configuration structurally impossible.
	\end{itemize}
	Consequently, the removal of $e$ leads to an immediate structural contradiction, confirming that this case is impossible.
   
	\item \textbf{Case 3: $e \notin F$ and $e$ is non-essential in $G$.} \\
		In this scenario, $e$ is non-essential in $G$ due to the existence of alternative feasible paths. However, the introduction of the terminal set $F$ prunes these alternatives, promoting $e$ to an essential edge in $G' - E_r$. This leads to the infeasibility of $W_f + f$ for some feasible walk $W_f$ previously feasible in $G$, as the history constraints imposed by an edge $f \in F$ render the computation walk invalid.
	    
	    The infeasibility arises at some ceiling edge $e_c$ of $W_f$, where $\init(e_c) \notin \IPrec(\term(f))$. If $e_c$ were not a cover edge, it would necessarily become a step-pendant edge, contradicting the assumption that $G' - E_r$ is step-pendant free. 
	    
	    Crucially, $e_c$ cannot be a cover edge because $G'$ is a grid-aligned footmark graph: the target nodes of all edges in $F$ share an identical time step and $(index, tier)$ coordinates. 
This synchronization renders the existence of a valid cover edge for such a configuration structurally impossible. Thus, $e_c$ is forced to be step-pendant, leading to an immediate contradiction. 
\end{itemize}
Through these exhaustive cases, we demonstrate that there exists at least one step-pendant edge, say $e'$. 
By designating $e'$ as the first edge extracted by the fixed-point pruning sequence, we reduce the graph $G'$ to a smaller configuration that remains within the scope of our inductive hypothesis.

\begin{itemize}[leftmargin=*]
    \item \textbf{Case A: $e'$ is not contained in any feasible walk.} \\
    In this case, removing $e'$ does not eliminate any valid feasible traces.
    Thus, $e \in E_r$ remains a strictly essential edge in the reduced graph $G'' = G' -e'$. Since $|E(G''-E_r)| = i$, the Inductive Hypothesis (IH) applies directly to $G''$. Consequently, the feasible graph of $G'' -E_r$ is empty, implying that the feasible graph of $G' - E_r$ is also empty.
    \item \textbf{Case B: $e'$ is contained within some feasible walks.} \\
	By including $e'$ in the set of pruned edges, we define $E_r' = E_r \cup \{e'\}$. Since $e'$ is a step-pendant edge, its removal reduces the total number of edges in the graph, i.e., $|E(G' - E_r')| = i$. This strictly satisfies the condition for the Inductive Hypothesis. 
  As the removal of a step-pendant edge does not invalidate the essentiality of $e$, the reduced graph $G' - E_r'$ collapses to an empty feasible graph by the IH, which in turn confirms the structural collapse of $G' - E_r$.
\end{itemize}

By mathematical induction, the removal of an essential edge $e$ located before the second splitting boundary triggers the total structural collapse of the graph: $G_f = \emptyset$ for $G-E_r$ where $E(G-E_r)=n \ge 1$.
\end{proof}

\paragraph{Discussion}
Conceptually, the original sweep-based algorithm is simpler and minimizes exceptional cases by reconstructing the graph layer by layer. 
In contrast, the improved algorithm is more faithful to the abstract definition of a feasible graph itself, as it directly eliminates edges that violate fundamental feasibility conditions.

Since all correctness lemmas regarding the feasible graph—as well as those dependent on it, such as walk verification and edge extension—rely solely 
on the total collapse and the preservation of feasible walks, replacing the original construction with the improved algorithm does not weaken the overall integrity of the framework. 
This shift from walk reconstruction to the systematic elimination of infeasibility is the primary source of the observed complexity improvement, 
while maintaining full correctness and compatibility with the NP verifier simulation framework.

\paragraph{Complexity Analysis}

\begin{lemma}[Time Complexity of Improved Feasible Graph Construction]
\label{lem:improved_feasible_graph_time_complexity}
Let $G$ be a computation graph of width $w$ and height $h$. The total worst-case time complexity of \textproc{ComputeFeasibleGraph()} is $\bigO(w h^{3} \log(wh))$.

\begin{proof}
The complexity analysis is divided into three main phases:

\textbf{1. Computing Cover Edges:}
The algorithm begins by computing the cover edges, which takes $\bigO(w h^{3} \log(wh))$ time according to \cref{lem:time_improved_computing_cover_edges}.

\textbf{2. Initial Step-Pendant Identification:}
The algorithm then identifies the initial step-reachable component and its step-pendant edges. By \cref{lem:step_pendant_time_complexity}, this procedure takes $\bigO(wh^3 \log(wh))$ time.

\textbf{3. Iterative Edge Removal:}
The algorithm maintains a queue $Q$ of candidate step-pendant edges for iterative removal. Since each edge in $G$ is enqueued and processed at most once, we analyze the cost per edge removal:
\begin{itemize}
    \item \textbf{Edge Enumeration and Counting}: The algorithm enumerates and counts neighbors ($\ISucc, \IPrec,$ $\Prev, \Next$). Each set has size at most $\bigO(h)$, and utilizing the indexing structure, these operations take $\bigO(h \log h)$ time.
    \item \textbf{Graph Modification}: Removing an edge from the adjacency list structure takes $\bigO(h)$ time.
    \item \textbf{Membership and Queue Operations}: Checking or updating membership in $E(H)$ costs $\bigO(h)$ via the adjacency list. Membership tests for the cover map $C[\cdot]$ take $\bigO(\log h)$ using ordered sets within each edge slice. Operations on the deque $Q$ take $\bigO(1)$ time.
    \item \textbf{Final Edge Verification}: Checking if $e \in E_f$ and updating the set takes $\bigO(\log(wh))$ time, as $|E_f| = \bigO(wh^2)$.
\end{itemize}

Given $|E(G)| = \bigO(w h^2)$ edges and that each removal involves local operations totaling $\bigO(h \log(wh))$, the total cost for the removal phase is:
\[
\bigO(w h^{2}) \times \bigO(h \log(wh)) = \bigO(w h^{3} \log(wh)).
\]

\textbf{Conclusion:}
Combining the costs of all phases, the overall time complexity is:
\[
\underbrace{\bigO(w h^{3} \log(wh))}_{\text{Cover Edges}} + \underbrace{\bigO(w h^{3} \log(wh))}_{\text{Initial Identification}} + \underbrace{\bigO(w h^{3} \log(wh))}_{\text{Iterative Removal}} = \bigO(w h^{3} \log(wh)).
\]
This complexity is polynomial with respect to the size of the input graph, completing the proof.
\end{proof}
\end{lemma}
\emph{Remark:} In practice, the number of edges considered in cover sets is typically much smaller than the worst-case bound, since not all edges can act as ceiling edges simultaneously.

\subsection{Impact of Improved Feasible Graph Construction}

\Cref{alg:walk_verification} presents the walk verification algorithm with minor structural modifications. 
It is important to note that these modifications do not alter the fundamental correctness of the procedure, nor are they the primary source of the improved time complexity. 
The reduction in total complexity is derived exclusively from the integration of the optimized feasible graph construction (\cref{alg:improved_feasible_graph}) within the verification loop. 

While specific implementation details of the verification steps are discussed in subsequent sections, the following analysis focuses on how the improved feasible graph construction leads to a superior overall time complexity. \Cref{lem:improved_verify_existence_of_walk_time}, together with the detailed logic of \cref{alg:extend_futile_walks} provided later, establishes the new complexity bounds by focusing on the dominant sub-algorithms. 
For a complete foundational analysis, readers may refer to the original framework.

\begin{algorithm}
\caption{Verification of Walk} \label{alg:walk_verification}
\begin{algorithmic}[1]
\Input{$G$: a computation graph, $f$: the target edge for verification}
\Output{The Verification Result \True/\False}
\Function{VerifyExistenceOfWalk}{$G_U, V_0, e_t$}
\State{Let $E_f \gets \{e_t\}$}
\State{Set $G \gets$ \Call{ComputeFeasibleGraph}{$G_U, V_0, E_f $}}
\While{$G$ contains any edge of $E_f$} \label{alg_line:walk_verification:start_of_while}
    \State{Let $(e, W) \gets$ \Call{FindTargetRedundantFutileEdge}{$G_U, V_0, e_t$}} \label{alg_line:walk_verification:find_futile_edge}
    \If{$e = e_t$}
        \State{\Return W}
    \ElsIf{$e =NIL$}
        \State{\Return $NIL$}
    \EndIf
    \State{Set $G \gets G-e$} \label{alg_line:walk_verification:removing_an_edge}
    \State{Set $G \gets$ \Call{ComputeFeasibleGraph}{$G, V_0, E_f$}}\label{alg_line:walk_verification:recompute_feasible_graph}
\EndWhile \label{alg_line:walk_verification:end_of_while}
\State{\Return $NIL$}
\EndFunction
\end{algorithmic}
\end{algorithm}

\begin{algorithm}
\caption{Find Computing-Redundant or Computing-Disjoint Edge for Pruning} \label{alg:find_futile_or_targeted_walk}
\begin{algorithmic}[1]
\Input{$e_t$: the verification target edge, $G_U$: $e_t$-augmented footmarks, $G$: Feasible graph, $V_0$: Initial vertices } 
\Output{Computing-Redundant or Computing-Futile Edge of feasible graph $G$ with respect to $E_f=\{e_t\}, V_0$.}
\Function{FindTargetRedundantFutileEdge}{$G_U, G, V_0, e_t$}
\State{Let $E_f \gets \{e_t\}$}\Comment{$E_f$: set of verification target edge}
\State{Let $R \gets$ an empty graph}
\While{$G$ is not empty}[Loop until a computing-targeted/computing-futile edge found]\label{alg_line:find_obsolete_walk:start_loop}
    \State{Let $W \gets$ \Call{TakeArbitraryWalk}{$G, V_0$}}\label{alg_line:take_arbitrary_walk}
    \If{$e_t$ in $W$}[$W$ is a computing-targeted walk]
        \State{\Return {$e_t$}} \Comment{$e_t$ is not always at the end of $W$} 
    \Else[$W$ can be computing-embedded walk]
        \State{Set $H \gets$ \Call{PruneWalk}{$G_U, G, V_0, E_f, W, \False$}} 
	 	\If{$H$ is empty}[$H$ does not contain any of $E_f$]
            \State{Set $G \gets$ \Call{PruneWalk}{$G_U, G, V_0, E_f, W, \True$}} \label{alg_line:prune_walk}
            \State{Let $f \gets$ \Call{FindDisjointEdge}{$R, G$}}\label{alg_line:find_disjoint_edge}
            \State{\Return $f$} \label{alg_line:return_walk_and_disjoint_edge} \Comment{Return futile or redundant edge with the walk}
        \Else
	     \State{Set $G \gets H$}
        \EndIf
    \EndIf
\EndWhile \label{alg_line:find_obsolete_walk:end_loop}
\State{\Return ($NIL$, $NIL$)}
\EndFunction

\Statex \Comment{Any consistent choice (e.g., always first edge) result in deterministic algorithm}
\Function{TakeArbitraryWalk}{$G, V_0$}
\State{Let $S$ be the empty dynamic array of `transition cases'}\Comment{Empty Surface}
\State{Let $W$ be the empty list of edges}
\State{Let $e$ be an edge in $G$ incident with a node in $V_0$}
\While{$e \ne NIL$}
    \State{Update surface $S[\text{index}(u)]$ with the transition case to which node $u$ belongs}
    \State{Append $e$ to walk $W$}
    \State{Set $e \gets$ an edge of $Next_G(e)$ on surface $S$ if exists, otherwise $NIL$}
\EndWhile 
\State{\Return $W$}
\EndFunction

\Function{FindDisjointEdge}{$R, G$}[Find disjoint edge from $R$ ]\label{alg:search_disjoint_edge}
\State{Let $e$ be the first edge of walk $R$}
\While{$e$ is not NIL}\label{alg_line:search_disjoint_edge:while_start}  
    \If{$e \not\in E(G)$} \label{alg_line:search_disjoint_edge:select_disjoint_edge}
		\State{Let $E_N \gets \Outgoing(\init(e))$}
        \State{\Return an edge $f \in E_N$ if $|E_N|>0$ otherwise \NIL}
    \EndIf
    \State{$e \gets \nextOf_R(e)$ }   \Comment{If $\nextOf_R(e)$ does not exist then $e$ is \NIL}
\EndWhile
\State{\Return \NIL}
\EndFunction
\end{algorithmic}
\end{algorithm}

\begin{algorithm}
\caption{Pruning an Edge Given a Computing-Futile Walk} \label{alg:puruning_an_edge_of_walk}
\begin{algorithmic}[1]
\Input{$G$: Feasible Graph, $e_t$: Verification target edge, $G_U$: $e_t$-augmented footmarks, $W$: Computation Walk}
\Output{The feasible graph $G'$ in which an edge of $W$ removed.}

\Function{PruneWalk}{$G_U, G, V_0, e_t, W, preserveFutile$}
\State{Let $E_o \gets \emptyset$ and $E_f \gets \{e_t\}$}
\State{Set $e' \gets$ \Call{FindFirstSplittingEdgeOrFinalEdge}{$G,W$}}\label{alg_line:find_first_splitting_edge}
\If{preserveFutile}
	\State{Set $E_o \gets$ \Call{ExtendFutileWalks}{$G_U, G,  E_f$}}\Comment{Add the extendable futile edges }
\EndIf
\State{Let $G' \gets$ \Call{ComputeFeasibleGraph}{$G-e', V_0, E_f \cup E_o$}}\Comment{Remove $e'$ from feasible graph if exists}\label{alg_line:find_obsolete_edge:recompute_feasible_graph}
\State{\Return $G'-E_o$}\Comment{Do not recover removed edge, this ensure polynomial time complexity}
\EndFunction

\Function{FindFirstSplittingEdgeOrFinalEdge}{$G, W$}
\State{Let $e$ be the first edge of walk $W$}
\While{$e$ is not the final edge of walk $W$}\label{alg_line:find_first_splitting_edge:while_start}
    \If{$e$ is splitting edge}
        \State{\Return $e$}
    \EndIf
    \State{$e \gets next_W(e)$ }
\EndWhile
\State{\Return $e$} \Comment{It returns final edge of computation walk if no splitting edge found}
\EndFunction
\end{algorithmic}
\end{algorithm}

\begin{lemma}[Improved Time Complexity of Verifying Existence of Walk]
\label{lem:improved_verify_existence_of_walk_time}
Let $G$ be a computation graph of width $w$ and height $h$, where $|E(G)| = \bigO(wh^{2})$. The worst-case time complexity of \textproc{VerifyExistenceOfWalk()} in \cref{alg:walk_verification} is bounded by
\[
\bigO(w^{3} h^{7} \log(wh)).
\]
\end{lemma}

\begin{proof}
Recall that \textproc{ComputeFeasibleGraph()} is implemented using the improved algorithm described in \cref{lem:improved_feasible_graph_time_complexity}, which runs in $\bigO(w h^{3} \log(wh))$ time. 
The \textproc{VerifyExistenceOfWalk()} procedure consists of a \textbf{while}-loop that repeatedly  invokes \textproc{FindTargetRedundantFutileEdge()} as follows:
\begin{itemize}
    \item \textproc{FindTargetRedundantFutileEdge()} (\cref{alg:find_futile_or_targeted_walk}) internally invokes \textproc{PruneWalk()} (\cref{alg:puruning_an_edge_of_walk}). With the improved feasible-graph computation, each call to \textproc{PruneWalk()} is dominated by a single execution of \textproc{ComputeFeasibleGraph()}, costing $\bigO(w h^{3} \log(wh))$ time.
    \item Within the internal loop of \textproc{FindTargetRedundantFutileEdge()}, at least one edge is removed from $G$ in each iteration, ensuring progress toward termination.
    \item Since $|E(G)| = \bigO(wh^{2})$, the number of iterations in the while loop of \textproc{FindTargetRedundantFutileEdge()} is bounded by $\bigO(wh^{2})$. 
\end{itemize}

Furthermore, since at least one edge is removed from $G$ per each iteration, the while loop of \textproc{VerifyExistenceOfWalk()} can itself involve up to $\bigO(wh^{2})$ iterations in the worst case. Therefore, the total time complexity is bounded by:
\[
\bigO\bigl((wh^{2}) \cdot (wh^{2}) \cdot (w h^{3} \log(wh))\bigr) = \bigO(w^{3} h^{7} \log(wh)),
\]
as claimed.
\end{proof}

\begin{table}[H]
\centering
\caption{Comparison Between the Original and Improved Feasible Graph Algorithms}
\label{tab:feasible_graph_comparison}
\begin{sloppypar}
\begin{tabularx}{\textwidth}{|p{3cm}|X|X|} 
\hline
\textbf{Aspect} 
& \textbf{Original Algorithm (Sweep-Based)} 
& \textbf{Improved Algorithm (Definition-Driven)} \\ \hline

Primary design goal 
& Preserve feasible walks by construction through bidirectional sweeping and index-wise expansion 
& Directly realize the definition of the feasible graph as the complement of a maximal step-extended component \\ \hline

Conceptual basis 
& Algorithm-first: Simple iteration of an additive sweep-based algorithm  
& Definition-first: Algorithm follows \cref{def:feasible_graph} \\ \hline

Treatment of step-pendant edges 
& Identified implicitly via repeated sweep and convergence checks 
& Explicitly removed using local step-pendant conditions and adjacency rules \\ \hline

Cover/ex-cover edge computation 
& Path existence is checked individually during ceiling-adjacency expansion 
& Path existence collected in a single backward traversal after expansion \\ \hline

Handling of step-reachability 
& No explicit step-reachability computed 
& Global reachability computed once during computing initial step-pendant edges  \\ \hline

Algorithmic structure 
& Iterative sweeping until convergence (left-to-right and right-to-left) 
& Single monotone elimination process based on queue-driven edge removal \\ \hline

Edge removal strategy 
& Edge slice based additive sweep with global stabilization criterion 
& Local removal of step-pendant edges adjacent to already removed edges \\ \hline

Worst-case time complexity of construction 
& $\bigO(w^{2}h^{4}(h \log h + \log w )$ 
& $\bigO(w h^3 \log(wh))$ \\ \hline

Impact on walk verification complexity 
& Leads to $\bigO(w^{4}h^{8}(h\log h  + \log w))$ bound for walk verification 
& Reduces walk verification bound to $\bigO(w^3 h^7 \log(wh))$ \\ \hline

No step-pendant edges
& Guaranteed by exclusion of step-pendant edges until convergence
& Guaranteed by iterative removal of step-pendant edges \\ \hline

Faithfulness to feasible graph definition 
& Indirect (proved by the absence of step-pendant edges) 
& Direct (definition realized by construction) \\ \hline

Feasible walk preservation 
& Guaranteed by construction and proven explicitly 
& Guaranteed by nearly the same proof of prior results \\ \hline

Intended role in the paper 
& Foundational construction in the original theoretical framework 
& Performance-oriented refinement suitable for implementation and experimentation \\ \hline
\end{tabularx}
\end{sloppypar}
\end{table}

\section{Improvements by Restricting Candidate Set for Verification}\label{sec:improvements_on_edge_extension}
In this section, we present both a theoretical reduction in the number of
verifier invocations and additional practical optimizations.

The key observation is that the verification cost is dominated by
\textproc{VerifyExistenceOfWalk()}.
By restricting the candidate set of edges to be verified,
we obtain a provable reduction in the number of verifier calls per edge.

First, we introduce a walk-first strategy; this significantly
reduces the average running time in practice,
without affecting the asymptotic upper bound.

Second, when edge extension becomes necessary, the original framework examines all
edges whose maximum tier is below a global threshold.
This results in a large candidate set, many of whose edges cannot contribute
to new computation walks.
We reduce this overhead by restricting the candidate set to edges that are structurally capable of participating in new computation walks. 
Specifically, an edge $e$ is considered a candidate for verification only if it satisfies at least one of the following structural conditions:
\begin{enumerate}
\item At least one index-precedent of $e$ is a \textit{merging edge}, a \textit{combining edge}, or an \textit{indirect index-precedent edge}.
\item $e$ is an index-succedent of a \textit{merging} or \textit{combining edge}, or $e$ itself acts as an \textit{indirect index-succedent} of another edge.
\end{enumerate}
These conditions identify the specific junctions in the computation graph where new computation walks can diverge only immediately above converged computation paths within the same edge slice structure.
A combining edge refers to an edge that converges into a unified transition case from distinct previous transition cases (as formally defined in \cref{subsec:reduction_of_extension_degree}). 
By focusing on these topological "hubs," we ensure that all potential extensions are captured while safely excluding edges that cannot contribute to any new computation walks.
As a result, the number of candidate edges examined during each extension
step is reduced from a global polynomial bound of $\bigO(wh^2)$ to a smaller polynomial bound $\bigO(wh)$ depending only on local graph parameters such as node degree and adjacency structure.

\subsection{Walk-First Strategy Prior to Edge Extension via verification}

Before invoking \textproc{VerifyExistenceOfWalk()} for a candidate edge,
the algorithm first attempts to extend the computation walk
using the computation walk information already obtained from previous verifications.

The baseline pruning algorithm is described in detail in our previous work.
In this paper, we do not restate the original procedure, and instead focus on
an implementation-oriented reformulation that reduces redundant verification
calls while preserving the same semantics.
The key difference is that \cref{alg:find_futile_or_targeted_walk} returns
not only the selected edge but also the corresponding computation walk,
so that \cref{alg:walk_verification} can return the walk itself
instead of \True, or \textsc{Nil} instead of \False.

Although \cref{alg:find_futile_or_targeted_walk} is presented as an
edge-oriented procedure, its verification mechanism is inherently walk-based.
Every verification step begins by constructing an arbitrary computation walk
$W$ from the initial set $V_0$ via \textproc{TakeArbitraryWalk()}.
If the verification succeeds, the algorithm necessarily holds a computing-targeted walk containing the verified edge.

This observation allows edge verification to be interpreted as the acquisition
of a valid computation walk containing the target edge.
In other words, returning a verified edge is equivalent to returning a computation walk reaching the verification target edge.
Consequently, once a new computation walk is obtained, the computation can be
extended along the walk without further verification, all the way to the final states.

Based on this observation, we adopt a \emph{walk-first strategy}.
Instead of repeatedly verifying individual edges, the algorithm first performs
verification to obtain a valid computation walk leading to the target edge, and then extends the walk
as far as possible using already verified edges.
Verification is resumed on another candidate edge 
only when the extension reaches a final state—even for edges 
that have been extended previously.
Since this step takes at most polynomial running time of the verifier, this reformulation does not affect correctness, but enables a more efficient
implementation by reducing redundant verifier invocations.

Once a valid computation walk can be extended directly,
the edges can be extended until accepted or rejected without triggering the verifier.
This strategy does not alter the worst-case asymptotic complexity,
since adversarial instances may still require full verification.
However, in practice, a substantial fraction of extensions
are resolved through direct walk construction,
thereby significantly reducing the average running time.

\begin{algorithm}
\caption{Extension of Edges of Footmarks} \label{alg:walk_first_extension}
\begin{algorithmic}[1]
\Procedure{ExtendByVerifiableEdges}{$V_0$: \In, $Q$:\In, $H$:\InOut, $E_b$:\Out} 
\ForAll{edge $e$ in $Q$}
     \If{$H$ already contains $e$}
        \State{\textbf{continue}}
    \EndIf
    \State{$W \gets$ \Call{VerifyExistenceOfWalk}{$H+e, V_0, e$}}
    \If{$W$ is not $NIL$}
        \State{$R \gets$ \Call{ExtendEdgeDirectlyWithWalk}{$G, H, W, E_b$}}
        \If{$R$ is not $NIL$}
            \State{\Return $R$}
        \EndIf
    \EndIf         
\EndFor
\EndProcedure
\end{algorithmic}
\end{algorithm}

\begin{algorithm}
\caption{Direct Extension Using a Verified Computation Walk}
\label{alg:extend_edge_directly_with_walk}
\begin{algorithmic}[1]
\Function{ExtendEdgeDirectlyWithWalk}{$G, H, W, E_b$}
    \State Let $S \gets$ empty ceiling-edge array;  $R \gets$ empty result array
    \State Let $T \gets$ empty stack
    \Comment{Initialize ceiling edges and result from the verified walk}

    \ForAll{edge $e=(u,v)$ in $W$}
        \If{$H$ does not contain $e$}
            \State \Break \Comment{$e$ can  be merging edge, not the last edge of $W$}
        \EndIf
        \State $S[\indexOf(e)] \gets e$
        \If{$\tier(u)=0$}
            \State Set $R[\indexOf(u)] \gets \symbol(u)$
        \EndIf
    \EndFor

    \State Push $(e,S,R)$ onto $T$

    \While{$T$ is not empty}
        \State $(e,S,R) \gets$ Pop from $T$
        \If{$H$ already contains $e$}
            \State \Continue
        \EndIf
	\State Let $(u,v) \gets e$
        \While{\True}[This loop iterates at most $|W|$ for any computation walk $W$]
	    \State{Let $isNewEdge \gets \False$} 
            \If{$H$ does not contain $e$}
                \State{Add $e$ to $H$; Set $isNewEdge \gets \True$}
            \EndIf

            \State Set $(u,v) \gets e$; Set $S[\indexOf(e))] \gets e$
            \If{$\mathrm{tier}(u)=0$}
                \State Set $R[\indexOf(v)] \gets \symbol(v)$
            \EndIf

            \If{$isNewEdge$ \textbf{and} $e$ is a merging edge}
                \State \Call{AddExtendableEdgeOnCeilingEdges}{$H,S,E_b$}
            \EndIf

            \If{$\state(v) \in \{\qacc,\qrej\}$}
                \State \Break
            \EndIf

            \State $e_p \gets S[\min(\indexOf(v),\nextIndex(v))]$
            \If{$isNewEdge$ \textbf{and} $e_p \neq NIL$ \textbf{and} $e_p$ is proper merging, combining, or pseudo-combining}
                    \State{Add $(\NIL, e)$ to $E_b$ if $\indexOf(v) \ne \nextIndex(v)$}\label{algline:to_collect_next_based_edge}
            \EndIf

            \State $E_n \gets$ \Call{GetNextEdgesAboveIPreds}{$G,v,\{e_p\}$}
            \State $e \gets$ an edge of $E_n$; Remove $e$ from $E_n$
            \If{$|E_n|>0$}
                \ForAll{$f \in E_n$}
                    \If{$H$ does not contain $f$}
                        \State Push $(f,\text{copy}(S),\text{copy}(R))$ onto $T$
                    \EndIf
                \EndFor
            \EndIf
        \EndWhile

        \If{$\state(v) = \qacc$}
            \State \Return $R$
        \EndIf
    \EndWhile
    \State \Return $NIL$
\EndFunction
\end{algorithmic}
\end{algorithm}

\begin{definition}[Maximal Direct Extension of Footmarks]\label{def:maximal_direct_extension_footmarks}
Let $W$ be a computation walk and $\mathcal{W}$ be a set of valid computation walks containing $W$ as their subwalk. Let $G \subseteq F(\mathcal{W})$ be an initial footmarks containing $E(W)$. For $G^{(0)} = G$, a \textit{direct extension of footmarks of $G$ with respect to $W$}, $G^{(i)}$, is constructed as follows:
\begin{enumerate}
\item For $i \ge 1$, let $W_i \in \mathcal{W}$ be a computation walk such that $E(W_i) \not\subseteq E(G^{(i-1)})$.
\item $W_i$ is admitted for extension if and only if there exists a previously included walk $W_j$ ($j < i$) that shares a common prefix with $W_i$ up to, but not including, the first edge $e \in E(W_i)$ such that $e \notin E(G^{(i-1)})$.
\item The graph is updated as $G^{(i)} = G^{(i-1)} \cup E(W_i)$.
\end{enumerate}
We define $G^{(k)}$ as a \textit{maximal direct extension footmarks} for the maximum such $k$ and denote it as $G^*$.
\end{definition}

\begin{lemma}[Correctness of Direct Edge Extension Algorithm] \label{lem:direct_extension_correct}
Let $W$ be a verified computation walk returned by \textproc{VerifyExistenceOfWalk()}, and let $G$ be the initial footmarks containing $E(W)$. 
Let $\mathcal{W}$ be the set of all valid computation walks containing $W$ as a subwalk that reach halting states of the verifier Turing machine $M$.
Let $H$ be the graph obtained by applying \textproc{ExtendEdgeDirectlyWithWalk()} to $G$. 
Then $H$ is the maximal direct extension of footmarks of $G$ with respect to $W$, as defined in \cref{def:maximal_direct_extension_footmarks}.
\end{lemma}
\begin{proof}
We prove the claim by contradiction. Suppose that the algorithm terminates with a graph $H$, but $H$ is not the maximal direct extension of footmarks of $G$ as defined in \cref{def:maximal_direct_extension_footmarks}. 
In other words, there exists a valid computation walk $W' \in \mathcal{W}$ containing $W$ as a subwalk that includes at least one edge $e \notin H$, 
where $W'$ shares a prefix prior to $e$ with the already extended computation walks. 
Conversely, it may be supposed that there exists an edge in $H$ that does not belong to any valid computation walk within the set of directly extended edges.

First, consider the case where a valid edge $e$ is missing from $H$. Let $e$ be the first such edge on $W'$. 
By this choice, the previous edge $e'$ of $e$ on $W'$ must already be in $H$. Since the underlying Turing Machine is deterministic, 
the index-successor of any given edge is uniquely determined by the current state and its structural context. 
Specifically, the transition function determines the next edge based on the node information and the boundary constraints provided by either the \textit{index-predecessor} 
or the \textit{ceiling-edges} stored in the dynamic array $S$.
During the execution of \textproc{ExtendEdgeDirectlyWithWalk()}, the algorithm systematically explores all potential extensions:
\begin{itemize}\item For \textbf{non-floor edges}, the transition is strictly deterministic and follows directly from $e'$ and the corresponding ceiling-edge in $S$.
\item For \textbf{floor edges} (splitting edges), the algorithm explores all possible branches by considering every valid certificate symbol through the stack $T$.
\end{itemize}
Therefore, regardless of whether $e$ is a floor or a non-floor edge, the extension procedure—utilizing the existing previous edge $e' \in H$ 
and the surface information from the stored ceiling-edges—will eventually process $e$ and add it to $H$. 
This contradicts the assumption that $e \notin H$.

Furthermore, every edge added to $H$ is constructed strictly following the transition rules of the Turing Machine starting from a verified walk $W$. 
Since each extension step preserves the feasibility of the computation walk, any edge $e \in H$ must, by construction, belong to at least one valid computation walk. 
This contradicts the existence of an edge that does not belong to any validcomputation walk in $H$.
Since the algorithm continues until all reachable edges along valid computation walks from $W$ have been processed, 
and because it explicitly skips further extensions if a splitting edge is already in $H$, 
the resulting graph $H$ satisfies the conditions of a maximal direct extension of footmarks. 

Consequently, $H$ contains all and only the edges belonging to the maximal direct extension of the footmarks of $G$ with respect to $W$, ensuring both the correctness and the maximality of the result.
\end{proof}

After the direct edge extension procedure terminates,
no edge belonging to any valid computation walk extending the current prefix
remains directly extendable from the current feasible graph.

The use of a verified computation walk $W$ as the anchor of direct extension
ensures that all iterative extensions are bounded by $|W| \le p(n)$.
Thus, although the algorithm performs local branching via a stack, the total
number of generated edges remains polynomially bounded as in the following \cref{lem:directed_walk_extension_time}.

\begin{lemma}[Time Complexity of Directed Walk Extension]\label{lem:directed_walk_extension_time}
Let $G$ be the computation graph of an NP verifier with width $w$ and height $h$, such that $|E(G)| = \bigO(wh^{2})$. 
Let $H_0 \subseteq G$ be an initial footmarks graph, and let $W$ be a verified computation walk.
The total running time of the directed walk extension procedure (\cref{alg:extend_edge_directly_with_walk}), 
which incrementally extends $H_0$ by exploring computation walks derived from $W$, is bounded by
\[
    \bigO\bigl( |E(G)| \cdot (p(n) w + T_o) \bigr) = \bigO\bigl( wh^{2} \cdot (p(n) w + T_o) \bigr),
\] where $T_o$ denotes the additional overhead for each extended edge and $n$ is the input size of the corresponding NP problem.
Furthermore, the cost per edge extension is defined as:
\[
    T_d = p(n) w + T_o.
\]
\end{lemma}
\begin{proof}
First, note that the length of any computation walk in $G$ is bounded by a polynomial $p(n)$, since the NP verifier halts in at most $p(n)$ steps. 
We bound the running time by accounting for the total number of branching events and the cost of the subsequent deterministic extensions.
\begin{itemize}
\item \textbf{1. Bounded number of walk explorations.} 
The algorithm triggers a full walk extension only when it identifies a \textit{first missing edge} relative to the current graph $H$.
Once an edge is added to $H$, it can no longer serve as a ``first missing edge'' for any future exploration. Since $H \subseteq G$ and $|E(G)| = \bigO(wh^2)$, 
the total number of such branching/activation events is bounded by the total number of edges in the computation graph, $\bigO(wh^2)$.
\item \textbf{2. Cost per branch/walk extension.} 
Whenever a new branch is identified, the algorithm performs a direct extension of the computation walk for at most $p(n)$ steps.
In each step, the algorithm registers a new edge to $H$ and updates the ceiling-edge array $S$ and the result array $R$. 
The dominant operation in each step is the update (or copying) of $S$ and $R$, which costs $\bigO(w)$. 
The invocation of \textproc{GetNextEdgesAboveIPreds()} results in only a constant cost per step, since $|E_p|=1$ when called in this context according to \cref{sublem:runningtime_get_next_edges_above_preds}. 
Thus, the cost for a single walk extension is $\bigO(p(n) w)$.
\item \textbf{3. Additional overhead for newly extended edges.}
Whenever a new edge is registered (i.e., the flag \texttt{isNewEdge} is \texttt{True}), it incurs an additional overhead $T_o$ per edge. 
Since a branch is activated only for edges not yet in $H$, the total cost incurred per new edge is $T_d = \bigO(p(n) w + T_o)$.
\item \textbf{Conclusion.}
Summing over all unique branching points (edges) in the graph, the total running time of the procedure is:
\[    \bigO\bigl( |E(G)| \cdot (p(n) w + T_o) \bigr) = \bigO\bigl( wh^{2} \cdot (p(n) w + T_o) \bigr). \]

This accounts for the fact that each edge in $H$ is processed as part of a deterministic walk extension precisely when its preceding branch is first explored.
\end{itemize}
No further computation walk scans are performed for branching edges already contained in $H$, ensuring the polynomial bound.
\end{proof}

Since \textproc{VerifyExistenceOfWalk()} is the dominant subroutine, 
even a moderate reduction in its invocations—by bypassing the verification process for directly extended edges—leads to a substantial empirical speedup.

\subsection{Theoretical Reduction of Candidate Edges} \label{subsec:reduction_of_extension_degree}

In the original NP verifier simulation framework, candidate edges for extension were collected by scanning all potential edges up to tier $t_m+1$ relative to the current graph. 
Specifically, for a candidate edge, this involved checking all vertices up to the maximum tier $t_m$ that share the same index as the head node of the candidate. 
This exhaustive approach resulted in $\bigO(wh^2)$ candidates per extension phase, leading to a per-edge verification bound of $\bigO(wh^2 T_v)$ for the height of the resulting graph $h$ and its width $w$,
where $T_v$ denotes the verification time complexity.

Since a surface can be derived from the ceiling edges, we refer to the set of ceiling edges of a computation walk as the set of \textbf{surface edges}, and we define its cardinality as the width $w$.
Under the restricted candidate selection strategy, the search space is significantly narrowed.
This is because every edge—whether added via formal verification or through the direct extension procedure—is extended exclusively along some valid computation walk.  
Each boundary edge is now constrained to the set of surface edges of width $w$, and is limited to at most $h$ of their index-succedent edges. 
Consequently, the verification of a single edge is confined strictly to the immediate index-succedent edges of these surface edges. 
This refined strategy reduces the per-edge verification bound to $\bigO(wh T_v)$, effectively eliminating a factor of $h$ from the extension overhead and streamlining the overall simulation process.

To formalize additional restrictions, we introduce the notion of a \emph{combined edge}.

\begin{definition}[Combined and Proper Merging Edges]\label{def:combined_proper_merging_edge}
Let $e=(u,v)$ and $e'=(u',v')$ be distinct edges in the computation graph ($e \neq e'$).
\begin{enumerate}
\item We say that $e$ and $e'$ form a \textbf{combined edge} if their respective endpoints belong to the same transition cases; that is, $u$ and $u'$ share the same transition case, and $v$ and $v'$ share the same transition case.
\item A pair of combined edges $e$ and $e'$ is called a \textbf{combined merging edge} if their target vertices coincide ($v = v'$).
\item Conversely, $e$ and $e'$ form a \textbf{proper merging edge} if they are merging edges ($v = v'$) but do not satisfy the condition of being a combined edge.
\end{enumerate}
These definitions capture the structural properties of edges that merge into the same index-successor or exhibit identical transition behaviors within the computation graph.
\end{definition}

\begin{remark}
It is important to note that the classification of an edge $e$ is not absolute but relative to the edge it is being compared with. 
A single edge $e$ can simultaneously participate in different types of merging relations depending on the reference edge.
For instance, an edge $e$ can be a combined merging edge with an edge $e'$, but still can be a proper merging edge with another edge $e''$
\end{remark}

\begin{definition}[Combining Edge] \label{def:combining_edge}
Let $e=(u,v)$ and $e'=(u',v')$ be edges in the computation graph, with $e \neq e'$.
We say that $e$ and $e'$ form a \emph{combining edge} if
\begin{enumerate}
    \item $u$ and $u'$ belong to different transition cases,
    \item $v$ and $v'$ belong to the same transition case, and
    \item $v \neq v'$, i.e., $e$ and $e'$ are not merging edges.
\end{enumerate}
\end{definition}
\medskip
\begin{remark}
A combining edge does \emph{not} necessarily share the edge index
(or equivalently, the edge direction).
Depending on the relative positions of $e$ and $e'$ in a computation walk,
three distinct configurations may occur:
\begin{enumerate}
    \item both $e$ and $e'$ are non-folding edges,
    \item exactly one of $e$ and $e'$ is a folding edge,
    \item both $e$ and $e'$ are folding edges.
\end{enumerate}
\begin{figure}
	\centering
	\includegraphics[width=0.8\columnwidth]{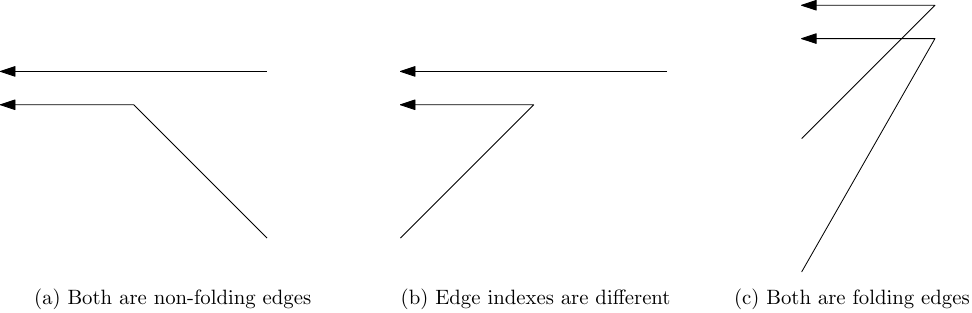}
	\caption{Types of Combining edges}
	\label{fig:types_of_combining_edges}
\end{figure}
In the first and third cases, $e$ and $e'$ share the same edge index,
whereas in the second case, the presence of a folding edge causes
the edge indices (and thus edge directions) of $e$ and $e'$ to differ.
\end{remark} 

Combined edges and combining edges provide a natural grouping of edges that share transition cases. 
By considering converging edges, such as merging and combining edges, collectively, 
we further reduce the number of \textproc{VerifyExistenceOfWalk()} calls required for each edge extension without affecting correctness.

Whereas merging and combining edges are introduced to examine splitting edges relative to a direct index-predecessor, 
the following notion is introduced to examine splitting edges relative to an indirect index-predecessor.

\begin{definition}[Pseudo-Combining Edge]\label{def:pseudo_combining_edge}
An edge $e = (u,v)$ is called a \emph{pseudo-combining edge} if
an index-succedent node of its head $v$ is a folding node and $v$ is not a folding node.
\end{definition}

To restrict the candidate space for indirect succedent edges, we introduce the concept of forward ceiling-adjacent  edges. 
By adopting the forward ceiling adjacency relation as defined below, it becomes evident that index-succedent edges correspond to the next edges (in the graph or walk sense) of those forward ceiling-adjacent edges.

\begin{definition}[Forward Ceiling Adjacency]\label{def:forward_ceiling_adjacency}
Let $e = (v,w)$ and $e' = (u,v')$ be ceiling edges of a computation walk.
We say that $e'$ is \emph{forward weak ceiling adjacent} to $e$ if there exists a
sequence of vertices
\[
v_0, v_1, v_2, \ldots, v_n
\]
such that
\begin{enumerate}
    \item $v_0 = v'$ and $v_n = v$ where $v_0$ and $v_n$ are non-folding nodes;
    \item for every $0 \le i < n$, $v_{i+1}$ is the index-succedent node of $v_i$;
    \item for every $0 \le i < n$, the vertex $v_i$ is a folding node.
\end{enumerate}
If there exists a path from $e'$ to $e$ where $e'$ is forward weak ceiling adjacent to $e$ and no edge in the path has the same edge index as $e'$ other than $e'$ itself, 
then we say that $e'$ is \textbf{forward ceiling adjacent} to $e$. Furthermore, we say $e$ is \textbf{reverse ceiling adjacent} to $e'$.
\end{definition}

We now demonstrate how the candidate space for edge extension can be restricted by applying the aforementioned definitions and lemmas.

\begin{lemma}[Restricted Next-Based Candidate Edges for Verification]
\label{lem:restricted_precedent}
Let $e$ be a newly added edge (the last edge) to a computation walk $W$ in a footmarks graph $H$, where $S$ is the set of ceiling edges of $W-e$ in $H$. Let $e' \in \Next(e)$ be an edge that becomes a new candidate for extension due to the insertion of $e$. 

If $W + e'$ remains a valid computation walk---that is, if the index-predecessor $\hat{e}$ of $e'$ in $W+e'$ belongs to $S$---then a direct edge extension is possible without additional verification. 

Otherwise, it is sufficient to consider only the index-succedent edges of those edges that are forward ceiling-adjacent to $e$ as candidate edges for verification. 
Specifically, these are the index-succedents of forward ceiling edges incident to the head node of $e$. 
Furthermore, a necessary condition for $e'$ to require verification is that there exists an index-precedent edge $\hat{e}'$ that is forward ceiling-adjacent to $e$ other than $\hat{e}$, where $\hat{e}'$ must be one of the following:
\begin{itemize}
    \item a merging edge,
    \item a combining edge,
    \item a pseudo-combining edge.
\end{itemize}
\end{lemma}

\begin{proof}
Let $e=(u,v)$, $e'=(v,w)$ and $\hat{e}=(x,y)$.

First, suppose $\ipred_{W'}(e')$ belongs to $S$ for some computation walk $W'$. Then, the edge index of the ceiling edge $\hat{e} \in S$ is identical to the next edge index of $e$ (and thus to the index of $e'$), 
which is uniquely determined by $\min\{\indexOf(v),$ $\nextIndex(v)\}$. Then the next state, next symbol, and next tier of $e$ are uniquely determined by $e$ and $\hat{e}$. These correspond to the state, symbol, and tier of $e'$, respectively. Hence, the edge $e'$, being the index-successor of $\hat{e}$ on the walk $W+e$, can be extended directly without verification.

\medskip
\noindent
Second, suppose $\ipred_{W'}(e') \notin S$ for some computation walk $W'$ where $e'=(v,w)$ is a valid edge of some computation walk $W' \neq W+e'$ containing both $e$ and $e'$. 
Then $e'$ must have an index-predecessor $\hat{e}' = \ipred_{W'}(e')$ on $W'$. By \cref{def:index-precedent_index-succedent_edges}, $e'=(v,w)$ is an index-succedent of $\hat{e}'=(x', y')$.
Beside, there should exist a path between $\hat{e}'$ and $e'$ where no edge has an edge index identical to that of $e'$ other than $\hat{e}'$ and $e'$.

By the definition of an index-succedent, there exists a sequence of vertices $v_0, v_1, \ldots, v_n$ satisfying:
\begin{enumerate}
    \item $v_0 = y'$ and $v_n = v$;
    \item for every $0 \le i < n$, $v_{i+1}$ is the index-succedent node of $v_i$;
    \item for every $0 < i < n$, the vertex $v_i$ is a folding node.
\end{enumerate}
From the index-succedent condition, $v_0$ and $v_n$ cannot be folding nodes unless $v_0=v_n$. If $v_0=v_n$, then $e'$ is a folding edge and can always be directly extended. 
Since $v_n$ is the head of $e$, this sequence forms a folding chain satisfying the forward ceiling-adjacency condition.
Hence, $\hat{e}'$ in $H$ is forward ceiling-adjacent to $e$ by definition.

\medskip
\noindent
We now show that $\hat{e}'$ must be restricted to the cases stated in the lemma. 
Suppose, for the sake of contradiction, that the forward ceiling-adjacent edge $\hat{e}' = (x', y')$ to $e$ is neither a merging edge, a combining edge, nor a pseudo-combining edge.

Since $\hat{e}'$ is not a pseudo-combining edge, the index-succedent of its head $y'$ is not a folding node. 
Furthermore, because $\hat{e}'$ is neither a merging nor a combining edge, the edge $\hat{e}'$ incoming to a node in the transition case of $y'$ must originate from the same transition case as the tail node $x$ of $\hat{e}$. 
By the determinism of the underlying Turing machine, the index-successor edge $e' = (v, w)$ of such an $\hat{e}'$ must be uniquely determined where $w \in \ISucc(x)=\ISucc(x')$.

Consequently, $e'$ would already be an index-succedent of $\hat{e}$ and thus eligible for direct extension. This implies that $e'$ does not require inclusion in the candidate set for verification, contradicting the assumption that $e'$ requires verification. 
Therefore, $\hat{e}'$ must be one of the three specified types.
\end{proof}

In the case of a merging edge, it is insufficient to examine only its direct index-succedent edges due to the potential existence of combined merging edges or internal splitting. 
As stated in \cref{lem:restricted_isuccedent}, we must examine all succedent edges of the ceiling edges of the computation walk where the merging edge is extended, as well as the edges converging to those ceiling edges.

\begin{sublemma}[Restricted ISuccedent-Based Candidate Edges for Verification]\label{lem:restricted_isuccedent}
Let $e$ be a newly added edge in the footmarks graph $H$, and let $e'$ be an edge that becomes extendable only after the insertion of $e$. 
If $e' \notin \Next(e)$, then $e'$ must be a next edge of an edge that is \textbf{reverse ceiling adjacent} to some edge $e_c$ identified as a proper merging edge, a combining edge, or a pseudo-combining edge.
\end{sublemma}
\begin{proof}
Let $G = H + e'$ be the footmarks graph containing $e'$. Since $G$ is a footmarks graph, let $W = e_0 e_1 \cdots e_m$ be a computation walk in $G$ where $e_m = e$ is the newly added edge. 
Let  $S$ be the set of ceiling edges of $W$ and let $W_m$ be a maximal computation walk containing $W$ in $G$.

First, if $e$ is not a merging edge, it would be the unique final edge of the subwalk $W$ that enables $e'$. 
However, if $e$ is an ex-pendant (non-merging) edge, any subsequent edges following $e$ cannot exist without its immediate next edge due to the continuity of the computation walk. 
This contradicts the assumption that $e' \notin \Next(e)$. Therefore, $e$ must be a merging edge to provide the structural foundation for $e' \notin \Next(e)$.

Second, by the maximality of $W_m$, there is no index $k < m$ such that $e_k \notin E(H)$ except for $e_m = e$. 
Thus, all edges of $W_m$ except possibly $e'$ must already belong to $H$ once $e$ is inserted. Let $W'$ be a computation walk in $G$ that contains both $e$ and $e'$. 
We consider the following cases regarding the relationship between $W_m$ and $W'$:
\begin{itemize}
    \item \textbf{Case 1 ($W_m = W'$):} $e'$ is directly extendable with an index-predecessor $e_c \notin S$. 
    \item \textbf{Case 2 ($W'$ splits from $W_m$ behind $e$ on ceiling edges $S$):} $e'$ is an index-successor of $e_c \in S$.
    \item \textbf{Case 3 ($W'$ splits from $W_m$ behind $e$ not on ceiling edges $S$):} $W'$ splits from an edge $f$ which is the index-successor of some edge $e_c \in S$, where $f$ already existed in $H$.
\end{itemize}

We now show that $e_c$ must be restricted to the specified configurations. 
Suppose, for the sake of contradiction, that the ceiling edge $e_c = (x, y)$ is neither a proper merging edge, a combining edge, nor a pseudo-combining edge. 
Then the head node $y$ of $e_c$ must belong to the same transition case as the head node $y'$ of some other potential index-predecessor $e'_c$, if one exists.
Since $e_c$ is not a pseudo-combining edge, the index-succedent of its head $y$ is not a folding node. Furthermore, because $e_c$ is neither a proper merging edge nor a combining edge, any incoming edge $e'_c = (x', y')$ to a node in the transition case of $y$ must originate from the same transition case as $x$. 

By the determinism of the underlying Turing machine, the index-successor edge of such an $e'_c$ must be uniquely determined and identical to $e'$. 
Consequently, $e'$ cannot be a splitting edge and would have been valid in $H$ even without $e$. This directly contradicts our initial assumption that $e'$ is extendable due to insertion of $e$.

Therefore, any edge $e'$ that becomes extendable only after the addition of $e$ must be an index-succedent of an edge $e_c$ that is a merging, combining, or pseudo-combining edge. 
By \cref{def:index-precedent_index-succedent_edges} and \cref{def:forward_ceiling_adjacency}, the previous edge $e_f$ of $e'$ establishes that $e_c$ is forward ceiling-adjacent to $e_f$. 
Inverting this relation via our definition, it follows that $e'$ is a next edge of an edge that is $e_f$ in reverse ceiling adjacent edges $E_f$  to $e_c$. 
This strictly restricts the search space for candidate edges to a well-defined set of index-succedents, completing the proof.
\end{proof}

\begin{remark}
In this framework, $E_f$ acts as the structural terminal of the computation walk completed by the insertion of $e$, thereby serving as the necessary foundation from which the validity of the next edge $e'$ is extended.
Furthermore, if $e$ itself is a merging edge, then either:
\begin{itemize}
\item a next edge of $e$ may become extendable, or
\item an index-succedent edge $e'$ may become extendable if its previous edges $E_f$ are forward ceiling-adjacent to a merging, combining, or pseudo-combining ceiling edge of the computation walk where $e$ is extended.
\end{itemize}
\end{remark}

\begin{algorithm}[!ht]
\caption{Add Extendable Edges on Ceiling Edges}
\label{alg:add_extendable_on_ceiling}
\begin{algorithmic}[1]
\Function{GetReverseCeilingAdjacentEdges}{$H, e$}
\State{Let $(u,v) \gets e$}
\State{Let $Q \gets$ a set of all the vertices of  $\ISucc_H(v)$}
\State{Let $E_f \gets$ Empty set}
\State{Let $V_v \gets$ Empty set}
\While{$Q \neq \emptyset$}
    \State{Select and remove a vertex $w$ from $Q$}
    \If{$w \in V_v$}
    	\State{\Continue}
    \Else
        \State{Add $w$ to $V_v$}
    \EndIf
    \If{$w$ is a folding node}
        \State{Set $Q \gets Q \cup \ISucc_H(w)$}
    \ElsIf{$\nextIndex(w) = \indexOf(u)$}
        \State{Let $E_t \gets$ \Call{FilterWithPathForward}{$H, e, \mathrm{IncomingEdges}_H(w)$}}
        \State{Set $E_f \gets E_f \cup E_t$}
    \EndIf
\EndWhile
\State{\Return $E_f$}
\EndFunction

\Function{AddExtendableEdgeOnCeilingEdges}{$H, S, E_b$}
\ForAll{edge $e \in S$}
    \If{$e$ is \textbf{NIL}}
        \State{\Continue}
    \EndIf
    \If{$e$ is a proper merging edge \textbf{or} a combining edge \textbf{or} a pseudo-combining edge}
            \State{Let $E_f \gets$ \Call{GetReverseCeilingAdjacentEdges}{$H, e$}}
            \ForAll{$f \in E_f$}
                \State{Add $(e, f)$ to $Q$}
            \EndFor
    \EndIf
\EndFor
\EndFunction
\end{algorithmic}
\end{algorithm}

\begin{algorithm}[!ht]
\caption{Compute Footmarks of Computation Walks and Determine Acceptance} \label{alg:compute_footmarks}
\begin{algorithmic}[1]
\Input{$G$: Dynamic Computation Graph, $H$: Graph of Visited Edges, $V_0$: Set of Initial Nodes}
\Output{Whether an accept state $q_{\text{acc}}$ is reachable by a maximal computation walk}
\Function{IsAcceptedOnFootmarks}{$G, H, V_0$}[$q_{\text{acc}}, q_{\text{rej}}$ is a constant]
\State{Let $Q$ be a set of all outgoing edges of $v_0 \in V_0$ in $G$}
\State{Let $isRetry \gets \False$}
\While{$Q \ne \emptyset$}[Extend $H$ by valid computation edges]
	\State{Let $E_b \gets \emptyset$}
	\State{Let $prevCntEdges \gets |E(H)|$}
    \State{Let $R \gets$ \Call{ExtendByVerifiableEdges}{$V_0, Q, H, E_b$}}\Comment{$E_b$: collected for boundary edges}
    \If{$R$ is not NIL}
        \State{Print $R$ as a valid Certificate}
        \State{\Return \True}
    \ElsIf{$|E(H) = prevCntEdges$ \textbf{and} $isRetry$} [No edges are extended for general-index-succedent based candidate edges]
       	\State{\Return \False}
    \EndIf
    \State{$Q \gets$ \Call{CollectRestrictedCandidateEdges}{$G,H, E_b$}}\Comment{Newly collected candidate edges}
    \If{$|Q|=0$} 
    	\State{Set $E_b \gets$ \Call{CollectForISuccedentBoundaryEdges}{$H, V_0$}}
    	\State{$Q \gets$ \Call{CollectRestrictedCandidateEdges}{$G,H, E_b$}}\Comment{General boundary edges}
    	\State{Set $isRetry \gets \True$}
    \Else
        \State{Set $isRetry \gets \False$}
    \EndIf
\EndWhile
\State{\Return \False}
\EndFunction
\end{algorithmic}
\end{algorithm}

\begin{algorithm}[!ht]
\caption{Collect Restricted Boundary Edges} \label{alg:collect_restricted_boundary_edges}
\begin{algorithmic}[1]
\Input{$H$: Graph of Visited Edges, $V_0$: Set of Initial Nodes}
\Output{The restricted candidate edges $Q$}
\Function{CollectForISuccedentBoundaryEdges}{$H, V_0$}
    \State{Let $T$ be a stack initialized with all outgoing edges of $v_0 \in V_0$ in $G$}
    \State{Let $Q \gets \emptyset$} \Comment{Set of candidate boundary edge pairs}
    \State{Let $E_v \gets \emptyset$}
    \While{$T$ is not empty}
        \State{Pop an edge $e$ from $T$}
        \If{$e \in E_v$}
            \State{\Continue}
        \EndIf
        \State{Add $e$ to $E_v$}
        \State{Push all edges in $Next_H(e)$ to $T$}
        
        \If{$e$ combining edge or pseudocombining edge}
            \State{Let $E_f \gets$ \Call{GetReverseCeilingAdjacentEdges}{$H, e$}}
            \ForAll{$f \in E_f$}
                \State{Add $(e, f)$ to $Q$}
            \EndFor
        \EndIf
    \EndWhile
    \State{\Return $Q$}
\EndFunction
  
\Function{CollectRestrictedCandidateEdges}{$G, H, E_b$}

\State{Let $Q$ be the empty set of edges}
\ForAll{$(e_p, e)$  in $E_b$} \Comment{Collect restricted candidate edges}
    \State{Let $(u,v) \gets e$}
    \If{$\state(v) \in \{\qacc, \qrej\}$ \textbf{or} $\indexOf(u)= \nextIndex(v)$}
         \State{\Continue}
    \ElsIf{$e_p$ is not \NIL}    
	    \State{$Let Ep \gets$ \Call{GetForwardWeakCeilingAdjacentEdges}{$H,e$}}
	    \State{$Set Ep \gets$ \Call{FilterWithPathBackward}{$H, e, Ep$}}
    \Else
           \State{$Let Ep \gets \{e_p\}$}
    \EndIf
    \State{$En \gets$ \Call{GetNextEdgesAboveIPreds}{$G, v, Ep$}}
    \ForAll{$e' \in En$}
    	\If{$e' \notin E(H)$}
             \State{Add $e'$ to $Q$}
        \EndIf
    \EndFor
\EndFor
\State{\Return $Q$}
\EndFunction
\end{algorithmic}
\end{algorithm}

In the original algorithm, \textproc{CollectBoundaryEdges()} identified candidate edges by scanning all edges in the tier immediately above the current tier (i.e., tier $m+1$), 
which potentially involved a large number of edges per surface.In contrast, 
the revised \textproc{CollectRestrictedCandidateEdges()} restricts the candidate set based on the structural properties of the visited-edge graph $H$. 
Specifically, only edges positioned above merging edges, combining edges, or pseudo-combining edges are considered for extension.
For merging structures, branching is permitted above all converging surface edges. 
Consequently, the candidate set is efficiently constructed by inspecting the surface and selecting only those edges identified as merging, combining, or pseudo-combining.
This restriction dramatically reduces the number of edges to be processed by \textproc{VerifyExistenceOfWalk()} while preserving correctness, 
as any valid computation walk must necessarily pass through at least one of these structural edges.

\begin{lemma}[Completeness of Restricted Verification Collection]
\label{lem:completeness_of_restricted_verification}
Let $G$ be the dynamic computation graph and $H$ be the footmarks graph of extended edges. Suppose an edge $e' \notin E(H)$ becomes extendable only after $H$ is extended by a direct edge extension (via \textproc{ExtendEdgeDirectlyWithWalk()}) along some valid computation walk.

Then $e'$ is necessarily collected by \textproc{CollectRestrictedCandidateEdges()} and added to the verification candidate set. In particular, no extendable edge requiring verification is omitted.
\end{lemma}

\begin{proof}
Assume, for the sake of contradiction, that there exists an edge $e'$ satisfying the following conditions:
\begin{enumerate}
    \item $e' \notin E(H)$;
    \item $e'$ becomes extendable only after extending $H$ via direct edge extension along some valid computation walk;
    \item $e'$ is not collected by \textproc{CollectRestrictedCandidateEdges()}.
\end{enumerate}

Let $e$ be a newly extended edge during the execution of \textproc{ExtendEdgeDirectlyWithWalk()}, and let $W = e_0 e_1 \cdots e_k$ be a computation walk such that $e=e_k$ is directly extended along this walk. Let $e'$ be the first edge in a computation walk $W'$ that is neither added to $H$ nor collected into $Q$, while all preceding edges of $W'$ already belong to $H$ or to the footmarks of previously extended computation walks.

We distinguish cases based on how $e'$ becomes extendable:

\begin{itemize}
    \item \textbf{Case 1: Direct continuation.}
    If $e_{k-1}$ is directly extended along the computation walk $W$ $(W'=W+e')$, then $e'$ (e.g., a non-floor next edge or a specific floor next edge) must be added immediately during the direct extension by \cref{lem:direct_extension_correct}. This contradicts the choice of $e'$ as the first non-added edge on $W'$.

    \item \textbf{Case 2: Next-edge-based extension.}
    Suppose $e' \in \Next(e)$ for some newly added edge $e$ in \textproc{ExtendEdgeDirectlyWithWalk()}. By \cref{lem:restricted_precedent} (Restricted Next-Based Verification), if $e'$ is not directly extendable, there must exist an edge $\hat{e}$ in $W'$ that is forward ceiling-adjacent to $e$ that is either a merging edge, a combining edge, or a pseudo-combining edge. In this case, $(\NIL, e)$ is added to $Ev$ at line \ref{algline:to_collect_next_based_edge} of \cref{alg:extend_edge_directly_with_walk}, and $e'$ is subsequently collected into $Q$ by \textproc{CollectRestrictedCandidateEdges()}. This contradicts the third assumption.
    
    \item \textbf{Case 3: Ceiling-edge-based extension.}
    Suppose $e'$ does not belong to the next edges of any previously extended edge but belongs to a valid computation walk $W'$. By \cref{lem:restricted_isuccedent} (Restricted Index-Succedent-Based Verification), this is possible only if $e$ is a merging edge and $e'$ is an index-succedent of a proper merging, combining, or pseudo-combining edge $e_c$, where $e_c$ is a ceiling edge of $W$. In this case, $e'$ is collected by \textproc{AddExtendableEdgeOnCeilingEdges()} and subsequently by \textproc{CollectRestrictedCandidateEdges()}, again contradicting assumption (3).
    
    \item \textbf{Case 4: General Index-succedent-based extension.}
    Suppose $e'$ does not belong to the next edges of any previously extended edge but belongs to a valid computation walk $W'$. By \cref{lem:restricted_isuccedent} (Restricted Index-Succedent-Based Verification), this is possible only if $e$ is a merging edge and $e'$ is a next edge of reverse ceiling adjacent edges of $e'_c$, where $e'_c \in E(H)$ is a combining, pseudo-combining, or proper merging edge. In this case, $e'$ is collected by \textproc{CollectForISuccedentBoundaryEdges()} and subsequently by \textproc{CollectRestrictedCandidateEdges()} when the $isRetry$ flag is \True, again contradicting assumption (3).
\end{itemize}

Since all possible cases lead to a contradiction, no such edge $e'$ exists. Thus, every edge that becomes extendable and requires verification is successfully collected by \textproc{CollectRestrictedCandidateEdges()}.
\end{proof}

\begin{sublemma}[Reverse Ceiling Adjacent Edge Computation Cost] \label{sublem:reverse_ceiling_adjacent_cost}
Let $G$ be a computation graph with width $w$ and height $h$, and let $H$ be a footmarks graph. 
Then \textproc{GetReverseCeilingAdjacentEdges()} runs in time
\[
\bigO(w h^5 \log(wh)).
\]

\begin{proof}
The algorithm's complexity is determined by the following stages:
\paragraph{(1) Traversal of index-succedent vertices.}
For a given ceiling edge $e=(u,v)$, the algorithm explores the index-succedent chain starting from $v$. Since the height of the graph is $h$, the total number of reached vertices $w$ through folding and succedent transitions is bounded by $\bigO(h)$.

\paragraph{(2) Incoming edge enumeration.}
For each vertex reached in the traversal (up to $\bigO(h)$ vertices), the algorithm inspects its incoming edges to identify candidate previous edges. Since the maximum in-degree in the computation graph is $\bigO(h)$, the total number of candidate incoming edges $E_t$ to be verified is at most $\bigO(h^2)$.

\paragraph{(3) Path-adjacency verification.}
Each candidate edge in $E_t$ is verified using \textproc{FilterWithPathForward()}, which checks for a valid forward path in $H$. This procedure takes $\bigO(w h^3 \log(wh))$ time per candidate according to \cref{sublem:runningtime_filter_forward}.

\paragraph{(4) Total cost.}
Summing these components, the total cost for one call is:
\[
\bigO(h^2) \cdot \bigO(wh^3 \log(wh)) = \bigO(w h^5 \log(wh)).
\]
This completes the proof.
\end{proof}
\end{sublemma}

\begin{corollary}[Computation of Edge Extension Overhead Cost $T_o$]
\label{lem:edge_extension_cost_To}
Let $G$ be a computation graph with width $w$ and height $h$, and let $H \subseteq G$ be a partially constructed footmarks graph. 
Define $T_o$ as the additional overhead of direct edge extension when $isNewEdge$ is \texttt{True} in \cref{alg:extend_edge_directly_with_walk}.
Then $T_o$ is bounded by
\[
T_o = \bigO(w^2 h^5 \log(wh)).
\]

\begin{proof}
The cost of computing the next-based candidates is $\bigO(h)$, as checking the combining or pseudo-combining edge cost is $\bigO(h)$ according to \cref{lem:complexity_combining_detection}. 
We now focus on \textproc{AddExtendableEdgeOnCeilingEdges()}.
The set $S$ of ceiling-adjacent edges contains at most $\bigO(w)$ edges, corresponding to the width of the computation graph. 
For each edge in $S$, \textproc{GetReverseCeilingAdjacentEdges()} is invoked with a cost of $\bigO(w h^5 \log(wh))$ by \cref{sublem:reverse_ceiling_adjacent_cost}.
Thus, the total overhead $T_o$ satisfies:
\[
T_o = \bigO(w) \cdot \bigO(w h^5 \log(wh)) = \bigO(w^2 h^5 \log(wh)).
\]
\end{proof}
\end{corollary}

In the unrestricted extension process of the original framework, a single extended edge may generate up to $O(wh^2)$ candidate edges. 
Each candidate edge triggers an invocation of $\textproc{VerifyExistenceOfWalk()}$. 
Since the verification cost is typically large, this results in a prohibitive total cost. 
However, the number of candidate edges can be reduced by only investigating restricted candidate edges as follows:

\begin{lemma}[Bound on Extension Candidates] \label{lem:restricted_candidate_bound}
Let $H$ be a footmarks graph of maximal computation walks in a computation graph $G$ of width $w$ and height $h$. 
Let the general extension candidate ratio $\Delta$ be the ratio of the number of edges having a combining or pseudo-combining edge as their index-predecessor to the total number of edges in $H' = H \cup E_b$, where $E_b$ denotes the set of boundary edges of $H$.
At any extension step, the number of restricted extension candidates is bounded by $\bigO(wh(1 + h p_{\Delta}\Delta))$ per newly extended edge, where $p_{\Delta}$ is the probability that no directly extendable, next-based, nor ceiling-edge-based edge exists after an edge insertion.
\end{lemma}

\begin{proof}
A new candidate edge must be either next-based, ceiling-edge-based, or general-index-succedent-based, unless it is directly extendable, as established in \cref{lem:completeness_of_restricted_verification}. By the structural properties of the computation graph:
\begin{itemize}
    \item \textbf{Next-based edges:} The number of such candidate edges is bounded by $\bigO(h)$, since $|\text{Next}(e)| = \bigO(h)$, reflecting the maximum number of incoming edges at the head node of $e$ within a single tier.
    \item \textbf{Ceiling-edge-based edges:} For each ceiling edge $e_c = (u, v)$, the number of index-succedent edges is bounded by a constant, specifically $(|\Gamma||Q|)^2$, determined by the transition rules. Since the surface contains at most $\bigO(w)$ such ceiling edges, and considering potential vertical verification across $h$ tiers, the total number of candidates is bounded by $\bigO(w) \times \text{constant} \times \bigO(h) = \bigO(wh)$.
    \item \textbf{General-index-succedent-based edges:} Let $H' = H \cup E_b$. The total number of candidate edges in this category is at most $\Delta|E(H')|$. While the original framework considered all $|E(H')|$ edges, the restricted approach reduces the search space by an asymptotic ratio of $\Delta$.
\end{itemize}
The effective fallback ratio during the verification phase is $p_{\Delta}\Delta$, where both $p_{\Delta}$ and $\Delta$ are strictly less than $1$ because floor edges and the edges whose index-precedent edges are none of combining nor pseudo-combining edges are excluded from this category. 
Consequently, during each extension phase, the number of candidate edges that must be examined per newly added edge is bounded by:
\[
\bigO(wh + |E(H')| \cdot h p_{\Delta} \Delta) = \bigO(wh + wh \cdot h p_{\Delta} \Delta) = \bigO(wh(1 + h p_{\Delta} \Delta)).
\]
This concludes the proof.
\end{proof}

\begin{corollary}[Reduction in Extension Complexity] \label{lem:extension_complexity_reduction}
Let $H$ be a footmarks graph of width $w$ and height $h$. The improved algorithm using restricted candidate edges reduces the candidate examination cost per each extended edge. 
The cost is bounded by $\bigO(wh(1 + h p_{\Delta} \Delta) T_v)$, compared to $\bigO(wh^2 T_v)$ in the original framework, where $T_v$ denotes the verification cost per candidate.
\end{corollary}

\begin{proof}
Under the verification algorithm using restricted candidates, for each newly extended edge, at most $\bigO(wh(1 + h p_{\Delta} \Delta))$ candidate edges are generated that satisfy the structural constraints, as established in \cref{lem:restricted_candidate_bound}. Therefore, the total number of verification calls is reduced and bounded by $\bigO(wh(1 + h p_{\Delta} \Delta))$, whereas the original framework required $\bigO(wh^2)$ calls.
\end{proof}

\begin{remark}[Verification Cost Per Edge Collected]\label{rem:verification_cost_per_edge}
During an extension step, direct extensions are performed along a computation walk up to its maximal length. Consequently, candidate edges may accumulate along the entire walk, and the total number of edges awaiting verification before the next verification step can reach $\bigO(wh^2)$ in the worst case.

Nevertheless, the verification cost attributed to each individual edge remains bounded. Specifically, each newly added edge contributes at most $\bigO(wh)$ candidates for verification when next-based or ceiling-edge-based candidates are found. Otherwise, it contributes at most $\bigO(wh^2\Delta)$ candidates. 

Therefore, although multiple edges may accumulate before a verification phase, the total verification cost per each extended edge is bounded by $\mathcal{O}(wh(1 + h p_\Delta \Delta) \cdot T_v)$.
For empirical analysis, we approximate these structural parameters using directly observable experimental metrics as follows:
\begin{itemize}
\item $p_\Delta \approx p_I$: The probability of triggering the general verification mode, defined as the ratio of the retry count to the total number of extended edges.
\item $\Delta \approx \rho$: The average workload of each general verification phase. To ensure statistical rigor and eliminate redundant counts, $\rho$ is defined as the average number of general index-succedent-based candidate edges per retry event, normalized by the total extended edges.
\end{itemize}
Consequently, the empirical complexity contribution $p_\Delta \Delta$ is evaluated by the product $p_I \cdot \rho$:
\[
r_I=p_I \cdot \rho = \left( \frac{N_{\text{retry}}}{N_{\text{total}}} \right) \cdot \left( \frac{N_{\text{cand}} / N_{\text{retry}}}{N_{\text{total}}} \right) = \frac{N_{\text{cand}}}{N_{\text{total}}^2}
\] where $N_{\text{retry}}$ is the retry count, $N_{\text{cand}}$ is the total number of general candidate edges allowing multiple counts, and $N_{\text{total}}$ is the total number of extended edges. 
This formulation ensures that even if $N_{\text{cand}}$ involves overlapping counts, the use of the per-retry average maintains a fair representation of the algorithm's actual efficiency.
\end{remark}

\begin{sublemma}[Complexity of General Index-Succedent Boundary Collection] \label{sublem:complexity_isuccedent_collection}
Let $G$ be a computation graph with width $w$ and height $h$, and $H \subseteq G$ be a footmarks graph. 
Let $N_I$ be the number of combining and pseudo-combining edges in $H$. 
Then \textproc{CollectForISuccedentBoundaryEdges()} runs in time 
\[\bigO(|E(H)| \log(wh) + N_I \cdot w h^5 \log(wh))=\bigO(w^2h^7 \log(wh)). \]
 
\begin{proof}
The complexity is analyzed by partitioning the operations into traversal overhead and structural verification cost:
\paragraph{(1) Graph Traversal Overhead.}
The algorithm performs a depth-first search (DFS) starting from $V_0$ using a stack $T$. 
Each edge $e \in E(H)$ is pushed and popped at most once due to the visited set $E_v$. 
Membership tests and insertions into $E_v$ take $\bigO(\log |E(H)|) = \bigO(\log(wh))$ time. Thus, the total traversal overhead is $\bigO(|E(H)| \log(wh))$.
\paragraph{(2) Conditional Structural Verification.}
The expensive \textproc{GetReverseCeilingAdjacentEdges()} function is invoked only when $e$ is a combining or pseudo-combining edge.
\begin{itemize}
\item Let $N_I$ be the number of such edges in $H$, which is at most $\bigO(wh^2)$.
\item By \cref{sublem:reverse_ceiling_adjacent_cost}, each call to \textproc{GetReverseCeilingAdjacentEdges()} costs $\bigO(w h^5 \log(wh))$.
\end{itemize}
The total cost for this stage is $\bigO(N_I \cdot w h^5 \log(wh))$.

\paragraph{(3) Total Complexity.}
Combining the two parts, and noting that $|E(H)|$ and $N_I$ are both bounded by $\bigO(wh^2)$, the total running time is dominated by the verification cost:
\[ 
\bigO(wh^2 \log(wh) + wh^2 \cdot wh^5 \log(wh)) = \bigO(w^2h^7 \log(wh)).
\]
\end{proof}
\end{sublemma}

\begin{sublemma}[Running Time of Collecting Restricted Candidate Edges] \label{sublem:runningtime_collect_restricted_boundary_edges_final}
Let $G$ be a computation graph of width $w$ and height $h$, and let $H \subseteq G$ be a partially constructed footmarks graph. 
Let $E_b$ be a set of edges from which restricted boundary edges are collected. 
Then collecting candidate edges (\textproc{CollectForISuccedentBoundaryEdges()} and \textproc{CollectRestrictedCandidateEdges()})   costs
\[
T_c = \bigO(w^2 h^7 \log(wh)).
\]

\begin{proof}
For each edge $e \in E_b$, the function performs the following steps:
\begin{enumerate}
    \item \textproc{GetForwardWeakCeilingAdjacentEdges()}, which takes $\bigO(h^2 \log h)$ time by \cref{sublem:time_complexity_forward_weakly_ceiling_adjacent_edges}.
    \item \textproc{FilterWithPathBackward()}, which takes $\bigO(wh^3 \log(wh))$ time according to \cref{sublem:runningtime_filter_backward}.
    \item \textproc{GetNextEdgesAboveIPreds()}, which takes $\bigO(h^2 \log h)$ time where $|Ep| = \bigO(h^2)$ by \cref{sublem:runningtime_get_next_edges_above_preds}.
    \item Membership checks in $H$ for each generated edge, costing $\bigO(\log(wh))$ per edge.

\end{enumerate}
Furthermore, if the initial collection of restricted candidate edges yields an empty set ($|Q|=0$), the algorithm performs a one-time execution of \textproc{CollectForISuccedentBoundaryEdges()} to retrieve general boundary edges,
 incurring an additional cost of $\bigO(w^2h^7 \log(wh))$ time according to \cref{sublem:complexity_isuccedent_collection}.

The per-edge cost in $E_b$ is dominated by the backward filtering, $\bigO(wh^3 \log(wh))$. Since the worst-case size of the boundary edge set is $|E_b| = \bigO(wh^2)$, the total running time is:
\[
T_c = \bigO(|E_b| \cdot wh^3 \log(wh) + w^2h^7 \log(wh)) = \bigO(wh^2 \cdot wh^3 \log(wh) + w^2h^7 \log(wh)) = \bigO(w^2 h^7 \log(wh)).
\]
Thus, the asymptotic running time is $T_c = \bigO(w^2 h^7 \log(wh))$.
\end{proof}
\end{sublemma}

\begin{sublemma}[Time per Edge Extension via Verification]
\label{sublem:per_edge_extension_time}

Consider a single edge $e$ processed by \textproc{ExtendByVerifiableEdges()} during the footmark-based extension procedure.  
Let $N_c$ be the maximum number of candidate edges per each extended edge.
Let $T_v$ denote the cost of \textproc{VerifyExistenceOfWalk()}, and $T_o$ the additional overhead of direct edge extension as defined in \cref{lem:edge_extension_cost_To}.

Then the worst-case time to process all candidate edges generated from $e$ is bounded by
\[
\bigO\Bigl( N_c \cdot T_v + T_o + p(n) w \Bigr).
\]

\begin{proof}
For a given edge $e$, the algorithm processes candidate edges from the queue $Q$, which contains at most $N_c$ edges. For each candidate edge, the algorithm performs:
\begin{itemize}
    \item One invocation of \textproc{VerifyExistenceOfWalk()}, costing $T_v$.
    \item If verification succeeds, the algorithm proceeds to the direct edge extension phase.
\end{itemize}

As established in \cref{lem:directed_walk_extension_time}, each newly added edge incurs a cost of $\bigO(p(n) w)$ for the edge extension including walk traversal and an additional overhead $T_o$ for collecting candidate edges. 

Summing these costs, the total cost per processed edge $e$ is:
\[
\bigO(N_c \cdot T_v + T_o + p(n) w).
\]
\end{proof}
\end{sublemma}

\begin{corollary}[Time Complexity of Footmark-Based Simulation]
\label{cor:simplified_simulation_time}

Let $G$ be the computation graph with $|E(G)|$ edges, and let
\[
T_a = N_c \cdot T_v + T_c + (T_o + p(n) w)
\]
denote the total overhead to process a single edge, including verification, direct edge extension with additional overhead, and candidate calculation.

Then the total running time of the footmark-based deterministic simulation is bounded by
\[
\bigO\bigl( |E(G)| \cdot T_a \bigr),
\]
where $T_a$ captures all per-edge costs.

\begin{proof}
Each edge of $G$ is added to the boundary queue at most once.  
For each edge, the simulation performs:
\begin{itemize}
    \item verifying the existence of a computation walk ($T_v$) in \textproc{ExtendByVerifiableEdges()},
    \item performing direct edge extension ($p(n) w$) in \textproc{ExtendByVerifiableEdges()},
    \item additional processing for direct edge extension and collecting candidate edges ($T_o$) in \textproc{ExtendByVerifiableEdges()},
	\item Calculating candidate edges ($T_c$) via the invocation of \textproc{CollectRestrictedCandidateEdges()} (and \textproc{CollectForISuccedentBoundaryEdges()} upon a retry condition).
\end{itemize}

Since these are the only per-edge costs incurred, the total running time is
\[
\bigO(|E(G)| \cdot T_a).
\]
\end{proof}

\end{corollary}

\paragraph{Impact on Overall Simulation Performance}
\label{par:cover-feasible-impact}

In this framework, the improved \textproc{IsAcceptedOnFootmarks()} replaces the original verification procedure within the \textproc{SimulateVerifierForAllCertificates()} algorithm. 
The integration of efficient cover-edge computation and feasible-graph construction—improved by step-pendency pruning—significantly reduces the overall verification overhead. 
Furthermore, by utilizing direct extensions and restricting candidate edges, the total number of required verification calls is drastically minimized during the deterministic simulation.
Consequently, the computational cost for deciding NP problems within this improved simulation framework is substantially decreased, providing both theoretical polynomial-time guarantees and practical efficiency.

\begin{lemma}[Overall Time Complexity of Improved Deterministic Verifier Simulation]
\label{lem:overall_simulation_time}
Let $G$ be the dynamic computation graph of an NP verifier with width $w$ and height $h$.
Let $T_v$ denote the cost of \textproc{VerifyExistenceOfWalk()}, and let $\Delta$ be the general extension candidate ratio, and $p_\Delta$ be the corresponding probability as defined in \cref{lem:restricted_candidate_bound}. 
The total running time of the deterministic simulation algorithm, incorporating direct walk extension, reduced verification cost and frequency, is bounded by
\[
\bigO\Bigl(w^2 h^3 (1 + h \Delta) \cdot T_v \Bigr) = \bigO\Bigl(w^5 h^{10} (1 + hp_\Delta \Delta) \log(wh) \Bigr) = \bigO\Bigl(w^5 h^{11} \log(wh) \Bigr).
\]
\end{lemma}
 
\begin{proof}
Let $p(n)$ denote the polynomial running time of the verifier on inputs of size $n$, $T_o$ the cost of extending and registering a new edge, and $T_c$ the cost of candidate edge calculation. 
We establish the total complexity bound by aggregating per-edge costs over the entire computation graph.

\paragraph{Step 1: Candidate edges per edge.}
Each edge $e \in E(G)$ may generate up to $N_c = \bigO(wh(1+p_\Delta\Delta))$ candidate edges during the execution.  

\paragraph{Step 2: Total cost over all edges.}

Since the dominant cost of the improved NP verifier simulation framework (\textproc{SimulateVerifierForAllCertificates()}) is the extension cost—specifically the cost of \textproc{IsAcceptedOnFootmarks()}—the total cost of \textproc{IsAcceptedOnFootmarks()} is $\bigO(|E(G)| \cdot T_a)$, where $T_a = p(n)w + T_o + T_c + N_c \cdot T_v$ as established in \cref{cor:simplified_simulation_time}.

\paragraph{Step 3: Substitute asymptotic costs.}
Substituting the asymptotic costs
\begin{align*}
&T_v = \bigO(w^3 h^7 \log(wh)), \quad N_c=O(wh(1+ p_\Delta \Delta h)), \quad \\ 
&T_o = \bigO(w^3 h^5 \log(wh)), \quad T_c = \bigO(w^2 h^7 \log(wh)), \quad p(n) = \bigO(wh).
\end{align*}
Noting that $N_c \cdot T_v$ dominates, we obtain
\[
\bigO(|E(G)| \cdot (N_c \cdot T_v + T_c + p(n) w + T_o)  = \bigO(wh^2 \cdot wh(1+ p_\Delta \Delta h) \cdot T_v) = \bigO(w^5 h^{10}(1+ p_\Delta\Delta h) \log(wh)).
\]
In the structural analysis of $G$, the ratio $p_\Delta \Delta$ remains strictly less than $1$ as in the \cref{lem:restricted_candidate_bound}. 
Consequently, while the simplified bound $\bigO(w^5 h^{11} \log(wh))$ serves as a valid theoretical upper limit, the actual computational complexity is significantly lower due to these structural constraints.
\end{proof}

The correctness of the improved simulation framework follows directly from the functional equivalence of its components; the improved procedures, such as \textproc{ExtendByVerifiableEdges()}, 
replace the original sub-routines while strictly preserving the transition invariants of the computation graph. Consequently, the focus of this analysis is the evaluation of the total time complexity required to solve NP-complete instances within this implemented framework.

\begin{theorem}[Polynomial-Time Simulation for SAT and Subset-Sum Instances] 
\label{thm:simulation_sat_subset_sum} 
Let $n$ denote the input size (e.g., number of variables for SAT or number of elements for Subset-Sum), and let $G$ be the dynamic computation graph of the corresponding NP verifier constructed as in this framework. Then, applying the directed walk extension and boundary-based verification strategies, the total running time for solving SAT or Subset-Sum instances is bounded by 
\[ 
\bigO(n^{16} \log n). 
\]
Consequently, the proposed deterministic simulation algorithm solves these NP-complete instances in polynomial time, confirming the theoretical feasibility of the construction for moderate input sizes.
\end{theorem}

\begin{proof}
Let $n$ be the input size of the NP problem corresponding to the verifier TM. From \cref{lem:overall_simulation_time}, the total running time of the deterministic verifier simulation is bounded by
\[
\bigO(w^5 h^{11} \log(wh)),
\]
where $w$ and $h$ denote the width and height of the dynamic computation graph $G$, respectively.

For both SAT and Subset-Sum instances, the construction of the computation graph ensures that the graph dimensions are proportional to the verifier's running time:
\[
w = \bigO(p(n)) = \bigO(n), \quad h = \bigO(p(n)) = \bigO(n),
\]
where $p(n)$ is the polynomial running time of the verifier. Substituting these bounds into the asymptotic expression yields
\[
\bigO(w^5 h^{11} \log(wh)) = \bigO(n^5 \cdot n^{11} \cdot \log n) = \bigO(n^{16} \log n),
\]
where the logarithmic factor arises from set operations during the extension process.

Note that the input size $n$ represents the size of the NP-complete problem, and the dimensions $w$ and $h$ are linearly bounded as established in \cref{subsec:refined_footmarks_bounds}. Therefore, the total time complexity for both NP-complete problems is bounded by 
\[
\bigO(n^{16} \log n).
\]
Consequently, the improved NP verifier simulation framework solves SAT and Subset-Sum instances in polynomial time relative to the input size $n$, establishing the claimed bound.
\end{proof}

As discussed in \cref{lem:overall_simulation_time}, the bound $\bigO(n^{16} \log n)$ is a conservative over-estimation. In practice, the complexity can be more tightly bounded by $\bigO(n^{15}(1 + r_I) \log n)$, As discussed in \cref{lem:overall_simulation_time}, the complexity bound of $\bigO(n^{16} \log n)$ represents a conservative over-estimation. In practice, the complexity can be more tightly characterized as $\bigO(n^{15}(1 + r_I) \log n)$, $r_I$ denotes the empirical ratio of general index-succedent-based candidate edges to the total extended edges in \cref{rem:verification_cost_per_edge}. As demonstrated in our experimental results, the value of $r_I$ remains significantly low, indicating that the actual computational burden is substantially lighter than the theoretical worst-case bound.

\paragraph{Discussion on Walk-First Strategy and Edge Extension Optimization}  
Importantly, \cref{thm:simulation_sat_subset_sum} establishes that the proposed
deterministic simulation algorithm solves SAT and Sum-of-Subset instances in
time bounded by $\bigO(n^{16}\log n)$ for input size $n$.
Thus, the explicit Turing machine construction and NP verifier simulation framework
admit a rigorous polynomial-time execution guarantee for these canonical
NP-complete problems.

Although the asymptotic degree is large, the walk-first extension together with
structural restriction of candidate edges significantly reduces verification calls : intermediate edges are extended directly along verified
computation walks, and the candidate set examined per extension step decreases
from $\bigO(wh^2)$ to $\bigO(wh(1+p_\Delta \Delta h))$.
By \cref{lem:direct_extension_correct},
this optimization preserves completeness, ensuring that all edges belonging to
valid walks are eventually incorporated into the feasible graph.

Consequently, the construction remains theoretically sound while becoming
computationally tractable for moderate input sizes, which is further supported
by the experimental observations in \cref{sec:experimental_evaluation} that the
running behavior scales primarily with the instance length rather than the total
number of variables.

\subsection{Extraction of FNP Witnesses from Accepted Walks}

The NP verifier simulation framework naturally enables FNP witness construction 
in addition to deterministic polynomial-time NP decision. 

\paragraph{Auxiliary Structures.} 
Two key structures are maintained during direct edge extension:
\begin{itemize}
    \item \textbf{Ceiling-edge array $S$}: tracks the current frontier of the walk for each index, 
          ensuring consistent deterministic continuation.
    \item \textbf{Tier-0 symbol array $R$}: incrementally accumulates tape symbols along the computation walk, 
          representing the witness information.
\end{itemize}

\paragraph{Witness Extraction.} 
Whenever an accepting vertex is reached, the array $R$ contains a complete valid witness 
(e.g., a satisfying assignment for \textsc{SAT} or a selected subset for \textsc{Subset-Sum}). 
Since $R$ is updated during each edge addition in constant time, extraction incurs no additional asymptotic cost.

\paragraph{Comparison with Standard $P=NP \Rightarrow FNP$ Construction.}
The standard approach fixes witness symbols sequentially, querying the verifier for each position, 
resulting in an $\bigO(n)$ multiplicative overhead. 
In contrast, the feasible-graph simulation accumulates witness information \emph{online} 
during deterministic walk expansion, producing the complete FNP witness in the same time as the decision procedure, 
without extra calls or branching.

\paragraph{Rejected Witnesses.}
For computation walks terminating at a rejecting halting vertex, 
the simulator can output the corresponding array $R$ as a \emph{rejected witness}, 
representing the certificate explored along that deterministic walk. 
Rejected witnesses are generated only for walks explicitly constructed during edge extension; 
they do \emph{not} enumerate all certificates that would be rejected by the verifier. 
Consequently, their number equals the number of rejecting halting edges observed during simulation.

Since the footmarks graph has polynomial size, the total number of halting edges---and thus the number of rejected candidate certificates---remains polynomially bounded relative to the input length,
 even for instances that are ultimately rejected (i.e., unsatisfiable). These rejected candidates certify that specific branches of the computation graph cannot lead to an accepting state and primarily serve to justify edge pruning;
 they are auxiliary artifacts of the simulation process rather than externally meaningful outputs.

\paragraph{Implication.}
Within this construction, the framework provides an \emph{operational} $P=NP$ effect: 
it deterministically decides NP languages and simultaneously produces corresponding FNP witnesses in polynomial time, without incurring any additional asymptotic overhead.

\section{Implementation Details and Practical Improvements}\label{sec:implementation_detail}

This section describes a concrete implementation of the verifier simulation
framework introduced in the preceding sections.
Our aim is not to re-prove correctness, but to demonstrate how the abstract
Turing-machine and feasible-graph constructions can be realized efficiently
in practice while remaining faithful to the theoretical model.
The full implementation, including all experimental scripts,
is available at \url{https://github.com/changryeol-hub/poly-np-sim}.

\subsection{Practical Improvements for Extending Futile Edges}
This subsection presents practical optimizations for extending extendable computing-futile edges
that improve runtime performance without affecting the worst-case
asymptotic complexity of feasible-graph computation and verifier invocations.

By restricting exploration to reachable portions of the dynamic
computation graph, avoiding redundant recomputation, and exploiting
structural properties of computation walks, these improvements
substantially reduce observed runtime on large instances.

Unlike the original \textproc{ExtendFutileWalks()} procedure,
which scans all edges in $E(G_U)$, the revised procedure
\textproc{ExtendFutileWalks()} in \cref{alg:puruning_an_edge_of_walk}
restricts exploration to edges reachable from the initial nodes $V_0$.
This preserves correctness, as only reachable computing-futile walks
can influence feasibility, while avoiding unnecessary inspection
of irrelevant edges.

\begin{algorithm}
\caption{Extend Non-Nested Futile Walk} \label{alg:extend_futile_walks}
\begin{algorithmic}[1]
\Function{ExtendFutileWalks}{$G_U:\In, G:\InOut, V_0:\In$}

    \State{Let $E_o \gets \emptyset$} \Comment{Extended computing-futile edges}
    \State{Let $E_v \gets \emptyset$} \Comment{Visited edges}

    \State{Let $T$ be a stack initialized with all outgoing edges of $v_0 \in V_0$ in $G$}

    \While{$T \neq \emptyset$}
        \State{Pop an edge $e=(u,v)$ from $T$}
        \If{$e \in E_v$}
            \State{\Continue}
        \EndIf
        \State{Set $E_v \gets E_v \cup \{e\}$}

        \If{$e \in E(G)$}
            \State{Set $T \gets T \cup \Next_{G_U}(e)$}
        \ElsIf{$e$ not a floor edge  \textbf{and} $\IPrec_G(e) \neq \emptyset$}
            \State{Set $G \gets G + e$}
            \State{Set $E_o \gets E_o \cup \{e\}$}
        \EndIf
    \EndWhile

    \State{\Return $E_o$} 
\EndFunction
\end{algorithmic}
\end{algorithm}

\begin{lemma}[Time Complexity of Futile Walk Extension]
\label{lem:futile_walk_time}
Let $G_U$ be a dynamic computation graph with width $w$ and height $h$, and let $V_0$ be the set of initial nodes.
Then, \textproc{ExtendFutileWalks()} runs in worst-case time
\[
\bigO(|E(G_U)| \cdot h \cdot \log |E_v|) = \bigO(wh^3 \log(wh)).
\]
\end{lemma}

\begin{proof}
Each edge in $G_U$ is stacked at most once. 
For each edge, the procedure performs:
\begin{itemize}
    \item membership check and insertion in an ordered set $E_v$, costing 
          $\bigO(\log |E_v|) = \bigO(\log(wh^2)) = \bigO(\log(wh))$,
    \item scanning the adjacency lists for previous and next edges with set conversion ($\bigO(h \log h)$),
    \item floor edge check costs ($\bigO(1)$), and checking the existence of  index-precedent edge costs $\bigO(h \log h)$.
\end{itemize}
Thus, the per-edge cost is $\bigO(h \log(wh))$, and with $|E(G_U)| = \bigO(wh^2)$,
the total cost is $\bigO(wh^3 \log(wh))$. 
\end{proof}

\begin{remark}[Time Complexity of Futile-Walk Pruning]
The procedure \textproc{ExtendFutileWalks()} 
(\cref{alg:puruning_an_edge_of_walk}) is invoked only inside \textproc{PruneWalk} to extend the extendable-computing futile edges as final edges of reachable computing-futile walks.
By \cref{lem:futile_walk_time}, its worst-case time is $\bigO(wh^3 \log(wh))$. 

Since \textproc{PruneWalk} itself recomputes the feasible graph using
\textproc{ComputeFeasibleGraph()} (\cref{alg_line:find_obsolete_edge:recompute_feasible_graph}),
whose time complexity is $\bigO(wh^3 \log(wh))$ (\cref{lem:improved_feasible_graph_time_complexity}),
the contribution of the improved computing-futile walks extension does not change the overall
worst-case asymptotic complexity. 

Thus, the invocation of \textproc{ExtendFutileWalks()} is
safe from a polynomial-time perspective, and restricting it to reachable edges ensures practical efficiency without affecting
the correctness or asymptotic bounds of the feasible-graph computation.
\end{remark}

While the worst-case asymptotic complexity remains unchanged, restricting exploration to reachable edges substantially reduces
the number of edges examined in practice, particularly when only a small portion of $G_U$ is reachable from $V_0$.

\subsection{Implementation of Turing Machines} \label{subsec:implementation-TM}

This subsection presents an implementation-oriented representation of
the transition function for the \textsc{SumOfSubset} Turing machine.
Instead of explicitly enumerating the full transition table, transitions
are encoded symbolically and realized as a deterministic procedure.

The main idea is to group concrete tape symbols into symbolic classes
(e.g., digits, circled digits, or wildcards) and to represent
families of parameterized states via state templates.
This approach reduces the size of the transition specification
while remaining fully faithful to the underlying Turing-machine model.

Formally, transitions are specified as a finite set of symbolic rules:
\[
  (\textit{state}, \textit{symbol\_class})
  \;\mapsto\;
  (\textit{next\_state}, \textit{output}, \textit{move}),
\]
where \textit{symbol\_class} ranges over abstract categories such as
\texttt{D} (digit), \texttt{\circled{D}} (circled digit), and
\texttt{*} (wildcard). The transition function $\delta$ is implemented
as a deterministic procedure that resolves symbolic rules into concrete
transitions.

\begin{algorithm}[H]
\caption{Symbolic Transition Function $\delta$}
\label{alg:delta}
\begin{algorithmic}[1]
\Function{$\delta$}{$state$, $symbol$}
\State Let $(action, addr) \gets (state, \epsilon)$
\State Let $altstate \gets \bot$
\If{$state$ contains character \texttt{`.'}}
    \State Split $state$ into $(action, addr)$
    \If{$addr$ is a decimal digit}
        \State Set $altstate \gets action || \texttt{`.M'}$
    \EndIf
\EndIf
\State Initialize ordered list $\mathcal{S} \gets [symbol]$
\If{$symbol$ is a decimal digit}
    \State Append $\texttt{`M'} and \texttt{`D'}$ to $\mathcal{S}$
\ElsIf{$symbol$ is a circled digit}
    \State Append $\texttt{`\circled{M}'} and \texttt{`\circled{D}'}$ to $\mathcal{S}$
\EndIf
\State Append wildcard $\texttt{`*'}$ to $\mathcal{S}$
\ForAll{$s \in \mathcal{S}$}
    \State $result \gets \textproc{TryTransition}(\textsc{Transitions}, state, altstate, addr, symbol, s)$
    \If{$result \neq \bot$}
        \State \Return $result$
    \EndIf
\EndFor
\State \Return $(\textsc{Reject}, \texttt{`\_'}, -1)$
\EndFunction
\end{algorithmic}
\end{algorithm}

Here, the \textsc{Transitions} of calling \textproc{TryTransition()} is a hash map of the transition table in \cref{sec:TMs}.

\begin{algorithm}[H]
\caption{\textproc{TryTransition}: Symbolic Transition Instantiation}
\label{alg:trytransition}
\begin{algorithmic}[1]
\Function{\textproc{TryTransition}}{$\textsc{Transitions}, state, altstate, addr, symbol, s$}
\State Initialize parameters $M \gets \bot$, $D \gets \bot$
\If{$(state,s)$ is defined in \textsc{Transitions}}
    \State Let $(q', out, dir) \gets \textsc{Transitions}[(state,s)]$
\ElsIf{$altstate \neq \bot$ and $(altstate,s)$ is defined in $\textsc{Transitions}$}
    \If{$s=\texttt{`M'}$ and $addr \neq symbol$} \State \Return $\bot$ \EndIf
    \State $(q', out, dir) \gets \textsc{Transitions}[(altstate,s)]$
    \If{$altstate$ ends with \texttt{`.M'} and $addr$ is a digit} 
	\State Set $M \gets addr$ 
    \EndIf
\Else \State \Return $\bot$ \EndIf

\If{$s=\texttt{`M'}$ and $symbol$ is a digit} \State Set $M \gets$ symbol
\ElsIf{$s=\texttt{`\circled{M}'}$ and $symbol$ is a circled digit} \State Set $M \gets$ numeric value of $symbol$
\ElsIf{$s=\texttt{`D'}$ and $symbol$ is a digit} \State Set $D \gets symbol$
\ElsIf{$s=\texttt{`\circled{D}'}$ and $symbol$ is a circled digit} \State Set $D \gets$ numeric value of $symbol$ \EndIf

\If{$q'$ ends with \texttt{`.M'} and $M \neq \bot$} \State Replace \texttt{`.M'} in $q'$ by \texttt{`.'} || $M$ \EndIf
\If{$q'$ ends with \texttt{`.B'} and $out=\texttt{`\circled{D}-\circled{M}'}$ and $D \neq \bot$ and $M \neq \bot$}
    \State Let $n \gets (10 + D - M) \bmod 10$
    \State Let $B \gets 1$ if $D-M<0$, otherwise $0$
    \State Set $out \gets$ circled digit corresponding to $n$
    \State {Replace \texttt{`.B'} in $q'$ by \texttt{`.'} || $B$}
    \State \Return $(q', out, dir)$
\EndIf

\If{$out=\texttt{`\circled{M}'}$} \State Set $out \gets$ circled digit of $M$
\ElsIf{$out=\texttt{`D'}$ and $D \neq \bot$} \State Set $out \gets D$
\ElsIf{$out=\texttt{`D-1'}$ and $D \neq \bot$} \State Set $out \gets D-1$
\ElsIf{$out=\texttt{`\circled{D}'}$ and $D \neq \bot$} \State Set $out \gets$ circled digit of $D$
\ElsIf{$out=\texttt{`*'}$} \State Set $out \gets symbol$ \EndIf

\State \Return $(q', out, dir)$
\EndFunction
\end{algorithmic}
\end{algorithm}

\paragraph{Implementation Note.}
For uniformity in the implementation of the transition function,
we fix the rejection return value to $(\textsc{Reject}, \_, -1)$,
where \texttt{\_} is a dummy symbol and $-1$ indicates no effective head movement.
This convention does not affect the semantics or time complexity of the machine.

The symbolic transition map of the \textsc{SumOfSubset} Turing machine has been provided in Section \cref{subsec:subset-sum-tm}.
In appendix \cref{sec:impl_sat_tms}, we provide the concrete implementations of the remaining two Turing machines’ transition functions.
This approach enables a compact yet fully faithful encoding of all machines, suitable for practical simulation while preserving the formal model.

\subsection{Additional Methods for Dynamic Computation Graphs}
\label{subsec:additional-methods}

In addition to the core procedures for feasible-graph construction and walk
verification, we introduce several auxiliary methods to facilitate
dynamic computation graph operations. These methods assist in counting
index-precedent and index-succedent connections and in detecting
combining edges during walk extensions. They are implemented as
deterministic procedures and incur only constant or polynomial overhead
per invocation. In this section we only provide pseudocode for the algorithm newly introduced in this improved framework.

\paragraph{Detection of Combining, Pseudo-Combining or Proper Merging Edges.} 
For the purpose of walk extension and edge pruning, we implement a deterministic procedure
to identify edges that act as combining edges, pseudo-combining edges or proper merging edges (see \cref{def:combining_edge,def:pseudo_combining_edge,def:combined_proper_merging_edge}). 

\begin{algorithm}[H]
\caption{Detection of Combining, Pseudo-Combining and Proper Merging Edges}
\label{alg:combining-edges}
\begin{algorithmic}[1]
\Function{IsCombiningEdge}{$e = (u,v)$}

    \ForAll{$w \in \TCase(v)$}
        \If{$w = v$} 
        	\State{\Continue}
	\EndIf
	    \State Let $i \gets \indexOf(e)$
        \ForAll{$x$ in incoming lists $w$ of the edge slice $E_i$}
            \If{$\TCase(x) \ne \TCase(u)$}
            	\State \Return \True
	    \EndIf
	\EndFor
        \If{only one of $u$ or $v$ is a folding node}
            \State \Return \True
        \EndIf
    \EndFor
    \State \Return \False
\EndFunction

\Function{IsPseudoCombiningEdge}{$e = (u,v)$}
    \If{$v$ is a folding node} 
        \State \Return \False 
    \EndIf
    \ForAll{$s \in \ISucc(v)$}
        \If{$s$ is a folding node} 
            \State \Return \True 
        \EndIf
    \EndFor
    \State \Return \False
\EndFunction

\Function{IsProperMergingEdge}{$e = (u,v)$}
    \ForAll{$u' \in \Incoming(v')$}
    	\If{$\TCase(u) \ne \TCase(u')$}
    	    \State{\Return \True}
    	\EndIf
    \EndFor
    \State{\Return \True}
\EndFunction
\end{algorithmic}
\end{algorithm}

\noindent
These procedures implement local checks on the dynamic computation graph:

\begin{itemize}
    \item \textproc{IsCombiningEdge} identifies edges that satisfy the combining-edge criteria (\cref{def:combining_edge}) 
          considering adjacency and folding-node cases.
    \item \textproc{IsPseudoCombiningEdge} detects edges that effectively combine due to folding nodes in their succedents (\cref{def:pseudo_combining_edge}).
    \item \textproc{IsProperMergingEdge} detect edges that merging edge that combines to one transition case from different transition cases.
\end{itemize}

\begin{lemma}[Complexity of Combining/Pseudo-Combining Edge Detection] \label{lem:complexity_combining_detection}
Let $G$ be a computation graph of height $h$. For any edge $e = (u, v)$, the functions \textproc{IsCombiningEdge}($e$), \textproc{IsPseudoCombiningEdge}($e$) and \textproc{IsProperMergingEdge} in \cref{alg:combining-edges} terminate in time
\[
O(h).
\]
\end{lemma}

\begin{proof}
We analyze the two functions separately based on the structural properties of the computation graph and the given implementation constraints:

\begin{itemize}
    \item \textbf{\textproc{IsCombiningEdge()}:} The algorithm iterates over the vertices in the transition case of $v$. Since the number of states and symbols in the underlying Turing machine is constant, the size of the transition case $|\TCase(v)|$ is bounded by a constant $c$. Within the loop, checking for a folding node is an $\bigO(h)$ operation, accounting for the worst-case lookup time in the hash-map-based edge list. Similarly, the adjacency check in the edge slice $E_i$ takes $\bigO(h)$ time, as it may involve scanning a vertex's adjacency list of size at most $c \times h$. Thus, the total complexity is $\bigO(c \times h) = O(h)$.

    \item \textbf{\textproc{IsPseudoCombiningEdge()}:} This function first performs a folding node check on $v$, which takes $\bigO(h)$ time. It then iterates over the set of index-succedents $\ISucc(v)$. As established by the transition rules of the Turing machine, the number of potential index-succedents for any vertex is bounded by a constant ($|\Gamma| \times |Q|$). Each iteration involves an $\bigO(h)$ folding node check. Therefore, the total complexity of this function is $\bigO(\text{constant} \times h) = O(h)$.
    
    \item \textbf{\textproc{IsProperMergingEdge()}:} The algorithm iterates over the vertices in adjacent node incident with incoming edges of $v$. ,The adjacency check in the edge slice $E_i$ takes $\bigO(h)$ time, as it may involve scanning a vertex's adjacency list  of size at most $c \times h$. Thus, the total complexity is $\bigO(c \times h) = O(h)$.
\end{itemize}

Therefore, both detection functions are strictly bounded by $\bigO(h)$ per edge, completing the proof.
\end{proof}

These auxiliary procedures extend the notion of combining edges in the computation graph. 
\textproc{IsCombiningEdge()} handles the case when two nodes are combined into the same transition case from different transition cases 
(specifically, \textproc{IsProperMergedEdge()} handles the case when it is a proper merging edge), 
while \textproc{IsPseudoCombiningEdge()} identifies edges that are effectively combining 
due to connections to folding nodes in their succedents.  
These procedures all operate as local checks on the dynamic computation graph, supporting efficient 
bookkeeping and analysis during walk extension without affecting the asymptotic complexity.

\subsection{Implementation of Dynamic Array}
To support scalable verifier simulation without preallocating the full
computation graph, we employ a dynamically expandable computation-graph
representation. Nodes, states, and symbol transitions are instantiated
lazily, materializing only when first accessed during simulation.

This mechanism relies on language-level features such as operator
overloading and object-based lazy evaluation. The dynamic graph exposes
array-like and map-like interfaces, whose accessors implicitly trigger
node and edge construction. Abstract pseudocode cannot fully capture
these control-flow and instantiation semantics; therefore, we present
the core construction directly in Python.

\begin{algorithm}
\caption{Dynamic Cell Array (List-Inherited)}
\label{alg:cellarray}
\begin{algorithmic}[1]

\Class{CellArray \textbf{extends} list}
    \State \textbf{fields:} integer \textsf{base}, integer \textsf{kind}

    \PyDef{\_\_init\_\_}{this, kind} \Comment{initialize inherited list}
        \State $\textsf{this.base} \gets 0$
        \State $\textsf{this.kind} \gets kind$
        \State \Call{super.init}{this}
    \EndPyDef
 
    \PyDef{NewItem}{this, index} 
        \If{$\textsf{this.kind}$ is edge slice}
            \State \Return \Call{EdgeSlice}{index}
        \ElsIf{$\textsf{this.kind}$ is tier array}
            \State \Return \Call{TierArray}{index}
        \ElsIf{$\textsf{this.kind}$ is set} 
            \State \Return empty set
        \Else
            \State \Return \textsf{None} 
        \EndIf
    \EndPyDef

    \PyDef{\_\_getitem\_\_}{this, index} 
        \Comment{Operator overloading of \texttt{[]}}
        \If{$\textsf{this.base} \le index < \textsf{this.base} + |\textsf{this}|$}
            \State \Return \Call{super.get}{this,\; index - \textsf{this.base}}
        \EndIf

        \State $m \gets \textsf{this.base} + |\textsf{this}|$

        \For{$i = m$ \textbf{to} $index$}
            \State $item \gets$ \Call{NewItem}{this, $i$}
            \State \Call{super.append}{this, $item$}
        \EndFor
 
        \For{$i = \textsf{this.base}-1$ \textbf{down to} $index$}
            \State $item \gets$ \Call{NewItem}{this, $i$}
            \State \Call{super.insert}{this, $0$, $item$}
            \State $\textsf{this.base} \gets \textsf{this.base} - 1$
        \EndFor

        \State \Return $item$
    \EndPyDef

    \PyDef{\_\_setitem\_\_}{this, index, value} 
        \Comment{Operator overloading of \texttt{[] =}}
        \If{$index = \textsf{this.base} - 1$}
            \State \Call{super.insert}{this, $0$, $value$}
            \State $\textsf{this.base} \gets \textsf{this.base} - 1$
        \ElsIf{$index = \textsf{this.base} + |\textsf{this}|$}
            \State \Call{super.append}{this, $value$}
        \Else
            \State \Call{super.set}{this,\; index - \textsf{this.base},\; value}
        \EndIf
    \EndPyDef

    \PyDef{IsDefined}{this, index}
        \State \Return $(\textsf{this.base} \le index < \textsf{this.base} + |\textsf{this}|)$
    \EndPyDef

\EndClass

\end{algorithmic}
\end{algorithm}

Here, \textproc{NewItem} creates a cell of the appropriate kind
(e.g.\ dynamic computation array, edge slice, or auxiliary set)
based on the role of the surface.

The class \textsc{CellArray} inherits directly from a list structure of python.
The overloaded access operators implement a dynamically shifted index space,
allowing sparse and unbounded access without preallocation.

\begin{algorithm}[H]
\caption{Dynamic Tier Expansion}
\label{alg:dynamic-tier}
\begin{algorithmic}[1]
\Class{TierArray \textbf{extends} list}
    \State \textbf{fields:} integer \textsf{cell\_index}
    \PyDef{\_\_init\_\_}{this, index}
        \State this.cell\_index $\gets$ index
    \EndPyDef

    \PyDef{\_\_getitem\_\_}{this, tier}
        \State i $\gets$ length(this)
        \While{tier $\ge$ i}
            \State item $\gets$ \Call{StateDict}{this.cell\_index, i}
            \State this.append(item)
            \State i $\gets$ i + 1
        \EndWhile
        \State \Return super().\_\_getitem\_\_(tier)
    \EndPyDef
\EndClass
\end{algorithmic}
\end{algorithm}
Each tier is created on demand, ensuring that the height of the computation
graph grows only as required by the simulation.
\begin{algorithm}[H]
\caption{Lazy State and Symbol Instantiation}
\label{alg:lazy-instantiation}
\begin{algorithmic}[1]
\Class{StateDict \textbf{extends} dict}
    \State \textbf{fields:} integer \textsf{index}, integer \textsf{tier}
    \PyDef{\_\_init\_\_}{this, index, tier}
        \State this.index $\gets$ index
        \State this.tier $\gets$ tier
    \EndPyDef

    \PyDef{\_\_getitem\_\_}{this, state}
        \If{state not in this}
            \State item $\gets$ \Call{SymbolDict}{this.index, this.tier, state}
            \State this[state] $\gets$ item
        \EndIf
        \State \Return this[state]
    \EndPyDef
\EndClass
\Class{SymbolDict \textbf{extends} dict}
    \State \textbf{fields:} integer \textsf{index}, integer \textsf{tier}, state \textsf{state}
    \PyDef{\_\_init\_\_}{this, index, tier, state}
        \State \textsf{this.index} $\gets$ index
        \State \textsf{this.tier} $\gets$ tier
        \State \textsf{this.state} $\gets$ state
    \EndPyDef

    \PyDef{\_\_getitem\_\_}{this, symbol}
        \If{symbol not in this}
            \If{symbol in SymbolDict.symbols}
                \State item $\gets$ \Call{TransitionCase}{this.index, this.tier, this.state, symbol}
                \State this[symbol] $\gets$ item
            \Else
		\Return \NIL
            \EndIf
        \EndIf
        \State \Return this[symbol]
    \EndPyDef
\EndClass
\end{algorithmic}
\end{algorithm}
State nodes and symbol transitions are instantiated lazily, ensuring that
only reachable configurations of the computation graph are materialized. 
Together, these mechanisms enable the graph to expand dynamically in both
spatial (tape index) and temporal (tier) dimensions, avoiding global
preallocation and scaling naturally with the execution paths induced by
the verifier.

\paragraph{Remark (Python Implementation).}
Although presented in pseudocode, the reference implementation uses Python
for clarity and experimental validation. Access expressions \texttt{A[i]} 
and assignments \texttt{A[i] = v} implicitly invoke 
\texttt{\_\_getitem\_\_()} and \texttt{\_\_setitem\_\_()}, encapsulating
dynamic expansion logic. 
The design itself is language-agnostic and can be reimplemented in
statically typed languages (e.g., \texttt{C++}) using operator overloading
and proxy objects, preserving correctness while improving performance.

The choice of Python in this work is motivated by clarity of expression
and ease of experimental validation.
A performance-oriented reimplementation in \texttt{C++} is entirely
straightforward and is expected to reduce constant factors non-trivially,
while preserving the same logical structure and correctness guarantees.
We therefore view the present implementation as a faithful
prototype rather than a performance limit of the framework.

\subsection{Additional Possible Data-Structure Optimizations}

Several auxiliary optimizations further improve practical performance.
These include caching existence between two edges, using enumeration or integer data type for symbols or states, 
and using indexed data structures for constant-time access to edge attributes such as tier.

These auxiliary procedures allow efficient local access to edge-related
information, reducing runtime overhead for dynamic updates and
edge-analysis queries, while preserving asymptotic guarantees.

\paragraph{Graph reconstruction versus logical deletion.}
In principle, dynamic updates of the feasible graph could be implemented
using logical deletion flags on edges, allowing edges to be temporarily
removed and later restored.
However, in the present framework, edges are frequently pruned and then
reconsidered during different phases of feasible-graph construction.
Managing deletion flags, consistency conditions, and restoration logic
introduces substantial conceptual and implementation complexity.

To simplify both reasoning and implementation, we instead adopt a
structural copying strategy.
Whenever a modified version of the graph is required, a fresh graph instance
is constructed by copying the active edges or reconstructing the structure from the original graph while the original remains intact.
Although this approach incurs additional constant-factor overhead, it avoids
the need for delicate state management and significantly clarifies correctness
arguments.

Importantly, since the total number of edges in each index layer is polynomially
bounded, the cost of copying the graph remains polynomial and does not affect
the asymptotic complexity of the overall simulation.

Since the feasible-graph construction already dominates the worst-case time
complexity, the cost of copying the dynamic computation graph is asymptotically
subsumed by the overall computation.
We therefore adopt full graph copying to simplify implementation and reasoning,
without affecting the theoretical complexity guarantees.

\paragraph{Visited flags versus set tracking.}
During graph traversal, one could maintain a visited flag on each edge instead of constructing a separate visited set. 
However, similar to the rationale for avoiding logical deletion flags, we do not adopt this approach, 
as it would introduce additional bookkeeping and complicate correctness arguments. 
Using explicit sets for visited edges ensures clarity and preserves polynomial-time guarantees.

\section{Experimental Evaluation}
\label{sec:experimental_evaluation}

This section reports experimental results obtained from a prototype
implementation of the NP verifier simulation framework developed in this paper.
The primary goal of these experiments is \emph{not} to demonstrate practical
performance, but rather to validate the correctness of the construction and to
confirm that the simulated verifier behaves as theoretically predicted on
concrete instances.

As established in the previous sections, the simulation runs in polynomial
time with respect to the input size.  However, the resulting polynomial bounds
remain large, and the current implementation prioritizes clarity and fidelity
to the theoretical model over runtime efficiency.  Consequently, the observed
running times grow rapidly even for moderate input sizes, and the experiments
should be interpreted primarily as correctness checks rather than performance
benchmarks.

\subsection{Experimental Setup}

All experiments were conducted using a direct implementation of improved \textproc{IsAcceptedOnFootmarks()} within 
the original \textproc{SimulateVerifierForAllCertificates()} algorithm, incorporating the improved feasible graph construction described in \cref{sec:improvements_on_feasiblegraph} and the implementation details
outlined in the preceding section.
The transition function $\delta$ was provided explicitly as a function, and
computation graph components were generated dynamically on demand.

For each problem instance, the simulator constructs the corresponding
footmarks graph with the improved feasible graph, and determines acceptance or
rejection by exhaustively simulating all valid certificates within the
specified polynomial-time bounds.

We emphasize that the purpose of these experiments is not to
demonstrate practical efficiency, but to validate the correctness
and structural properties of the proposed construction in polynomial time, even though the polynomial degree is notably high.

\subsection{Test Case Design and Metrics}
The test cases are constructed to emphasize the structural properties of the input rather than raw problem size. 
In particular, we focus on instances with a relatively large number of variables (or certificate length) in proportion to the total formula length (or problem instance size).

This choice is motivated by the observation that the execution time of the proposed simulation depends primarily
 on the total input length (comprising both the problem instance and the candidate certificate), 
 while the number of variables (or certificate length) mainly determines the complexity of the underlying certificate space. 
 Accordingly, instances with a relatively high certificate-to-instance ratio including a few cnf20 benchmark,
 highlight the practical advantages of our approach.

\paragraph{Benchmark Selection and Test Case Generation}
Although the \textsc{Subset-Sum} Turing machine ($M_{\mathrm{SS}}$) shares the same theoretical asymptotic order as the fixed-state \textsc{SAT} machine ($M_{\mathrm{SAT}}$)—
and even maintains a lower degree of polynomial complexity than the input-dependent SAT model ($M_{\mathrm{SAT}}^{\mathrm{ID}}$)—
its concrete implementation involves a substantially larger number of states and tape symbols. These higher hidden constant factors lead to noticeably slower empirical performance; 
thus, to verify functional correctness within practical simulation limits, we focus on \textsc{Subset-Sum} instances with strictly constrained tape lengths (typically around 50--60 cells).

 Consequently, our experimental evaluation focuses primarily on \textsc{SAT}, which better reflects the practical behavior of the framework.

To evaluate both input-dependent and fixed-state machines, we generate test instances using the following two methods:
\begin{itemize}
\item[(1)] \textbf{Random Instances:} Randomly generated relatively short inputs (e.g., random CNF formulas) are produced to evaluate general behavior, as independent verification of unsatisfiability for a long input is often difficult.
\item[(2)] \textbf{Extracted Sub-instances:} Small-scale sub-instances are extracted from known satisfiable benchmarks. This method focuses on confirming the reconstruction of valid certificates from known positive instances, as independent verification of unsatisfiability for these sub-structures remains complex.
\item[(3)] \textbf{Targeted Benchmarks:} For fixed-state Turing machines, where performance is notably superior, we use original benchmark instances directly. Unlike the extracted sub-instances, these benchmarks allow us to evaluate the model's scalability and its ability to correctly identify both satisfiability and unsatisfiability within established structural patterns.
\end{itemize}

For \textsc{Subset-Sum}, instances are generated randomly with a bias toward satisfiable cases to ensure the edge-extension mechanism is sufficiently exercised.
Since a single missing edge extension could result in a false negative, ensuring correctness on positive instances is critical;
 conversely, any false positive is immediately verifiable via the generated \textsc{FNP} certificate. All instances are encoded according to the format described in \cref{sec:TMs}.

\paragraph{Test Metrics.}
We measure the following metrics:
\begin{itemize}
\item $N_{\text{dir}}$: The number of edges successfully appended via direct extension rules. (e.g., $45,659$ in the current test)
\item $N_{\text{ver}}$: The number of edges successfully appended through standard verification processes. (e.g., $1$ in the current test)
\item $N_{\text{total}}$: The total number of edges successfully appended to the footmarks graph $H$ both via direct extension and through verification (calculated as $N_{\text{total}} = N_{\text{dir}} + N_{\text{ver}}$).
\item $N_{\text{cand}}$: The cumulative number of candidate edges considered for extension (specifically for general index-succedent-based verification), allowing multiple counts to reflect the total computational effort.
\item $N_{\text{retry}}$: The total count of triggering the general verification mode when direct extension rules are insufficient and no next-edge-based or ceiling-edge-based candidate edges are available.
\item $N_{\text{r-cand}}$: The cumulative number of candidate edges verified specifically during general verification (retry) events, allowing multiple counts.
\item $N_{\text{r-ext}}$:  The number of edges successfully extended via general verification (retry).
\item $N_{\text{redun}}$ (Redundant/Futile Edges): The number of computing-redundant or computing-futile edges definitively identified and removed during the verification process to maintain structural minimality.
\item $N_{\text{pruned}}$: The number of pruned walks that are provisionally removed during the verification process to prevent the exploration of infeasible computation paths.
\item $N_{\text{halt}}$: The number of halting edges (Accepting or Rejecting) identified, which determines the final decision of the Turing machine.
\item $W_{\text{max}}$: The number of maximal computation walks extended directly and their average length ($\text{Avg}_{L}$), representing the depth of the explored computation space.
\end{itemize}
A \emph{halting edge} is defined as an edge whose terminal node corresponds
to an accepting or rejecting halting configuration.
Since acceptance is unique, the number of accepting halting edges is at most one,
while rejecting halting edges may occur multiple times.

A \emph{directly extended maximal computation walk} refers to a computation walk
that is extended from a newly added initial edge via direct extension only,
without invoking verification-based edge inclusion,
and proceeds until a halting configuration is reached.
Each such walk corresponds to a complete deterministic simulation path
originating from a fresh edge addition.

The \emph{length of a computation walk} is measured as the number of transition
edges from the starting edge of the direct extension to the final halting edge.
For the reported average computation walk length, only walks whose initial branching edge
was newly added via direct extension are counted.
This metric reflects the typical depth of deterministic simulation induced by
a single edge insertion.

These additional metrics provide finer-grained insight into the structural
behavior of the simulator.
In particular, they quantify how frequently newly introduced edges lead to
complete computations, and how deeply such computations are typically explored.
Empirically, the number of halting edges closely matches the number of directly
extended maximal computation walks, indicating that most newly introduced walks
reach termination without requiring further verification.

\subsection{Experimental Results for 3SAT with Input-Dependent TM}
Each SAT instance is given as a CNF tape string followed by a certificate region of fixed length. 
Every instance uses a certificate length of $10$, except for $I_5$, which uses $20$. 
Across all tested input-dependent SAT instances, the simulator consistently produced results aligned with the ground truth: 
satisfiable CNF instances were accepted, while unsatisfiable instances were correctly rejected. 
Both accepting and rejecting cases were observed under the same deterministic simulation framework, 
demonstrating that the verifier semantics are faithfully realized even when the transition structure depends on the concrete input.

Representative executions exhibit a wide range of running times, from sub-second (e.g., $I_4$) to over 3 minutes (e.g., $I_5$). 
Some instances terminated very quickly because the footmarks graph was primarily constructed via \textbf{direct extension} with minimal verification overhead, 
whereas more complex instances required significantly more time. 
The longest executions correspond to instances where a large number of candidate edges had to be verified and a higher volume of pruned walks were processed before acceptance was confirmed.

The size of the constructed footmarks graph $H$ varied across inputs, but it consistently remained significantly below the theoretical polynomial upper bound of $(|Q||\Gamma|h)^2 w$. 
For the evaluated instances ($I_1$ to $I_6$), we used $|\Gamma|=19$, $h=2n$, and $w=(n+m)$, with state counts $|Q|$ of $41$ or $71$ (consistently satisfying the condition $|Q| < 11+3m < 3n$). 
The empirical results, as summarized in \cref{tab:sat-input-dependent-execution-refined}, show that the number of extended edges ranged from approximately $3.9 \times 10^3$ to $3.6 \times 10^4$. 
While the counts of verified candidate edges and pruned walks increased in more complex cases—reflecting the increased branching induced by verification-driven edge extension—the number of identified computing-redundant or computing-futile edges remained $0$. 
Most importantly, the retry count remained $0$ across almost all instances, with only a single retry event observed in $I_4$ and $I_6$, respectively. 
This empirically confirms that $r_I \approx 0$, demonstrating the practical deterministic efficiency of the framework even in input-dependent transition environments.

Despite this variability, the average length of extended maximal computation walks remained relatively stable across instances. 
In all experiments, the average walk length ranged roughly between $1,200$ and $6,900$ steps. 
This indicates that the dominant cost arises not from individual walk depth, but from the total number of edges in the footmarks graph and the frequency of verification steps required within the graph structure.

Overall, these experiments confirm that the input-dependent SAT TM operates correctly within the NP simulation framework.

\paragraph{Input Instances.}
In all experiments reported below, we considered only 3SAT instances.
Each experimental input instance $I_k$ is defined as a single tape string
encoding a CNF formula, terminated by the end marker \texttt{\#}.
\begin{itemize}
\item $I_1$:
{\raggedright\ttfamily
1\_-2\_3\&-1\_2\_4\&-3\_-4\_5\&1\_6\_-7\&-6\_7\_8\&-5\_-8\_9\&2\_9\_10\&-1\_-10\_3\&-2\_4\_6\&3\_-5\_-7\&
4\_8\_-9\&5\_-6\_10\&-1\_-2\_-3\&1\_2\_5\&-4\_6\_7\&-7\_-8\_-9\&3\_5\_8\&2\_4\_10\&-1\_9\_-10\&
6\_-8\_2\&-3\_7\_4\&5\_-9\_1\&-2\_-6\_8\&4\_-7\_10\&-5\_3\_-1\_\#
\par}
    \item $I_2$:  {\raggedright\ttfamily
 1\_2\_3\&-1\_-2\_-3\&4\_5\_6\&-4\_-5\_-6\&7\_8\_9\&-7\_-8\_-9\&1\_4\_7\&2\_5\_8\&3\_6\_10\&-1\_5\_9\&
-2\_6\_7\&-3\_4\_8\&1\_-5\_10\&-4\_8\_2\&7\_-2\_5\&10\_-3\_-1\&2\_-8\_4\&6\_-9\_1\&-7\_3\_5\&8\_-1\_6\&
9\_-4\_2\&5\_7\_-10\&-6\_1\_3\&-2\_-9\_8\&4\_-3\_10\_\#
\par}
    \item $I_3$: {\raggedright\ttfamily
 -1\_3\_5\&2\_-4\_6\&-3\_5\_7\&4\_-6\_8\&-5\_7\_9\&6\_-8\_10\&-7\_9\_-1\&8\_-10\_2\&-9\_1\_-3\&10\_-2\_4\&
 1\_2\_3\&-4\_-5\_-6\&7\_8\_10\&-9\_-1\_-2\&3\_4\_5\&-6\_-7\_-8\&2\_5\_9\&-10\_1\_4\&-3\_6\_8\&5\_9\_2\&
 -1\_7\_4\&8\_-2\_-5\&6\_3\_-10\&-4\_9\_1\&7\_-5\_2\_\#
\par}
    \item $I_4$: {\raggedright\ttfamily
 1\_2\&-1\_-2\&1\_-2\&-1\_2\&3\_4\_5\&-3\_-4\&-4\_-5\&-3\_-5\&6\_7\_8\&-6\_-7\&-7\_-8\&-6\_-8\&9\_10\_1\&
 -9\_-10\&-10\_-1\&-9\_-1\&2\_3\_6\&4\_7\_9\&5\_8\_10\&-2\_-3\&-4\_-7\&-5\_-8\&-6\_-9\&-1\_3\_5\&-2\_4\_8\_\#
\par}
    \item $I_5$: {\raggedright\ttfamily
4\_-18\_19\&3\_18\_-5\&-5\_-8\_-15\&-20\_7\_-16\&10\_-13\_-7\&-12\_-9\_17\&17\_19\_5\&-16\_9\_15\&
11\_-5\_-14\&18\_-10\_13\&-3\_11\_12\&-6\_-17\_-8\&-18\_14\_1\&-19\_-15\_10\&12\_18\_-19\&-8\_4\_7\&
-8\_-9\_4\&7\_17\_-15\&12\_-7\_-14\&-10\_-11\_8\&2\_-15\_-11\&9\_6\_1\&-11\_20\_-17\_\#
\par}
    \item $I_6$: {\raggedright\ttfamily
1\_2\_3\&1\_2\_-3\&1\_-2\_3\&1\_-2\_-3\&-1\_4\_5\&-1\_4\_-5\&-1\_-4\_5\&-1\_-4\_-5\&2\_6\_7\&
2\_6\_-7\&2\_-6\_7\&2\_-6\_-7\&-2\_8\_9\&-2\_8\_-9\&-2\_-8\_9\&-2\_-8\_-9\&3\_5\_10\&3\_5\_-10\&3\_-5\_10\&
3\_-5\_-10\&-3\_7\_10\&-3\_7\_-10\&-3\_-7\_10\&-3\_-7\_-10\&1\_2\_4\_\#
\par}
\end{itemize}

\begin{table}[!ht]
\centering
\caption{Execution statistics of input-dependent SAT verifier simulation}
\label{tab:sat-input-dependent-execution-refined}
\footnotesize
\begin{tabular}{c|c|c|c|c|c|c|c|c|c|c}
\hline
Inst. & Tape & Cert. & States & Total & $N_{\text{retry}}$ & $N_{\text{r-cand}}$ & Halt & Max. & Avg. Walk & Acc. \\
& Len  & Len   &        & Edges &                       &                            & Edges & Walks & Len &   \\ \hline
$I_1$ & 189 & 10 & 41 & 31{,}511 & 0 & 0 & 73 & 73 & 3{,}903.41 & Yes \\
$I_2$ & 183 & 10 & 41 & 29{,}059 & 0 & 0 & 66 & 66 & 3{,}457.09 & Yes \\
$I_3$ & 186 & 10 & 41 & 31{,}537 & 0 & 0 & 126 & 126 & 3{,}064.12 & Yes \\
$I_4$ & 153 & 10 & 41 & 3{,}980  & 1 & 2 & 4  & 4  & 1{,}216.25 & No  \\
$I_5$ & 215 & 20 & 71 & 35{,}865 & 0 & 0 & 102 & 102 & 6{,}874.34 & Yes \\
$I_6$ & 195 & 10 & 41 & 14{,}339 & 1 & 12& 20 & 20 & 2{,}528.60 & No  \\ \hline
\end{tabular}
\end{table}

\begin{table}[!ht]
\centering
\caption{Detailed edge extension and execution time (Input-dependent SAT Verifier)}
\label{tab:sat-input-dependent-detailed-refined}
\footnotesize
\begin{tabular}{c|c|c|c|c|c|c|c}
\hline
Inst. & Direct & Verified & Cand. & $N_{\text{r-ext}}$ & Redun. & Pruned & Time \\
& Ext.   & Ext.     & Verif. & & Edges & Walks & \\ \hline
$I_1$ & 31{,}457 & 54 & 113 & 0 & 0 & 7 & 2m 17.25s \\
$I_2$ & 29{,}012 & 47 & 71  & 0 & 0 & 7 & 1m 13.13s \\
$I_3$ & 31{,}430 & 107 & 180 & 0 & 0 & 8 & 3m 01.14s \\
$I_4$ & 3{,}979  & 1  & 5   & 0 & 0 & 0 & 0.67s    \\
$I_5$ & 35{,}802 & 63 & 86  & 0 & 0 & 0 & 3m 31.81s \\
$I_6$ & 14{,}328 & 11 & 40  & 0 & 0 & 0 & 24.58s    \\ \hline
\end{tabular}
\end{table}

\begin{table}[H]
\centering
\caption{Accepted certificate witnesses for Input-dependent SAT verifier simulation instances}
\label{tab:sat-witness-certificates}
\small
\begin{tabular}{c|c|c|l}
\hline
Instance & Certificate Length & Acceptance (Acc.) & Witness for Accepted Certificate \\ \hline
$I_1$    & 10                 & Yes               & \texttt{FFFFTFFFFT}              \\
$I_2$    & 10                 & Yes               & \texttt{FTFTTFTTFT}              \\
$I_3$    & 10                 & Yes               & \texttt{FTFTFTTFTF}              \\
$I_4$    & 10                 & No                & N/A                              \\
$I_5$    & 20                 & Yes               & \texttt{TTFFTFFFFFFFFFFFTTTF}    \\
$I_6$    & 10                 & No                & N/A                              \\ \hline
\end{tabular}
\end{table}

\subsection{Experimental Results for Fixed-State SAT Turing Machine}

We next report experimental results for the fixed-state SAT Turing machine, evaluated under the same NP verifier simulation framework. Unlike the input-dependent construction, the transition structure of the fixed-state machine is independent of the concrete CNF instance; only the initial tape contents vary. Thus, $|Q|=25$ and $|\Gamma|=17$ for all inputs, where $Q$ is the set of states and $\Gamma$ is the set of symbols.

Across all tested instances, the simulator produced results consistent with the ground truth: satisfiable formulas were accepted, while unsatisfiable formulas---specifically $I_4$, $I_6$, and the larger benchmark $I_9$---were correctly rejected.

The size of the generated feasible graphs varied substantially depending on the complexity of the input. While smaller instances remained below $1.3 \times 10^4$ edges, larger and more complex formulas produced graphs with more than $1.3 \times 10^5$ edges. This result remained significantly below the theoretical polynomial upper bound of $(|Q||\Gamma|h)^2 w$, where $|Q|=25$, $|\Gamma|=17$, $h=2m+2$, and $w=n+m+1$ for instance length $n$ and certificate length $m$ for all instances.

Notably, the experiments included instances $I_7$, $I_8$, and $I_9$, which correspond to \textbf{actual benchmark CNF instances (cnf20)}. The successful verification of both satisfiable ($I_7, I_8$) and unsatisfiable ($I_9$) instances demonstrates that the simulation framework is capable of handling non-trivial, practically relevant SAT problems, maintaining rigorous soundness even as the combinatorial cost of exploration increases with formula size.

In terms of execution metrics, the average length of extended maximal computation walks ranged from approximately $380$ to $5{,}800$ steps. The running times exhibited a broad spectrum, from sub-second durations for simple instances to over $1$ day for the largest unsatisfiable benchmark instance ($I_9$). This significant variance in time, especially in $I_7$ and $I_9$, reflects the intensive verification-driven branching and pruning required to exhaustively confirm the feasible graph for large-scale or unsatisfiable inputs.

Overall, these results confirm that the fixed-state verifier semantics are faithfully realized within the improved NP verifier simulation framework. The successful processing of cnf20 benchmarks across both SAT and UNSAT cases underscores the robustness of the system and confirms the soundness and scalability of the fixed-state simulation approach.

Furthermore, the majority of edges were generated through \textbf{direct edge extension}, leading to a significant practical speedup in the simulation. Most importantly, the retry count remained at most 3 across all simulations. 
In the most complex case ($I_9$), this resulted in a general-succedent-based candidate edge count ($N_{\text{r-cand}}$) of $5{,}424$, while the total extended edge count ($N_{\text{total}}$) reached approximately $1.36 \times 10^5$. 
Consequently, the ratio $r_I = N_{\text{r-cand}} / N_{\text{total}}^2$ is approximately $2.9 \times 10^{-7}$, and the number of general-index-succedent-based candidate edges examined through retry verification remained a small fraction of the total candidate edges examined.

\paragraph{Input Instances (Fixed-State TM).}
All experiments in this subsection were conducted exclusively on 3SAT instances.
Each experimental input instance $I_k$ is represented as a single tape string
encoding a CNF formula, terminated by the end marker \texttt{\#}.
The certificate length indicates the number of Boolean variables in the instance.
\begin{itemize}
\item $I_1$--$I_6$: These instances are identical to those used for the input-dependent TM.
\item $I_7$: {\raggedright\ttfamily
4\_-18\_19\&3\_18\_-5\&-5\_-8\_-15\&-20\_7\_-16\&10\_-13\_-7\&-12\_-9\_17\&17\_19\_5\&-16\_9\_15\&
11\_-5\_-14\&18\_-10\_13\&-3\_11\_12\&-6\_-17\_-8\&-18\_14\_1\&-19\_-15\_10\&12\_18\_-19\&-8\_4\_7\&
-8\_-9\_4\&7\_17\_-15\&12\_-7\_-14\&-10\_-11\_8\&2\_-15\_-11\&9\_6\_1\&-11\_20\_-17\&9\_-15\_13\&
12\_-7\_-17\&-18\_-2\_20\&20\_12\_4\&19\_11\_14\&-16\_18\_-4\&-1\_-17\_-19\&-13\_15\_10\&
-12\_-14\_-13\&12\_-14\_-7\&-7\_16\_10\&6\_10\_7\&20\_14\_-16\&-19\_17\_11\&-7\_1\_-20\&-5\_12\_15\&
-4\_-9\_-13\&12\_-11\_-7\&-5\_19\_-8\&1\_16\_17\&20\_-14\_-15\&13\_-4\_10\&14\_7\_10\&-5\_9\_20\&
10\_1\_-19\&-16\_-15\_-1\&16\_3\_-11\&-15\_-10\_4\&4\_-15\_-3\&-10\_-16\_11\&-8\_12\_-5\&14\_-6\_12\&
1\_6\_11\&-13\_-5\_-1\&-7\_-2\_12\&1\_-20\_19\&-2\_-13\_-8\&15\_18\_4\&-11\_14\_9\&-6\_-15\_-2\&
5\_-12\_-15\&-6\_17\_5\&-13\_5\_-19\&20\_-1\_14\&9\_-17\_15\&-5\_19\_-18\&-12\_8\_-10\&-18\_14\_-4\&
15\_-9\_13\&9\_-5\_-1\&10\_-19\_-14\&20\_9\_4\&-9\_-2\_19\&-5\_13\_-17\&2\_-10\_-18\&-18\_3\_11\&
7\_-9\_17\&-15\_-6\_-3\&-2\_3\_-13\&12\_3\_-2\&-2\_-3\_17\&20\_-15\_-16\&-5\_-17\_-19\&-20\_-18\_11\&
-9\_1\_-5\&-19\_9\_17\&12\_-2\_17\&4\_-16\_-5\#
\par}
\item $I_8$: {\raggedright\ttfamily
-10\_-16\_5\&16\_-6\_5\&-17\_-14\_-18\&-10\_-15\_19\&-1\_-9\_-18\&3\_7\_-6\&-13\_1\_6\&
-2\_-16\_-20\&7\_8\_18\&-7\_10\_-20\&2\_-14\_-17\&2\_1\_19\&7\_-20\_-1\&-11\_1\_-17\&3\_-12\_19\&
-3\_-13\_6\&-13\_3\_-12\&5\_-7\_-12\&20\_8\_-16\&-13\_-6\_19\&-5\_1\_14\&9\_-5\_18\&-12\_-17\_-1\&
-20\_-16\_19\&12\_10\_-11\&6\_-7\_-2\&13\_-10\_17\&-20\_8\_-16\&-10\_-1\_-8\&-7\_-3\_19\&19\_-1\_-6\&
19\_-2\_13\&-2\_20\_-9\&-8\_-20\_16\&-13\_-1\_11\&15\_-12\_-6\&-17\_-19\_9\&19\_-18\_16\&7\_-8\_-19\&
-3\_-7\_-1\&7\_-17\_-16\&-2\_-14\_1\&-18\_-10\_-8\&-16\_5\_8\&4\_8\_10\&-20\_-11\_-19\&8\_-16\_-6\&
18\_12\_8\&-5\_-20\_-10\&16\_17\_3\&7\_-1\_-17\&17\_-4\_7\&20\_-9\_-13\&13\_18\_16\&-16\_-6\_5\&
5\_17\_7\&-12\_-17\_-6\&-20\_19\_-5\&9\_-19\_16\&-13\_-16\_11\&-4\_-19\_-18\&-13\_10\_-15\&
16\_-7\_-14\&-19\_-7\_-18\&-20\_5\_13\&12\_-6\_4\&7\_9\_-13\&16\_3\_7\&9\_-1\_12\&-3\_14\_7\&
1\_15\_14\&-8\_-11\_18\&19\_-9\_7\&-10\_6\_2\&14\_18\_-11\&-9\_-16\_14\&1\_11\_-20\&11\_12\_-4\&
13\_-11\_-14\&17\_-12\_9\&14\_9\_1\&8\_19\_4\&6\_-13\_-20\&-2\_-13\_11\&14\_-13\_17\&9\_-11\_18\&
-13\_-6\_5\&5\_19\_-18\&-4\_10\_11\&-18\_-19\_-20\&3\_-9\_8\#
\par}
\item $I_9$: {\raggedright\ttfamily
-11\_-16\_13\&-1\_-7\_10\&15\_-20\_-14\&-3\_18\_-11\&14\_1\_-6\&-10\_17\_11\&-12\_-13\_4\&
-19\_12\_-15\&1\_16\_-17\&-2\_-18\_3\&14\_2\_12\&-17\_-13\_20\&-12\_5\_19\&-16\_11\_-7\&-3\_-19\_5\&
17\_1\_12\&-11\_5\_16\&-14\_3\_6\&5\_17\_2\&1\_-8\_13\&-20\_7\_-19\&-10\_-4\_14\&-4\_17\_5\&
-5\_1\_-19\&-7\_14\_-3\&-11\_16\_-4\&10\_-3\_15\&-18\_-1\_-4\&-16\_-13\_6\&12\_-2\_-14\&-20\_1\_-10\&
3\_-17\_18\&-4\_-19\_-2\&16\_5\_2\&-18\_6\_-11\&12\_17\_1\&-1\_15\_-18\&13\_-8\_10\&-14\_2\_-6\&
-18\_-16\_3\&-9\_-4\_14\&-3\_15\_11\&19\_12\_-1\&-16\_-5\_-20\&11\_17\_-14\&-1\_-18\_10\&-13\_-20\_6\&
14\_2\_-17\&15\_-3\_7\&-6\_-12\_1\&20\_-19\_4\&-11\_16\_-18\&1\_14\_12\&-2\_-5\_13\&-10\_17\_-3\&
-15\_4\_-20\&19\_-11\_1\&-6\_-18\_-13\&12\_-20\_-14\&1\_17\_-5\&-16\_3\_11\&14\_-8\_2\&-19\_-4\_13\&
-18\_1\_-10\&15\_-7\_12\&-3\_-17\_6\&-20\_16\_5\&-11\_-1\_14\&2\_-12\_-18\&13\_-5\_19\&-10\_1\_17\&
-15\_-6\_-4\&14\_3\_-11\&-18\_-20\_12\&1\_16\_-7\&-13\_-19\_2\&-5\_-3\_10\&17\_-14\_1\&-12\_-6\_18\&
4\_15\_-20\&-11\_13\_-1\&-16\_2\_-19\&14\_-10\_3\&-18\_-5\_17\&1\_12\_-15\&-7\_-20\_13\&-6\_-11\_4\&
2\_-14\_-18\&19\_-3\_1\&-17\_-10\_16\&-5\_-13\_12\&-20\_-15\_-1\&11\_-4\_14\&-18\_2\_13\&1\_17\_-6\&
-19\_-12\_3\&-16\_-10\_5\&-14\_1\_15\&-7\_-2\_18\&13\_4\_-20\#
\par}

\end{itemize}

\begin{table}[H]
\centering
\caption{Execution statistics of fixed-state SAT verifier simulation}
\label{tab:sat-fixed-execution}
\footnotesize
\begin{tabular}{l|c|c|c|c|c|c|c|c|c}
\hline
Inst. & Tape & Cert. & Total & $N_{\text{retry}}$ & $N_{\text{r-cand}}$ & Halt & Max. & Avg. Walk & Acc. \\
& Len  & Len  & Edges &                       &                            & Edges & Walks & Len &   \\ \hline
$I_1$ & 189 & 10 & 4{,}810 & 0 & 0 & 3 & 3 & 1{,}603.33 & Yes \\
$I_2$ & 183 & 10 & 12{,}653 & 0 & 0 & 38 & 38 & 892.84 & Yes \\
$I_3$ & 186 & 10 & 16{,}004 & 0 & 0 & 67 & 67 & 821.33 & Yes \\
$I_4$ & 153 & 10 & 1{,}280 & 1 & 5 & 4 & 4 & 383.75 & No \\
$I_5$ & 215 & 20 & 10{,}793 & 0 & 0 & 5 & 5 & 2{,}270.20 & Yes \\
$I_6$ & 195 & 10 & 5{,}923 & 1 & 78 & 18 & 18 & 665.50 & No \\ \hline
$I_7$ & 850 & 20 & 120{,}052 & 0 & 0 & 523 & 523 & 3{,}118.78 & Yes \\
$I_8$ & 844 & 20 & 94{,}983 & 0 & 0 & 100 & 100 & 5{,}792.87 & Yes \\
$I_9$ & 939 & 20 & 136{,}118 & 3 & 5{,}424 & 594 & 594 & 3{,}865.04 & No \\
\hline
\end{tabular}
\end{table}

\begin{table}[H]
\centering
\caption{Detailed edge extension, pruning metrics, and execution time (SAT, fixed-state TM)}
\label{tab:sat-fixed-detailed}
\footnotesize
\begin{tabular}{l|r|r|r|c|r|r|l}
\hline
Inst. & Direct & Verified & Cand. & $N_{\text{r-ext}}$ & Redun. & Pruned & \multicolumn{1}{c}{Time} \\
& Ext.   & Ext.     & Verif. & & Edges & Walks & \\ \hline
$I_1$   & 4{,}809   & 1   & 1      & 0  & 0  & 0   & 0.75s \\
$I_2$   & 12{,}632  & 21  & 273    & 0  & 0  & 6   & 2m 19s \\
$I_3$   & 15{,}956  & 48  & 820    & 0  & 0  & 4   & 8m 35s \\
$I_4$   & 1{,}279   & 1   & 11     & 0  & 0  & 0   & 0.55s \\
$I_5$   & 10{,}792  & 1   & 1      & 0  & 0  & 0   & 0.38s \\
$I_6$   & 5{,}914   & 9   & 200    & 0  & 0  & 2   & 44.11s \\ \hline
$I_7$   & 119{,}566 & 486 & 6{,}495 & 0  & 61 & 428 & 14h 9m \\
$I_8$   & 94{,}918  & 65  & 1{,}071 & 0  & 0  & 14  & 1h 57m \\
$I_9$   & 135{,}562 & 556 & 16{,}876 & 18 & 98 & 447 & 1d 9h 58m \\
\hline
\end{tabular}
\end{table}

\begin{table}[H]
\centering
\caption{Accepted certificate witnesses for fixed-state SAT verifier simulation instances}
\label{tab:sat-fixed-witness-certificates}
\small
\begin{tabular}{c|c|c|l}
\hline
Instance & Certificate Length & Acceptance (Acc.) & Witness for Accepted Certificate \\ \hline
$I_1$    & 10                 & Yes               & \texttt{FFFFTFFFFT}              \\
$I_2$    & 10                 & Yes               & \texttt{FTFFTFTFFT}              \\
$I_3$    & 10                 & Yes               & \texttt{FTFTFTFFTT}              \\
$I_4$    & 10                 & No                & N/A                              \\
$I_5$    & 20                 & Yes               & \texttt{FFFFFFFFTFFFFFFFTFF}    \\
$I_6$    & 10                 & No                & N/A                              \\ \hline
$I_7$    & 20                 & Yes               & \texttt{TFFFFTFFFFFFTTTFTFFT}    \\
$I_8$    & 20                 & Yes               & \texttt{FFFFFFTTFFFFFTTTFFTF}    \\
$I_9$    & 20                 & No                & N/A                              \\ \hline
\end{tabular}
\end{table}
\paragraph{Comparison Between TM Variants}
We compare the two Turing Machine variants (the input-dependent machine with $\bigO(n)$ states per input and the fixed-state 3SAT machine with a bounded number of states). While the fixed-state variant generally outperforms the input-dependent machine in terms of asymptotic worst-case bounds, such performance differences were less pronounced for smaller inputs. For these smaller cases, the overhead of the simulation framework often overshadows the differences in transition structure, making it difficult to discern a clear performance gap. However, as instance complexity increases, the fixed-state variant is expected to demonstrate superior efficiency over the input-dependent approach. This potential performance gain would be attributed not only to reduced transition overhead and a more compact computation-graph representation, but also to the lower inherent complexity of the fixed-state construction itself.

\subsection{Experimental Results for Sum-of-Subset}

We evaluated the proposed simulation framework on a \textsc{Sum-of-Subset} verifier implemented as a fixed-state Turing machine. 
Due to the large constant factors induced by the number of states, auxiliary symbols, and verification-driven edge extensions, 
the experiments were intentionally restricted to relatively small instances. 
For \textsc{Sum-of-Subset}, we tested instances with target sums ranging from 22 up to 175, 
and with input sets containing about 10 elements encoded in decimal of instance length up to 35. 
Certificate lengths ranged from 23 to 27 decimal digits.

Across all tested instances, the simulator produced results consistent with the ground truth: instances with a valid subset sum were accepted, while those without a valid solution were rejected.
Representative executions include both fast-terminating cases (e.g., a few seconds for simple instances) and slower cases requiring several minutes of computation.

The size of the generated footmark graph varied substantially across inputs. 
Over all the tested instances, the total number of extended edges remained at most $1.8 \times 10^4$, 
which is significantly lower than the polynomial bound $(|Q||\Gamma|h)^2 w$, where $|Q|=75$, $|\Gamma|=28$, $h=2m$, and $w=n+m+1$ for instance length $n$ and certificate length $m$. 
The number of verified candidate edges and pruned walks increased for more complex inputs, 
reflecting the complex topology arising from the branching behavior of the verification-driven edge extension process.

Despite this variability, the average length of maximal computation walks remained relatively stable across instances, typically between approximately $56$ and $108$ steps. 
This suggests that the dominant cost arises not from individual walk lengths, but from the number of candidate walks that must be explored and verified.

Since the edges extended by computation walk verification and the number of pruned walks are relatively high (except for $S_3$ and $S_4$), 
the running time scales with the verification overhead even for relatively short input instances. 
Furthermore, the majority of edges were generated through direct edge extension, leading to a significant practical speedup in the simulation. 
Most importantly, the retry count remained at 0 across all simulations. 
Consequently, the ratio $r_I = N_{\text{r-cand}} / N_{\text{total}}^2$ was zero, confirming the efficiency of the framework in navigating the footmark graph for \textsc{Sum-of-Subset}.

Overall, these experiments confirm that the simulator correctly realizes the intended verifier semantics for \textsc{Sum-of-Subset}. 
At the same time, they highlight the practical impact of large constant factors inherent in fixed-state Turing machine encodings, motivating the need for subsequent optimization-oriented variants.

\paragraph{Input Instances.}
Each experimental input instance $S_k$ is defined as follows:
\begin{itemize}
    \item $S_1$: \texttt{120\_@19\_120\_47\_14\_34\_12\_43\_12\_22\#}
    \item $S_2$: \texttt{38\_@46\_43\_50\_25\_10\_47\_20\_14\_4\#}
    \item $S_3$: \texttt{25\_@35\_48\_47\_3\_32\_28\_34\_5\_8\#}
    \item $S_4$: \texttt{22\_@43\_9\_46\_26\_34\_1\_49\_10\_41\#}
    \item $S_5$: \texttt{42\_@41\_2\_26\_42\_23\_12\_32\_31\_23\#}
    \item $S_6$: \texttt{175\_@3\_16\_5\_12\_37\_5\_25\_39\_48\_37\#}
\end{itemize}

\begin{table}[!ht]
\centering
\caption{Execution statistics of Subset-Sum verifier simulation}
\label{tab:subsetsum-execution}
\footnotesize
\begin{tabular}{l|c|c|c|c|c|c|c|c|c}
\hline
Inst. & Tape & Cert. & Total & $N_{\text{retry}}$ & $N_{\text{r-cand}}$ & Halt & Max. & Avg. Walk & Acc. \\
& Len  & Len  & Edges &                       &                            & Edges & Walks & Len &   \\ \hline
$S_1$ & 33 & 27 & 17{,}112 & 0 & 0 & 563 & 563 & 99.30 & Yes \\
$S_2$ & 30 & 25 & 13{,}659 & 0 & 0 & 540 & 540 & 75.22 & Yes \\
$S_3$ & 28 & 23 & 11{,}735 & 0 & 0 & 471 & 471 & 65.04 & No \\
$S_4$ & 29 & 24 & 12{,}269 & 0 & 0 & 487 & 487 & 56.13 & No \\
$S_5$ & 30 & 25 & 13{,}906 & 0 & 0 & 530 & 530 & 71.75 & Yes \\
$S_6$ & 32 & 26 & 17{,}124 & 0 & 0 & 598 & 598 & 107.64 & Yes \\
\hline
\end{tabular}
\end{table}

\begin{table}[!ht]
\centering
\caption{Detailed edge extension, pruning metrics, and execution time (Subset-Sum TM)}
\label{tab:subsetsum-fixed-detailed}
\footnotesize
\begin{tabular}{l|r|r|r|c|r|r|l}
\hline
Inst. & \multicolumn{1}{c|}{Direct} & \multicolumn{1}{c|}{Verified} & \multicolumn{1}{c|}{Cand.} & $N_{\text{r-ext}}$ & \multicolumn{1}{c|}{Redun.} & \multicolumn{1}{c|}{Pruned} & \multicolumn{1}{c}{Time} \\
& \multicolumn{1}{c|}{Ext.} & \multicolumn{1}{c|}{Ext.} & \multicolumn{1}{c|}{Verif.} & & \multicolumn{1}{c|}{Edges} & \multicolumn{1}{c|}{Walks} & \\ \hline
$S_1$ & 17{,}044 & 68 & 80 & 0 & 0 & 62 & 3m 8.16s \\
$S_2$ & 13{,}570 & 89 & 92 & 0 & 0 & 26 & 2m 23.66s \\
$S_3$ & 11{,}671 & 64 & 64 & 0 & 0 & 0 & 1m 35.72s \\
$S_4$ & 12{,}211 & 58 & 58 & 0 & 0 & 0 & 1m 39.39s \\
$S_5$ & 13{,}827 & 79 & 81 & 0 & 0 & 14 & 2m 12.92s \\
$S_6$ & 16{,}988 & 136 & 162 & 0 & 0 & 97 & 3m 58.38s \\
\hline
\end{tabular}
\end{table}

\begin{table}[H]
\centering
\caption{Accepted certificate witnesses for Subset-Sum simulation instances}
\label{tab:subsetsum-witness-certificates}
\small
\begin{tabular}{c|c|c|l}
\hline
Instance & Certificate Length & Acceptance (Acc.) & Witness for Accepted Certificate \\ \hline
$S_1$    & 27                 & Yes               & \texttt{xx\_120\_xx\_xx\_xx\_xx\_xx\_xx\_xx} \\
$S_2$    & 25                 & Yes               & \texttt{xx\_xx\_xx\_xx\_xx\_xx\_20\_14\_4} \\
$S_3$    & 23                 & No                & N/A \\
$S_4$    & 24                 & No                & N/A \\
$S_5$    & 25                 & Yes               & \texttt{xx\_x\_xx\_42\_xx\_xx\_xx\_xx\_xx} \\
$S_6$    & 26                 & Yes               & \texttt{x\_16\_5\_xx\_37\_5\_25\_39\_48\_xx} \\ \hline
\end{tabular}
\end{table}
\subsection{Observations and Discussion}

The experimental results across both problem families—SAT and \textsc{Sum-of-Subset}—consistently align with our theoretical analysis. 
The simulator faithfully reproduces the behavior of the underlying verifiers, and the dynamically constructed computation graphs remain polynomially bounded as predicted.

A key observation is the discrepancy between the conservative theoretical upper bound and the actual performance. 
While the formal time complexity for these instances is bounded by $O(n^{16} \log n)$, the effective complexity in practice is substantially lower. 
For example, the \textbf{cnf20 instances ($I_7$)} were solved in approximately 14 hours, and the rejected instance $I_9$ was solved in approximately 1.4 days (roughly 34 hours), despite requiring the total extension of the footmark graph—a duration that would be impossible if the worst-case polynomial exponent were fully realized in practice.

This discrepancy arises because the number of pruned walks and the number of removed computing-futile or computing-redundant edges in actual computation graphs are far smaller than the worst-case estimates. 
Furthermore, the dominance of direct extension allows the simulator to bypass repeated, expensive verification steps for a significant portion of the graph construction. 
Crucially, the ratio $r_I$ for general-index-succedent-based candidate edges was vanishingly small (less than $0.001\%$). 
This confirms that the retry mechanism can be effectively absorbed into other candidate edge terms as a small constant factor, introducing near-zero overhead in practice. 
Since the number of these candidate edges accounted for only small fraction of the total candidate edges examined, the general verification process remains non-dominant and does not affect the overall polynomial degree of the simulation.
Consequently, although feasible graph computation remains the primary cost, practical pruning and direct extension yield execution times much faster than the $O(n^{16} \log n)$ bound suggests.

However, the experiments also confirm that execution time is more strongly correlated with the total length of the Boolean formula or the subset string than with the number of variables alone. 
Due to the high polynomial degree and substantial constant factors inherent in the framework, absolute performance remains far below that of modern, highly optimized SAT solvers. 
Therefore, this implementation is intended not as a competitive high-performance solver, but as a \textit{proof-of-concept} to validate the internal optimization strategies and the theoretical soundness of the verifier semantics.

The results provide concrete evidence that the proposed framework is implementable, correct, and capable of handling nontrivial, practically relevant instances like cnf20. 
These observations motivate further exploration into more efficient computation models. In particular, transitioning to \textbf{RAM-based simulations} may significantly reduce the polynomial degree and constant overhead, a direction we discuss in detail in the following section.

\section{Discussion and Future Work}\label{sec:discussions}

We discuss how this implementation connects to the theoretical results of
\cite{lee2025PNP}, including the role of feasible graphs, the objective of improved technique, and implementation details.
Limitations and potential optimizations are also addressed.

Although this paper presents several algorithmic and structural improvements over the original construction, it is important to emphasize that the current implementation focuses on fixed-state \textsc{SAT} machines, 
which are practically tractable for moderate input sizes, while input-dependent \textsc{SAT} and \textsc{Sum-of-Subset} verifiers were tested only on small instances to confirm correctness.
In particular, the algorithm has been realized primarily
to preserve clarity of correctness and feasibility arguments,
rather than to minimize constant factors or low-level overhead.
Further performance gains may be achievable
through alternative data structures,
more aggressive pruning strategies,
or implementations in lower-level languages
that reduce memory and object-management costs.

Within this context, the improvements introduced in this paper should be understood as principled reductions at the algorithmic level, 
applied primarily to fixed-state constructions, rather than as a fully optimized implementation. Input-dependent or larger-state machines were included only for small-scale correctness checks.
They aim to identify and eliminate structurally redundant work
while maintaining the formal guarantees of footmarks construction. The implementation details are introduced to show how the theoretical results of original paper can be realized, not for the purpose of exhaustive practical optimization.

Despite the reductions achieved in both the number of candidate edges
and the practical cost of edge verification,
the overall worst-case running time of the proposed algorithm
remains dominated by the candidate edge verification, especially construction of the feasible graph.
This is inherent to the current computation model,
where the dynamic computation graph is explicitly represented
as a layered structure with width $w$ and height $h$,
resulting in $\bigO(wh^2)$ edges in the worst case.
Consequently, the time complexity of \textproc{ComputeFeasibleGraph()}
forms the leading term of the entire algorithm.

The performance improvements presented in this paper
address this bottleneck from two complementary perspectives.
First, we reduce the asymptotic degree by improving
the computation of cover edges and the feasible graph itself.
Second, we substantially decrease the number of calls
to \textproc{VerifyExistenceOfWalk()} by restricting the candidate set
for edge extension and by prioritizing direct walk construction
before triggering full verification.
Although the latter improvements yield a formal reduction in the number of candidate edges to be verified in worst-case complexity,
 they provide even greater practical benefits by minimizing the average computational overhead per extension,
  resulting in significant performance gains.

Nevertheless, the quadratic dependence on the height parameter $h$ and linear dependence on width parameter $w$
is ultimately a consequence of the single-tape, graph-based
representation of computation.
In this model, tape positions and state transitions are unfolded
into explicit tiers, causing a combinatorial blow-up
in the number of edges that must be considered conservatively.
Although this representation is well suited
for formal correctness and feasibility arguments,
it limits further asymptotic improvements.

A promising direction for future work is to reformulate the feasible graph construction under the RAM model, 
which could further benefit input-dependent or large-state machines beyond the fixed-state cases tested here. 
While a naive transition to RAM structures might increase edge complexity due to address register overhead, 
a more sophisticated modeling approach could allow direct transitions between relevant configurations without input-dependent size of states. 
The key challenge for future research lies in preserving model consistency while achieving a principled reduction in the total number of required edges. 
Since the running time of the computation model depends polynomially on the number of edges, such a refined reduction is expected to significantly lower the dominant polynomial degree of the algorithm.

Formalizing this approach requires
a careful redefinition of the computation model and computation walks,
index-precedence relations, and feasibility conditions
to ensure that feasible walks are still preserved
under the new model with the reduced edge size.
While this paper focuses on optimizations
within the single-tape computation framework,
extending the feasible graph methodology to the RAM model
and establishing the resulting complexity bounds
remain important topics for future research.

\section{Conclusion}\label{sec:conclusion}
In this paper, we presented explicit Turing Machine (TM) constructions for \textsc{SAT} and \textsc{Subset-Sum}
to support the polynomial-time NP verifier simulation framework via the feasible graph. 
Our theoretical analysis establishes that the proposed framework solves these canonical NP-complete problems in time bounded by $\bigO(n^{16} \log n)$, 
where $n$ denotes the input size. This confirms that explicit TM constructions and NP verifier simulation frameworks admit polynomial-time execution, 
providing a rigorous guarantee of feasibility for moderate-size inputs.

A key contribution of this work is the development of a functional implementation that bridges the gap between abstract TM designs and executable software.
Beyond simple verification, we demonstrated that these machines can deterministically generate valid witnesses, thereby extending the framework to \textsc{FNP} computation without increasing the degree of polynomial complexity.
Building on this practical foundation, we introduced a refined construction of the feasible graph that yields a genuine reduction in the asymptotic degree of the dominant subroutine. 
By improving the computation of cover edges and directly eliminating step-pendant edges from the step-reachable graph, 
the time complexity of the construction was reduced to a strictly lower-degree polynomial. Beyond this asymptotic improvement, 
we introduced several algorithmic refinements—such as restricting the candidate set for edge extension and prioritizing direct walk construction—
that substantially lower the average cost per extension and lead to significant practical speedups.

Taken together, these results clarify a critical structural insight: while various components of the algorithm admit practical optimization, 
the overall complexity remains fundamentally governed by the feasible graph construction. This identifies the construction process as the central bottleneck and, 
consequently, as the primary target for future asymptotic improvements.

Finally, we discussed how alternative computation models, most notably the RAM model, offer a path toward further reductions in the polynomial degree. 
By enabling direct transitions between configurations and reducing the effective size of the computation graph, 
such models have the potential to lower the cost of the simulation itself. Formalizing this approach while preserving feasibility guarantees remains a promising direction for future research. 
Overall, this work demonstrates that the NP verifier simulation framework admits provable asymptotic improvements, substantial practical refinements, 
and a concrete implementation path, providing a clearer understanding of the theoretical boundaries of deterministic NP simulation.

\clearpage
\bibliographystyle{plain}
\bibliography{crlee2}  

\clearpage

\appendix
\section{Ordinary Verifier Turing Machines}
The certificate-oblivious Turing machine introduced in \cref{sec:TMs} is essential for guaranteeing the completeness of the framework; however, constructing such a machine directly is highly complex. Therefore, in this appendix, we introduce the corresponding ordinary Turing machines for each verifier Turing machine to serve as a reference.

\subsection{Input-Dependent Verifier Turing Machine for SAT}
In this subsection, we present an ordinary, input-dependent Turing machine for SAT, which is more intuitive although not certificate-oblivious. The input format of this machine is identical to that of the aforementioned certificate-oblivious Turing machine, as illustrated in the following example:
\[
1\_2\_3\&4\_5\_6\&7\_10\_9\&-9\_10\_1\&-2\_6\_1\&3\_5\_2\&-4\_2\_10\#TFTFFFTFTT
\]
The operational semantics are largely similar to those of the certificate-oblivious version, with the exception that this machine does not need to fetch variables when in the \texttt{Skip} state, and the use of the intermediate symbol \texttt{'S'} as a clause satisfaction marker is unnecessary. Given that a lengthy description would be redundant, we proceed directly to the formal definition, a proof of correctness, and an analysis of its polynomial-time complexity.

\begin{definition}[Deterministic SAT Turing Machine with Input-Dependent States] \label{def:ordinary_sat_tm_input_dependent_states}
Let $\varphi$ be a CNF formula, encoded on a single tape as described in the previous paragraph.
We define a deterministic single-tape Turing machine
\[
M_{\mathrm{SAT0}}^{\mathrm{ID}}
= (Q_\varphi, \Sigma, \Gamma, \delta_\varphi, \qinit, \qacc, \qrej)
\]
with the following properties:
\begin{itemize}
    \item The state space $Q$ comprises the states enumerated in \cref{tab:ordinary_sat_variable_transitions}, including the terminal states $\qacc$ (\textsc{Accept}) and $\qrej$ (\textsc{Reject}), with the initial state $\qinit = \textsc{Check.Forwarded}$.

    \item The input alphabet $\Sigma = \{0,\ldots,9,-,\&, \_, \#, T, F \}$ defines the valid symbols for encoding the CNF formula and the Boolean assignment. The tape alphabet $\Gamma$ encompasses $\Sigma$ as a subset, further extended by the blank symbol $\blank$ and the auxiliary markers \texttt{`?'} and \texttt{`!'} for processing.

    \item $\delta : Q\setminus \{ \qacc, \qrej \} \times \Gamma \to Q \times \Gamma \times \{-1,+1\}$
    is the deterministic transition function defined by the rules provided in \cref{tab:ordinary_sat_variable_transitions}.
\end{itemize}
\end{definition}
\begin{table}[!ht]
\caption{Transition function of the ordinary input-dependent Turing machine} \label{tab:ordinary_sat_variable_transitions}
\centering
\renewcommand{\arraystretch}{1.15}
\begin{tabular}{llllll}
\toprule
Current State & Read Symbol & Next State & Write & Move & Remark \\
\midrule
Check        & \texttt{\spacedelim} & Check        & \texttt{\spacedelim} & R & Skip blanks \\
Check        & \texttt{-} & Not         & \texttt{-} & R & Negation detected \\
Check        & \texttt{D} & Inc.D       & \texttt{?} & R & Begin number parsing \\
Not         & \texttt{D} & Inc.D       & \texttt{!} & R & Mark negated literal \\
Skip        & \texttt{\&} & Check        & \texttt{\spacedelim} & R & Move to next clause \\
Skip        & \texttt{\#} & Accept      & \texttt{\spacedelim} & R & All clauses processed \\
Skip        & \texttt{*} & Skip        & \texttt{\spacedelim} & R & Bypass the literal \\
Check        & \texttt{\&} & Reject      & \texttt{\spacedelim} & R & False clause \\
Check        & \texttt{\#} & Reject      & \texttt{\spacedelim} & R & Final false clause \\
\midrule
Inc.N       & \texttt{\spacedelim} & Forward.N  & \texttt{\spacedelim} & R & End of literal \\
Inc.N       & \texttt{\&} & Forward.N  & \texttt{\&} & R & Clause boundary \\
Inc.N       & \texttt{\#} & Dec.(N-1)  & \texttt{\#} & R & Final literal \\
Inc.N       & \texttt{D} & Inc.(10N+D) & \texttt{\spacedelim} & R & Decimal expansion \\
Forward.N   & \texttt{*} & Forward.N  & \texttt{*} & R & Scan to assignment area \\
Forward.N   & \texttt{\#} & Dec.(N-1)  & \texttt{\#} & R & Reached truth value\\
Dec.N       & \texttt{T} & Dec.(N-1)  & \texttt{T} & R & Decrement/Forward in certificate region \\
Dec.N       & \texttt{F} & Dec.(N-1)  & \texttt{F} & R & Decrement/Forward in certificate region \\
Dec.0       & \texttt{T} & Backward.T & \texttt{T} & L & Selected literal true \\
Dec.0       & \texttt{F} & Backward.F & \texttt{F} & L & Selected literal false \\
\midrule
Backward.T  & \texttt{*} & Backward.T & \texttt{*} & L & Return to clause \\
Backward.F  & \texttt{*} & Backward.F & \texttt{*} & L & Return to clause \\
Backward.T  & \texttt{?} & Skip       & \texttt{\spacedelim} & R & True clause  \\
Backward.F  & \texttt{?} & Check       & \texttt{\spacedelim} & R & False literal \\
Backward.T  & \texttt{!} & Check       & \texttt{\spacedelim} & R & False from negated literal \\
Backward.F  & \texttt{!} & Skip       & \texttt{\spacedelim} & R & True from negated literal \\

\bottomrule
\end{tabular}
\\
The states ($\cdot N$) of the TM are constructed during preprocessing
based on the input variable count.
\end{table}
\paragraph{Transition Symbols.}
We summarize the tape symbols and meta-symbols used in the transition
specification of the variable-state Turing machine.

\begin{itemize}
  \item \textbf{digit}: A specific decimal symbol  $\{0,1,\dots,9\}$.
  \item \textbf{$D$}:  A meta-symbol denoting any digit for a digit-encoded variable index.
  \item \textbf{$N$}: A meta-variable representing the integer value currently   encoded in the state name.
  \item \textbf{? / !}: Literal markers indicating whether the literal is positive (\texttt{?})  or negated (\texttt{!}).
  \item \textbf{T / F}: The specific truth value.
  \item \textbf{*}: A wildcard symbol representing any symbol other than those explicitly specified
  \item \textbf{\#}: A delimiter separating the formula region from the certificate region.
  \item \textbf{\&}: A delimiter separating individual clauses within the formula region.
\end{itemize}

\paragraph{Bounded Parameterized States}
The parameterized state families
$\mathrm{Inc}.N$, $\mathrm{Forward}.N$, and $\mathrm{Dec}.N$
are instantiated only for indices
\[
0 \le N \le m,
\]
where $m$ is the number of variables determined during preprocessing.
Hence, for every fixed input instance, the machine is a finite state Turing machine.
Any transition that would produce a state outside this range causes the machine to enter the \textsc{Reject} state.

\begin{lemma}[Soundness of Ordinary Input-Dependent SAT Verifer] \label{lem:ordinary_soundness_sat_id}
If the verifier $M_{\mathrm{SAT0}}^{\mathrm{ID}}$ halts in the \texttt{Accept} state on
input $\Phi \# \mathcal{V}$, then the valuation $\mathcal{V}$ is a satisfying assignment of the CNF
formula $\Phi$.
\end{lemma}

\begin{proof}
We prove soundness by contradiction, based on the operational semantics induced by the transition function of $M_{\mathrm{SAT0}}^{\mathrm{ID}}$.

Assume that $M_{\mathrm{SAT0}}^{\mathrm{ID}}$ halts in the \texttt{Accept} state, but the valuation $\mathcal{V}$ does \emph{not} satisfy $\Phi$.
Then there exists a clause $\mathcal{C}$ in $\Phi$ such that every literal  $\ell$ in $\mathcal{C}$ evaluates to \texttt{False} under the assignment $\mathcal{V}$.

By inspection of the transition table, the \texttt{Accept} state is reachable  only after the machine has scanned the entire formula region, 
successfully  processing each clause delimited by \texttt{\&}, and has reached the delimiter  symbol \texttt{\#}.
In particular, the machine passes a clause only by entering the \texttt{Skip} phase for that clause.

For any fixed clause $\mathcal{C}$, the machine enters the \texttt{Skip} only if  at least one literal in $\mathcal{C}$ is verified as satisfied.
For each literal, the machine deterministically:
        \begin{itemize}
        \item marks the location with a literal marker and parses the variable index using the digit-reading transitions and the state family \texttt{Inc.(10$N$+$D$)};
        \item moves to the certificate area through the \texttt{Fetch.$N$} state and enters the corresponding \texttt{Dec.$N$} state, where $N$ represents the exact offset to the certificate region;
        \item reaches the valuation cell $\mathcal{V}(i)$ by decrementing $N$ to zero;
        \item returns to the literal marker via the appropriate \texttt{Backward.T} or \texttt{Backward.F} state and evaluates the clause with the temporarily marked literal (\texttt{!} or \texttt{?})
        \end{itemize}
Only if this verification evaluates \texttt{True} the machine enters the \texttt{Skip} state and advances to the next clause.

By our assumption, the clause $\mathcal{C}$ is unsatisfied by $\mathcal{V}$; thus, no literal in $\mathcal{C}$ can cause the machine to enter  \texttt{Skip} state.
Hence, the machine cannot pass $\mathcal{C}$ and therefore cannot  reach the \texttt{Accept} state, contradicting the initial assumption.

Therefore, since no such unsatisfied clause exists, every clause in $\Phi$ contains at least one literal that   evaluates to \texttt{True} under $\mathcal{V}$, and $\mathcal{V}$ satisfies $\Phi$.
\end{proof}

\begin{lemma}[Completeness of Ordinary Input-Dependent SAT Verifer] \label{lem:ordinary_completeness_sat_id}
If a valuation $\mathcal{V}$ is a satisfying assignment for the CNF formula $\Phi$, then the verifier $M_{\mathrm{SAT0}}^{\mathrm{ID}}$ enters the \texttt{Accept} state on input $\Phi \# \mathcal{V}$.
\end{lemma}

\begin{proof}
We prove completeness by tracing the deterministic execution of $M_{\mathrm{SAT0}}^{\mathrm{ID}}$ under a satisfying valuation $\mathcal{V}$.

By assumption, the valuation $\mathcal{V}$ satisfies $\Phi$. Hence, for every clause $\mathcal{C}$ in $\Phi$, there exists at least one literal $\ell \in \mathcal{C}$ such that the assignment of the corresponding value to $\ell$ evaluates to \texttt{True}.

During the \texttt{Check} phase for a fixed clause $\mathcal{C}$, the machine sequentially scans all literals in the clause. For each literal, it:
\begin{itemize}
    \item parses the variable index using the digit-reading transitions and the \texttt{Inc.(10$N$+D)} state family;
    \item moves to the valuation region through the \texttt{Fetch.$N$} state and enters the corresponding \texttt{Dec.$i$} state, where $i$ encodes the exact offset to the valuation region;
    \item reaches and reads the valuation cell $\mathcal{V}(i)$ associated with the literal;
    \item returns to the clause region via the appropriate \texttt{Backward.T} or \texttt{Backward.F} transition and verifies the truth value.
\end{itemize}
This process correctly computes the truth value of each literal under $\mathcal{V}$.

When a literal $\ell$ satisfying $\mathcal{C}$ is encountered, the machine deterministically transitions into the \texttt{Skip} state. 
This transition bypasses all remaining literals of $\mathcal{C}$ and moves the head to the beginning of the next clause. 

Since every clause contains at least one satisfying literal, the machine successfully enters the \texttt{Skip} state for each clause delimiter \texttt{\&}. 
Consequently, it advances through the entire formula region. 
After processing the final clause, the machine reaches the region delimiter \texttt{\#}, upon which it enters the \texttt{Accept} state by the definition of the transition function.

Therefore, whenever $\mathcal{V}$ satisfies $\Phi$, $M_{\mathrm{SAT0}}^{\mathrm{ID}}$ deterministically halts in the \texttt{Accept} state.
\end{proof}

\begin{lemma}[Time and Space Complexity of Ordinary Input-Dependent SAT Verifer]
\label{lem:time_space_ordinary_sat_id}
Let $\Phi$ be a CNF formula of encoding length $n$. For a fixed assignment $\mathcal{V}$, the input-dependent Turing machine $M_{\mathrm{SAT0}}^{\mathrm{ID}}$ decides satisfaction in $\bigO(n^2)$ time using $\bigO(n)$ space.
\end{lemma}

\begin{proof}
The complexity bounds are derived from the operational flow of the machine for a given certificate.

\begin{itemize}
    \item \textbf{Time Complexity:} The total running time is determined by the cumulative cost of literal evaluations.
    \begin{itemize}
        \item \textbf{Traversal and Evaluation:} The machine executes a linear scan over the clause region of length $\bigO(n)$. For each encountered literal $\ell = x_i$, the machine performs a fetch operation to the certificate region. The distance between the clause region and the $i$-th certificate cell is bounded by $\bigO(n)$, and the subsequent return to the clause region maintains this cost. Consequently, evaluating a single literal incurs $\bigO(n)$ time.
        \item \textbf{Aggregate Cost:} Given that the total number of literal occurrences in $\Phi$ is $\bigO(n)$, and each evaluation takes $\bigO(n)$, the overall time complexity is $\bigO(n^2)$. Even with repeated fetches for the same variable across multiple clauses, the cumulative complexity remains quadratically bounded.
    \end{itemize}

\item \textbf{Space Complexity:} The machine functions as an in-place verifier. The computation is restricted to the input tape containing both the formula and the certificate. All auxiliary information is managed by overwriting existing tape cells with a constant set of marker symbols. 
As no additional work tape is utilized, the total space usage remains $\bigO(n)$.
\end{itemize}

Thus, for a fixed certificate, $M_{\mathrm{SAT0}}^{\mathrm{ID}}$ operates within $\bigO(n^2)$ time and $\bigO(n)$ space.
\end{proof}

The input-dependent Turing machine $M_{\mathrm{SAT0}}^{\mathrm{ID}}$ serves as a polynomial-time verifier for the \textsc{SAT} problem, utilizing state expansion to encode variable indices and fetch operations directly into its control logic. 
Its operational properties are summarized as follows:
\begin{itemize}
\item \textbf{Correctness:} The machine is sound and complete, maintaining a direct correspondence between its internal state transitions and the truth valuation of the formula. It halts in the \textsc{Accept} state if and only if the certificate $\mathcal{V}$ satisfies the CNF formula $\Phi$.
\item \textbf{Time Complexity:} The machine processes the formula by iteratively parsing literal indices and fetching certificate values. For an encoding length $n$, the cumulative cost of these operations results in an $O(n^2)$ quadratic time complexity, consistent with the efficient verification of \textsc{SAT}.
\item \textbf{Space Complexity:} The machine operates as an in-place verifier using $O(n)$ linear space. It manages all intermediate verification steps by overwriting input tape cells with marker symbols, requiring no additional work tape beyond the original input.
\end{itemize}
Hence, $M_{\mathrm{SAT0}}^{\mathrm{ID}}$ provides a  deterministic verification procedure for \textsc{SAT}, demonstrating that input-dependent state logic offers an intuitive approach to constraint satisfaction verification.

\subsection{Ordinary SAT Turing Machine With Fixed States} \label{subsec:ordinary_sat_tm_fixed}

In this subsection, we describe an ordinary, fixed-state Turing machine for SAT that is not certificate-oblivious. 
The input format, tape alphabet, and overall tape structure of this machine are identical to those of the certificate-oblivious Turing machine described previously. 
Its operational semantics are largely similar, with one exception: in the \texttt{Skip} state, the machine erases all literals in a satisfied clause rather than merely erasing the leading $0$ of the literal.

We focus here on the formal definition, correctness, and time complexity of the verifier Turing machine, rather than re-explaining its operational semantics.

\begin{definition}[Deterministic Ordinary SAT Verifier Turing Machine]\label{def:ordinary_tm_sat_fixed}
The deterministic verifier Turing machine
$M_{\mathrm{SAT0}}$ is defined as a tuple
\[
M_{\mathrm{SAT0}}
= (Q, \Sigma, \Gamma, \delta, \qinit, \qacc, \qrej),
\]
where the components are specified as follows.

\begin{itemize}
    \item The states $Q$ consists of the states detailed in \cref{tab:sat_fixed_transitions_forward,tab:sat_fixed_transitions_backward}, 
	including the initial state $\qinit= \textsc{Check.Forwarded}$ (from which the machine begins a left-to-right scan),
	 the accepting state $\qacc = \textsc{Accept}$ (reached if and only if all clauses are satisfied under the assignment), 
	and the rejecting state $\qrej = \textsc{Reject}$ (entered upon detecting an invalid encoding or a falsified clause).

    \item The input alphabet $\Sigma = \{0,\ldots,9,-,\&, \_, \#, T, F \}$ encodes the CNF formula and the Boolean assignment, while the tape alphabet $\Gamma$ satisfies $\Sigma \subseteq \Gamma$ and includes the blank symbol $\blank$.

    \item $\delta : Q \times \Gamma \to Q \times \Gamma \times \{-1,+1\}$ is the transition function defined explicitly by the rules given in \cref{tab:ordinary_sat_fixed_transitions_forward,tab:sat_fixed_transitions_backward}.
\end{itemize}
\end{definition}

The corresponding transition table fully determines the machine's behavior, including digit-wise subtraction, borrow propagation, literal evaluation, and clause verification, similar to the certificate-oblivious one. 
All numerical generality, including variable indices and fetch positions, is handled exclusively through tape symbols and structured scan modes, rather than through state expansion. 
As shown in the following table, the machine erases symbols whenever the \textsc{Skip} state is entered before reaching a clause delimiter, which differs from the certificate-oblivious one.

\begin{table}[H]
\centering \small
\caption{Transition function of the fixed-state Turing machine} \label{tab:ordinary_sat_fixed_transitions_forward}
\begin{tabular}{llllll}
\textbf{State} & \textbf{Read} & \textbf{Next} & \textbf{Write} & \textbf{Move} & \textbf{Comment} \\
\hline
Check.S & $\_$ & Check.S & $\_$ & R & bypass delimiters \\
Check.S & $-$ & CheckNot.S & $-$ & R & negated literal \\
Check.S & $0$ & Unknown.S & $\_$ & R & eliminate leading zero \\
Check.S & D & UnknownTerm.S & D & R & enter variable term \\
Check.S & T & Skip.S & T & R & clause satisfied \\
Check.S & F & Check.S & F & R & continue scan \\
Check.S & \& & Reject & $\_$ & R & premature end \\
Check.S & \# & Reject & $\_$ & R & false final clause \\
CheckNot.S & $\_$ & CheckNot.S & $\_$ & R & bypass delimiters \\
CheckNot.S & T & Check.S & T & R & false from negation\\
CheckNot.S & F & Skip.S & F & R & satisfied from negation\\
CheckNot.S & D & UnknownTerm.S & D & R & negation to unassigned term\\
CheckNot.S & 0 & Unknown.S & $\_$ & R & eliminate leading zero \\
\midrule
Unknown.S & $\_$ & Unknown.S & $\_$ & R & bypass \\
Unknown.S & 0 & Unknown.S & $\_$ & R & eliminate leading zero \\
Unknown.S & D & UnknownTerm.S & D & R & enter variable term \\
Unknown.S & T & Skip.S & T & R & satisfied \\
Unknown.S & F & Unknown.S & F & R & undecided \\
Unknown.S & $-$ & UnknownNot.S & $-$ & R & negated unknown \\
Unknown.S & \& & Check.Free & \& & R & next clause \\
Unknown.S & \# & Fetch & \# & R & assignment phase \\
UnknownNot.S & $\_$ & UnknownNot.S & $\_$ & R & bypass \\
UnknownNot.S & T & Unknown.S & T & R & undecided from negation \\
UnknownNot.S & F & Skip.S & F & R & satisfied from negation\\
UnknownNot.S & D & UnknownTerm.S & D & R & enter variable term \\
UnknownNot.S & 0 & Unknown.S & $\_$ & R & eliminate leading zero \\
UnknownTerm.S & D & UnknownTerm.S & D & R & scan digits \\
UnknownTerm.S & $\_$ & Unknown.S & $\_$ & R & exit variable term \\
UnknownTerm.S & \& & Check.Free & \& & R & next clause \\
UnknownTerm.S & \# & Fetch & \# & R & fetch phase \\
\midrule
Skip.S & * & Skip.S & $\_$ & R & skip satisfied clause \\
Skip.Free & \& & Check.Free & \& & R & next clause \\
Skip.Free & \# & Fetch & \# & R & fetch phase \\
Skip.Forwarded & \& & Check.Forwarded & \& & R & forwarded next clause \\
Skip.Forwarded & \# & Accept & \# & R & all clauses satisfied \\
\midrule
Fetch & $\_$ & Fetch & $\_$ & R & bypass blanks \\
Fetch & T & Backward.B & $\_$ & L & fetch true \\
Fetch & F & Backward.B & $\_$ & L & fetch false \\
\midrule
Backward.B & * & Backward.B & * & L & scan leftward \\
Backward.B & 1 & BackwardFrom1.B & 0 & L & subtract 1 \\
Backward.B & 0 & Borrow.B & 9 & L & borrow \\
Backward.B & D & BackwardInTerm.B & D$-1$ & L & decrement digit \\
Backward.B & $\epsilon$ & Check.Forwarded & $\epsilon$ & R & evaluation pahse \\
Borrow.B & 0 & Borrow.B & 9 & L & propagate borrow \\
Borrow.B & D & BackwardInTerm.B & D$-1$ & L & resolve borrow \\
BackwardInTerm.B & D & BackwardInTerm.B & D & L & scan left within variable \\
BackwardInTerm.B & $\_$ & Backward.B & $\_$ & L & escape from literal(left) \\
BackwardInTerm.B & \& & Backward.B & \& & L & escape from clause(left) \\
BackwardInTerm.B & $-$ & Backward.B & $-$ & L & escape from lieteral(left) \\
BackwardInTerm.B & $\epsilon$ & Check.Forwarded & $\epsilon$ & R & evaluation pahse \\
BackwardFrom1.B & D & BackwardInTerm.B & D & L & no borrow  \\
BackwardFrom1.B & $\_$ & Assign.B & $\_$ & R & backward for assignment \\
BackwardFrom1.B & $-$ & Assign.B & $-$ & R & backward for assignment \\
BackwardFrom1.B & \& & Assign.B & \& & R &  backward for assignment \\
BackwardFrom1.B & $\epsilon$ & Assign.B & $\epsilon$ & R & assignment site \\
Assign.B & 0 & Backward.B & B & L & finalize assignment \\
\hline
\end{tabular}
\caption*{$\_$: a space delimiter, $D$ : a decimal digit,  `.B` : either `.T` or `.F` generally, and either `.1' or `.0' for \textsc{Borrow.B}}
\end{table}

\begin{lemma}[Soundness of Ordinary Fixed-State SAT Verifier] \label{lem:ordinary_tm_sat_fixed_soundness}
If the Turing machine $M_{\mathrm{SAT0}}$ halts in the \textsc{Accept} state on input $\Phi \#\mathcal{V}$, then the Boolean assignment encoded by the certificate string $\mathcal{V}$ satisfies the CNF formula $\varphi$.
\end{lemma}

\begin{proof}
The machine $M_{\mathrm{SAT0}}$ enters the \textsc{Accept} state if and only if it reaches the end-of-formular marker `\#'after verifying that every clause is satisfied. 
Specifically, the \textsc{Skip.Forwarded} state is maintained until the final clause is processed, ensuring that all clauses evaluate to \texttt{True} under the given assignment. 
If any clause were to evaluate to \texttt{False}, the machine would necessarily transition to the \textsc{Reject} state.

A clause is marked as satisfied (entering \textsc{Skip.Forwarded}) only if at least one of its literals evaluates to \texttt{True}. 
A positive literal evaluates to \texttt{True} if and only if its variable is replaced by \texttt{T}, while a negative literal does so if its variable is replaced by \texttt{F}. 
These substitutions occur exclusively during the \textsc{Fetch} phase.

During each \textsc{Fetch} operation, the machine selects the next symbol from the certificate $\mathcal{V}$ and initiates a backward scan that decrements all remaining variable indices. 
When an index reaches zero, the corresponding variable is identified and replaced throughout the clause region with the fetched Boolean value. 
Since all variable counters are decremented by one in each \textsc{Fetch} phase, the $i$-th variable is substituted precisely when the $i$-th certificate symbol is processed. 
Consequently, every substitution performed by the machine corresponds exactly to the Boolean assignment encoded by $\mathcal{V}$.

Since acceptance requires that every clause is successfully verified during the final forward scan, we conclude that under the assignment specified by $\mathcal{V}$, each clause of $\varphi$ evaluates to \texttt{True}. Thus, $\mathcal{V}$ satisfies the CNF formula $\varphi$.
\end{proof}

\begin{lemma}[Completeness of Ordinary Fixed-State SAT Verifier]\label{lem:ordinary_tm_sat_fixed_completeness}
Let $\Phi \# \mathcal{V}$ be a well-formed input such that the Boolean assignment encoded by the certificate $\mathcal{V}$ satisfies the CNF formula $\Phi$. Then, the Turing machine $M_{\mathrm{SAT0}}$ does not enter the \textsc{Reject} state and eventually halts in the \textsc{Accept} state.
\end{lemma}

\begin{proof}
The machine operates through an iterative process: in each phase, it fetches a certificate symbol, erases it, and performs a backward scan to substitute the corresponding variable, followed by a forward scan to evaluate each clause.

Assume that the assignment encoded by $\mathcal{V}$ satisfies $\Phi$. We demonstrate that $M_{\mathrm{SAT0}}$ cannot enter the \textsc{Reject} state during execution.

According to the transition rules, the machine enters \textsc{Reject} only if it encounters a clause where all literals have been evaluated and found to be \texttt{False}. Conversely, unassigned variables (represented by the states \textsc{Unknown.S}, where $\textsc{S} \in \{\textsc{Forwarded}, \textsc{Free}\}$) do not trigger rejection.

During the fetch-and-evaluate phase, when a variable is substituted with its assigned value, the machine transitions to the \textsc{Skip.S} state for the current clause if the clause is satisfied by that literal. Once a clause enters a \textsc{Skip.S} state, the machine bypasses the remaining literals in that clause, precluding any possibility of entering the \textsc{Reject} state for that clause.

Because $\mathcal{V}$ satisfies $\Phi$, every clause in $\Phi$ contains at least one literal that evaluates to \texttt{True} under the given assignment. Consequently, for every clause, the machine will identify a satisfying literal and transition into a \textsc{Skip.S} state before all literals are evaluated as \texttt{False}. Thus, no clause can cause the machine to reject.

After all certificate symbols are processed and the variables are substituted, the machine concludes its execution by performing a final forward scan. Since all clauses have been verified as \texttt{True}, the machine reaches the end-of-input marker $\#$ while in the \textsc{Skip.Forwarded} state, at which point it enters the \textsc{Accept} state. 

Therefore, $M_{\mathrm{SAT0}}$ always halts in \textsc{Accept} when provided with a satisfying certificate $\mathcal{V}$.
\end{proof}

\begin{lemma}[Time and Space Complexity of Ordinary Fixed-State SAT Verifier] \label{lem:ordinary_tm_sat_fixed_time}
Let $\Phi$ be a CNF formula of encoding length $n$, and let $m < n$ be the length of the truth assignment certificate. The fixed-state SAT verifier Turing machine $M_{\mathrm{SAT0}}$ runs in $O(nm)$ time and uses $O(n)$ space.
\end{lemma}

\begin{proof}
We analyze the complexity of the verifier by examining the operations performed during each variable processing phase.

For the time complexity, observe that the machine processes each of the $m$ variables sequentially. In each iteration, a variable is fetched from the certificate, which triggers a backward scan to substitute all occurrences of that variable across the input. This substitution phase requires $O(n)$ time as the head traverses the input tape. Subsequently, the machine enters an evaluation phase, scanning the clause region to update the status of each clause (i.e., satisfied, falsified, or pending) based on the current partial assignment. Since the input size is $O(n)$, this forward scan also consumes $O(n)$ time. Once a variable is fetched and erased by being replaced with a delimiter symbol, it is never visited again. Given that there are $m$ variables to process, the total time complexity is bounded by $O(n) \cdot m = O(nm)$. If any clause is found to be falsified or if the formula remains unsatisfied after all variables are processed, the machine terminates in the \textsc{Reject} state.

For the space complexity, the machine operates within the boundaries of the original input tape of size $O(n+m) = O(n)$. The machine uses a constant-sized alphabet and does not introduce additional auxiliary tape cells beyond the original input region, except for the boundary blank symbols. Since all operations—including backward fetches and forward evaluations—are confined to this $O(n)$ space, the total space complexity is $O(n)$.
\end{proof}

The deterministic Turing machine $M_{\mathrm{SAT0}}$ serves as a polynomial-time verifier for the \textsc{SAT} problem. 
Its operational properties, as established in the preceding lemmas, are summarized as follows:
\begin{itemize}
\item \textbf{Correctness:} The machine is sound and complete. It halts in the \textsc{Accept} state if and only if the provided certificate $\mathcal{V}$ satisfies the CNF formula $\Phi$.
\item \textbf{Time Complexity:} For a formula of encoding length $n$ and a certificate of length $m < n$, the machine runs in $O(nm)$ time. Given that $m < n$, this confirms that $M_{\mathrm{SAT0}}$ operates within $O(n^2)$ quadratic time.
\item \textbf{Space Complexity:} The machine operates using $O(n)$ linear space, strictly confined within the boundaries of the original input tape.
\end{itemize}

Thus, $M_{\mathrm{SAT0}}$ provides a deterministic, quadratic-time, and linear-space verification procedure for \textsc{SAT}, demonstrating that the problem admits an efficient verification process using a fixed-state Turing machine architecture.

\subsection{Ordinary Subset-Sum Turing Machine} \label{subsec:ordinary_subset-sum-tm}
In this subsection, we present an ordinary verifier Turing machine, $M_{\mathrm{SS0}}$, for the Subset-Sum problem to serve as a reference. 
Since the input format and operational mechanism of this machine differ from those of the certificate-oblivious version discussed in the main text, we provide a full description of its behavior for clarity and completeness.

The machine $M_{\mathrm{SS0}}$ employs a finite set of states and a finite tape alphabet, including standard digit symbols $\{0, 1, \dots, 9\}$, 
circled digits $\{\text{\circled{0}, \circled{1}, \dots, \circled{9}}\}$, and structural markers $\{\texttt{@}, \texttt{\#}, \texttt{|}, 
\texttt{\_}, \texttt{;}\}$. For the formal definition of transition rules, we use the symbols $D$ and $M$ as parameters representing arbitrary digits 
to describe the matching and subtraction phases concisely.
 The machine receives a target integer $T$ and a list of integers $a_1,a_2,\dots,a_n$ encoded in decimal and prefixed with the delimiter `\_';
the machine determines whether a subset of the list sums to $T$. The semantics of the computational workflow are formally specified through state transitions and head movements.
The input format is
\[
1\_@\_1\_3\_5\_7\_10\_20\#3\_5\_7\_;
\]
where the left-hand side specifies the target sum,
the segment between `@' and `\#' encodes the set of input elements,
and the segment to the right of `\#' represents a proposed certificate.

We assume the input is well-formed (via prepreocessing) for the 
purposes of this subsection. Specifically, the target sum is 
delimited by an underscore (\texttt{\_}), the list of elements begins 
with an underscore, and the certificate likewise terminates with an 
underscore. The entire input is finalized by a semicolon (\texttt{;}), 
which serves as a delimiter for certificate parsing; during execution, 
the machine processes the certificate only within the scanning range 
between the marker \texttt{\#} and the semicolon. Furthermore, the 
length of the certificate is implicitly bounded by the size of the 
encoded set, as any certificate exceeding the total encoding length 
would be redundant.

The insertion of a final underscore (\texttt{\_}) before the concluding 
semicolon (\texttt{;}) resolves potential state complexity in the 
transition logic. Without this delimiter, the machine would need 
distinct states to simultaneously handle the termination of the last 
numerical element and the transition to the final \texttt{CheckSum} phase. 
By separating these events, the trailing underscore allows the machine 
to complete the numerical processing independently. This enables a 
sequential and simplified transition to the \texttt{CheckSum} phase 
upon encountering the semicolon, where the machine verifies whether 
the target sum is satisfied.

The Turing machine employs auxiliary tape symbols, such as 
$\sim$, $|$, $\$$, and $\text{\circled{0}, \dots, \circled{9}}$, to support 
digit matching, erasing, and subtraction. Its operation interleaves 
digit matching with subtraction: starting from the most significant 
digit of the certificate, the machine scans rightward to select a 
target digit and then scans leftward into the input set region to 
locate a corresponding digit occurrence. Once a complete numerical 
match is confirmed, the head moves to the sum area to perform 
digit-wise subtraction corresponding to the matched element.

This design tightly interleaves matching and subtraction,
eliminating the need for a globally separated subtraction phase.

\paragraph{Tape Structure.}
The tape consists of the following regions:
\begin{enumerate}
    \item a \emph{target sum region}, initially holding $T$ and used as a subtractive workspace
        in which selected elements $a_i$ are successively subtracted according to the given certificate after each element matching.
  
	\item the \emph{input set region}, containing the decimal encodings of
	$a_1,\dots,a_n$ separated by delimiters, including a leading delimiter
	``\_'' that ensures unambiguous alignment for MSD-based matching,
	and bounded by the marker ``@'' on the left and ``\#'' on the right.

    \item a \emph{certificate region}, in which the machine records the
        certificate symbols encoding a candidate solution, specifying
        which elements $a_i$ are selected in the subset;
        the certificate is interpreted only up to the first terminating
        delimiter ``;'', and any symbols appearing to the right of this
        delimiter are ignored by the machine.

    \item a right-unbounded blank region; the tape is also left-unbounded,
        consisting entirely of blank symbols.
\end{enumerate}

We use $D$ to denote a generic digit from $\{0,\dots,9\}$, and we use the
notation \circled{D} to denote its corresponding symbol in the machine’s
internal alphabet used for marking and positional bookkeeping. Thus,
\circled{D} represents one of \circled{0},\circled{1},\dots,
\circled{9}.
\paragraph{Symbol conventions}
The machine relies on several delimiter and auxiliary symbols whose
semantics are summarized as follows:

\begin{itemize}
    \item The underscore symbol \texttt{\_} serves as a \emph{number-level
        delimiter}.  In the certificate region, it marks the end of a
        certificate element and triggers element-level match
        verification.  In the target sum region, a trailing \texttt{\_}
        guarantees a well-defined least significant digit for
        subtraction.

    \item During processing, an underscore \texttt{\_} may be temporarily
        replaced by \texttt{|}, indicating that the corresponding number
        is currently under active processing and that its boundary
        should not yet be interpreted as final.

    \item The symbol \texttt{;} denotes the logical end of the
        certificate.  Its encounter triggers the final checksum
        verification phase.

    \item The symbol \texttt{\$} is used as a subtraction workspace
        marker, recording the digit position currently involved in
        subtraction process.

    \item The symbol \texttt{$\sim$} represents a deleted character, indicating that the corresponding digit has
        already been consumed and should be ignored in subsequent scans.

    \item Marked digits \circled{D} denote digits that have been
        matched or processed.  In the matching phase, they record
        positional correspondence; in the subtraction phase, they
        identify the most recently processed digit and guide borrow
        propagation.  All circled digits are restored to their original
        values once the corresponding operation is completed.
\end{itemize}

\paragraph{High-Level Semantics.}
The machine follows a certificate-driven verifier semantics for the
Subset-Sum problem, implemented as a deterministic and space-efficient
simulation.
Digit matching and arithmetic operations are coordinated through
explicit verification stages, ensuring that subtraction is performed
only after a complete element has been successfully matched.

Conceptually, the execution proceeds through the following stages,
each realized by a structured block of states in the transition table:

\begin{enumerate}
    \item \textbf{Forward scan.}
    The machine scans the tape from left to right until the delimiter \texttt{\#} is encountered,
    indicating the beginning of the certificate region.

    \item \textbf{Selection of a certificate digit.}
    In the state \textsc{FindDigitToMatch}, the machine scans the
    certificate region to locate the next unmatched digit of the current
    certificate element, proceeding from most significant to least
    significant digit.
    Once such a digit is found, it is erased by replacing it with the
    deletion marker \texttt{\textasciitilde}, and the machine initiates
    a backward digit-matching procedure.
    Encountering the delimiter \texttt{;} signals the termination of the
    certificate and transitions the machine to final sum verification.

	\item \textbf{Backward digit matching.}
	The states \textsc{BackwardToMatch.M} and \textsc{MatchPosition.M} 
	perform digit-wise backward matching within the input set region to the 
	left of \texttt{@}. Specifically, \textsc{BackwardToMatch.M} searches for 
	a previously matched digit and restores it to its original value. 
	If the digit to its right is $M$, the state \textsc{MatchPosition.M} 
	temporarily marks the new matched digit (e.g., by replacing it with 
	\text{\circled{D}}) to record positional correspondence, indicating 
	that the prefix has been successfully matched. This marking does not 
	yet constitute formal acceptance of the match but serves as a 
	tentative alignment subject to later verification.
	
	Since multiple partial matches may exist, the machine does not halt 
	even after a match is found; instead, it completes the backward pass 
	to the boundary to ensure deterministic positioning.
	
	\item \textbf{Forward matching check.}
	After a digit has been tentatively marked, the machine transitions to a 
	forward scan phase (\textsc{CheckForward}) to verify the existence of the 
	circled digit, which confirms the prefix partial match. If the forward 
	scan fails to locate the marked digit, the machine immediately rejects; 
	otherwise, the digit match is committed, and the machine resumes 
	processing the remaining digits of the current certificate element.

	\item \textbf{Verification of complete element matching.}
	When the machine encounters the delimiter \texttt{\_} in the 
	certificate region, it interprets this as the end of the current 
	certificate element. The state \textsc{BackwardToCheckMatch} initiates 
	a verification to ensure that all digits of this element have been 
	consistently matched. Upon locating the end-of-element marker, the 
	state \textsc{MatchedDigits} verifies the existence of the circled 
	digit at the LSD position; since only one digit is marked at a time, 
	the presence of the LSD marker confirms that the entire element has 
	been successfully matched. If this verification fails, the machine 
	rejects. Otherwise, the match is confirmed at the element level, 
	rendering the association between the certificate element and the 
	input-set element irreversible.

	\item \textbf{Subtraction after confirmed matching.}
	Only after a complete element has been successfully matched does the 
	machine enter the subtraction phase. The states \textsc{SumArea.M}, 
	\textsc{Subtract.M}, and \textsc{Borrow.B} perform digit-wise 
	subtraction of the matched element from the target-sum region. This 
	process proceeds from the least significant digit (LSD), utilizing 
	the positional marker \texttt{\$} and the circled digit for 
	alignment. Crucially, the trailing delimiter of the target-sum region 
	immediately preceding \texttt{@} serves as the LSD marker for the 
	operation. Upon completion, the matched element is erased from the 
	input-set region to prevent reuse, and all temporary markers are 
	restored to their original state.

    \item \textbf{Final sum verification.}
    After all certificate elements have been processed and the
    corresponding subtractions have been applied, the machine scans the
    target sum region to check whether the remaining value is exactly
    zero.
    If only zeros (or their marked equivalents) remain, the machine
    enters an accepting configuration; otherwise, it rejects.
\end{enumerate}

\paragraph{Invariant}
After the completion of the \(j\)-th subtraction phase,
the target-sum region encodes
\[
    S - \sum_{i=1}^{j} a_{\pi(i)},
\]
where each \(a_{\pi(i)}\) is the unique input-set element whose
digit sequence has been fully verified and subsequently subtracted
during the \(i\)-th subtraction phase.

In particular, the value stored in the target-sum region changes only
upon completion of a subtraction phase and remains unchanged during all
other matching and verification phases.

\begin{definition}[Ordinary Subset-Sum Verifier Turing Machine] \label{def:ordinary_sum_of_subset_verifier_tm}

The deterministic subset-sum verifier Turing machine
\[
M_{\mathrm{SS0}}
=
(Q,\Sigma,\Gamma,\delta,\qinit,\qacc,\qrej)
\]
is defined as follows.

\begin{itemize}
\item
$Q$ is a finite set of states consisting of all states explicitly appearing
in the transition table in \cref{tab:ordinary_tm_subsetsum_transitions}, including
parameterized state families such as $\textproc{Backward.M}$,
$\textproc{MatchPosition.M}$, $\textproc{Subtract.M}$, and $\textproc{Borrow.B}$.
Each parameterized family is instantiated over a finite domain
($M,D \in \{0,\dots,9\}$, $B \in \{0,1\}$), and therefore contributes only a constant number of
states.

\item
$\Sigma$ is the input alphabet, consisting of decimal digits
$\{0,\dots,9\}$ and delimiter symbols
$\texttt{\_}$, \texttt{@}, \texttt{\#}, and \texttt{;}.

\item
$\Gamma$ is the tape alphabet, extending $\Sigma$ with auxiliary marker
symbols used during verification, including
$\sim$, $|$, $\$$, \circled{0}$,\dots,$\circled{9},  and blank symbol $\epsilon$.

\item
$\delta$ is the deterministic transition function defined explicitly by the
rules in \cref{tab:ordinary_tm_subsetsum_transitions}.
It implements digit matching, marker-based bookkeeping, and digit-wise
subtraction with borrow propagation.

\item
$\qinit = \textproc{Forward}$ is the initial state, in which the
machine begins by scanning the input from left to right.

\item
$\qacc = \textproc{Accept}$ is the accepting halting state.

\item
$\qrej = \textproc{Reject}$ is the rejecting halting state.
\end{itemize}
\end{definition} 

\begin{table}[!ht]
\centering
\small
\begin{tabular}{llllll}
\hline
\textbf{State} & \textbf{Read} & \textbf{Write} & \textbf{Move} & \textbf{Next State} & \textbf{Comment} \\
\hline

Forward & \# & \# & +1 & FindDigitToMatch & Enter certificate area \\
Forward & * & * & +1 & Forward & Scan input. \\

FindDigitToMatch & $\sim$ & $\sim$ & +1 & FindDigitToMatch & Skip deleted symbols \\
FindDigitToMatch & M & $\sim$ & -1 & BackwardToMatch.M & Select MSD digit to match \\
FindDigitToMatch & \circled{D} & \circled{D} & +1 & FindDigitToMatch & Skip matched digits \\
FindDigitToMatch & \_ & $\sim$ & -1 & BackwardToCheckMatch & End of number; verify full match \\
FindDigitToMatch & ; & ; & -1 & BackwardToCheckSum & No digits left; check sum \\

BackwardToMatch.M & D & D & -1 & BackwardToMatch.M & Skip digits \\
BackwardToMatch.M & | & | & -1 & BackwardToMatch.M & Skip separators \\
BackwardToMatch.M & \_ & | & +1 & MatchPosition.M & Candidate position found \\
BackwardToMatch.M & $\sim$ & $\sim$ & -1 & BackwardToMatch.M & Skip deleted cell \\
BackwardToMatch.M & \# & \# & -1 & BackwardToMatch.M & Continue scan \\

BackwardToMatch.M & \circled{D} & D & +1 & MatchPosition.M & Skip circled digit \\

MatchPosition.M & | & | & -1 & BackwardToMatch.M & No match; try next element \\
MatchPosition.M & $\sim$ & $\sim$ & -1 & BackwardToMatch.M & No match; try next element \\
MatchPosition.M & M & \circled{M} & -1 & BackwardToMatch.M & Mark digit matched \\
MatchPosition.M & D & D & -1 & BackwardToMatch.M & No match; try next element \\

BackwardToMatch.M & @ & @ & -1 & CheckForward & End of set scan \\

CheckForward & \circled{D} & \circled{D} & +1 & Forward & Matching digit confirmed \\
CheckForward & * & * & +1 & CheckForward & Scan right \\
CheckForward & \# & \_ & -1 & Reject & No matching digit found \\

BackwardToCheckMatch & \# & \# & -1 & MatchedDigits & Verify LSD for element matching\\
BackwardToCheckMatch & | & \_ & -1 & MatchedDigits & Verify LSD for element matching \\
BackwardToCheckMatch & \circled{D} & D & -1 & BackwardToCheckMatch & Restore digit \\
BackwardToCheckMatch & * & * & -1 & BackwardToCheckMatch & Scan left \\
BackwardToCheckMatch & @ & \_ & -1 & Reject & Invalid termination \\

MatchedDigits & \circled{M} & \$ & -1 & BackwardToSubtract.M & All digits matched \\
MatchedDigits & D & D & -1 & BackwardToCheckMatch & Continue restore \\
MatchedDigits & $\sim$ & $\sim$ & -1 & BackwardToCheckMatch & Skip deleted \\ 

BackwardToSubtract.M & @ & @ & -1 & SumArea.M & Enter sum area \\
BackwardToSubtract.M & * & * & -1 & BackwardToSubtract.M & Scan left \\

SumArea.M & D & D & -1 & SumArea.M & Traverse sum digits \\
SumArea.M & | & | & -1 & SumArea.M & Skip marker \\
SumArea.M & \_ & | & -1 & Subtract.M & prepare LSD subtraction \\
SumArea.M & \circled{D} & D & -1 & Subtract.M & Prepare subtraction \\

Subtract.M & D & \circled{D}-\circled{M} & -1 & Borrow.B & Perform subtraction. \\

Borrow.0 & * & * & +1 & Forward & No borrow. \\
Borrow.1 & 0 & 9 & -1 & Borrow.1 & Propagate borrow. \\
Borrow.1 & D & $D{-}1$ & +1 & Forward & Borrow resolved. \\
Borrow.1 & $\epsilon$ & \_ & -1 & Reject & Underflow. \\

BackwardToCheckSum & @ & @ & -1 & CheckSum & Begin check sum phase. \\
BackwardToCheckSum & \circled{D} & \circled{D} & -1 & Reject & Incompleted matching \\
BackwardToCheckSum & * & * & -1 & BackwardToCheckSum & Scan left. \\

CheckSum & \_ & \_ & -1 & CheckSum & Skip separators. \\
CheckSum & 0 & 0 & -1 & CheckSum & Continue check. \\
CheckSum & $\epsilon$ & \_ & -1 & Accept & Sum equals zero. \\
CheckSum & * & \_ & -1 & Reject & Nonzero remainder. \\

\hline
\end{tabular}
\caption{Transition table for the Subset-Sum Turing Machine $M_{\mathrm{SS}}$.
Here $D\in\{0,\dots,9\}$, $\text{\circled{D}}\in \text{\{\circled{0}},\dots, \text{\circled{9}}\}$,
and $M$ denotes the currently matched digit stored in the state.} \label{tab:ordinary_tm_subsetsum_transitions} 
\end{table}

The transition table employs schematic symbols to compactly represent families of transitions. 
The meanings and application order of these symbols—\(D\), \circled{D}, \(M\), \circled{D}, \(B\), and \(*\)—are identical to those of the certificate-oblivious subset-sum verifier.

\begin{lemma}[Token Identification and Single-Use Correspondence]
\label{lem:ordinary_token-identification}
During any accepting computation of $M_{\mathrm{SS0}}$,
each confirmed certificate element corresponds to exactly one uniquely
determined delimiter-bounded input-set element.

Let $b_j$ be the $j$-th delimiter-bounded integer encoded in the
certificate region that reaches the confirmed matching phase.
Then there exists a unique input-set element $a_k$ such that:

\begin{enumerate}
    \item Starting from the most significant digit, every digit of $b_j$
    is matched sequentially with the digits of $a_k$.
    \item The match terminates exactly at the delimiter of both tokens,
    establishing equality of the encoded integers $b_j=a_k$ rather than
    prefix agreement.
    \item The subtraction phase is executed over the same digit interval,
    from least significant digit to most significant digit.
    \item After completion, all digits of $a_k$ are erased and can never
    participate in future matches.
\end{enumerate}

Consequently, each certificate element $b_j$ is used exactly once and
corresponds to exactly one input-set element $a_k$, and no input element
can be reused in later phases.
\end{lemma}

\begin{proof}
The digits of a certificate element $b_j$ are processed in the state \textsc{FindDigitToMatch} from the most significant to the least significant digit.
Backward matching ensures that each element exists in the input-set region, provided it is not masked by 'x'.
The matching process verifies that all matched digits belong to a single delimiter-bounded input-set element and that both token boundaries are reached, as matching occurs at the exact relative position for each digit.
Consequently, the equality $b_j = a_k$ is established before subtraction begins.
The subtraction routine operates on each element in reverse order to perform the subtraction, including borrow propagation; subsequently, the subtracted digits are erased.
Thus, each subtraction corresponds to exactly one complete input-set element and one certificate element; if an element is masked by 'x', zero values are subtracted accordingly.
\end{proof}

\begin{lemma}[Soundness of Ordinary Subset-Sum Verifier]\label{lem:soundness-ordinary_subsetsum}
Let $W$ be an input string encoded in the Subset-Sum format
\[
W = S\_@\_a_1\_a_2\_\cdots\_a_k\_\# b_1\_b_2\_\cdots\_b_m\_ ;,
\]
where $S, a_i, b_j$ are nonnegative integers represented in decimal.
If $M_{\mathrm{SS0}}$ enters the accepting configuration on input $W$,
then
\[
\{b_1,\dots,b_m\} \subseteq \{a_1,\dots,a_k\}
\quad\text{and}\quad
\sum_{j=1}^m b_j = S.
\]
\end{lemma}

\begin{proof}
Assume that $M_{\mathrm{SS0}}$ enters the accepting configuration
on input $W$. We show that this implies
\[
\{b_1,\dots,b_m\} \subseteq \{a_1,\dots,a_k\}
\quad\text{and}\quad
\sum_{j=1}^m b_j = S.
\]

\begin{itemize}
\item[(1)] \textbf{Acceptance implies zero target sum.}
The machine enters the accepting configuration only from the
\textsc{CheckSum} state.
By the transition rules in \textsc{CheckSum},
acceptance is possible only if the target-sum region $T$ entirely contains  zero(`0') digits.
Hence, at acceptance, the tape encodes
\[
T = 0.
\]

\item[(2)] \textbf{Zero target sum implies complete and correct subtraction.}
The target-sum region is initially set to $S$. The subtraction phase is triggered by the \textsc{BackwardToSubtract} state. 
Once the machine enters this state at the LSD of a matched element, each digit of that element is subtracted iteratively from LSD to MSD. 
This process is implemented by the states \textsc{Subtract.M} and \textsc{Borrow.B}, which return control to \textsc{BackwardToSubtract} until the operation is complete. 
Any unresolved borrow, digit underflow, or malformed subtraction during this phase forces an immediate transition to the rejecting state. 

Let $c_j$ denote the numerical value subtracted during the $j$-th subtraction phase. 
Reaching the \textsc{CheckSum} state with the target-sum region being exactly zero implies that every subtraction phase was executed successfully without rejection. 
Therefore, the final state of the tape satisfies:
\[
T = S - \sum_{j=1}^m c_j = 0, \quad \text{which implies} \quad S = \sum_{j=1}^m c_j.
\]

\smallskip
\item[(3)] \textbf{Successful subtraction implies matching of an input-set element.}
A subtraction phase is initiated only after the machine confirms that a 
digit-matching process corresponds to a complete and valid input-set 
element. Specifically, subtraction is triggered only after the backward 
confirmation phase (via \textsc{BackwardToCheckMatch} and 
\textsc{MatchedDigits}) verifies that the LSD of a candidate element is 
properly marked, indicating that the element has been matched. 

If this verification fails at any digit or at the element level, the 
machine never enters the subtraction states. Consequently, every value 
$c_j$ subtracted from the target-sum region must be identical to some 
input-set element $a_{i_j}$. Thus, for each $j \in \{1, \dots, m\}$, 
there exists an index $i_j \in \{1, \dots, k\}$ such that $c_j = a_{i_j}$. 
Furthermore, the input-set digits are erased immediately following their 
subtraction to prevent reuse, ensuring that each input-set element is 
utilized at most once.

\item[(4)] \textbf{Element matching requires complete digit-wise correspondence.}
For each certificate element, the machine processes its digits from the most significant digit (MSD) to the least significant digit (LSD) in the state \textsc{FindDigitToMatch}. 
Each selected digit is tentatively marked within some input-set element, while the previously matched digit is restored to its original value by the backward matching procedure. 
Only after all digits of the delimiter-bounded token have been processed does the machine transition to the states \textsc{BackwardToCheckMatch} and \textsc{MatchedDigits}. 

After this phase, the subtraction phase is triggered only if every digit of the certificate element has been successfully matched, leaving the LSD properly marked. 
A partially matched element cannot possess a marked LSD due to the restoration logic applied at each digit-matching step. 
This ensures that the $j$-th successfully matched element corresponds exactly to the $j$-th certificate element $b_j$; otherwise, the machine rejects. 
Furthermore, this mechanism ensures that the $j$-th matched element always triggers the $j$-th subtraction phase, 
implying that each $b_j$ is associated with a unique input-set element $a_{i_j}$ and ensuring $\{b_1, \dots, b_m\} \subseteq \{a_1, \dots, a_k\}$. 

Since the value subtracted from the target-sum region is $c_j = a_{i_j}$, and the matching process guarantees $a_{i_j} = b_j$, it follows that $c_j = b_j$. 
Thus, each subtraction step corresponds exactly to one delimiter-bounded input-set element as per \cref{lem:ordinary_token-identification}. 
Consequently, we conclude that $S = \sum_{j=1}^m b_j$ and $\{b_1, \dots, b_m\} \subseteq \{a_1, \dots, a_k\}$.

\end{itemize}
Combining (1)–(4), acceptance implies that the certificate encodes a subset of
the input set whose elements sum exactly to $S$, completing the proof.
\end{proof}

\begin{lemma}[Completeness of Ordinary Subset-Sum Verifier]\label{lem:completeness-ordinary_subsetsum}
Let $W$ be an input string encoded in the Subset-Sum format:
\[
W = S\_@\_a_1\_a_2\_\cdots\_a_k\_\# b_1\_b_2\_\cdots\_b_m\_ ;,
\]
where $S, a_i, b_j$ are nonnegative integers represented in decimal. 
If $\{b_1, \dots, b_m\} \subseteq \{a_1, \dots, a_k\}$ and 
\[
\sum_{j=1}^m b_j = S,
\]
then $M_{\mathrm{SS0}}$ enters the accepting configuration on input $W$.
\end{lemma}

\begin{proof}

The machine processes the certificate element by element, and within each
element matches digits in the order determined by the tape layout, starting
from the most significant digit.

\begin{itemize}
\item[(1)] \textbf{Successful digit matching for each certificate digit.}
In state \textsc{Forward}, the machine reaches the beginning of the next
unprocessed certificate element and enters \textsc{FindDigitToMatch}.
For each digit of the current certificate element, a corresponding digit
exists in the input-set region by validity of the certificate.
State \textsc{BackwardToMatch.M} deterministically locates the matched digit and circles this digit with restoring the previously matched digit.
After each successful match, the machine returns to \textsc{CheckForward} to
 validate the match by verifying the presence of the circled digit.
Thus, all digits of each certificate element are matched without triggering
any rejecting transition, and LSD digit remains circled for each element.

\item[(2)] \textbf{Verification of a complete element.}
When \textsc{FindDigitToMatch} encounters the trailing delimiter
of a certificate element, the machine enters
\textsc{BackwardToCheckMatch}.
This phase confirms that the matched digits form a complete input-set element by checking the existence of the circled LSD digit at the \textsc{MatchedDigits}
triggered by a valid trailing delimiter, and restores all temporary markers.
Thus, $j$-th matched digit is the exactly $j$-th certificate element and identical to some element of $\{a_1, \dots, a_k\}$. 
\item[(3)] \textbf{Element-wise subtraction corresponding to the certificate.}
For each successfully confirmed element, the machine enters the subtraction
states \textsc{BackwardToSubtract}.
Using the states \textsc{SumArea.$M$}, \textsc{Subtract.$M$}, and \textsc{Borrow.$B$},
the subtraction is performed from LSD to MSD, and the subtracted digit is erased. 
The digits marked during the matching phase encode exactly the value of the
current certificate element $b_j$, and the subtraction performed corresponds
precisely to subtracting $b_j$ from the target-sum region.
Because
\[
\sum_{j=1}^m b_j = S,
\]
all borrow propagations terminate successfully and no underflow occurs.
After processing all certificate elements, the target-sum region contains only
zeros.

\item[(4)] \textbf{Final zero check and acceptance.}
Upon reaching the end of the certificate, the machine enters
\textsc{CheckSum}.
Since the target-sum region encodes zero, the machine reaches the accepting
configuration.
\end{itemize}

Hence, for every valid certificate, $M_{\mathrm{SS0}}$ never enters a
rejecting state and eventually accepts, establishing completeness.
\end{proof}

We analyze the time and space complexity of $M_{\mathrm{SS0}}$ as a function of the Subset-Sum problem instance length, 
rather than the total input length of the verifier TM. Let $n$ denote the length of the initial tape contents—including the target sum, 
the input-set encoding, and delimiters—excluding the certificate. We then assume that the certificate length is at most $n$

\begin{lemma}[Time and Space Complexity of Ordinary Subset-Sum Verifier]
\label{lem:time-space-ordinary_subsetsum}
Let $n$ denote the length of the initial tape contents (including the target sum, the input set, and delimiters). 
Assuming the certificate length is at most $n$, the machine $M_{\mathrm{SS0}}$ halts in $ \bigO(n^2)$ time and uses $ \bigO(n)$ tape space.
\end{lemma}

\begin{proof}
We bound the time and space complexity separately.

\smallskip
\noindent\textbf{Time complexity.} 
The machine operates through a sequence of deterministic phases, each implemented using linear tape scans, marker-based position recovery, and irreversible digit erasure. While individual scans are linear ($ \bigO(n)$), certain phases are repeated proportional to the input length, resulting in quadratic total time.

\begin{enumerate}
\item \textbf{Forward scan.}
After each backward matching or subtraction operation, the machine performs a left-to-right scan to locate the next unprocessed digit or delimiter. Each scan requires $ \bigO(n)$ steps. Since this is performed for each of the at most $n$ digits in the certificate, the cumulative time is $ \bigO(n^2)$.

\item \textbf{Digit-to-match search.}
In the state \textsc{FindDigitToMatch}, the machine scans the certificate region to locate the most-significant unmatched digit. A single scan takes $ \bigO(n)$ time, and since at most $n$ digits are selected over the entire computation, the total time for this phase is $ \bigO(n^2)$.

\item \textbf{Backward matching and element verification.}
For each selected certificate digit, the machine scans leftward to locate a corresponding digit in the input-set region. Additionally, upon encountering a certificate delimiter, a backward scan verifies the complete element match (via LSD marker check). Each such scan takes $ \bigO(n)$, and at most $n$ such operations are performed, totaling $ \bigO(n^2)$.

\item \textbf{Subtraction phase.}
Digit-wise subtraction with borrow propagation involves repeated scans between the matched element and the target-sum region. Each subtraction step (per digit) requires $ \bigO(n)$ time. Since there are at most $n$ digits in the certificate to be subtracted, the total time spent in subtraction is $ \bigO(n^2)$.

\item \textbf{Final sum check.}
A final linear scan of the target-sum region verifies that all digits are zero. This requires a single $ \bigO(n)$ pass.
\end{enumerate}

Summing the contributions from all phases, the total running time is:
\[
 \bigO(n^2) +  \bigO(n^2) +  \bigO(n^2) +  \bigO(n^2) +  \bigO(n) =  \bigO(n^2).
\]

\smallskip
\noindent\textbf{Space complexity.}
The machine never extends the tape beyond the cells initially occupied by the input and the provided certificate, which terminates with a `;' marker. 
All auxiliary information is managed by overwriting existing symbols. The machine prevents the head from moving beyond the tape's boundaries by treating the blank symbols (on the left) and the `;' marker (on the right) as impenetrable barriers; 
the head reverses its direction upon encountering these boundaries. Thus, the total number of tape cells accessed remains $\bigO(n)$.
\smallskip
Therefore, $M_{\mathrm{SS0}}$ operates in quadratic time and linear space relative to the input length.
\end{proof}

The Turing machine $M_{\mathrm{SS0}}$, as formally defined in \cref{def:ordinary_sum_of_subset_verifier_tm}, 
serves as a rigorous deterministic polynomial-time verifier for the \textsc{Subset-Sum} problem. 
Given an input string $W$ of length $n$ and a corresponding certificate of length at most $n$, the machine operates under the criteria of both soundness and completeness. 
Specifically, $M_{\mathrm{SS0}}$ enters an accepting configuration if and only if the provided certificate corresponds to a subset of the input multiset that satisfies the target summation constraint, thereby ensuring that no invalid certificate is accepted. 
Furthermore, the machine is computationally efficient, requiring $\bigO(n^2)$ time and $\bigO(n)$ space. 
Collectively, these properties establish that the \textsc{Subset-Sum} problem admits a deterministic, quadratic-time, and linear-space verification procedure, satisfying the requirements for its classification within the complexity class NP.

\section{Implementation Detail of Two SAT TMs}\label{sec:impl_sat_tms}
The algorithms in this appendix detail the implementation of the transition function for the Turing machine in Section \cref{sec:TMs},
 showing how concrete next states, output symbols, and head movements are derived deterministically from the symbolic specification.
The table matches exactly the implementation given in the reference Python code.
\begin{algorithm}[H]
\caption{Symbolic Transition Resolution for Input-dependent SAT Turing Machine}
\label{alg:delta_sat_id}
\begin{algorithmic}[1]
\Function{$\delta$}{$state, symbol$}

    \State Let $action \gets state$, $addr \gets \emptyset$, $altstate \gets \emptyset$

    \Comment{Parse parameterized state}
    \If{$state$ contains the character `.'}
        \State Parse $state$ into $(action, addr)$
        \State Set $altstate \gets action || \texttt{`.N'}$
    \EndIf

    \Comment{Build prioritized symbol classes}
    \State Let $Symbols \gets [symbol]$
    \If{$symbol$ is a digit}
        \State Append \texttt{`D'} to $Symbols$
    \EndIf
    \State Append \texttt{`*'} to $Symbols$

    \Comment{Resolve transition by priority}
    \ForAll{$s \in Symbols$}
        \If{$(state, s)$ is defined in \textsc{Transitions}}
            \State Let $(q', out, dir) \gets \textsc{Transitions}[(state, s)]$
        \ElsIf{$altstate \neq \emptyset$ \textbf{and} $(altstate, s)$ is defined}
            \State Let $(q', out, dir) \gets \textsc{Transitions}[(altstate, s)]$
        \Else
            \State{\Continue}
        \EndIf

        \Comment{Re-instantiate parameterized states}
        \If{$q'$ ends with \texttt{`.D'} \textbf{and} $symbol$ is a digit}
            \State Replace \texttt{`.D'} in $q'$ with \texttt{`.'} || $symbol$
        \EndIf

        \If{$q'$ ends with \texttt{`.N'}}
            \State Replace \texttt{`.N'} in $q'$ with \texttt{`.'} || $addr$
        \ElsIf{$q'$ ends with \texttt{`.(N-1)'}}
            \State Replace \texttt{`.(N-1)'} with \texttt{`.'} || $(addr-1)$
        \ElsIf{$q'$ ends with \texttt{`.(10N+D)'}}
            \State Replace \texttt{`.(10N+D)'} with \texttt{`.'} || $(10\cdot addr + symbol)$
        \EndIf
        
        \If{$q'$ not belongs to the machine's states $Q$}[Out-of-range parameterized state]
        	\State \Return $(\textsc{Reject}, \texttt{`\_'}, +1)$
	\EndIf

        \Comment{Resolve output symbol}
        \If{$out = \texttt{`*'}$}
            \State Set $out \gets symbol$
        \EndIf

        \State \Return $(q', out, dir)$
    \EndFor

    \Comment{No transition applicable}
    \State \Return $(\textsc{Reject}, \texttt{`\_'}, -1)$

\EndFunction
\end{algorithmic}
\end{algorithm}

\begin{algorithm}[H]
\caption{Symbolic Transition Function $\delta$ for SAT Turing Machine with Fixed States}
\label{alg:delta_sat_fixed}
\begin{algorithmic}[1]
\Function{$\delta$}{$state$, $symbol$}

\State Let $action \gets state$, $sub \gets \bot$, $altstate \gets \bot$

\If{$state$ contains character \texttt{`.'}}
    \State Parse $state$ into $(action, sub)$
    \If{$sub$ is a digit}
        \State Set $altstate \gets action || \texttt{`.D'}$
    \ElsIf{$sub \in \{ \texttt{`T'}, \texttt{`F'} \}$}
        \State Set $altstate \gets action || \texttt{`.B'}$
    \Else
        \State Set $altstate \gets action || \texttt{`.S'}$
    \EndIf
\EndIf

\State Initialize list $\mathcal{S} \gets [symbol]$
\If{$symbol$ is a digit}
    \State append \texttt{`D'} to $\mathcal{S}$
\ElsIf{$symbol \in \{ \texttt{`T'}, \texttt{`F'} \}$}
    \State append \texttt{`B'} to $\mathcal{S}$
\EndIf
\State append \texttt{`*'} to $\mathcal{S}$

\ForAll{$s \in \mathcal{S}$}
    \If{$(state, s)$ is defined in $\textsc{Transitions}$}
        \State Let $(q', out, dir) \gets \textsc{Transitions}[(state, s)]$
    \ElsIf{$altstate \neq \bot$ and $(altstate, s)$ is defined in $\textsc{Transitions}$}
        \State Let $(q', out, dir) \gets \textsc{Transitions}[(altstate, s)]$
        \If{$q'$ ends with \texttt{`.S'}}
            \State{ Replace \texttt{`.S'} in $q'$ by \texttt{`.'} || $sub$}
        \EndIf
        \If{$q'$ ends with \texttt{`.B'} and $out = \texttt{`B'}$ and $sub \in \{ \texttt{`T'}, \texttt{`F'} \}$}
            \State{Set $out  \gets sub$}
        \EndIf
    \Else
        \State{\Continue}
    \EndIf

    \If{$q'$ ends with \texttt{`.B'} and $symbol \in \{ \texttt{`T'}, \texttt{`F'} \}$}
        \State{ Replace \texttt{`.B'} in $q'$ by \texttt{`.'} || $symbol$}
    \EndIf

    \If{$out = \texttt{`D'}$ and $symbol$ is a digit}
        \State Set $out \gets symbol$
    \ElsIf{$out = \texttt{`D-1'}$ and $symbol$ is a digit}
        \State Set $out \gets symbol - 1$
    \ElsIf{$out = \texttt{`*'}$}
        \State Set $out \gets symbol$
    \EndIf

    \State \Return $(q', out, dir)$
\EndFor

\State \Return $(\textsc{Reject}, \texttt{`\_'}, -1)$
\EndFunction
\end{algorithmic}
\end{algorithm}

Both SAT Turing machines are implemented using a uniform symbolic transition resolution scheme.
Rather than explicitly enumerating the full transition table, transitions are computed at runtime by instantiating parameterized states according to a finite set of symbolic rules.

\section{Functions for Extension Candidate Edges}
\begin{algorithm}[H]
\caption{Get outgoing edges of $u$ whose index-precedent edges in $E_p$}\label{alg:next_edges_above_index_precedent}
\begin{algorithmic}[1]
\Description{Given a vertex $u$ and a set of precedent edges $E_p$, returns the set of next edges above the index-precedents.}
\Function{GetNextEdgesAboveIPreds}{$G, u, E_p$}
    \State Let $E_n \gets \emptyset$  \Comment{Set of next edges to return}
    
    \ForAll{$e \in E_p$}
        \If{$e = \texttt{NIL}$}
            \State Add all the floor outgoing edges of $u$ in $G$ to $E_n$
            \State \Continue
        \EndIf 
        \State Let $(v, w) \gets e$
        \State Let $i \gets \indexOf(v)$, $t \gets \tier(v) + 1$, $q \gets \nextState(u)$, $s \gets \output(v)$
        \State Let $z \gets$ the vertex in $V[i][t][q][s]$ of $G$ whose precedent is transition case of $v$
        \State Set $E_n \gets E_n \cup \{(u, z)\}$
    \EndFor
    
    \State \Return $E_n$
\EndFunction
\end{algorithmic}
\end{algorithm}

\begin{sublemma}[Running Time of \textproc{GetNextEdgesAboveIPreds}] \label{sublem:runningtime_get_next_edges_above_preds}
Let $G$ be a computation graph with height at most $h$.
Then \textproc{GetNextEdgesAboveIPreds()} in \cref{alg:next_edges_above_index_precedent} runs in time 
\[ O(|E_p| \log |E_p|). \]
If $E_p$ is a set of edges within a single edge slice, it runs in
\[ O(h^2 \log h). \]
\end{sublemma}

\begin{proof}
The algorithm iterates over each edge $e \in E_p$ exactly once. For each iteration, it performs a constant number of operations: accessing entries in the vertex table $V$, identifying the target vertex $z$, and constructing the next edge $(u, z)$. 
Each such operation, including the vertex table lookup, is performed in $\bigO(1)$ time.

Since the resulting edges are stored in an ordered set to maintain uniqueness and a consistent structure, each insertion incurs a cost of $\bigO(\log |E_p|)$. 
The total running time is proportional to $\bigO(|E_p| \log |E_p|)$. If $E_p$ belongs to a single edge slice, the number of edges is bounded by $h^2$, which simplifies the complexity to $\bigO(h^2 \log h)$.
\end{proof}

\begin{algorithm}[ht]
\caption{GetForwardWeakCeilingAdjacentEdges}
\label{alg:forward_weakly_ceiling_adjacent_edges}
\begin{algorithmic}[1]
\Description{Returns the set of forward ceiling-adjacent edges of $e_0$ in computation graph $G$.}
\Function{GetForwardWeakCeilingAdjacentEdges}{$G, e_0$}
    \State{Let $C \gets \emptyset$} \Comment{Collected forward weakly-ceiling-adjacent edges}
    \State{Let $(u_0,v_0) \gets e_0$}
    \State{Let $V_v \gets \emptyset$} \Comment{Visited vertex set}

    \If{$v_0$ is an accepting or rejecting node}
        \State \Return $C$
    \EndIf

    \State{Let $Q$ be a queue initialized with $v_0$}

    \While{$Q$ is not empty}
        \State{Dequeue $u$ from $Q$}
        \If{$u \in V_v$}
            \State \Continue
        \EndIf
        \State{Add $u$ to $V_v$}

        \If{$u$ is a folding node \textbf{or} $u=v_0$}
            \State Let $P \gets \IPrec_G(u)$
            \If{$P=\emptyset$}
                \State Add $\bot$ to $C$ \Comment{Boundary marker}
            \EndIf
            \ForAll{$v \in P$ such that $v \notin V_v$}
                \State Enqueue $v$ to $Q$
            \EndFor
        \Else
            \State Add all incoming edges of $u$ to $C$
        \EndIf
    \EndWhile

    \State \Return $C$
\EndFunction
\end{algorithmic}
\end{algorithm}

\begin{sublemma}[Time Complexity of Forward Weak Ceiling-Adjacent Edge Computation]
\label{sublem:time_complexity_forward_weakly_ceiling_adjacent_edges}

Let $G$ be a computation graph of height $h$ and width $w$, and let $e_0 \in E(G)$ be a given edge. 
Then \textproc{GetForwardWeakCeilingAdjacentEdges()} in \cref{alg:forward_weakly_ceiling_adjacent_edges} terminates in time 
\[
O(h^2 \log h).
\]
\end{sublemma}

\begin{proof}
If the terminal vertex $v_0$ is an accepting or rejecting node, the algorithm returns immediately in constant time. 

The algorithm performs a backward traversal along index-precedent edges similar to \textproc{GetWeakCeilingAdjacentEdges()}, with the primary difference being the omission of the final edge inclusion check. The only additional operation is the insertion of the boundary marker $\bot$ when a vertex has no precedents; this takes standard set insertion time and does not affect the asymptotic complexity (refer to the appendix of the original paper for further details).

Since the computation graph has height $h$, at most $\bigO(h)$ vertices are reachable through index-precedent relations. Each vertex is inserted into and removed from the queue at most once; thus, the while-loop iterates $\bigO(h)$ times.

For each visited vertex $u$:
\begin{itemize}
    \item If $u$ is a folding node (or $u=v_0$), the algorithm enumerates $\IPrec_G(u)$, the size of which is bounded by a constant determined by the machine. Membership tests in $V_v$ cost $\bigO(\log h)$ time using an ordered structure.
    \item Otherwise, all incoming edges of $u$ are inserted into $C$. Each vertex has at most $\bigO(h)$ incoming edges. Since an edge slice contains at most $\bigO(h^2)$ edges, each insertion into the ordered set $C$ costs $\bigO(\log h^2) = O(\log h)$ time. Consequently, the cost for a single vertex in this case is $\bigO(h \log h)$.
\end{itemize}

Aggregating these costs over $\bigO(h)$ iterations, the total traversal and insertion cost is
\[
\bigO(h \cdot h \log h) = O(h^2 \log h).
\]
\end{proof}

\begin{algorithm}[ht]
\caption{Filtering the edges with a path to $e_f$ in the backward direction} \label{alg:filter_path_backward} 
\begin{algorithmic}[1]
\Function{FilterWithPathBackward}{$G, e_f, E_s$}
\State Let $Q \gets$ empty deque
\State Let $E_c \gets \emptyset$
\State push $e_f$ into $Q$

\If{$E_s = \varnothing$}
    \State \Return $E_c$
\EndIf
\State Let $(u0, v0)=e_f$
\State Let $i_0 \gets \min(\indexOf(v0), \nextIndex(v0))$
\State Let $E_v \gets \emptyset$

\While{$Q$ not empty}
    \State Let $e \gets$ pop from $Q$

    \If{$e = \texttt{NIL}$}
        \State Set $E_c \gets E_c \cup \{e\}$
    \Else
    	\State $(u,v) \gets e$
    	\If{$e \in E_v$}
        	\State \Continue
    	\EndIf
    	\State $E_v \gets E_v \cup \{e\}$
    	\If{$e \in E_s$}
        	\State Set $E_c \gets E_c \cup \{e\}$
    	\EndIf

    	\If{$\indexOf(e) = i_0$}
        	\State \Continue
    	\EndIf
    	\State push all edges of $\Prev_G(e)$ into $Q$
    \EndIf
\EndWhile
\State \Return $E_c$
\EndFunction
\end{algorithmic}
\end{algorithm}

\begin{sublemma}[Running time of \textproc{FilterWithPathBackward}]\label{sublem:runningtime_filter_backward}
Let $G$ be a computation graph of width $w$ and height $h$.
Assume the visited edge set $E_v$ is implemented as an ordered set.
Then \textproc{FilterWithPathBackward()} in \cref{alg:filter_path_backward} runs in time
\[
\bigO(wh^3 \log(wh)).
\]

\begin{proof}
The algorithm performs a backward traversal over edges starting from $e_f$. 
An edge can be enqueued up to its out-degree,  which is $\bigO(h)$. Thus, the total number of enqueued edges is at most $\bigO(|E(G)|h)$. The visited set $E_v$ ensures that each edge is processed only once.

For each edge $e$ popped from the queue, the algorithm performs:
\begin{itemize}
    \item A membership test $e \in E_v$, costing $\bigO(\log |E(G)|)$ time.
    \item If $e \notin E_v$ (which occurs at most $|E(G)|$ times):
    \begin{itemize}
        \item An insertion into $E_v$ and a membership test $e \in E_s$, each costing $\bigO(\log |E(G)|)$ time.
        \item Enumeration of $\textproc{Prev}_G(e)$, where the total number of predecessors generated over the entire execution is $\bigO(h \cdot |E(G)|)$.
    \end{itemize}
\end{itemize}

The total running time is dominated by the $\bigO(|E(G)| h)$ membership tests. With $|E(G)| = \bigO(wh^2)$, the total complexity is
\[
\bigO(wh^3 \log(wh)).
\]
\end{proof}
\end{sublemma}

\begin{algorithm}[ht]
\caption{Filtering the edges with a path from $e_s$ in the forward direction} \label{alg:filter_path_forward}
\begin{algorithmic}[1]
\Function{FilterWithPathForward}{$G, e_s, E_f$}
\State Let $Q \gets$ empty deque
\State push $e_s$ into $Q$

\State Let $i_0 \gets \indexOf(e_s)$
\State Let $E_v \gets \emptyset$
\State Let $E_c \gets \emptyset$

\While{$Q$ not empty}
    \State Let $e \gets$ pop from $Q$
    \State Let $(u,v) \gets e$

    \If{$e \in E_v$}
        \State \Continue
    \EndIf
    \State Set $E_v \gets E_v \cup \{e\}$

    \If{$e \in E_f$}
        \State Set $E_c \gets E_c \cup \{e\}$
    \EndIf

    \If{$e \neq e_s$ \textbf{and} $\indexOf(e) = i_0$}
        \State \Continue
    \EndIf

    \State push all edge of $\Next_G(e)$ into $Q$
\EndWhile
\State \Return $E_c$
\EndFunction
\end{algorithmic}
\end{algorithm}

\begin{sublemma}[Running time of \textproc{FilterWithPathForward}]\label{sublem:runningtime_filter_forward}
Let $G$ be a computation graph of width $w$ and height $h$.
Assume the visited edge set $E_v$ is implemented as an ordered set.
Then \textproc{FilterWithPathForward()} in \cref{alg:filter_path_forward} runs in time
\[
\bigO(wh^3 \log(wh)).
\]

\begin{proof}
The algorithm performs a forward traversal over edges starting from $e_s$. 
Note that an edge can be enqueued multiple times up to its in-degree, which is $\bigO(h)$. 
Thus, the total number of enqueued edges is at most $\bigO(|E(G)|h)$. However, due to the visited set $E_v$, each edge is processed (i.e., its successors are explored) at most once.

For each edge $e$ popped from the queue, the algorithm performs:
\begin{itemize}
    \item A membership test $e \in E_v$, costing $\bigO(\log |E(G)|)$ time.
    \item If $e \notin E_v$ (which occurs at most $|E(G)|$ times):
    \begin{itemize}
        \item An insertion into $E_v$ and a membership test $e \in E_f$, each costing $\bigO(\log |E(G)|)$ time.
        \item Enumeration of $\textproc{Next}_G(e)$, where the total number of next edges generated over the entire execution is $\bigO(h \cdot |E(G)|)$.
    \end{itemize}
\end{itemize}

Summing these, the total running time is dominated by the membership tests for all enqueued elements:
\[
\bigO(|E(G)| h \log |E(G)|).
\]
Since $|E(G)| = \bigO(wh^2)$ and $\log |E(G)| = \bigO(\log(wh))$, we obtain
\[
\bigO(wh^3 \log(wh)).
\]
\end{proof}
\end{sublemma}

\begin{algorithm}[ht]
\caption{Get Weakly Ceiling Adjacent Edges with Visited Node Set Input} \label{alg:weakly_ceiling_adjacent_edges}
\begin{algorithmic}[1]
\Function{GetWeakCeilingAdjacentEdges}{$G, e_0, E_f, \InOut:V_v$}  \Comment{Visited vertex set}
    \State{Let $C \gets \emptyset$}\Comment{Collected wealky-ceiling-adjacent edges}
    \State{Let $(u_0,v_0) \gets e_0$}
    \State{Let $Q$ be a queue initialized with $u_0$}
    \If{$e_0 \in E_f$} 
        \State{Enqueue $v_0$ to $Q$}
    \EndIf
    \While{$Q$ is not empty} 
        \State{Dequeue $u$ from $Q$}
        \If{$u \in V_v$} 
            \State{\Continue}
        \EndIf
        \State{Add $u$ to $V_v$}
        \If{$u$ is a folding node \textbf{or} $u=v_0$} 
            \State Let $P \gets \IPrec_G(u)$
            \ForAll{$v \in P$ such that $v \notin V_v$} 
                \State{ Enqueue $v$ to $Q$ }
            \EndFor
        \Else 
            \State{Add all incoming edges of $u$ to $C$}
        \EndIf
    \EndWhile
    \State{\Return $C$}
\EndFunction
\end{algorithmic}
\end{algorithm}

\begin{sublemma}[Time Complexity of Weak Ceiling-Adjacent Edge Computation With Visited Node Set]
\label{sublem:time_complexity_weakly_ceiling_adjacent_edges}
Let $G$ be a computation graph of height $h$ and width $w$.
Then \textproc{GetWeakCeilingAdjacentEdges()} in \cref{alg:weakly_ceiling_adjacent_edges}
terminates in time $\bigO(h^2\log(wh))$ for a given edge $e_0 \in E(G)$.
\end{sublemma}

\begin{proof}
The algorithm performs a backward traversal over nodes $u$ of the computation graph,
starting from the initial endpoint of $e_0$ and proceeding along backward folding paths.

First, note that the membership test $e_0 \in E_f$ at the initialization step
can be implemented in time $\bigO(\log |E_f|)$ using an ordered set.
Since $|E_f| \le |E(G)| = \bigO(w h^2)$, this check costs
$\bigO(\log (wh))$ time.

Since the computation graph has height $h$, it contains at most $\bigO(h)$ nodes.
Each node is enqueued and dequeued at most once due to the visited node set $V_v$.
Therefore, the while-loop iterates at most $\bigO(h)$ times.

For each dequeued node $u$, the algorithm performs one of the following operations:
\begin{itemize}
    \item If $u$ is a folding node or $u = \term(e_0)$, it enumerates
    $\IPrec_G(u)$.
    The number of index-precedent nodes is bounded by a constant determined by the
    tape alphabet and the deterministic transition function, hence
    $|\IPrec_G(u)| = \bigO(1)$.
    Each membership test or insertion into $V_v$ costs $\bigO(\log (wh))$ time due to the existence of $\bigO(wh)$ nodes.
    
    \item Otherwise, all incoming edges of $u$ are added to the set $C$.
    Each node has at most $\bigO(h)$ incoming edges.
    For each such edge, membership testing and insertion into $C$ costs
    $\bigO(\log |C|) = \bigO(\log (h^2))$ time since each edge slice have at most $\bigO(h^2)$ edges.
\end{itemize}

Thus, each iteration of the while-loop costs $\bigO(h \cdot \log(wh))$.
Since there are at most $\bigO(h)$ iterations, the total running time is $\bigO(h^2 \log(wh))$.

The algorithm terminates in time
\[
\bigO(\log(wh)) + \bigO(h^2 \log(wh))
= \bigO\bigl((h^2 \log(wh)) \bigr).
\]

\end{proof}

\end{document}